\documentclass[11pt]{article}
\usepackage[letterpaper, margin=1.2in]{geometry}
\usepackage[utf8]{inputenc}
\usepackage{mathrsfs}
\usepackage{amssymb,amsfonts}
\usepackage{amsmath}
\usepackage{amsthm}
\usepackage{enumerate}
\usepackage{braket}
\usepackage{bbm}
\usepackage{tikz-cd}
\usepackage{tikz}
\usepackage{setspace}
\usepackage{tikz-cd}
\usepackage{authblk}
\usepackage{appendix}
\usepackage{url}

\usepackage[breaklinks]{hyperref}
    \hypersetup{
        colorlinks=true,
        linktoc=all,
        linkcolor=blue,
        breaklinks=true,
        pdftitle={Taming 3d N=4 line operators I: abelian theories and VOA's},
		pdfauthor={Andrew Ballin, Thomas Creutzig, Tudor Dimofte and Wenjun Niu},
    }
\usepackage{dsfont}
\usepackage{graphicx}
\usepackage{enumitem}
\pdfoutput=1

\graphicspath{{./figures}}

\makeatletter
\newcommand*{\defeq}{\mathrel{\rlap{%
	\raisebox{0.3ex}{$\m@th\cdot$}}%
	\raisebox{-0.3ex}{$\m@th\cdot$}}%
	=}
\makeatother

\newcommand{\mc}{\mathcal}
\newcommand{\normord}[1]{:\mathrel{\mspace{2mu}#1\mspace{2mu}}:}

\newcommand{\T}{{\mathsf{T}}}

\numberwithin{equation}{section}

\newtheorem{Def}{Definition}[section]

\newtheorem{Prop}[Def]{Proposition}
\newtheorem{Thm}[Def]{Theorem}
\newtheorem{PThm}[Def]{Physics Theorem}
\newtheorem{PProp}[Def]{Physics Proposition}
\newtheorem{Lem}[Def]{Lemma}
\newtheorem*{Rmk}{Remark}
\newtheorem{Cor}[Def]{Corollary}
\newtheorem*{Exp}{Example}
\newtheorem{innercustomthm}{Theorem}
\newenvironment{customthm}[1]
  {\renewcommand\theinnercustomthm{#1}\innercustomthm}
  {\endinnercustomthm}
\newtheorem{innercustomprop}{Proposition}
\newenvironment{customprop}[1]
  {\renewcommand\theinnercustomprop{#1}\innercustomprop}
  {\endinnercustomprop}

\newcommand{\sm}[1]{\left(\begin{smallmatrix} #1 \end{smallmatrix}\right)}

\newcommand{\cbg}{\mathcal{C}_{\beta\gamma}}
\newcommand{\grho}{{{\mathfrak g}_\rho}}

\newcommand{\be}{\begin{equation}}
\newcommand{\ee}{\end{equation}}
\newcommand{\pd}{{\partial}}
\newcommand{\ol}{\overline}
\newcommand{\wt}{\widetilde}
\newcommand{\mb}{\mathbf}
\newcommand{\ds}{\displaystyle}

\newcommand{\cf}{\emph{cf.}}
\newcommand{\ie}{\emph{i.e.}}

\newcommand{\R}{{\mathbb R}}

\renewcommand{\L}{{\mathbb L}}
\newcommand{\V}{{\mathbb V}}
\newcommand{\W}{{\mathbb W}}
\newcommand{\C}{{\mathbb C}}
\newcommand{\Z}{{\mathbb Z}}
\newcommand{\X}{{\mathbb X}}

\renewcommand{\S}{{\mathbb S}}
\newcommand{\CA}{{\mathcal A}}

\newcommand{\CC}{{\mathcal C}}

\newcommand{\CF}{{\mathcal F}}
\newcommand{\CG}{{\mathcal G}}
\newcommand{\CH}{{\mathcal H}}
\newcommand{\CI}{{\mathcal I}}

\newcommand{\CM}{{\mathcal M}}
\newcommand{\CN}{{\mathcal N}}
\newcommand{\CO}{{\mathcal O}}

\newcommand{\CT}{{\mathcal T}}

\newcommand{\CU}{{\mathcal U}}
\newcommand{\CV}{{\mathcal V}}  
\newcommand{\CW}{{\mathcal W}}  

\newcommand{\CZ}{{\mathcal Z}}
\newcommand{\CL}{{\mathcal L}}
\newcommand{\Fock}{{\mathcal F}}
\newcommand{\bg}{\beta\gamma}
\newcommand{\vgl}{\widehat{\mathfrak{gl}(1|1)}}
\newcommand{\wh}{\widehat}
\newcommand{\id}{{1\!\!1}}
\newcommand{\m}{m}
\newcommand{\mv}{{m}}
\newcommand{\uv}{{\upsilon}}
\newcommand{\LL}{\mathbb{L}}

\begin{document}
\title{3d mirror symmetry of braided tensor categories}
\author[1]{Andrew Ballin}
\author[2]{Thomas Creutzig}
\author[3]{Tudor Dimofte}
\author[1]{Wenjun Niu}

\affil[1]{\small Department of Mathematics and Center for Quantum Mathematics and Physics, UC Davis, One~Shields~Ave., Davis,~CA~95616, United States}
\affil[2]{Department of Mathematical \& Statistical Sciences 632 CAB, University~of~Alberta, Edmonton,~Alberta~T6G~2G1,~Canada}
\affil[3]{School of Mathematics, University~of~Edinburgh, James~Clerk~Maxwell~Building, Peter~Guthrie~Tait~Rd., Edinburgh~EH9~3FD, United Kingdom}

\date{\today}

\maketitle

\begin{abstract}
We study the braided tensor structure of line operators in the topological A and B twists of abelian 3d $\CN=4$ gauge theories, as accessed via boundary vertex operator algebras (VOA's).
We focus exclusively on abelian theories.
We first find a non-perturbative completion of boundary VOA's in the B twist, which start out as certain affine Lie superalebras; and we construct free-field realizations of both A and B-twist VOA's, finding an interesting interplay with the symmetry fractionalization group of bulk theories. We use the free-field realizations to establish an isomorphism between A and B VOA's related by 3d mirror symmetry. Turning to line operators, we extend previous physical classifications of line operators to include new monodromy defects and bound states. We also outline a mechanism by which continuous global symmetries in a physical theory are promoted to higher symmetries in a topological twist --- in our case, these are infinite one-form symmetries, related to boundary spectral flow, which structure the categories of lines and control abelian gauging. Finally, we establish the existence of braided tensor structure on categories of line operators, viewed as non-semisimple categories of modules for boundary VOA's. In the A twist, we obtain the categories by extending modules of symplectic boson VOA's, corresponding to gauging free hypermultiplets; in the B twist, we instead extend Kazhdan-Lusztig categories for affine Lie superalgebras. We prove braided tensor equivalences among the categories of 3d-mirror theories.
All results on VOA's and their module categories are mathematically rigorous; they rely strongly on recently developed techniques to access non-semisimple extensions.
\end{abstract}

\tableofcontents

\section{Introduction}
\label{sec:intro}

Three-dimensional $\CN=4$ gauge theories have enjoyed a rich interplay with geometry and representation theory since the 90's, particularly motivated by the discovery of 3d mirror symmetry \cite{IS,dBHOO}. This interplay has only intensified in recent years, due in part to the advent of new, non-perturbative, physical and mathematical methods to access algebraic structures in 3d $\CN=4$ gauge theories. Just a small and incomplete sampling includes \cite{GW-S,CHZ,BFN,BDG,BDGH,Dedushenko:2017avn,Webster-Coulomb, AganagicOkounkov,BullimoreZhang,DedushenkoNekrasov,Gaiotto-twisted,CG,CCG,AsselGomis,lineops,RW-Coulomb,BFK,CDGG,hilburn2021tate,gammage2022perverse,BBFS-3d,NSWZ}. In turn, many of these new methods can be viewed from the perspective of \emph{holomorphic} or \emph{topological twists} --- in that the non-perturbative structure being computed lies in the cohomology of a holomorphic or topological supercharge.

A basic but important example is the moduli space of vacua of a 3d $\CN=4$ gauge theory. Functions on the Higgs and Coulomb branches in the space of vacua are given by the expectation values of certain half-BPS local operators, elements of the so-called Higgs and Coulomb chiral rings, \emph{cf.} \cite{CHZ,BDG}. These same local operators appear, respectively, in the B and A topological twists of a 3d $\CN=4$ gauge theory. The B twist (sometimes known as the $C$ twist \cite{Gaiotto-twisted}) was introduced by Blau and Thompson \cite{BlauThompson}, and generalized to hyperk\"ahler sigma-models by Rozansky and Witten \cite{RW}. Its 3d mirror, the A twist (also known as the $H$ twist \cite{Gaiotto-twisted}) is a reduction of Witten's Donaldson twist of 4d $\CN=2$ gauge theory \cite{Witten-Donaldson}. Schematically,
\be \text{A-twist local operators}\simeq \C[\CM_{\rm Coul}]\,,\qquad \text{B-twist local operators}\simeq \C[\CM_{\rm Higgs}]\,. \ee
The celebrated Braverman-Finkelberg-Nakajima construction \cite{Nak,BFN}, which opened up the study of Coulomb branches and 3d $\CN=4$ mirror symmetry to much of mathematics, is a careful computation of the algebra of local operators in the A twist, using a topological state-operator correspondence.

In the current paper, we further develop the physics and mathematics of \emph{line operators} and \emph{boundary vertex operator algebras} (VOA's) preserved by the two topological twists 3d $\CN=4$ gauge theories. The two objects are intrinsically related, since line operators correspond to modules for boundary VOA's \cite{CG,CCG}. We will focus exclusively on the simplest class of 3d $\CN=4$ theories: abelian gauge theories, with continuous gauge group acting faithfully on linear hypermultiplet matter. (We call these ``simple'' abelian theories.) This class is particularly nice, because it is closed under 3d $\CN=4$ mirror symmetry. As we shall see, it is also sufficiently nontrivial to make the analysis in this paper interesting.

Our main results include:
\begin{itemize}
    \item A non-perturbative definition of boundary VOA's in the $B$ twist, generalizing their original construction \cite{CG} to include the contribution of boundary monopole operators.
    \item Free-field realizations of boundary VOA's in both A and B twists, which turn out to involve the finite ``symmetry fractionalization group'' of bulk 3d theories in an intrinsic way.
    \item A proof of 3d mirror symmetry for boundary VOA's, in A and B twists of mirror abelian theories.
    \item An extended physical classification of line operators in A and B twists, including previously unstudied monodromy defects, and bound states of Wilson lines and of vortex lines.
    \item A mathematically rigorous construction of abelian module categories for boundary VOA's that have a braided tensor structure and include a basic collection of bulk line operators as simple modules. We construct these using independent methods in the A and B twists, and then prove that 3d mirror symmetry induces an isomorphism of categories. We expect that the full physical categories of bulk line operators are the derived categories of the abelian VOA module categories.
\end{itemize}
The results on free-field realizations, braided tensor categories of modules, and 3d mirror symmetry thereof generalize work of the first and last author \cite{BN22}, who studied the simplest abelian mirror pair: the A twist of a free hypermultiplet and the B twist of SQED.

Along the way, we introduce some new constructions that may be of independent interest:
\begin{itemize}
    \item We establish an efficient computation of the symmetry fractionalization group of abelian gauge theories, and show that its order counts massive vacua. This group thus measures the non-triviality of abelian theories.

    Mathematically, this result amounts to a simple new formula for counting fixed points of hypertoric varieties.
    \item We find new, infinite one-form global symmetry in topological twists of 3d $\CN=4$ theories, which emerges in a natural way from zero-form symmetries that become cohomologically trivial in the twist. Gauging of physical zero-form symmetries is promoted to gauging of this one-form symmetry in topological twists --- which in turn translates to infinite simple-current extensions of boundary VOA's. 
\end{itemize}

We'll elaborate a bit further on our results in the remainder of the Introduction.

While this paper was in preparation, we learned of related work of C. Beem and A. Ferrari on free-field realizations of  boundary VOA's; we've coordinated our preprints to appear together. Results relating a special class of abelian boundary VOA's (for the A twist of SQED with multiple hypermultiplets) to W-algebras also appeared very recently  \cite{Yoshida:2023wyt}.

A few other recent works related to braided tensor categories of line operators associated to 3d $\CN=4$ theories include: \cite{Mikhaylov,GeerYoung,GarnerNiu} on line operators in abelian supergroup Chern-Simons theories, of which the B twists of a free hypermultipet and SQED are examples; \cite{Dedushenko:2018bp,Gang:2021hrd} on semisimple categories of line operators associated to massive 3d $\CN=4$ theories; and  \cite{creutzig2020vertex,FrenkelGaiotto,CDGG} on a few particulary interesting examples of boundary VOA's for nonabelian theories and their module categories, related to geometric Langlands and quantum groups. 
Modular tensor categories and VOA characters associated to 3d $\CN=2$ theories have also been studied, \emph{e.g.} in \cite{3dmodularity,VOAcharacters}.

The general story for nonabelian 3d $\CN=4$ gauge theories is unfortunately very far from complete.

\subsection{Boundary VOA's}
\label{sec:intro-VOA}

The first half of our paper is devoted to developing the structure of boundary VOA's in the A and B twists of abelian 3d $\CN=4$ gauge theories.

These boundary VOA's were first introduced in \cite{CG}. Their existence came as a bit of a surprise. One might naively expect that boundary conditions for a physical 3d $\CN=4$ theory that are compatible with the bulk topological supercharge $Q$ would also become topological after twisting. \cite{CG} found an exception to this, constructing boundary conditions that remained holomorphic --- thus supporting an entire VOA --- after a topological twist of the bulk.

The VOAs that we are dealing with are very non-rational in the sense that they admit infinitely many inequivalent simple modules and they also have reducible, but indecomposable modules, \ie\ they are examples of VOA's of logarithmic type. Presently there are few such logarithmic VOA's whose representation theory is under full control. Some most relevant members of well understood logarithmic theories are the $\beta\gamma$-VOA, also called symplectic bosons, which we denote by $\CV_{\beta\gamma}$ VOA \cite{allen2020bosonic}, the affine VOA of $\mathfrak{gl}(1|1)$ at any non-zero level \cite{creutzig2020tensor}, symplectic fermions \cite{Tsuchiya:2012ru} and the singlet vertex algebras \cite{Creutzig:singlet1, Creutzig:singlet2}. For an introduction to these examples and logarithmic VOA's in general, see \cite{Creutzig:2013hma}.

These VOA's enjoy a zoo of relations, e.g. the symplectic fermions are a BRST-cohomology of the $\CV_{\beta\gamma} \otimes \CV_{bc}$-VOA very similar to \eqref{VOA-A-intro}, while the affine VOA of $\mathfrak{gl}(1|1)$ is a $\C^*$-orbifold of $\CV_{\beta\gamma} \otimes \CV_{bc}$. Conversely this means that $\CV_{\beta\gamma} \otimes \CV_{bc}$ is a simple current extension of the affine VOA of $\mathfrak{gl}(1|1)$. We review some of these relations while setting up some VOA formalism in Section \ref{subsecVOAconvention}.

It is interesting to note that the relation between the affine VOA of $\mathfrak{gl}(1|1)$, the $\beta\gamma$ VOA, and symplectic fermions is the by far simplest example of the Gaiotto-Rapcak triality of ``corner'' VOA's \cite{GaiottoRapcak}, proven in \cite{Creutzig:2020zaj, Creutzig:2021dda}.
Relating corner VOA's involves taking cosets/orbifolds and then VOA extensions; these two operations can be performed in a single step via a BRST cohomology \cite{Creutzig:2022yvr} involving the kernel VOA's of \cite{creutzig2020vertex,FrenkelGaiotto}. 
The kernel VOA of $\mathfrak{g}= \mathfrak{gl}(1)$ is nothing but the lattice VOA of the integer lattice, \emph{i.e.} of a pair of free fermions $\CV_{bc}$, times a Heisenberg VOA. 

Our VOA's $\mc{V}_{\rho}^A$ and $\mc{V}_{\rho}^B$ are constructed very much in this spirit; however it turns out that we are able to be much more explicit. The reason is that all VOA's involved are extensions of several singlet algebras together with some Heisenberg VOA's. Moreover, the singlet algebra itself is a subalgebra of a Heisenberg VOA characterized as the kernel of a certain screening charge. This means that all our VOA's $\mc{V}_{\rho}^A$ and $\mc{V}_{\rho}^B$ allow for free field realizations inside large lattice VOA's times Heisenberg VOA's. These free field realizations are ultimately of great aid in allowing us to explicitly analyze modules.

For a 3d $\CN=4$ theory with gauge group $G$ and hypermultiplet matter representation $T^*V$, the resulting boundary VOA in the A twist is of the form
\be \CV_{G,V}^A = H^\bullet_{\rm BRST}\big(\mathfrak{g}_\C,\CV_{\beta\gamma}^{T^*V}\otimes \CV_{bc}^{W}\big)\,, \label{VOA-A-intro} \ee
a BRST reduction (\emph{a.k.a.} semi-infinite cohomology) of a number of copies of beta-gamma VOA's (corresponding to bulk hypermultiplets) and some additional free-fermion ($bc$) VOA's to soak up a boundary gauge anomaly. Namely, one must choose a  representation $W$ so that quadratic Casimirs satisfy $T_V-T_{\mathfrak g}=T_W$.
When the boundary gauge anomaly vanishes on its own (meaning $T_V=T_{\mathfrak g}$) the 3d $\CN=4$ theory turns out to be the dimensional reduction of a superconformal 4d $\CN=2$ theory, and the boundary VOA \eqref{VOA-A-intro} coincides with the 4d $\CN=2$ VOA of \cite{Beem:2013sza}. For our abelian gauge theories, the boundary anomaly never vanishes on its own, but there is a canonical way to cancel it: by choosing $W=V$.

The VOA \eqref{VOA-A-intro} is defined for any abelian theory, but difficult to analyze. In Section \ref{secNeumannA}, we introduce free field realization of \eqref{VOA-A-intro} that lets us explicitly take the BRST cohomology. One immediate consequence is that the BRST cohomology is supported in degree zero.

To further describe the structure we find, we need bit more notation for abelian theories, which we develop in Section \ref{sec:setup} and then use throughout the paper. A theory with $G=U(1)^r$ acting faithfully on $T^*V=T^*\C^n$ is characterized by an $n\times r$ charge matrix $\rho$. We denote the theory $\CT_\rho$,
\be \CT_{\rho}:\qquad G=U(1)^r\,,\quad V = \C^n\,,\quad\text{charge matrix\;$\rho$} \ee
The charge matrix fits into an exact sequence
\be \raisebox{-.2in}{\includegraphics[width=2.3in]{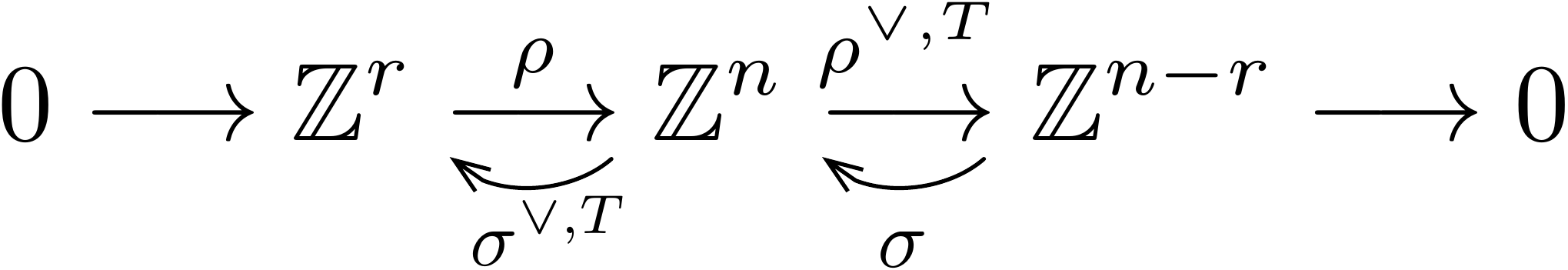}} \label{split-SES2-intro} \ee
A splitting $\sigma$ of this sequence determines charges of $T_F=U(1)^{n-r}$ flavor symmetry on hypermultiplets, and a co-splitting $\sigma^\vee$  determines charges of $T_{\rm top}=U(1)^r$ topological flavor symmetry on monopole operators. 3d mirror symmetry acts by transposing the entire exact sequence \eqref{split-SES2-intro} \cite{dBHOOY},
\be \CT_\rho \;\overset{\text{3d mirror}}\leftrightarrow\; \CT_{\rho^\vee}\,. \ee
The ``symmetry fractionalization group'' of $\CT_\rho$ is the finite group
\be H = \Z^r/\rho^T\rho(\Z^r)\,. \ee
We prove in Section \ref{sec:H} (Prop. \ref{prop:H} and Physics Thm. \ref{Thm:H}) that $|H|$ counts massive vacua.

Coming back to the A-twist VOA for abelian theories, we will establish in Section \ref{sec:VA-BRST} a decomposition schematically of the form
\begin{customprop}{\ref{prop:A}}
There is an isomorphism
\begin{align} \label{propA-intro}
\mc{V}_{\rho}^A &= H_{\rm BRST}\left(\mathfrak{gl}(1)^r, \mc{V}_{\beta\gamma}^{T^*V} \otimes \mc{V}_{bc}^{V}\right) 
\,\cong\, \bigoplus\limits_{h\in H} \CM_{\rho,h}\otimes \CV_{\rho,h} 
\end{align}
\end{customprop}
\noindent where the RHS is an $H$ simple current extension of a product of two VOA's: $\CM_{\rho,0}$ (a particular extension of a product of singlet VOA's) and $\CV_{\rho,0}$ (a lattice VOA with bilinear form $\rho^T\rho$). 

We note that free-field realizations of very closely related 4d $\CN=2$ VOA's have been constructed in \cite{Beem:2019tfp,Kiyoshige:2020uqz,Beem:2021jnm,Beem:2022vfz}; they are of a different character, based intrinsically on the geometry of the Higgs branch, and it would be interesting to connect them to expressions such as \eqref{propA-intro}. Another interesting direction that we won't investigate here is to determine the effect of resolution/deformation parameters in bulk theories boundary VOA's, as well as their categories of modules.
Upon turning on FI parameters to resolve the Higgs branch of a 3d $\CN=4$ theory, the A-twist boundary VOA is expected to become a sheaf of beta-gamma systems on the Higgs branch. The latter has been constructed, specifically in the case of abelian theories, by Kuwabara \cite{Kuwabara}.

The boundary VOA of a 3d $\CN=4$ gauge theory in the B twist has a different sort of complication. In \cite{CG} (see also \cite{garner2022vertex}) the \emph{perturbative} boundary VOA was shown to be equivalent to an affine VOA $\widehat{\mathfrak g_{G,V}}$ associated to a Lie superalgebra
\be \mathfrak g_{G,V} := \mathfrak g \oplus \mathfrak g^* \oplus \Pi V \oplus \Pi V^*\,,  \label{gGV-intro} \ee
whose even part is $T^*\mathfrak g$ and whose odd part is $T^*V$. The nonperturbative boundary VOA further includes a contribution from boundary monopole operators. In nonabelian theories, it is still not known how to describe this contribution; boundary monopoles have nonzero cohomological degree \cite{BDGH,dimofte2018dual} and their presence is likely to induce a particular quotient of $\widehat{\mathfrak g_{G,V}}$.

A close analogue of this setup that may be useful to keep in mind is Chern-Simons theory with compact gauge group at integer level \cite{Witten-Jones,EMSS}. There, a chiral boundary condition supports a Kac-Moody VOA perturbatively, which is corrected to a WZW VOA (a maximal simple quotient) nonperturbatively. The bulk B twist of a 3d $\CN=4$ gauge theory seems to be version of Chern-Simons theory for the Lie superalgebra \eqref{gGV-intro}; but the analogue of the WZW VOA on its boundary is not yet known in general.

For abelian 3d $\CN=4$ theories, we find (applying methods of \cite{dimofte2018dual}) that boundary monopoles all have cohomological degree zero. We identify them as particular simple currents of the perturbative VOA, denoted $\widehat{\mathfrak{g}_\rho}$ for abelian theories; and we propose that they should be accounted for by taking a simple current extension (Definition \ref{def-VB} of Section \ref{subsecperturbalg}):
\be \CV_\rho^B := \text{$\Z^r$ simple current extension of the affine VOA $\widehat{\mathfrak{g}_\rho}$}\,. \ee
We then introduce a free field realization of $\CV_\rho^B$ (Section \ref{sec:Bff}), and use it to prove 3d mirror symmetry of abelian boundary VOA's:
\begin{customthm}{\ref{thm:VOA}}
For 3d-mirror abelian theories, there is an isomorphism of boundary VOA's
\be \CV_\rho^A\,\cong\,\CV_{\rho^\vee}^B\,, \ee
which is compatible with conformal grading, cohomological grading, and global (flavor) symmetries on the two sides.
\end{customthm}

An important special case of this theorem is the mirror symmetry between the A twist of a free hypermultiplet (whose boundary VOA is $\CV_{\beta\gamma}\otimes \CV_{bc}$) and the B twist of SQED with one hypermultiplet (whose boundary VOA is a $\Z$ extension of the affine VOA $\widehat{\mathfrak{gl}(1|1)}$). The equivalence of these VOA's was established in \cite{creutzig2009gl,Creutzig:2011cu,creutzig2013w}, which motivated \cite{BN22} and the more general free-field realizations we use in this paper.

We include a final section (Section \ref{secMorita}) analyzing the structure of VOA's in the A and B twists in order to prepare to describe their categories of modules. The most important result is seemingly technical at this point. We define ``alternative'' versions of the boundary VOA's, schematically (Definitions \ref{def:altA} and \ref{def:altB})
\be \label{ext-orb-intro}
\begin{array}{rl} \wt \CV_{\rho}^A &:= \text{$\Z^r$ simple current extension of $\CV_{\beta\gamma}^{T^*V}$} \\[.2cm]
 \wt\CV_{\rho}^B &:= \text{$(\C^*)^r$ orbifold of $T^*V$ symplectic fermions times a self-dual lattice}\,. \end{array}
\ee
Then we prove
\begin{customthm}{\ref{ThmMirSymbdVOA}}
There are equivalences
\be \CV_\rho^A \;\cong\; \wt\CV_\rho^A \;\cong\; \wt\CV_{\rho^\vee}^B\;\cong\; \CV_{\rho^\vee}^B \ee
up to tensoring with (or removing) self-dual lattice VOA's. In particular, these are all Morita-equivalent VOA's, whose module categories are equivalant on the nose.
\end{customthm}

The significance of this result becomes clearer in Section \ref{sec:highersym} and Section \ref{sec:VOAlines}. The operations on the RHS of \eqref{ext-orb-intro} roughly correspond to starting with free hypermultiplets (in either the A or B twists) and gauging part of their flavor symmetry to obtain a rank$-r$ abelian gauge theory $\CT_\rho$. They are not literally what's induced on a boundary VOA from bulk gauging, as that would lead to the usual $\CV_\rho^A,\CV_\rho^B$. Rather, they mimic at the level of VOA's the simplest manifestations of gauging of module categories. Ultimately, Theorem \ref{ThmMirSymbdVOA}) says that a simple, naive gauging operation on VOA's is sufficient for analyzing their module categories.

\subsection{Infinite higher symmetry}
\label{sec:sym-intro}

The remainder of the paper is devoted to studying the line operators of abelian 3d $\CN=4$ gauge theories, in the A and B topological twists. One of our main tools, both conceptually and practically, is gauging and ungauging operations in the bulk and on the boundary. To this end, we make a small detour in Section \ref{sec:highersym} to explain a rather interesting feature of global symmetries in a topological twist.

The physical 3d $\CN=4$ theories $\CT_\rho$ that we study are defined in such a way that they have no one-form higher/generalized symmetries, in the sense of \cite{GKSW}. Their ordinary zero-form symmetry includes flavor symmetry $T_F=U(1)^{n-r}$ acting on hypermultiplets and $T_{\rm top}=U(1)^r$ acting on monopole operators. The gauge group can be enlarged by gauging a subgroup of $T_F$ and effectively reduced (ungauged) by gauging a subgroup of $T_{\rm top}$ \cite{Witten-SL2}. Physically, these operations are on the same footing.

In the topological A and B twists, the two flavor symmetries go very different ways. In the A twist, $T_{\rm top}$ is naturally complexified, and allows a theory to couple to background complex flat connections, in a manner studied in \cite{CDGG}. However, $T_F$ acts trivially on local operators, because its current becomes $Q$-exact. Mathematically, this setup would be expected to lead to a potentially nontrivial ``topological action'' of $T_F$ --- an action that factors through the homotopy type of $T_F$ \cite{Teleman-ICM}. 
In our case, the twisted theory $\CT_\rho^A$ gains a new, infinite 1-form symmetry $\pi_1(T_F)\simeq\Z^r$ replacing $T_F$ (Physics Theorem \ref{thm:sym}). In the B twist, the role of symmetries is reversed, so that altogether
\be \label{oneform-intro}
\begin{array}{l|cc}
& \text{zero-form} & \text{one-form} \\\hline
\CT_\rho & T_F\times T_{\rm top} \simeq U(1)^{n-r}\times U(1)^r & \text{none}  \\
\CT_\rho^A & T_{\text{top},\C}\simeq (\C^*)^{r} & \pi_1(T_F)\simeq \Z^{n-r} \\
\CT_\rho^B & T_{F,\C}\simeq (\C^*)^{n-r} & \pi_1(T_{\rm top})\simeq \Z^r
\end{array}\ee

The generators of one-form symmetry are flavor vortex lines in the A twist, and Wilson lines in the B twist. Both are intrinsically related to spectral-flow modules in boundary VOA's, in a way that resembles classic constructions in Chern-Simons theory \cite{MooreSeiberg-taming}.

An important consequence of \eqref{oneform-intro} is that while gauging part of $T_F$ (say) induces a $\C^*$ orbifold on the category of line operators in $\CT_\rho^B$, it also induces $\Z$ extension (gauging of the one-form symmetry) on the category of line operators in $\CT_\rho^A$. These are the operations that are being mimicked by alternative VOA's \eqref{ext-orb-intro}. The operations are closely analogous to guaging and unguaging of finite symmetry in physical (not topologically twisted) 3d gauge theories, studied \emph{e.g.} in \cite{HLS-gauging,BBFS,BBFS-3d,NSWZ}. It would be interesting to further study the interplay between topological one-form symmetry and higher symmetries in physical theories, as discussed \emph{e.g.} in \cite{Eckhard:2019jgg,GukovHsinPei} and the preceding references.

\subsection{Modules and line operators}
\label{sec:lines-intro}

In a topological twist of a 3d $\CN=4$ theory, one expects line operators to have the structure of a DG braided tensor category $\wh\CC$. The first examples of such categories, for B twists of 3d $\CN=4$ sigma models, were developed in \cite{KRS,KapustinRozansky,RobertsWillerton}. The objects of the category $\wh\CC$ are the line operators themselves; the morphism space $\text{Hom}_{\wh\CC}(\ell_1,\ell_2)$ among a pair of objects is the space of local operators at their junction; tensor product (\emph{a.k.a.} fusion) is induced from collision of parallel lines; and the braiding isomorphisms $R_{\ell_1,\ell_2}\in \text{Hom}_{\wh\CC}(\ell_1\otimes\ell_2,\ell_2\otimes\ell_1)$ are induced from topologially braided configurations of lines: 
\be  \raisebox{-.7in}{\includegraphics[width=2.1in]{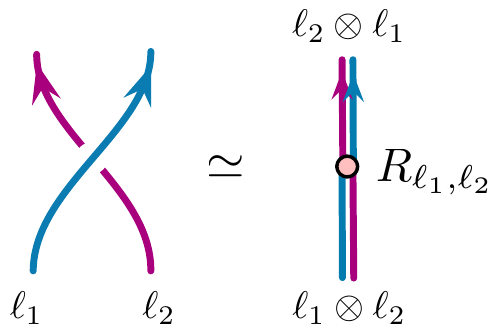}} \ee

Much of this is standard in any 3d TQFT. A feature that is special for twists of supersymmetric theories, as we are considering now, is the DG (differential graded) nature of the category. This is reflected in morphism spaces having defined as the cohomology of a larger space of physical local operators, with respect to a topological supercharge $Q$. Relatedly, one almost always expects line operators to form nontrivial bound states. If $\widehat \CC$ can be modeled as the derived category of an abelian category $\CC$, then $\CC$ will nearly always be non-semisimple. These ideas were further reviewed from a physical perspective in (\emph{e.g.}) \cite{KRS,CCG,CDGG}; related derived structure in modular tensor was studied mathematically by \cite{SchweigertWoike}.

Yet another feature is that twisted 3d $\CN=4$ theories may give rise to Spin or Spin$^c$ TQFT's, depending on the gauge group and matter content, \emph{cf.} \cite{CGP-spin,RW-Coulomb}. This enhances categories of line operators with extra twisted sectors and extra structure. In this paper, however, we will only consider trivial spin structures ---  by ``category of line operators'' we mean the untwisted sector.

An important expectation for the categories of line operators in A and B twists of 3d $\CN=4$ theories is that they reproduce the rings of functions on Coulomb and Higgs branches. Namely, the trivial line operator $\id$  belongs to both A and B twist categories. It is the identity for the tensor product. Junctions between $\id$ and $\id$ simply carry bulk local operators, so
\be \text{Hom}_{\wh\CC_A}(\id,\id) = \C[\CM_{\rm Coulomb}]\,,\qquad \text{Hom}_{\wh\CC_B}(\id,\id) = \C[\CM_{\rm Higgs}]\,. \label{Endid-intro} \ee
The fact that there are nontrivial endomorphims of $\id$ is yet another sign that $\wh\CC$ cannot be a semisimple category.

Several classifications and geometric models of categories of line operators in 3d $\CN=4$ gauge theories are known. Assel and Gomis gave brane constructions of half-BPS line operators \cite{AsselGomis}, adapting techniques previously used to analyze surface operators in 4d gauge theory \cite{HananyHori,GukovWitten,AGGTV,Dimofte:2010tz}. The physical classification was extended in \cite{lineops}, which further gave a computation of morphism spaces. The works \cite{CG,CCG} related line operators to boundary VOA modules, and then used the expected relation \eqref{Endid-intro} to give a new computation of Higgs and Coulomb chiral rings.

Geometric models for categories of line operators were proposed by Hilburn and Yoo, based on a mathematical analysis of A and B twists using techniques from derived geometry --- roughly, they predict that line operators in the A twist of $G,T^*V$ gauge theory look like D-modules on the algebraic loop space of $V/G_\C$, whereas line operators in the B twist look like coherent sheaves on the derived (locally constant) loop space of $V/G_\C$, \cf\ \cite{BF-notes,hilburn2021tate}. Versions of these models have appeared in mathematical work on Coulomb branches  \cite{BFN-lines,Webster-Coulomb,Webster-tilting}.
Hilburn and Raskin proved a 3d-mirror equivalence for categories in the A twist of a free hypermultiplet and the B twist of SQED. An alternative description of line operators in the B twist, essentially an equivariantization of \cite{KRS}, has appeared in work of Oblomkov-Rozansky on knot homology \cite{OR-Chern}.
Most recently, line operators were reconstructed from DG 2-categories of boundary conditions in \cite{HG-Betti}.

However, none of the existing models of categories of line operators in twists of (massless) 3d $\CN=4$ gauge theories capture braided tensor structure. We ultimately access it via boundary VOA's. (The perspective on line operators via Koszul duality proposed by \cite{CostelloPaquette} may provide another route to determining their braided tensor structure, directly from bulk physics. This would be interesting to investigate further.)

The relation between bulk line operators in a topological twist and boundary VOA modules was outlined in \cite{CCG}; it generalizes the celebrated relationship in Chern-Simons theory between Wilson lines and boundary WZW modules \cite{Witten-Jones,EMSS}. For every line operator $\ell\in \wh\CC$, there is an associated module $M(\ell)$, given by local operators at a junction of $\ell$ and the boundary:
\be \label{VOA-intro} \raisebox{-.8in}{\includegraphics[width=2.7in]{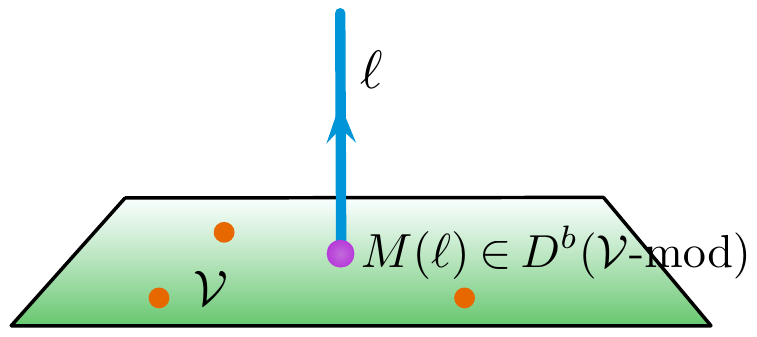}}\qquad  
\ee
This defines a functor
\be \begin{array}{ccc}\widehat \CC &\to & D^b(\CV\text{-mod}) \\
\ell & \mapsto & M(\ell) \end{array}  \ee
which one expects to be a derived equivalence, between the bulk DG category and the \emph{derived} category of VOA modules.

In Section \ref{sec:VOAlines}, we will define abelian categories of modules $\CC_\rho^A$, $\CC_\rho^B$ for the boundary VOA's $\CV_\rho^A$, $\CV_\rho^B$, which both have braided tensor structures, and are equivalent under 3d mirror symmetry. This is a major novelty of the present work as well as \cite{BN22} compared to previous literature. Neither the precise identification of the categories nor their braided tensor structure is entirely straightforward.
A general theory of braided tensor categories of VOA modules --- \emph{a.k.a.} vertex tensor categories --- has been developed by Huang, Lepowsky and Zhang \cite{huang2014logarithmic1,huang2010logarithmic2, huang2010logarithmic3, huang2010logarithmic4, huang2010logarithmic5, huang2010logarithmic6, huang2011logarithmic7, huang2011logarithmic8}. Often it is difficult to show that this theory applies and one of the very few logarithmic theories whose vertex tensor categories are completely understood are the singlet VOA categories \cite{Creutzig:singlet1, Creutzig:singlet2}. The VOAs $\mc{V}_{\rho}^A$ and $\mc{V}_{\rho}^B$ are extensions of many singlet VOAs times some Heisenberg VOA's. The theory of vertex algebra extensions \cite{creutzig2017tensor, creutzig2020direct} then allows us to also understand the braided tensor categories of modules of $\mc{V}_{\rho}^A$ and $\mc{V}_{\rho}^B$. 

More precisely, we first review the work of \cite{BN22}, which implemented this program for the basic mirror pair, the A-twist of a free hypermultiplet and the B twist of SQED (Section \ref{subsec:baby}.) Then we use the idea that gauging in $\CT_\rho^A$ should be implemented by $\Z$ simple current extensions to define (schematically)
\be \CC_\rho^A := \CC_{\beta\gamma}^{T^*V}\big/ \Z^r\,, \label{CA-intro} \ee
and bootstrap the results of \cite{BN22} to prove that this has a braided tensor structure. In Section \ref{subsec:LineB} we independently define
\be \CC_\rho^B := KL_\rho\big/\Z^r\,, \ee
where $KL_\rho$ denotes the Kazhdan-Lusztig category of the affine superalgebra $\widehat{\mathfrak g_\rho}$ (consisting of finite-length grading-restricted modules --- see Definition \ref{def:KLrho}), and now the $\Z^r$ simple current extension is due to boundary monopole operators. We prove (Theorem \ref{Thm:KLKLrho}) that $\CC_\rho^B$ also has a braided tensor structure. Finally, we establish
\begin{customthm}{\ref{ThmCACB}}
There is an equivalence of abelian braided tensor categories
\be \CC_\rho^A\simeq \CC_{\rho^\vee}^B\,. \ee
\end{customthm}

\subsection{Comparison between bulk and boundary categories}

The somewhat technical VOA analysis of Section \ref{sec:VOAlines} is complemented by a comparison with line operators, constructed purely in the bulk 3d theories in Section \ref{sec:bulklines}. In Section \ref{sec:bulklines}, we organize the bulk categories of lines according to the infinite one-form symmetry \eqref{oneform-intro}. In the B twist, the generators of the $\Z^r$ one-form symmetry are half-BPS Wilson lines $\{\W_m\}_{m\in \Z}$, well known from \cite{AsselGomis,lineops}. The ``charged objects'' include hybrids of monodromy defects and Wilson lines $\W_m^{(\upsilon)}$, labelled by an additional parameter $\upsilon\in (\C/\Z)^r$ (a character of $\Z^r$). These have not been previously studied. They are analogues of Wilson-'t Hooft lines in 4d gauge theory

We define a minimal DG category containing the $\W_m^{(\upsilon)}$, which has a block decomposition
\be \wh\CC_\rho^B = \bigoplus_{\upsilon\in(\C/\Z)^r} \wh\CC_{\rho,\upsilon}^B\,, \ee
with each subcategory $\wh\CC_{\rho,\upsilon}^B$ generated by $\{\W_m^{(\upsilon)}\}_{m\in \Z^r}$. We argue in Section \ref{sec:VOAlines} (in particular, by comparing derived morphisms in Section \ref{subsec:HomLines}) that $\wh\CC_\rho^B$ is the derived category of our VOA module category,
\be \wh\CC_\rho^B \simeq D^b \CC_\rho^B\,, \label{DB-intro} \ee
in such a way that the generators $\W_m^{(\upsilon)}$ of $\wh\CC_\rho^B$ map to the simple objects of $\CC_\rho^B$.

Similarly, in the A twist, we recall that the bulk theory has ``flavor vortex'' lines $\{\V_m\}_{m\in \Z^{n-r}}$, which now generate the $\Z^{n-r}$ one-form symmetry. In addition, there are previously unstudied BPS defects $\V_m^{(\upsilon)}$ labelled by $m\in\Z^{n-r}$ and $\upsilon\in (\C/\Z)^{n-r}$. Altogether, these generate a DG category that we argue is the derived category of our VOA module category,
\be \wh\CC_\rho^A = \bigoplus_{\upsilon\in(\C/\Z)^{n-r}} \wh\CC_{\rho,\upsilon}^A \simeq D^b \CC_\rho^A\,. \label{DA-intro} \ee

The results of Section \ref{sec:VOAlines} immediately imply a 3d-mirror equivalence of bulk categories
\be \wh\CC_\rho^A  \simeq  \wh\CC_{\rho^\vee}^B\,. \ee
We also expect that the braided tensor structure of $\CC_\rho^A,\CC_\rho^B$ will lift to the derived categories. In particular, by a result of  \cite{creutzig2020tensor} for affine $\mathfrak{gl}(1|1)$ together with properties of the lifting functor that we use to construct $\CC_\rho^A,\CC_\rho^B$, it follows that the tensor product in the abelian categories is exact. Under some mild conditions, it should then lift to a tensor product on the derived categories. We do not analyze this rigorously here.

We expect that the derived categories $\wh\CC_\rho^A,\wh\CC_\rho^B$ are also closely related to the ``Betti'' geometric models of line operators in the B twist studied in \cite{OR-Chern,HG-Betti}. This appears to be instance of Koszul duality, which we will  pursue elsewhere.

\subsection{Acknowledgements}

We are grateful to many for insightful comments and advice in regards to this work, especially Mathew Bullimore, Niklas Garner, Jusin Hilburn, Eugene Gorsky, Pavel Safronov, and Daniel Zhang. We also thank Andrea Ferrari and Christopher Beem for coordinating a manuscript submission with us.
TC's research is supported by NSERC Grant RES0048511.
TD's research is supported by EPSRC Open Fellowship EP/W020939/1. WN's research is supported by a UC Davis dissertation year fellowship.
AB, TD, and WN were also supported by NSF CAREER Award DMS-1753077.

\section{Setup and notation}
\label{sec:setup}

\subsection{Simple abelian 3d gauge theories}
\label{sec:Trho}

As noted in the introduction, we consider in this paper abelian 3d $\CN=4$ theories, with a few further restrictions:

\begin{enumerate}
    \item The gauge group $G$ is continuous (no discrete factors), and thus isomorphic to $U(1)^r$ (for some $r\geq 0$).
    \item The gauge group acts faithfully on a matter representation of the form $T^*V \simeq V^*\oplus V$, where $V$ is a complex linear representation of $V$, and $V^*$ its dual. We let $V=\C^n$, and may assume that the action has been diagonalized, meaning the action map $G\to U(V)$ factors through a homomorphism $\varphi:G\to U(1)^n$, \emph{i.e.}
    \be \label{gauge-hom} \varphi: U(1)^r \to U(1)^n\,. \ee
    Faithfulness of the action means that $\varphi$ is injective.
\end{enumerate}
We will call such theories ``simple.'' One of the things that is simple about them is the action of 3d mirror symmetry. We review this in a moment, after first translating the requirements above to statements about charge lattices.

\subsubsection{Charge lattices}\label{sec:chargelattice}

We first translate faithfulness to a statement about charge lattices.
Using the exponential map $\theta\mapsto e^{2\pi i \theta}$ to identify $\R/\Z$ with $U(1)$, we see that the homomorphism \eqref{gauge-hom} is induced from a linear map
\be \rho :\R^r \to \R^n \ee
that preserves integer sublattices, $\rho(\Z^r)\subseteq \Z^n$. (This restriction ensures that $\varphi=e^{2\pi i \rho}$ is well defined.) One would usually call the integer-valued matrix $\rho$ the charge matrix of the theory. We'll write its components as
\be \rho = \rho^i{}_a\,,\qquad i=1,...,n\,,\quad a=1,...,r \ee
such that $\rho^i{}_a$ is the charge of the $i$-th hypermultiplet under the $a$-th factor in the gauge group.

We also define
\be \Lambda := \rho(\Z^r) \ee
to be the charge lattice itself. It is a sublattice of $\Z^n$, with basis vectors $\rho_a$ (the columns of $\rho$) and bilinear form $$(\rho^\T\rho)_{ab},$$ where $\rho^\T$ denotes the transpose matrix. Denote by $\Lambda_\R:=\Lambda\otimes_\Z\R$ the $\R$-span of $\Lambda$.

We then have:
\begin{Lem}
The faithfulness condition (2) is equivalent to the cokernel of the restricted map $\rho:\Z^r\to \Z^n$ being rank $n-r$ and torsion-free. In other words, the gauge group acts faithfully if and only if there exists an exact sequence
\be 0\longrightarrow \Z^r \overset{\rho}{\longrightarrow} \Z^n \overset{\tau}{\longrightarrow} \Z^{n-r} \longrightarrow 0\,. \label{SES} \ee
Alternatively, in lattice terminology, condition (2) is equivalent to $\Lambda$ being a \emph{complete sublattice}, meaning
\be \Lambda_\R\cap \Z^n = \Lambda\,. \ee
\end{Lem}

\noindent\emph{Proof.} 
After an integral change of basis in both $\Z^r$ and $\Z^n$, we can transform $\rho$ to a \emph{diagonal} map $\hat \rho:\Z^r\to \Z^n$ whose kernel and cokernel are isomorphic to those of $\rho$. (This is sometimes known as Smith normal form. Diagonal means $\hat\rho^i{}_a=0$ unless $i=a$.) Injectivity of $\varphi$ implies injectivity of $\rho,\,\hat \rho$, which in turn requires $n\geq r$. We then have a direct sum decomposition
\be \hat\rho = \hat\rho^{diag} \oplus 0 \,:\, \Z^r\to \Z^r\oplus \Z^{n-r}\,, \label{Smith} \ee
with a nondegenerate $r\times r$ diagonal matrix $\hat\rho^{diag}$.
Injectivity of $\varphi$ also leads to a more subtle condition.
Even when $\rho$ is injective, $\varphi$ may still have a kernel consisting of non-integral points of $\R^r$ that are mapped by $\rho$ to integer points of $\R^n$. Namely, denoting by $(\hat\rho^{diag})^{-1}(\Z^r)$ the preimage {in $\R^r$} of the first summand in \eqref{Smith}, we have
\be \text{ker}\,\varphi \simeq (\hat\rho^{diag})^{-1}(\Z^r)\,\big/\,\Z^r \simeq \Z^r / \hat\rho^{diag}(\Z^r) = \text{coker}\,\hat\rho^{diag}\,. \ee
This is a finite group. Moreover, since $\text{coker}\,\rho \simeq \text{coker}\hat\rho = \text{coker}\,\hat\rho^{diag} \oplus \Z^{n-r}$, we conclude that $\text{ker}\,\varphi$  is trivial if and only if the restriction of $\rho$ to integer points can be completed to an exact sequence \eqref{SES}.

The statement about complete lattices follows by further observing that $\Lambda_\R\cap\Z^n$ contains precisely the points of $\R^r$ mapped by $\rho$ integral points of $\Z^n$, so
\be \text{ker}\,\varphi  \simeq \big(\Lambda_\R\cap\Z^n\big)\big/\Lambda\,,  \ee
whence faithfulness is equivalent to $\Lambda$ being complete.
\hfill $\square$

\subsubsection{Global symmetry}\label{sec:global}

We summarize for future reference the global symmetries of a 3d $\CN=4$ theory $\CT_\rho$ with $n\times r$ charge matrix $\rho$, \emph{i.e.} with gauge group $U(1)^r$ acting faithfully on $n$ hypermultiplets. The ordinary (zero-form) symmetry structure is well known, see \emph{e.g.} \cite{IS}. There is:
\begin{itemize}
    \item An $SU(2)_H$ R-symmetry acting on hypermultiplet scalars (and vectormultiplet fermions), inducing hyperkahler rotations of the Higgs branch.
    \item An $SU(2)_C$ R-symmetry acting on vectormultiplet scalars (and both hypermultiplet and vectormultiplet fermions), inducing hyperk\"ahler rotations of the Coulomb branch.
    \item A flavor symmetry with maximal torus $T_F=U(1)^{n-r}$ acting on hypermultiplets, and thus on the Higgs branch.
    \item A topological flavor symmetry with maximal torus $T_{top}=U(1)^r$ acting on monopole operators, and thus on the Coulomb branch.
\end{itemize}
Both of the flavor symmetries may have nonabelian enhancements in special cases, which do not play a role in this paper.

The charges of hypermultiplets under the flavor symmetry $T_H\simeq U(1)^{n-r}$ are only defined modulo gauge charges. One can fix this ambiguity by choosing a splitting of the exact sequence \eqref{SES},
\be \raisebox{-.2in}{\includegraphics[width=2.3in]{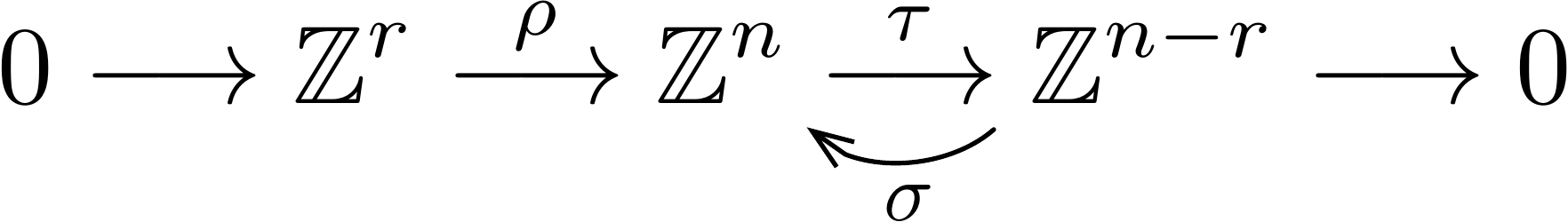}} \label{split-SES} \ee
(Recall that a splitting is a map $\sigma$ satisfying $\tau\sigma=\text{id}_{\Z^{n-r}}$. The change of basis to Smith normal form \eqref{Smith} shows a splitting always exists.) Then the components $\sigma^i{}_{\alpha}$ encode the charge of the $i$-th hypermultiplet under the $\alpha$-th factor in the flavor group. Correspondingly the flavor symmetry embeds into the full abelian $U(1)^n$ symmetry of $V=\C^n$ via
\be e^{2\pi i\sigma}: U(1)^{n-r}\to U(1)^n\,.\ee

In general, 3d $\CN=4$ theories can also have one-form symmetries \cite{GKSW,HLS-gauging}. They would arise from discrete factors in the gauge group, giving rise to one-form symmetry acting on Wilson lines; or from a discrete subgroup of the gauge group acting trivially, giving rise to one-form symmetry acting on vortex lines. These phenomena and their interaction with zero-form symmetry can be quite rich, and were recently studied in a 3d $\CN=4$ context in \cite{BBFS}. However, the two origins of one-form symmetry are precisely what we have ruled out in our definition of \emph{simple} abelian gauge theories. We thus expect that \emph{simple abelian 3d $\CN=4$ theories are precisely the theories with trivial one-form symmetry}.

It should be quite interesting to extend the analysis in current paper beyond simple theories in the future.

\subsubsection{3d mirror symmetry}

For a simple theory with exact sequence \eqref{SES}, the theory $\CT_\rho$ defined by charge matrix $\rho$ has a 3d mirror \cite{dBHOOY,KapustinStrassler}, denoted $\CT_{\rho^\vee}$, with gauge group $G=U(1)^{n-r}$, $n$ hypermultiplets $V=\C^n$, and mirror charge matrix
\be \rho^\vee := \tau^\T\,, \ee
fitting into the dual exact sequence
\be 0\longrightarrow \Z^{n-r} \overset{\rho^\vee}{\longrightarrow} \Z^n \overset{\tau^\vee=\rho^\T}{\longrightarrow} \Z^{r} \longrightarrow 0\,. \label{SES-mirror} \ee
Note that duality here uses the standard Euclidean inner product on $\Z^r,\Z^n,\Z^{n-r}$, and $M^\T$ denotes the usual transpose of a matrix $M$.

In terms of lattices, mirror symmetry swaps $\Lambda$ and the \emph{orthogonal lattice}
\be \Lambda^\perp := \{\mu\in \Z^n\,|\, \mu\cdot\lambda=0\;\forall\,\lambda\in\Lambda\} = \rho^\vee(\Z^{n-r})\,. \label{Lambda-perp} \ee
The requirement that $\Lambda$ be a complete lattice guarantees that $(\Lambda^\perp)^\perp=\Lambda$.

3d mirror symmetry is also sometimes phrased in terms of gauge \emph{and} flavor charges. Upon choosing a splitting $\sigma$ as in \eqref{split-SES}, one can form an $n\times n$ matrix $(\rho\;\sigma)$ that is unimodular. 3d mirror symmetry than acts by taking the inverse-transpose of this matrix:
\be (\rho^\vee,\sigma^\vee) := (\rho,\sigma)^{-1,\T}  \label{mirror-matrix} \ee
where $\rho^\vee$ is the mirror charge matrix, satisfying $\rho^{\vee\T}\rho=0$; and $\sigma^\vee$, encoding mirror flavor charges, is a splitting of the dual sequence, chosen to satisfy $\sigma^{\vee\T}\sigma=0$.

Physically, 3d mirror symmetry swaps the role of the $T_F$ flavor symmetry (acting on hypermultiplets) and the $T_{top}$ topological flavor symmetry (acting on monopole operators, and thus on the Coulomb branch). It also swaps the two R-symmetries:
\be \text{3d MS:} \qquad \begin{array}{c@{\;\;\leftrightarrow\;\;}c}
 SU(2)_H & SU(2)_C \\
 T_F & T_{\rm top} \\
 \text{vortex lines} & \text{Wilson lines} \\
  \text{A twist} & \text{B twist} \end{array} \ee
Thus, if we keep track of the R-symmetry action, hypermultiplets are replaced in a mirror theory with twisted hypermultiplets (which by definition are hypermultiplets with a swapped action of $SU(2)_H$ and $SU(2)_C$) and vectormultiplets are similarly replaced by twisted vectormultiplets. The swap of R-symmetries is also why A and B topological twists are exchanged in mirror theories.

\subsubsection{Vacuum structure of simple theories}
\label{sec:smooth}

The Higgs and Coulomb branches of simple abelian gauge theories are hypertoric varieties, whose geometry may be analyzed via the combinatorics of hyperplane arrangements. This is reviewed mathematically in \cite{Proudfoot,Gale} and from a physical 3d $\CN=4$ perspective in \cite{BDGH}. We collect a few results here.

After turning on generic real mass and real FI parameters (associated to $T_F$ and $T_{top}$ flavor symmetries), both the Higgs and Coulomb branches become massive. They support a finite number of vacua, located at points on a partially resolved Higgs branch, or a partially resolved Coulomb branch (the same vacua appear on both). More precisely, the vacua are in 1-1 correspondence with subsets $I\subset\{1,...,n\}$ of size $r$, such that the corresponding minor of the charge matrix satisfies
\be \text{det}(\rho^I) \neq 0\,,\quad\text{where}\quad \rho^{I}:= (\rho^i{}_a)^{i\in I}_{a\in\{1,...,r\}} \label{det-vac}. \ee
The complement $I^c$ contains the hypermultiplets that are set to zero in the vacuum, due to mass parameters. The set $I$ itself contains hypermultiplets that are set to non-zero values to satisfy real moment-map constraints, and are then fixed by the action of the rank-$r$ gauge group $G$. The constraint \eqref{det-vac} guarantees that the Lie algebra of $G$ acts freely. However, when $\text{det}(\rho^I) \neq \pm 1$, a finite subgroup of $G$ may remain unbroken. In this case the local neighborhood of the vacuum has the structure of an orbifold singularity of order $|\text{det}(\rho^I)|$ on both the Higgs and Coulomb branches, and the theory is not quite fully massive.

This vacuum analysis is beautifully compatible with mirror symmetry, due to a simple identity relating minors of a matrix and its inverse. In the case at hand, the inverse relation \eqref{mirror-matrix} guarantees that for each subset $I$ of size $r$ and its complement $I^c$ we have
\be \text{det}\,\rho^I\, =\, \text{det}\,(\rho^\vee)^{I^c}\,. \ee
This gives another way to see that both vacua and orders of orbifold singularities on Higgs and Coulomb branches must match.

A rough measure of the non-triviality of a simple abelian theory, which we'll refine in Section \ref{sec:H}, is the count of vacua (upon generic mass and FI deformation) weighted by the square of their orbifold orders --- or, equivalently, weighted by the product of their Higgs-branch and Coulomb-branch orbifold orders. We'll denote this quantity $N_{\rm vac}(\rho)$. It is given by summing over all minors of the charge matrix,
\be N_{\rm vac}(\rho) :=  \sum_{I\subset\{1,...,n\}} \big(\text{det}\,\rho^I\big)\big(\text{det}\,(\rho^\vee)^{I^c}\big) = \sum_{I\subset\{1,...,n\}}\big|\det\,\rho^I\big|^2\,. \label{Nvac} \ee

\subsubsection{Examples}
\label{sec:eg1}

The theory $\CT_{(1)}$ with $\rho=(1)$ is 3d $\CN=4$ SQED. It has $G=U(1)$ acting on a single hypermultiplet, with charge 1. Its exact sequence is
\be \text{SQED}_1:\quad  0 \longrightarrow \Z \xrightarrow{\rho=1} \Z \xrightarrow{\tau=0} 0 \longrightarrow 0\,. \label{SQED-seq} \ee
Its 3d mirror $\CT_{0}$ is a free hypermultiplet. (More accurately, the mirror of SQED$_1$ is a twisted hypermultiplet, in that the actions of $SU(2)_H$ and $SU(2)_C$ R-symmetries are swapped.)
It has $G=1$ acting (trivially) on $V=\C$. Its exact sequence is
\be \text{hyper}:\quad  0 \longrightarrow 0 \xrightarrow{\rho^\vee=0}  \Z \xrightarrow{\tau^\vee=(1)} \Z \longrightarrow 0\,. \label{hyper-seq} \ee
Both of these theories have $N_{\rm vac}=1$.

A more interesting theory is SQED with $n$ hypermultiplets of charge 1, meaning $G=U(1)$, $V=\C^n$, and $\rho=(1,1,...,1)^\T$,
\be \text{SQED}_n:\quad  0 \longrightarrow \Z \xrightarrow{\rho=\left(\begin{smallmatrix}1\\\vdots \\1\end{smallmatrix}\right)} \Z^n \overset{\tau}\longrightarrow \Z^{n-1} \longrightarrow 0\,.  \ee
The splitting in the exact sequence is given by the $(n-1)\times n$ matrix
\be \tau = \left(\begin{smallmatrix} 1&-1&0&\cdots & 0 \\ 0&1&-1&\cdots & 0 \\ &&&\ddots& \\ 0&0&0&&-1 \end{smallmatrix}\right)\,. \ee
This theory has $N_{\rm vac}(\rho)=n$ smooth massive vacua, as is most easily seen from the minors of $\rho$. Its Higgs branch resolves to $T^*\mathbb{CP}^{n-1}$, and the massive vacua are fixed points of a generic $U(1)$ action thereon.

Finally, a simple example of a theory with orbifold vacua is a modification of SQED$_2$, with a hypermultiplet of charge one and a hypermultiplet of charge \emph{two}.
\be \widetilde{\text{SQED}_2} :\quad 0\longrightarrow \Z \xrightarrow{\rho = \left(\begin{smallmatrix}1\\2\end{smallmatrix}\right)} \Z^2 \xrightarrow{\tau=\left(\begin{smallmatrix}2 & -1\end{smallmatrix}\right)} \Z \longrightarrow 0  \ee
It admits a splitting $\sigma = \left(\begin{smallmatrix}0\\-1\end{smallmatrix}\right) $. Looking at $1\times 1$ minors of $\rho$, it is clear that this theory has two vacua, one smooth and one that is a $\Z_2$ orbifold on both the Higgs and Coulomb branches. Now we have $N_{\rm vac}(\rho) = 1+2^2=5$.

\subsection{An invariant of abelian theories}
\label{sec:H}

A natural question to ask is how far one can simplify an abelian 3d $\CN=4$ theory using field redefinitions and dualities. We will give a partial answer here, together with a necessary and sufficient criterion for a theory to be IR free --- meaning it is a product of free hypermultiplets \eqref{hyper-seq} and SQED$_1$'s \eqref{SQED-seq} (which flow to free twisted hypermultipelts in the IR%
\footnote{This flow is just rephrasing the statement that SQED$_1$ is mirror to a free twisted hypermultiplet. 
We remind the reader that 3d mirror symmetry relates pairs of gauge theories that become the same \emph{in the infrared}.}%
).

Given a simple (in the sense of Section \ref{sec:setup}) abelian theory with charge matrix $\rho$, we can obtain an equivalent theory by
\begin{itemize}
\item Permuting the hypermultiplets, and/or negating their charges. Abstractly, this is an action of the Weyl group of the full symmetry $USp(n)$ of $n$ hypermultiplets, which is the stabilizer of the torus $U(1)^n$ through which we have required the gauge group to act. In other words, this is the permutation group that preserves weight spaces. Such a permutation acts on $\rho$ by
\be \rho \mapsto P \rho \ee
where $P \in GL(n,\Z)$ has exactly one $\pm1$ entry in each row and each column. (Recall that this condition on an integer matrix $P$ is equivalent to $P^TP=I$, \emph{i.e.} $P\in O(n,\Z)$.)

Geometrically, these permutations are the automorphisms of the Euclidean lattice $\Z^n$, and act on the embedding of the charge lattice $\Lambda\subseteq\Z^n$.

\item Redefining the vectormultiplets, by a transformation that is invertible over the integers. This amounts to a redefinition of gauge charges. The charge matrix transforms as
\be \rho \mapsto \rho M \ee
for any $M\in GL(r,\Z)$. These are automorphisms of the charge lattice $\Lambda$ itself.
\end{itemize}
All together, field redefinitions correspond to isomorphisms of the charge lattice $\Lambda$ and its embedding in $\Z^n$. Therefore, isomorphism classes of simple abelian theories are classified \emph{at most} by isomorphism classes of embedded sublattices in $\Z^n$.

However, dualities may further relate such lattice isomorphism classes. For example, 3d mirror symmetry swaps a lattice with its orthogonal complement, as in \eqref{Lambda-perp}. We are therefore led to define a coarser invariant of abelian theories, which may not be a complete invariant, but is guaranteed to be preserved by any quantum duality.

We consider the group
\be H := \text{coker}(\rho^\T\rho) =  \Z^r / \text{im}({\rho^\T\rho})\,, \ee
noting that $\rho^T\rho$ is the intrinsic bilinear form on the charge lattice $\Lambda$ (independent of its embedding in $\Z^n$. The abelian group $H$ is finite (since $\rho$ has full rank $r$, whence $\rho^\T\rho$ is non-degenerate). Its order is simply $|H| = \text{det}(\rho^\T\rho)$. It has several beautiful mathematical and physical properties:

\begin{Prop} \label{prop:H} \phantom{x}
\begin{itemize}
\item[(i)]  A field redefinition $\rho\mapsto P\rho M$ induces an isomorphism of $H$.
\item[(ii)] Recalling the definition $\rho^\vee=\tau^\T$ from \eqref{SES-mirror}, there are isomorphisms
  \be\label{H-iso} H =  \Z^r / \text{im}({\rho^\T\rho}) \simeq \Z^{n-r} / \text{im}({\rho^\vee{}^\T\rho^\vee})   \simeq \Z^n/(\Lambda\oplus\Lambda^\perp) \ee
  In particular, $H$ is invariant under 3d mirror symmetry.
\item[(iii)] The order of $H$ is the weighted count of massive vacua introduced in \eqref{Nvac}:
\be |H| = N_{\rm vac}(\rho)\,. \ee
\item[(iv)] $H$ is the ``symmetry fractionalization group'' for Wilson lines (or, by (ii) and 3d mirror symmetry, for vortex lines). Namely,
\be H \,\simeq\, \frac{\text{possible flavor charges of local operators at junctions of Wilson lines}}{\text{possible flavor charges of bulk local operators}} \,.
\label{H-frac} \ee
(Here we are ignoring R-symmetry, and considering only global flavor symmetry.) A restatement of \eqref{H-frac}, is that if $F$ denotes the hypermultiplet flavor symmetry in the presence of Wilson lines, the symmetry acting faithfully on bulk local operators is $F/H$.
\end{itemize}
\end{Prop}

\noindent\emph{Proof.} (i) is straightforward: for any $N,M\in GL(r,\Z)$, the abelian groups $\Z^r/\rho^\T\rho$ and $\Z^r/N\rho^\T\rho M$ are isomorphic, and we observe that the field redefinition $\rho\mapsto P\rho M$ induces $\rho^\T\rho \mapsto M^\T(\rho^\T\rho)M$.

For (ii), it is enough to show that $H\simeq \Z^n/(\Lambda\oplus\Lambda^\perp) = \Z^n / ( \text{im}(\rho)\oplus \text{im}(\rho^\vee))$, since the middle isomorphism in \eqref{H-iso} will then follow by symmetry between $\rho$ and $\rho^\vee$. Note that from the dual exact sequence \eqref{SES-mirror} we have $\text{im}(\rho^\vee) =\text{ker}(\rho^\T)$. 
Then the third isomorphism theorem for groups gives
\be \frac{\Z^n}{\text{im}(\rho)\oplus \text{im}(\rho^\vee)} = \frac{\Z^n}{\text{im}(\rho)\oplus \text{ker}(\rho^\T)} \overset{\sim}\longrightarrow \frac{\rho^\T(\Z^n)} {\rho^T\big(\text{im}(\rho)\oplus \text{ker}(\rho^\T)\big) }  = \frac{\Z^r}{\text{im}(\rho^\T\rho)}\,. \ee

(iii) is a consequence of the Cauchy-Binet formula, which expresses the determinant of a product of rectangular matrices as a sum of corresponding minors:
\be |H| = \text{det}(\rho^T\rho) \overset{C.B.}= \sum_{I\in \{1,...,n\}} (\text{det}\,\rho^I)(\text{det}\,\rho^\T_I) = \sum_{I\in \{1,...,n\}} |\det\rho^I|^2\,,\ee
where the sum is over all size-$r$ subsets of $\{1,...,n\}$. This reproduces the RHS of \eqref{Nvac}.

(iv) amounts to unravelling the physical definition of symmetry fractionalization. This is a concept introduced in the condensed matter literature, which has recently played a key role in the study of higher/generalized symmetries, \emph{cf.} \cite{BCH,BBFS}. The Wilson-line analysis is perturbative, and only involves hypermultiplet fields. 
 We consider perturbative local operators monomials built from hypermultiplet scalars, fermions, and their derivatives. Each monomial can be labelled by a vector $m\in \Z^n$, corresponding to its charge under the abelian $U(1)^n$ symmetry of $n$ hypermultiplets. Before one imposes gauge invariance, all charges are possible. For example, denoting the chiral hypermultiplet scalars $(X_i,Y_i)_{i=1}^n$, there is an operator $\CO_m$ of each charge $m\in \Z^n$, given by
\be \CO_m = \prod_{i=1}^n \begin{cases} (X_i)^{m_i} & m_i\geq 0 \\ (Y_i)^{-m_i} & m_i < 0 \end{cases}\,. \ee 
At junctions of (appropriate) Wilson lines, any such monomial may appear. If it is not gauge invariant, its lack of gauge invariance can be compensated for by the charges of the lines. The flavor charges of all possible local operators at all possible junctions of Wilson lines are then given by the image $\tau(\Z^n)=\Z^{n-r}$ of the second map in the exact sequence \eqref{SES}. 

In contrast, in the bulk, gauge invariance requires local operators to have $m\in \Z^n$ orthogonal to the sublattice $\text{im}(\rho)\subset \Z^n$ (since the charge of a monomial under the $a$-th gauge group is $m\cdot \rho^a$). Thus, gauge-invariant operators have $m\in (\text{im}\,\rho)^\perp$. By the dual exact sequence \eqref{SES-mirror},  $ (\text{im}\,\rho)^\perp = \text{ker}\,\rho^\T = \text{im}\,\tau^\T$. The possible flavor charges of bulk local operators are therefore $\tau(\text{im}\,\tau^\T) = \text{im}(\tau\tau^\T) \subset \Z^{n-r}$. The symmetry fractionalization group is then $\Z^{n-r} / \text{im}(\tau\tau^\T)$, which is isomorphic to $H$ by part (ii) of the Proposition. \hfill $\square$

\medskip

Note that symmetry fractionalization is a topological physical observable. It is invariant under RG flow, and invariant under all quantum dualities, whether exact dualities or only IR dualities (like 3d mirror symmetry). Thus, a corollary of Part (iv) of the Proposition is:

\begin{PThm} \label{Thm:H}
The group $H$ is an invariant of IR isomorphism classes of simple abelian 3d $\CN=4$ gauge theories.
\end{PThm}

\begin{Exp}
Consider SQED$_n$, the theory with gauge group $U(1)$ and $n$ hypermultiplets of charge $1$, which has charge matrix $\rho=(1,1,...,1)^\T$. We have $\rho^\T\rho = (n)$, and so
\be H = \Z/n\Z\,. \label{H-SQEDn} \ee

Correspondingly, the theory has $|H|=n$ massive vacua after a generic real mass and FI deformation. They were described in 
Section \ref{sec:eg1}, where it was also shown that they are individual smooth vacua rather than orbifold points.

Let us also look at symmetry fractionalization. We learned of this example of symmetry fractionalization from M. Bullimore and A. Ferrari \cite{BF-talk} (now published in \cite{BBFS-3d}). Bulk local operators must have an equal number of positively and negatively charged fields in order to be gauge invariant. (For example, the monomials $X_iY_j$ are all gauge invariant.)
 They are labelled by $m\in \Z^n$ orthogonal to $(1,1,...,1)$; after quotienting by $(1,1,...,1)$, these may be identified with points on the root lattice of $SU(n)$. In contrast, local operators at junctions of Wilson lines can be labelled by any $m\in \Z^n$ (for example, individual $X_i$ or $Y_i$ can appear at junctions where Wilson line charges change by $+1$ or $-1$, respectively). After quotienting by $(1,1,...,1)$, the flavor charges of operators at junctions may be identified with points on the weight lattice of $SU(n)$. The symmetry fractionalization group is
 \be \frac{\text{$SU(n)$ weight lattice}}{\text{$SU(n)$ root lattice}} \simeq \Z_n \ee

The group $SU(n)$ appears naturally in this example due to a nonabelian symmetry enhancement. The global symmetry of the bulk theory is $PSU(n)$, while the global symmetry in the presence of Wilson lines is $SU(n)$, with $PSU(n)\simeq SU(n)/H$\,.
\end{Exp}

It is interesting to ask whether $H$ is a complete IR invariant of simple abelian gauge theories. The answer is likely negative, given the great richness of the classification of definite integral lattices. In particular, there are in general many more lattice isomorphism types of a given rank and determinant (given $r$ and $|H|$) than there are abelian groups of fixed order. Even for $|H|=1$ there are nontrivial isomorphism types of lattices when $r\geq 8$, the first example (at $r=8$) being the $E_8$ lattice. It seems highly unlikely that these would correspond to trivial 3d $\CN=4$ theories.

Nonetheless, we can give one stronger criterion for a theory to be infrared free:

\begin{Prop} \label{prop:triv} If a simple theory $\CT_\rho$ is such that its charge lattice is (intrinsically) isomorphic to a trivial lattice, meaning that there exists $M\in GL(n,\Z)$ such that $M^\T(\rho^\T\rho) M=I_{r\times r}$, then $\CT_\rho$ is isomorphic to an IR-free theory.

In particular, by the classification of definite unimodular lattices, any theory with $r<8$ and $|H|=1$ is IR free.
\end{Prop}

\noindent\emph{Proof.} This is intuitively clear given that lattices together with embeddings in $\Z^n$ classify simple abelian theories, and the permutations that change embeddings are fairly inocuous. However, we give a short self-contained proof.

If $M^\T(\rho^\T\rho) M=I_{r\times r}$, then the $r$ column vectors of $\rho M$ are orthonormal. A vector of \emph{integers} of norm one must contain a single $\pm 1$, with all other entries 0. Orthogonality of the columns implies that each $\pm 1$ occurs in a different row. Then we can find a permutation $P \in O(n,\Z)$ such that
\be P\rho M = \left(\begin{smallmatrix} 1&0&\cdots &0 \\ 0 & 1 &\cdots& 0 \\ &&\ddots \\ 0&0&\cdots &1 \\ 0&0&\cdots&0 \\ & \vdots \\ 0&0&\cdots&0 \end{smallmatrix}\right)\,. \ee
This means that after a field redefinition corresponding to $P$ (acting on hypers) and $M$ (acting on gauge fields), our theory becomes a tensor product of $r$ copies of SQED$_1$ and $n-r$ hypermultiplets. \hfill $\square$

\subsection{Free field vertex algebras}\label{subsecVOAconvention}

In this section we introduce the basic vertex algebras and their modules that appear in this work. A free field algebra is a vertex algebra that is strongly generated by fields that have the property that only the identity appears in their operator product algebra. There are four classes of free field algebras that admit a Virasoro structure (stress tensor): the free boson/Heisenberg VOA, the free fermion, the symplectic fermions, and the symplectic bosons. These algebras are related in various way that we now recall.

\subsubsection{Heisenberg VOA's and Fock modules}
\label{sec:Heis}

The basic Heisenberg VOA $H_J$ is generated by a single even (bosonic) field $J(z) = \sum\limits_{n \in\Z} J_n z^{-n-1}$ with OPE
\be
J(z) J(w) \sim \frac{1}{(z-w)^2}\,.
\ee
Its simple modules are Fock modules $\Fock_\lambda$ of highest weight $\lambda \in \C$. These are generated by a highest-weight vector $|\, \lambda\, \rangle$ on which $J_0$ acts by multiplication with $\lambda$, the $J_n$ for positive $n$ annihilate $|\, \lambda\, \rangle$, and the negative modes act freely. 
We will denote vector space tensor products by the usual symbol $\otimes$; while the fusion product, \emph{i.e.} the tensor product as modules of a VOA $\CV$, will be denoted by the symbol $\otimes_\CV$. The Fock modules of the Heisenberg VOA satisfy the fusion rules
\be
\Fock_\lambda \otimes_{H_J} \Fock_\mu = \Fock_{\lambda + \mu}\,.
\ee
The vertex operator associated to the highest-weight vector $|\, \lambda\, \rangle$ is denoted by $Y(|\, \lambda\, \rangle, z)$, and the fusion rules of the Fock modules are reflected in the following OPE formula
\be
Y(|\, \lambda\, \rangle, z) Y(|\, \mu\, \rangle, w) \sim (z-w)^{\lambda\mu} Y(|\, \lambda+\mu\, \rangle, w) + \dots 
\ee

More generally, let $V$ be a finite-dimensional complex vector space, say of dimension $n$, with symmetric bilinear form $B: V \times V \rightarrow \C$, and fix a basis $\{v^1, \dots, v^n\}$ of $V$. Then the Heisenberg VOA associated to $(V, B)$, which we denote compactly as $H_{\{v^i\}}$, is strongly and freely generated by fields $J^i$ for $i=1, \dots, n$ with OPE
\be
H_{\{v^i\}}:\quad J^i(z) J^j(w) \sim \frac{B(v^i, v^j)}{(z-w)^2}\,. \label{JVB-OPE}
\ee
We also introduce formal fields $v^i(z)$, obeying $J^i(z)=\pd v^i(z)$, with a non-analytic OPE
\be
v^i(z)v^j(w)\sim B(v^i,v^j)\log (z-w)\,, \label{vv-OPE}
\ee
which implies \eqref{JVB-OPE}. The $v^i(z)$ themselves are not part of the Heisenberg VOA, but they provide a useful way to describe modules.

Fock modules (\emph{a.k.a.} Verma modules) for the generalized Heisenberg VOA $H_{\{v^i\}}$ are in one-to-one correspondence with linear maps $V\to \C$. In all applications in this paper, the bilinear form $B$ will be non-degenerate, and can thus be used to establish an isomorphism between linear maps $V\to \C$ and elements of $V$ itself.
We will mainly use the latter to describe Fock modules.

Given an element $\lambda\in V$, with associated map $B(\lambda,-):V\to \C$, there is a unique Fock module denoted
\be \Fock_\lambda \qquad \text{or (for clarity)}\quad \Fock_\lambda^{v^1,...,v^n}\,, \ee
generated by a highest-weight state $|\lambda\rangle$ that satisfies
\be J_0^i|\lambda\rangle  = B(\lambda,v^i)|\lambda\rangle\,,\qquad J_{n>0}^i|\lambda\rangle = 0\,, \ee
and on which the $J_{n<0}^i$ act freely. The vacuum module of the Heisenberg VOA is simply $\Fock_0$.
If we expand $\lambda=\sum_i\lambda_iv^i$ and correspondingly set $\lambda(z):= \sum_i \lambda_i v^i(z)$, then we can formally express the vertex operator corresponding to the highest-weight state $|\lambda\rangle$ as 
\be 
Y(|\, \lambda\, \rangle, z) \,=\, \normord{e^{\lambda (z)}} \,.
\ee
The OPE between vertex operators follows from \eqref{vv-OPE}. In particular,
\be \normord{e^{\lambda(z)}} \normord{e^{\eta(w)}} \;\sim\; (z-w)^{B(\lambda,\eta)} \normord{e^{\lambda(w)+\eta(w)}} + \ldots\,. \label{lattice-OPE}\ee
Thus the fusion rules are $\Fock_\lambda\otimes_{H_{v^i}} \Fock_\eta = \Fock_{\lambda+\eta}$.

We finally note that, fixing an orthogonal basis $v^1,...,v^n$ of $V$, there is a decomposition 
\be H_{v^1,...,v^n}\, \cong\, \bigotimes_i H_{v^i}\,. \ee 
Correspondingly, letting $\lambda = \sum_i \lambda_i v^i$, there is a decomposition of Fock modules
\be \CF_\lambda^{v^1,...,v^n} \,\cong\, \bigotimes_i \Fock_{\lambda_iv^i}^{v^i}\,.\ee

\subsubsection{Free fermions}
\label{sec:FF}

Let us choose the basis $v^1,...,v^n$ such that $B(v^i, v^j)$ is real for all pairs $(v^i, v^j)$, and consider the subcategory of those Fock modules $\Fock_\lambda$ that have the property that all $\lambda_i$ are real. This category is a braided tensor category \cite{creutzig2019schur}. Let $L \subset V$ be a lattice, meaning a $\mathbb Z$-submodule of $V$, with the property that $B$ restricted to $L$ is integral. 
Then 
\be
\CV_L := \bigoplus_{\lambda \in L} \Fock_\lambda
\ee
is itself a vertex superalgebra, the lattice VOA of the lattice $L$;  and if $L$ is even, then it is actually a vertex algebra. 

In particular, in rank one with $V=\C\langle v\rangle$ generated by a vector with $B(v,v)=1$, choosing $L = \Z\langle v\rangle \simeq \Z$ gives rise to a vertex algebra
\be \CV_\Z \,\simeq\, \mc{V}_{bc} \label{bose-fermi} \ee
strongly generated by a pair of free fermions. This is the classic bose-fermi correspondence. The two fermionic generators may be chosen as
\be b(z) = Y(\vert 1\rangle, z) = \normord{e^{v(z)}}\,,\qquad c(z) = Y(\vert -1\rangle, z) = \normord{e^{-v(z)}}\,, \ee
with OPE
\be
b(z)c(w) \sim \frac{1}{(z-w)}\,.
\ee

The free fermions $\mc{V}_{bc}$ are a \emph{holomorphic} VOA, in the sense that the VOA itself is the only simple module and every module is completely reducible. Thus the module category of $\mc{V}_{bc}$ is isomorphic to the ``trivial'' category of vector spaces.
 (The terminology ``holomorphic VOA,'' sometimes also called a ``holomorphic CFT,'' originates from the fact that all spaces of conformal blocks are automatically one-dimensional, implying that the VOA itself has well-defined partition functions in any genus, and carries the structure of a full CFT. For example, the free-fermion VOA $\mc{V}_{bc}$ is equivalent to the well-defined physical CFT containing a free complex-valued 2d chiral fermion.)

\subsubsection{Symplectic fermions and the singlet VOA}
\label{sec:SF}

The symplectic fermion vertex algebra $\CV_{SF}$ may be defined as the vertex algebra strongly generated by two fermionic fields $\chi^\pm(z)$ with OPE
\be \CV_{SF}:\quad \chi^+(z)\chi^-(w) \,\sim\, \frac{1}{(z-w)^2} \label{SF-OPE} \ee
More generally, given any complex vector space $W$ with a non-degenerate anti-symmetric bilinear form $\Omega:W\times W\to \C$ (a symplectic form), there is an asociated symplectic fermion VOA $\CV_{SF}^W$. Given a basis $\{\chi^i\}$ for $W$, the VOA is generated by fermionic fields $\{\chi^i(z)\}$ with
\be \CV_{SF}^W:\quad \chi^i(z)\chi^j(w) \,\sim\, \frac{\Omega(\chi^i,\chi^j)}{(z-w)^2}\,.  \ee

Symplectic fermions can be embedded in free fermions (and thus in a lattice VOA), as the kernel of certain screening charges. Define the screening operator $S(z):\mc{V}_{bc}\to \mc{V}_{bc}(\!(z)\!)$ by
\be S(z):= b(z)\,, \ee
and the ``screening charge'' $S_0:\mc{V}_{bc} \to \mc{V}_{bc}$ by
\be S_0 := \frac{1}{2\pi i}\oint S(z)dz = b_0\,. \ee
In this case, this is just the zero-mode of the fermionic field $b(z) = \sum_{n \in\Z}  b_n z^{-n-1}$.

The kernel of $S_0$ is simply the subalgebra generated by $b(z)$ and $\partial c(z)$, since $b_0$ commutes with all modes of $c(z)$ except its zero-mode. Letting $\chi^+(z) = b(z)$ and $\chi^-(z)=\partial c(z)$, we find that $\chi^\pm$ satisfy the symplectic fermion OPE \eqref{SF-OPE}. Thus
\be \CV_{SF} \,\cong\, \text{ker}(S_0:\mc{V}_{bc}\to \mc{V}_{bc})\,. \label{SF-screen} \ee

There is an action of $\C^*$ on both free fermions and symplectic fermions, such that $b,\chi^+$ have weight $+1$ and $c,\chi^-$ have weight $-1$. This makes $\mc{V}_{bc}$ and $\CV_{SF}$ $\Z$-graded vertex algebras. Decomposing symplectic fermions (as a vector space) into graded components
\be\label{eqSFfreefield}
\CV_{SF} = \bigoplus_{\mu \in \Z} M_\mu\,,
\ee
we find that the degree-zero subspace $M:=M_0$ (\emph{a.k.a.} the $\C^*$ orbifold of $\CV_{SF}$) is a vertex algebra itself, while the other components $M_n$ are modules for $M$. Conversely, $\CV_{SF}$ is an extension of $M$ by the modules $\{M_\mu\}_{\mu\in \Z}$.

The vertex algebra $M=M_0$ is known as the $p=2$ singlet algebra.
It contains the fields with equal numbers of $\chi^+$ and $\chi^-$, such as $\normord{\chi^+\chi^-}$, $\normord{\chi^+\partial^a\chi^-}$, $\normord{}\chi^+\partial^a \chi^+\chi^-\partial^b\chi^-$, etc.
The modules $M_\mu$ that appear in the decomposition of symplectic fermions are simple currents with fusion rules
\be
M_\mu \otimes_M M_\nu = M_{\mu+\nu}\,.
\ee
The singlet algebra has many other modules, however; its full representation theory is rather complicated and has only been completely understood in the past year \cite{Creutzig:singlet1, Creutzig:singlet2}.

Later in the paper we will encounter multiple copies of symplectic fermions and singlet algebras/modules. We summarize some notation and relations.  By combining the relation to free fermions \eqref{SF-screen} with the bose-fermi correspondence \eqref{bose-fermi}, we find that $n$ symplectic fermions are embedded in a rank-$n$ lattice VOA
\be \CV_{SF}^{\otimes n} \hookrightarrow \CV_{\Z^n}\,,\qquad \chi_+^i \mapsto\,\normord{e^{v^i}}\,,\quad \chi_-^i \mapsto\, -\normord{\pd v^i\,e^{-v^i}}\,, \ee
where $\CV_{\Z^n}$ is the extension of the Heisenberg algebra $H_{v^1,...,v^n}$ with $B(v^i,v^j)=\delta^{ij}$ by Fock modules $\CF_{\mu\cdot v}$ for all $\mu\in \Z^n$. The embedding is the kernel of screening operators
\be \CV_{SF}^{\otimes n} = \bigcap_{i=1}^n \text{ker}\, S^i_0\big|_{\CV_{\Z^n}}\,,\qquad S^i(z) = \frac{1}{2\pi i} \oint \normord e^{v^i(z)}\,. \ee
Furthermore, there is a $(\C^*)^n$ action on $\CV_{SF}^{\otimes n}$, induced from the $(\C^*)^n$ action on $\CV_{\Z^n}$ under which $\normord e^{\mu\cdot v}$ has charge $\mu\in \Z^n$. (And so in particular $\chi_\pm^i$ have charges $(0,...,0,\underset{i}{\pm 1},0,...,0)$.) We denote the weight spaces of this action $M_{\mu\cdot v}^{\{v^i\}}$ or simply $M_{\mu\cdot v}$, with
\be \CV_{SF}^{\otimes n} = \bigoplus_{\mu\in \Z^n} M_{\mu}\,, \qquad M_{\mu\cdot v} = \bigcap_{i=1}^n \text{ker}\, S^i_0\big|_{\CF_{\mu\cdot v}}\,. \label{SF-ff-n} \ee
Here $M_0^{\{v^i\}}\cong \bigoplus_{i=1}^n M_0^{v^i}$ is $n$ copies of the singlet VOA, and each of the $M_{\mu\cdot v}$'s are simple currents thereof.

\subsubsection{Symplectic bosons}
\label{sec:betagamma}

The basic symplectic boson VOA $\mc{V}_{\beta\gamma}$, \emph{a.k.a.} a beta-gamma system, is strongly generated by two bosonic fields $\beta(z)$, $\gamma(z)$ with OPE
\be \mc{V}_{\beta\gamma}:\quad \beta(z)\gamma(w) \,\sim\, \frac{-1}{z-w}\,. \label{betagamma-OPE} \ee
More generally, given a symplectic vector space $(W,\Omega)$, there is an associated symplectic boson VOA $\mc{V}_{\beta\gamma}^{W}$. Given a basis $\{\beta^i\}$ for $W$, the VOA is generated by fields $\{\beta^i(z)\}$ with OPE
\be  \CV_{\beta\gamma[W]}: \quad \beta^i(z)\beta^j(w) \,\sim\, \frac{\Omega(\beta^i,\beta^j)}{z-w}\,.\ee
Symplectic bosons are closely related to the other free field VOA's above, in several interesting ways.

To begin, let $V=\C\langle \eta\rangle$ be the one-dimensional vector space with negative inner product $B(\eta,\eta)=-1$, and let $\sqrt{-1}\Z := \Z\langle \eta\rangle$ denote the integer lattice therein. Correspondingly, we have a Heisenberg VOA $\CH_\eta$ and its lattice extension
\be \CV_{\sqrt{-1}\Z} =  \bigoplus_{n \in \Z} \Fock_n\,,\ee
where each Fock module $\Fock_n$ is generated by $\normord{e^{n\eta}}$. Now consider the $\C^*$ action on $\CV_{\sqrt{-1}\Z}$ under which the subspace $\CF_n$ has weight $-n$. Also recall the $\C^*$ action on symplectic fermions in \eqref{eqSFfreefield}. The invariant part of the tensor-product VOA $\CV_{SF}\otimes \CV_{\sqrt{-1}\Z}$ under the diagonal $\C^*$ action is
\be
(\CV_{SF} \otimes \CV_{\sqrt{-1}\Z})^{\C^*} = \bigoplus_{n \in \Z} M_n \otimes \Fock_n\,, \label{betagamma-MF0}
\ee
and it is generated by fields
\be
\beta = \chi^+ \otimes \normord{e^\eta}\quad \text{and} \quad 
\gamma = -\chi^- \otimes \normord{e^{-\eta}}
\ee
precisely satisfying the symplectic boson OPE \ref{betagamma-OPE}. Thus
\be (\CV_{SF} \otimes \CV_{\sqrt{-1}\Z})^{\C^*} \simeq \mc{V}_{\beta\gamma}\,. \label{betagamma-MF} \ee

Alternatively, by combining the embedding \eqref{SF-screen} of symplectic fermions in a lattice VOA with \eqref{betagamma-MF}, we obtain a well-known free field realization of symplectic bosons, \emph{cf.} \cite{allen2020bosonic}. Let $H_{\phi,\eta}$ be the rank-two Heisenberg algebra corresponding to a two-dimensional vector space with basis $\{\phi,\eta\}$ and inner product $B(\phi,\phi)=1$, $B(\eta,\eta)=-1$, $B(\phi,\eta)=0$. Consider the one-dimensional lattice $L=\Z\langle \phi+\eta\rangle$ and the corresponding lattice VOA
\be \CV_L = \bigoplus_{n\in \Z} \Fock_{n(\phi+\eta)}\,. \label{Lphieta}  \ee
There is an embedding $\mc{V}_{\beta\gamma}\hookrightarrow \CV_L$ given by 
\be \beta\,\mapsto\, e^{\phi+\eta}\,,\quad \gamma\,\mapsto\, \normord{\partial\phi\, e^{-\phi-\eta}}\,. \ee

To characterize this as the kernel of a screening charge, we note that the lattice VOA $\CV_L$ has modules $\CV_{L,k\phi}$ defined by lifting the Fock modules $\Fock_{k\phi}$ of the Heisenberg VOA to the lattice $\CV_L$. Explicitly,
\be \CV_{L,k\phi} = \bigoplus_{n\in \Z} \Fock_{(k+n)\phi+n\eta}\,. \ee
Consider the intertwining operator $S(z):\CV_{L,k\phi}\to \CV_{L,(k+1)\phi}(\!(z)\!)$ defined by
\be S(z) \,=\,  \normord{e^{\phi(z)}} \ee
and corresponding screening charges $S_0 = \frac{1}{2\pi i} \oint S(z)\,dz : \CV_{L,k\phi}\to \CV_{L,(k+1)\phi}$. Then the embedding of symplectic bosons in $\CV_L$ coincides with the kernel
\be \mc{V}_{\beta\gamma} \,=\,\text{ker}(S_0: \CV_L \to \CV_{L,\phi} )\,. \label{betagamma-screen} \ee
This follows from decomposing $\Fock_{n(\phi+\eta)}=\Fock_{n\phi}\otimes \Fock_{n\eta}$ in \eqref{Lphieta} as modules for two Heisenberg VOA's $H_\phi\otimes H_\eta$, and combining the decomposition with \eqref{SF-screen} and \eqref{betagamma-MF0}.

\subsubsection{BRST cohomology}
\label{sec:BRST}

Conversely, we may go back from symplectic bosons to symplectic fermions using BRST cohomology.  By BRST cohomology we mean relative semiinfinite Lie algebra cohomology as in \cite{Frenkel:BRST}; and the version that we need here is exactly the one used in \cite{Creutzig:BRST}.

Let $\CV$ be a vertex algebra with an internal $\C^*$ Kac-Moody action at level zero. In other words, $\CV$ contains a field $J(z)$ (the Kac-Moody current) that has non-singular OPE with itself, and generates an action of the loop group $\C(\!(z)\!)^*$ on $\CV$ by sending:
\be \text{for $\alpha(z)\in \C(\!(z)\!)$,}\quad v\,\mapsto\, \frac{1}{2\pi i}\oint \alpha(z)J(z)\cdot v\,. \ee
Intuitively, BRST cohomology takes the symplectic quotient of $\CV$ by the $\C^*(\!(z)\!)$ action --- setting the current $J(z)$ to zero and taking $\C(\!(z)\!)^*$ invariants --- in a derived way, \emph{i.e.} by expressing this quotient as the zeroth cohomology of a complex. Concretely, let $\mc{V}_{bc}$ be a free-fermion VOA, $Z$-graded such that $c$ has degree 1 and $b$ has degree $-1$. The tensor product $\CV\otimes \mc{V}_{bc}$ inherits this grading, and has a differential
\be Q_{\rm BRST} := \frac{1}{2\pi i} \oint c(z)\,J(z)\,dz \ee
of degree $+1$. Let the \emph{relative complex} $(\CV\otimes \mc{V}_{bc})^{\rm rel}$ be the subspace of $\CV\otimes \mc{V}_{bc}$ annihilated by the zero-modes $b_0$ and $J_0$. Then one defines BRST cohomology as the $Q_{\rm BRST}$-cohomology of the relative complex, denoted
\be H_{\rm BRST}(\CV) := H^\bullet\big( (\CV\otimes \mc{V}_{bc})^{\rm rel},Q_{\rm BRST}\big)\,.\ee
Note that $Q_{\rm BRST}$ sends $b(z)\mapsto J(z)$ (whence $J(z)$ is effectively set to zero in cohomology); it also sends any element of $\CV$ to its image under the $GL\big(1,\C(\!(z)\!)\big)$ action with generator $c(z)$ (whence cohomology also takes invariants for the action).

Similarly, if $M$ is any module for $\CV$, its BRST cohomology is defined by
\be H_{\rm BRST}(M) := H^\bullet\big( (M\otimes \mc{V}_{bc})^{\rm rel},Q_{\rm BRST}\big)\,, \label{def-BRST}\ee
where again $(M\otimes \mc{V}_{bc})^{\rm rel}$ denotes the subspace of $M\otimes \mc{V}_{bc}$ annihilated by $b_0$ and $J_0$.

Going from symplectic bosons to symplectic fermions uses a particularly well-behaved instance of BRST cohomology. Let $H_{\phi,\eta}$ be the rank-two Heisenberg algebra associated to vectors $\phi,\eta$ of norms $+1,-1$, respectively. The field $J(z)=\partial(\phi+\eta)$ is a level-zero $\C^*$ Kac-Moody current. Consider the Fock module $\Fock_{\lambda\phi+\mu\eta}$ (here $\lambda,\eta\in \C$). Its BRST cohomology has the property that
\be \label{BRST-triv0}
H^i_{\text{BRST}}(\Fock_{\lambda\phi+\eta\mu}) 
= \delta_{i, 0} \delta_{\lambda- \mu, 0} \C [|\lambda\phi+\mu\eta\rangle]\,.
\ee
In other words, the cohomology vanishes unless $\lambda-\mu=0$, in which case the cohomology is one-dimensional and given by the class of the highest-weight vector.

Now, the symplectic boson VOA $\mc{V}_{\beta\gamma}$ has a current $J_{\beta\gamma}\,=-\,\normord{\beta\gamma}$ at level $-1$ ($J_{\beta\gamma}$ generates a Heisenberg subalgebra $H_\eta$ with inner product $-1$). The free fermion VOA $\mc{V}_{bc}$ has a current $J_{bc}\,=\,\normord{bc}$ at level 1 (generating a Heisenberg subalgebra $H_\phi$ with norm $+1$). Therefore, their tensor product $\mc{V}_{bc}\otimes \mc{V}_{\beta\gamma}$ has a diagonal current $J=J_{\beta\gamma}+J_{bc}$ at level zero. Using the bose-fermi correspondence and \eqref{betagamma-MF0}, we know that as modules for $H_{\phi,\eta}=H_\phi\otimes H_\eta$ we have
\be \mc{V}_{bc}\otimes \mc{V}_{\beta\gamma} = \bigoplus_{m,n\in \Z} \Fock^\phi_{n\phi}\otimes \Fock^\eta_{m\eta}\otimes M_m\,. \ee
Taking BRST cohomology for the diagonal $\C^*$ action, we get
\begin{equation}
\begin{split}
H_{\text{BRST}}(\mc{V}_{bc} \otimes \mc{V}_{\beta\gamma}) &= \bigoplus_{n, m \in \Z } H_{\text{BRST}}\left(\Fock^\phi_{n\phi} \otimes  \Fock^\eta_{m\eta} \right) \otimes M_{m} \\
&= \bigoplus_{n, m \in \Z }  \delta_{n-m, 0}\C \otimes M_{m} \\
&= \bigoplus_{m \in \Z} M_m \\
&= \CV_{SF}.
\end{split}    
\end{equation}

\subsubsection{Affine $\mathfrak{gl}(1|1)$}
\label{sec:gl11}

There is in fact one more player that we can introduce, the affine VOA of $\mathfrak{gl}(1|1)$. It is the diagonal $\C^*$ orbifold of $\mc{V}_{bc} \otimes \mc{V}_{\beta\gamma}$, that is 
\be
\widehat{\mathfrak{gl}(1|1)} \,\cong\, (\mc{V}_{bc} \otimes \mc{V}_{\beta\gamma})^{\C^*} =
\bigoplus_{n \in \Z } \Fock^{\phi}_{n\phi} \otimes  \Fock^{\eta}_{n\eta} \otimes M_{n}
\ee
Here the level $k$ can be any non-zero number and in particular it can always be set to one. 
It is generated by the fields $N = :bc:, \psi^+ = :\beta b:, \psi^- = -:\gamma c:$ and $E = :bc: -:\beta\gamma:$.

All the relations between the different free field algebras above have been explored in detail in \cite{creutzig2009gl, Creutzig:2011cu}. 
The representation theory of $\mc{V}_{\beta\gamma}$ is worked out in \cite{allen2020bosonic}; the one for $\widehat{\mathfrak{gl}(1|1)}$ in \cite{creutzig2020tensor}, and the one for the singlet in \cite{Creutzig:singlet1, Creutzig:singlet2}. All the module categories are non-finite and non-semisimple ribbon (super)categories.
Orbifolds, simple current extensions, and BRST cohomologies provide nice functors between representation categories, as explained in \cite{Creutzig:BRST, Creutzig:2022ugv}.

\subsection{Spectral flow}
\label{sec:spectralflow}

In this final introductory section, we review the idea of spectral-flow automorphisms, and spectral-flow modules of a VOA, which will play a central role in many of our constructions. Spectral flow is associated with abelian Kac-Moody symmetries of VOA's; loosely speaking, it mixes the Kac-Moody symmetry with conformal symmetry. 
The basic idea appeared in physics in the 80's, in particular in the context of of worldsheet superstring theory. In the mathematical theory of VOA's, spectral flow is implemented by Haisheng Li's $\Delta$-operator \cite{Li:1997}.

Let $\CV$ be a VOA that has a Heisenberg subVOA, say of rank $n$ with associated bilinear form $B$ as in \eqref{JVB-OPE}. (The Heisenberg subVOA is another name for an abelian current algebra, or an abelian Kac-Moody symemtry.)
The Heisenberg VOA has a huge group of automorphisms. In particular, for any vector $\ell = (\ell_1, \dots, \ell_n)\in \C^n$, there is a spectral automorphism $\sigma^\ell$ that acts on the modes of the Heisenberg VOA as
\begin{equation} \label{specflow-zero}
    \sigma^\ell(J^i_n) = J^i_n + \delta_{n,0}\ell_i\,,  
\end{equation}
doing nothing but shifting the zero-modes by a scalar. Not all of these automorphisms lift from the Heisenberg subVOA to an automorphism of the mode algebra of the full VOA; this depends on the way that $\CV$ is graded by the $J^i_0$. The automorphisms that do lift are called spectral flows automorphisms of $\CV$.

For example, if $\CV$ is a lattice VOA $\CV_L$, there are spectral-flow automorphism so long as $\ell$ lies in the lattice $L'$ that is dual to $L$.  If $\CV$ is an affine VOA, then spectral-flow automorphisms correspond to coweights. The example of $\mathfrak{sl}(2)$ is instructive; it is discussed in detail in section 2 of \cite{Creutzig:2013yca}.

If $\sigma^\ell$ is an automorphism of the algebra of modes of $\CV$, then to any $\CV$-module $M$ one defines the \emph{spectral-flow module} $\sigma^\ell(M)$ as follows. The underlying vector space of $\sigma^\ell(M)$ is isomorphic to $M$, the isomorphism mapping $m \in M$ to $\sigma^\ell(m) \in \sigma^\ell(M)$;
but an element $x$ of the mode algebra of $\CV$ acts as
\begin{equation}
x\cdot \sigma^\ell(m)= \sigma^\ell(\sigma^{-\ell}(x)\cdot m)\,.
\end{equation}

Haisheng Li's $\Delta$-operator \cite{Li:1997}, implementing spectral flow in a mathematical context,  has several nice properties.  By Proposition 2.11 of \cite{Li:1997} together with skew-symmetry of intertwining operators, spectral flow respects fusion:
\begin{equation}
    \sigma^\ell(M) \otimes_\CV \sigma^{\ell'}(M') \cong \sigma^{\ell+\ell'}(M \otimes_\CV M').
    \end{equation}
In particular, the spectral flow image of the VOA itself, $\sigma^\ell(\CV)$ is always a simple current, with fusion rules
\begin{equation} \label{spec-fusion}
\sigma^\ell(\CV) \otimes_\CV \sigma^{\ell'}(\CV) \cong \sigma^{\ell+\ell'}(\CV), \qquad
\sigma^\ell(\CV) \otimes_\CV M \cong \sigma^{\ell}(M).
   \end{equation}
Another property is that spectral flow is exact, \ie\ it maps simple to simples, non-split short exact sequences to non-split short-exact sequences, Loewy diagrams to Loewy diagrmas and so on --- see Proposition 2.5 of \cite{creutzig2019schur}.

\begin{Exp}\label{exp:heisenberg}
When $\CV$ itself is a Heisenberg VOA, spectral flow acts on its own Fock modules. One has that $\sigma^{-\lambda}(\mathcal F_\mu) \cong \mathcal F_{\lambda+\mu}$ in particular
$\sigma^{-\lambda}(H_{v^i}) \cong \mathcal F_\lambda$ and the well-known fusion rules 
\[
\mathcal F_\lambda \otimes_{H_{v^i}} \mathcal F_\mu \cong \sigma^{-\lambda}(H_{v^i}) \otimes_{H_{v^i}}  \sigma^{-\mu}(H_{v^i}) \cong \sigma^{-\lambda-\mu}(H_{v^i}) \cong \mathcal F_{\lambda+\mu}
\]
follow immediately.
\end{Exp}

\begin{Exp}\label{exp:affineVOA}
 The example of the affine VOA of $\mathfrak{gl}(1|1)$
 is found in section 3.2 of \cite{creutzig2009gl}. The spectral-flow automorphism acts on the modes as
   \[
  \sigma^\ell(N_r) = N_r, \qquad \sigma^\ell(E_r) = E_r - \ell k \delta_{r, 0}, \qquad \sigma^\ell(\Psi^\pm_r) = \psi^\pm_{r \mp \ell}.
  \]
  This illustrates that spectral flow changes the mode labels, \emph{i.e.} it doesnot leave the horizontal subalgebra of  $\widehat{\mathfrak{gl}(1|1)}$ invariant.
\end{Exp}

\begin{Rmk}
    We will denote spectral-flow automorphisms by both $\sigma$ or $\Sigma$ in this paper, in order to distinguish different contexts in which they are used.
\end{Rmk}

\section{Boundary VOA for $\CT_{\rho}^A$}
\label{secNeumannA}

In this section, we develop the structure of boundary VOA's for abelian 3d $\CN=4$ theories in the topological A twist. We begin by reviewing the physical definition of holomorphic boundary conditions, following \cite{CG}, and describe their continuous global symmetries. We then consider the VOA's supported on holomorphic boundary conditions, which in \cite{CG} were defined using abelian BRST cohomology. We use free field realizations to show that the BRST cohomologies are supported in degree zero, and explicitly relate them to extensions of products of symplectic fermions, symplectic bosons, and singlet algebras.

We work in the conventions of Section \ref{sec:Trho}: we consider abelian theories $\CT_\rho$, with gauge group $G=U(1)^r$
 acting faithfully on hypermultiplet representation $T^*V = T^*\C^n$, with weights encoded in the $n\times r$ charge matrix $\rho$, which fits into the split exact sequence \eqref{split-SES}. We denote $\rho^\vee:=\tau^T$. We also choose a co-splitting as in \eqref{mirror-matrix}, meaning a map $\sigma^\vee:\Z^r\to\Z^n$ such that $\rho^T\circ \sigma^\vee = \text{id}_{\Z^r}$, satisfying the additional condition that $\sigma^{\vee,T}\sigma = 0$. Altogether, we have
\be \raisebox{-.2in}{\includegraphics[width=2.3in]{split2}} \label{split-SES2} \ee

\subsection{Neumann boundary condition and its symmetry}
\label{sec:Nbc}

The physical 3d $\CN=4$ gauge theory $\CT_\rho$ has a half-BPS boundary condition that preserves 2d $\CN=(0,4)$ supersymmetry and is Neumman-like, in that it imposes a Neumann b.c. on the gauge fields and on the hypermultiplet scalars \cite{CG}. The boundary condition is intrinsically chiral; it preserves ``right-handed'' hypermultiplet fermions (transforming in representation $T^*V$) and ``left-handed'' vectormultiplet fermions (transforming in the adjoint representation $\mathfrak g$, which of course is trivial since $G$ is abelian). This leads to a boundary gauge anomaly that must be cancelled, by introducing additional boundary matter \cite{dimofte2018dual}  (see also \cite{Gadde:2013sca,ChungOkazaki}). There is a canonical choice of boundary matter: left-handed complex fermions in representation $V$.

To analyze the anomaly cancellation, and more generally the symmetry structure at the boundary, we introduce some notation. We schematically denote
\begin{itemize}
\item $\lambda_-,\lambda_-':$ $2r$ bulk vectormultiplet fermions (gauginos in representation $T^*\mathfrak g_\C\simeq \C^r\oplus \C^r$) that survive at the boundary; they have negative (left-handed) boundary chirality
\item $X,Y$: $2n$ bulk hypermultiplet scalars (in representations $V,V^*$) that have a Neumann b.c., and thus survive at the boundary
\item $\psi_+^X,\psi_+^Y$: $2n$ bulk hypermultiplet fermions (in representations $V,V^*$) that survive at the boundary; they have positive (right-handed) boundary chirality
\item $\Gamma_-$: $n$ 2d right-handed fermions (in representation $V$) added at the boundary; they have negative boundary chirality
\end{itemize}
In addition to the $G=U(1)^r$ bulk gauge symmetry, there is bulk $SU(2)_H\times SU(2)_C$ R-symmetry (we keep track of a maximal torus $U(1)_H\times U(1)_C$), and bulk $T_F=U(1)^{n-r}$ and $T_{\rm top}=U(1)^r$ flavor symmetry. Moreover, the boundary fermions $\Gamma_+$ have their own $T_B=U(1)^{n-r}$ flavor symmetry, independent of the bulk flavor symmetry. The charges of the above fields under these various symmetries are:
\be \label{Ncharges}
\begin{array}{c|c|c|c|c|c}
\text{symmetry} & \text{f.s.} & \lambda_-, \lambda_-' &  X,Y & \psi_+^X,\psi_+^Y & \Gamma_- \\\hline
G = U(1)^r & \mathbf{G} & 0, 0 & \rho, -\rho & \rho,-\rho &  \rho  \\
T_F = U(1)^{n-r} & \mathbf{F} & 0,0 & \sigma,-\sigma & \sigma,-\sigma &  \sigma \\
T_{\rm top} = U(1)^r & \mathbf{T} & 0,0 & 0,0&0,0& \sigma^\vee \\
T_B = U(1)^{n-r} & \mathbf{B} & 0,0 & 0,0 & 0,0 & \rho^\vee \\
U(1)_H & \mathbf{H} & 1,1 & 1,1 & 0,0 & 0 \\
U(1)_C & \mathbf{C} & 1,-1 & 0,0 & -1,-1 & 0 
\end{array}
\ee
Note that some choices are involved here. The $T_F$ charges of bulk fields depend on a choice of splitting of \eqref{split-SES}, and the $T_{\rm top}$ charges of boundary fermions depend on a choice of co-splitting $\sigma^\vee$.
 All other charges (in particular, charges of boundary fermions $\Gamma_+$ under all the other symmetries) are completely fixed by cancellation of boundary anomalies.

The boundary anomaly polynomial is a quadratic function of the field strengths of the various gauge and global symmetries, computed from UV data by the schematic formula \cite{dimofte2018dual}
\be \CI_\pd \;= (\text{bulk CS level})\;- \!\!\!\sum_{\text{fermions $\chi$ at the bdy}} d_\chi \epsilon_\chi [\text{charge}(\chi)\cdot (\text{field strengths})]^2\,, \label{anom-def} \ee
with $d_\chi = 1$ for purely 2d fermions and $d_\chi=\frac12$ for 3d fermions that survive at the boundary; and $\epsilon_\chi\in \pm1$ denoting the boundary chirality. We label the field strengths as in \eqref{Ncharges}, with $\mathbf{G} = (\mathbf{G}_1,...,\mathbf{G}_r)$, $\mathbf{F} = (\mathbf{F}_1,...,\mathbf{F}_r)$ (etc.) denoting vectors when the symmetry has rank $>1$. The bulk Chern-Simons (CS) contribution to \eqref{anom-def} is $-2\mb G\cdot \mb T$ (this defines the bulk topological symmetry, and determines topological charges of bulk monopole operators).
Then the anomaly polynomial of the Neumann b.c. for $\CT_\rho$ is
\begin{align} \CI_\pd &= -2\,\mb G\cdot \mb T+ \frac{r}{2}(\mb H+\mb C)^2 + \frac{r}{2}(\mb H-\mb C)^2 -\frac12|\rho\mb G+\sigma\mb F-\mb C|^2  -\frac12|-\rho\mb G-\sigma\mb F-\mb C|^2 \notag \\ & \qquad + |\rho\mb G+\sigma \mb F +\sigma^\vee \mb T +\rho^\vee\mb B|^2  \notag \\
&= r\mb H^2-(n-r)\mb C^2 +2\,\mb F\cdot \mb B +|\sigma^\vee\mb T+\rho^\vee\mb B|^2 \label{anom-N}
\end{align}

Crucially, the terms quadratic and linear in $\mb G$ cancel in the anomaly polynomial. This ensures, respectively, that there is no gauge anomaly or mixed gauge-global anomaly. The remaining terms in the anomaly polynomial encode boundary 't Hooft anomalies for the global symmetries.

\subsection{The VOA on the boundary}
\label{subsecNeumannAVOA}

The topological A-twist of the 3d bulk is compatible with a holomorphic twist of the Neumann boundary condition above, and leads to the boundary VOA of \cite{CG}. We review the definition of this VOA here; for further discussion see also \cite{garner2022vertex}.

We place the boundary along a complex $z,\bar z$ plane, at Euclidean time $t=0$. The boundary VOA is generated by
\begin{itemize}
\item $n$ beta-gamma systems $\mc{V}_{\beta\gamma}^{T^*V}=\mc{V}_{\beta\gamma}^{\otimes n}$, whose fields come from the restriction of the bulk hypermultiplet scalars $X,Y$ to the boundary
\be \gamma^{i}(z):=X^i(z,t=0)\,,\qquad \beta^{i}(z):= Y^i(z,t=0)\,,\qquad i=1,...,n \ee
\item $n$ free fermions $\mc{V}_{bc}^V=\mc{V}_{bc}^{\otimes n}$ whose fields come from the additional 2d fermions used to cancel the boundary gauge anomaly,
\be b^i(z) = \Gamma_-^i(z)\,,\qquad c^i(z) = \overline{\Gamma}^i_-(z)\,,\qquad i=1,...,n\,. \ee
\item $r$ free fermions $\mc{V}_{bc}^{\mathfrak g} = \mc{V}_{bc}^{\otimes r}$ whose fields come from bulk gauginos restricted to the boundary,
\be \pd c^a(z) = \lambda_-^a(z,t=0)\,,\qquad b_a(z)=\lambda_-'(z,t=0)\,,\qquad a=1,...,r\,. \label{partialc} \ee
\end{itemize}
The A-twisted QFT has a cohomological $\Z$ grading given by $U(1)_C$ charge. Comparing with \eqref{Ncharges}, we see that the degrees of generators are $|c^a|=1$, $|b_a|=-1$, $|\beta|=|\gamma|=|c^i|=|b^i|=0$.

Moreover, there is a nontrivial differential $Q$ on $\mc{V}_{\beta\gamma}^{T^*V}\otimes \mc{V}_{bc}^V\otimes \mc{V}_{bc}^{\mathfrak g}$ induced from the A-twist supercharge. It coincides with the BRST differential for the $\mathfrak g_\C=\mathfrak{gl}(1)^r$ action on $\mc{V}_{\beta\gamma}^{T^*V}\otimes \mc{V}_{bc}^V$. The currents generating the action are 
\be J_a = \sum_i\rho_{ia}(-\normord{\beta^i\gamma^i}+\normord{b^ic^i})\,, \ee
and they generate an affine $\mathfrak{gl}(1)^r$ Kac-Moody algebra at level zero. The differential is then
\be Q = \frac{1}{2\pi i}\oint dz\, \sum_a c^a J_a\,.\ee
In particular,
\be Q(c^a)=0\,,\quad Q(b_a) = J_a\,,\quad \begin{array}{c} Q(\beta^i)=-\sum_a \rho_{ia}c^a\beta^i\,,\quad Q(\gamma^i)= \sum_a \rho_{ia}c^a\gamma^i\,,\\[.2cm]
Q(b^i)= \sum_a \rho_{ia}c^ab^i\,,\quad Q(c^i)= -\sum_a \rho_{ia}c^ac^i\,. \end{array} \ee
Altogether, we find
\begin{Def}
The boundary VOA on a Neumann boundary condition in the A-twisted theory $\CT_{\rho}^A$ is given by the relative BRST cohomology
\begin{equation} \label{def-VA}
    \mc{V}_{\rho}^A:=H_{\rm BRST}\left(\mathfrak{gl}(1)^r, \mc{V}_{\beta\gamma}^{T^*V}\otimes \mc{V}_{bc}^{V}\right)=H^\bullet\left(\left(\mc{V}_{\beta\gamma}^{T^*V}\otimes \mc{V}_{bc}^{V}\otimes \mc{V}_{bc}^{\mathfrak g}\right)^{\mathrm{rel}}, Q\right)\,.
\end{equation}
\end{Def}
\noindent The appearance of relative cohomology, as defined in Section \ref{sec:BRST}, is due physically to the absence of the zero mode of $c^a$, \emph{cf}. \eqref{partialc}\,. For further explanation of this phenomenon, see \cite[Sec. 6.2.1]{CDG-chiral}.

\begin{Rmk} The fact that the  $\mathfrak{gl}(1)^r$ Kac-Moody algebra above has level zero is a consequence of cancellation of boundary gauge anomalies. The currents $J_a^{\beta\gamma}=-\sum_i\rho_{ia}\normord{\beta^i\gamma^i}$ acting on the beta-gamma system alone (the boundary values of bulk hypermultiplets) generate a Kac-Moody algebra at level $-\rho^T\rho$, meaning
\begin{equation}\label{eqAanomaly}
   J_a^{\beta\gamma}(z)J_b^{\beta\gamma}(w)\sim \frac{-\sum_i\rho_{ia} \rho^i{}_b}{(z-w)^2}\,.
\end{equation}
The level $\rho^T\rho$ is simply the contribution of bulk hypermultiplets to the boundary anomaly, \emph{cf}. the $-|\rho\mb G|^2$ term in \eqref{anom-N}. The currents $J_a^{bc} = \sum_i \rho_{ai} \normord{b^ic^i}$ acting on the 2d boundary fermions generate a Kac-Moody at level $+\rho^T\rho$, in agreement with the $+|\rho\mb G|^2$ contribution from the final term in \eqref{anom-N}.
\end{Rmk}

\subsubsection{Symmetries of $\CV_\rho^A$}
\label{sec:VA-sym}

The global symmetries of the Neumann boundary condition all induce symmetries of the boundary VOA $\CV_\rho^A$. However, exactly how these symmetries act depends on further properties of the topological A twist.

We already saw that the R-symmetry $U(1)_C$ gives the A-twisted boundary VOA a cohomological $\Z$ grading, under which $|c^a|=-|b_a|=|Q|=1$. Later we will prove that the BRST cohomology \eqref{def-VA} is supported entirely in cohomological degree zero, so ultimately $U(1)_C$ acts \emph{trivially} on $\CV_\rho^A$.

The R-symmetry $U(1)_H \subset SU(2)_H$ is used in the twisting homomorphism that defines the topological A-twist. It does not survive the twist. However, a diagonal combination of boundary rotations and $U(1)_H$ (twisted spin) defines a $\frac12 \Z$-valued conformal grading for the boundary VOA, with charges
\be \label{conf-A} |Q|=0\,,\quad |\pd|=|z^{-1}|=1\,,\qquad |\beta|=|\gamma|=|c^i|=|b^i|=\tfrac12\,,\quad |c^a|=0\,,\quad |b_a|=1\,.\ee 
Moreover, since the bulk stress tensor is trivialized ($Q$-exact) in the topological A-twist, there can be a conserved boundary stress tensor $L(z)$ whose mode $L_{0}$ acts with charges \eqref{conf-A}. The boundary stress tensor is given by 
\be L(z) =\frac{1}{2}\sum\limits_i \normord{\partial \beta_i \gamma^i}-\normord{\beta_i\partial \gamma^i}+\normord{\partial b_i c^i}+\normord{\partial c_i b^i}+\sum\limits_a\normord{\partial c_ab^a} \label{L-A} \ee
This commutes with the BRST operator $Q$, and induces a stress tensor on $\CV_\rho^A$.

The remaining flavor symmetry $T_F\times T_{\rm top}\times T_B$ all acts on $\CV_\rho^A$ as well. As $T_B$ acts exclusively on 2d boundary matter, it should be enhanced to a Kac-Moody action. Indeed, it is easy to identify $T_B$ currents
\be J^B_\alpha  = \sum_i (\rho^\vee)_{\alpha i} \normord{b^ic^i}\,,\qquad \alpha=1,...,n-r \ee
that commute (have trivial OPE) with all gauge currents $J_a$, and thus descend to currents in the BRST quotient $\CV_\rho^A$. Similarly, in the A-twist, the bulk current for $T_F$ is $Q$-exact, allowing the existence of a conserved boundary current that should generate a Kac-Moody action on $\CT_\rho^A$. (This mechanism is reviewed in greater detail later in Section \ref{sec:sym-spectral}.) The current for the boundarty $T_F$ Kac-Moody symmetry is
\be J^F_\alpha = \sum_i \sigma_{\alpha i}(-\normord{\beta^i\gamma^i}+\normord{b^ic^i})\,,\ee
and it again commutes with the gauge currents $J_a$\,.

In constrast, the bulk current for  $T_{\rm top}$ is \emph{not} exact in the A-twist. Indeed, the A-twist has bulk monopole operators that are nontrivially charged under $T_{\rm top}$. Correspondingly, we do not expect a Kac-Moody enhancement of the finite symmetry $T_{\rm top}$ acting on $\CV_\rho^A$. The lack of enhancement can be seen as follows. Prior to BRST reduction, the VOA has $T_{\rm top}$ currents
\be J^{\rm top}_a =\sum_i (\sigma^\vee)_{ai}\normord{b^ic^i}\,. \ee
However, these do not commute with the BRST operator, since
\be J^{\rm top}_a(z)J_b(w) \sim \frac{\delta_{ab}}{(z-w)^2}\,.\ee
(This calculation corresponds to the contribution $+2\,\mb G\cdot \mb T$ from the last term in \eqref{anom-N}, which is cancelled by the bulk Chern-Simons level $-2\,\mb G\cdot \mb T$.) Thus, the Kac-Moody action does not survive the BRST quotient, even though the finite $T_{\rm top}$ action, generated by $J_0^{\rm top}$, does.

In summary, the VOA $\CV_\rho^A$ has the following symmetries: 
\be
\begin{array}{r|c|c|c|c|c|c|c|c}
& \pd & Q& c^a & b_a &\beta^i & \gamma_i & b^i & c_i \\\hline
\text{cohomological}\; U(1)_H & 0 & 1 & 1 & -1 & 0 & 0 & 0 &0 \\
\text{Virasoro}\;\ni L_{0} & 1 & 0 & 0 & 1 & \tfrac12 & \tfrac12 & \tfrac12 & \tfrac12 \\
\text{Kac-Moody}\; \supset T_F & 0&0& 0&0 & -\sigma & \sigma & \sigma & -\sigma \\
\text{Kac-Moody}\; \supset T_B & 0&0&0&0 & 0&0&\rho^\vee & -\rho^\vee \\
\text{finite}\; T_{\rm top} & 0&0&0&0 &0&0 & \sigma^\vee & -\sigma^\vee
\end{array}
\ee

\subsection{BRST cohomology via free field realizations}\label{subsecBRSTAff}

To describe the VOA $\CV_\rho^A$ more explicitly, we will use the free field realization of $\mc{V}_{\beta\gamma}$ as an extensions of a Heisenberg VOA by its Fock modules, from Section \ref{sec:betagamma}, combined with the special cases of BRST cohomology reviewed in Section \ref{sec:BRST}.

Let $H_\phi,H_\eta$ be Heisenberg algebras defined by the vector space $\C^n$ with bases $\{\phi^i\}_{i=1}^n$ and $\{\eta^i\}_{i=1}^n$ and opposite pairings $B( \phi^i,\phi^j)=-B( \eta^i,\eta^j)=\delta^{ij}$. 
For $\mu\in \Z^n$, recall that $\CF_{\mu\cdot\phi}^\phi$ denotes the Fock module of $H_\phi$ generated by highest-weight vector $\normord{e^{\mu\cdot \phi}}$, and let $M_{\mu\cdot\eta}^\phi \subset \CF_{\mu\cdot\eta}^\phi$ be the subspace defined by
\be M_{\mu\cdot\phi}^\phi = \bigcap_{i=1}^n \text{ker}\, S^i_0\big|_{\CF_{\mu\cdot\phi}^\phi}\,,\qquad S^i(z) = \normord{ e^{\phi^i(z)}}\,. \ee
Recall from Section \ref{sec:SF} that $M_{0}^\phi$ is $n$ copies of the singlet VOA, and that the $M_{\mu\cdot\phi}^\phi$ are simple modules thereof. (Alternativelty, $M_{\mu\cdot\phi}^\phi$ is the charge-$\mu$ subspace of $n$ symplectic fermions.)
Then $n$ copies of the beta-gamma VOA have a decomposition 
\be \mc{V}_{\beta\gamma}^{T^*V}\,\cong\, \bigoplus_{\mu\in \Z^n} M_{\mu\cdot \phi}^\phi \otimes \CF_{\mu\cdot\eta}^\eta\,.\label{nbg-ff} \ee
Explicitly, the generators of the beta-gamma VOA are identified as
\be \beta^i(z) = \normord{e^{\phi^i(z)+\eta^i(z)}}\,,\qquad \gamma^i(z) = -\normord{\pd\phi^i(z)e^{-\phi^i(z)-\eta^i(z)}}\,. \ee

Let's also bosonize the $n$ boundary fermions $\mc{V}_{bc}^V = \mc{V}_{bc}^{V}$, by introducing a third Heisenberg algebra $H_{\theta}$ with vector space spanned by $\{\theta^i\}_{i=1}^n$ and positive inner product $B(\theta^i,\theta^j)=\delta^{ij}$. The $\mc{V}_{bc}^{\otimes n} \cong \bigoplus_{\lambda\in \Z^n} \CF_{\lambda\cdot\theta}^\theta$ and
\be \mc{V}_{\beta\gamma}^{T^*V}\otimes \mc{V}_{bc}^{V}\,\cong\, \bigoplus_{\mu,\lambda\in\Z^n} M_{\mu\cdot \phi}^\phi \otimes \CF_{\mu\cdot\eta}^\eta \otimes \CF_{\lambda\cdot\theta}^\theta\,. \label{bgbc-ff}\ee

\subsubsection{Decomposing the charge lattice}
\label{sec:VA-lattice}

We'd like to find the BRST cohomology of $\mathfrak{gl}(1)^r$ acting on this, as in \eqref{def-VA}. To do so, we use a  decomposition of the vector space $V=\C^n$ into a subspace spanned by the gauge charges, and its orthogonal complement:
\be V = \rho(\C^r)\oplus \rho^\vee(\C^{n-r})\,. \ee
These two subspaces are orthogonal under the standard Euclidean inner product, as a consequence of the exact sequence \eqref{split-SES2}.
Let $\pi:V \to \rho(\C^r)$ and $\pi^\perp: V\to \rho^\vee(\C^{n-r})$ denote the orthogonal projections. These orthogonal projections can be written using the matrices $\rho$ and $\rho^\vee$ as: $\pi=\rho (\rho^\T \rho)^{-1} \rho^\T $ and $\pi^\perp= \tau^\T (\tau \tau^\T)^{-1}\tau$. As projections, these operators satisfy $\pi^2=\pi$, $(\pi^\perp)^2=\pi^\perp$, and moreover $\pi+\pi^\perp=\mathrm{Id}$ and $\pi^\perp\cdot \pi=0$.

If we could similarly decompose the integral lattice $\Z^n$ as $\rho(\Z^r)\oplus \rho(\Z^{n-r})$, BRST cohomology would be quite easy to evaluate. However, we know from Proposition \ref{prop:H} in Section \ref{sec:H} that
\be \frac{\Z^n}{\rho(\Z^r)\oplus \rho^\vee(\Z^{n-r})} = H \ee
may be a nontrivial finite group, the symmetry fractionalization group of the theory $\CT_{\rho}$, whose elements label massive vacua.

To account for nontrivial $H$ we consider the image of $\Z^n$ under the orthogonal projections $\pi$ and $\pi^\perp$. Since $\pi(\Z^n)$ is a lattice inside $\rho(\C^r)$, we can take its inverse image under $\rho$ to obtain a lattice $N$ in $\C^r$. Doing this also for the orthogonal projection $\pi^\perp$, we obtain
\be N := \rho^{-1}(\pi(\Z^n)) \subset \C^r\,,\qquad N^\vee := \rho^{\vee,-1}(\pi^\perp(\Z^n)) \subset \C^{n-r}\,. \ee 
By Proposition \ref{prop:H} there are natural isomorphisms
\be N/\Z^r \simeq \Z^n/\big(\rho(\Z^r)\oplus\rho^\vee(\Z^{n-r})\big) \simeq  N^\vee/\Z^{n-r} \simeq H\,. \ee
For $\lambda\in N$ and $\lambda^\vee\in N^\vee$, denote $[\lambda],[\lambda^\vee] \in H$ the corresponding equivalence classes in $H$. We then have an orthogonal decomposition
\begin{subequations}  \label{Z-decomp}
\be \Z^n = \bigsqcup_{\footnotesize\begin{array}{c}\lambda\in N,\,\lambda^\vee\in N^\vee \\\text{s.t.}\; [\lambda]=[\lambda^\vee]\end{array}} \rho(\lambda)+\rho^\vee(\lambda^\vee) \qquad\text{with}\quad \rho(\lambda)\cdot \rho^\vee(\lambda^\vee)=0\,.\ee
Equivalently, if for each $h\in H$ we choose representatives $h_N\in N$ and $h_{N^\vee}\in N^\vee$ of its equivalence classes, we have
\be \Z^n = \bigsqcup_{\footnotesize\begin{array}{c}\lambda\in \Z^r,\,\lambda^\vee\in \Z^{n-r} \\ h \in H\end{array}} \rho(\lambda+h_N) + \rho^\vee(\lambda^\vee+h_{N^\vee})\,. \ee
\end{subequations}

\begin{Exp}
Consider the case:
\begin{equation}
    \rho=\left(\begin{array}{c}1 \\ 1\end{array}\right). 
\end{equation}
Namely the gauge group is $U(1)$ and the representation is $V=\mathbb{C}^2$ with weights $1,1$. In this case, $\rho^\vee$ is given by $\left(\begin{smallmatrix} 1\\-1\end{smallmatrix}\right)$, and $H \simeq \Z_2$ as in \eqref{H-SQEDn}. The images $\rho(\Z)$ and $\rho^\vee(\Z)$ in $\Z^2$ are shown on the left of Figure \ref{fig:lattices}. The projection $\pi(\Z^2)$ is shown on the right, and its inverse image is $N = \rho^{-1}(\pi(\Z^2)) = \frac12\Z$. Similarly, $N^\vee= \rho^{\vee,-1}(\pi^\perp(\Z^2)) = \frac12 \Z$\,. We can write every element $v=\left(\begin{smallmatrix} a\\ b\end{smallmatrix}\right)\in\Z^2$ uniquely as
\be  v = \rho(n)+\rho^\vee(m) \ee
where $n,m$ are both integers (corresponding to the class $0\in H$) or both half-integers (corresponding to the class $1\in H$). Explicitly, $n=\frac12(a+b)$ and $m=\frac12(a-b)$\,.
\end{Exp}

\begin{figure}[htb]
\centering
\includegraphics[width=5in]{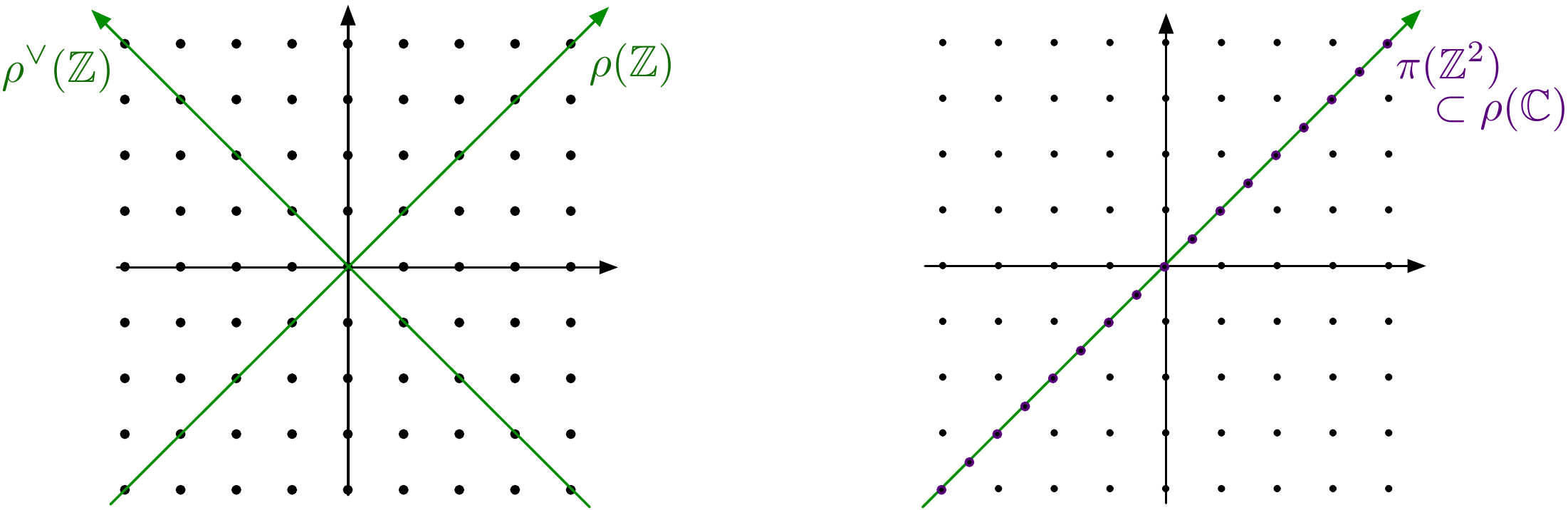}
\caption{The $\Z^2$ lattice of all hypermultiplet charges in gauge theory with $G=U(1)$, $V=\C^2$, and $\rho=\left(\begin{smallmatrix}1\\1\end{smallmatrix}\right)$. Left: the image of the gauge charge lattice $\rho(\Z)$ and its orthogonal complement $\rho^\vee(\Z)$. Right: the projection $\pi(\Z^2)$ of the full charge lattice to the space spanned by $\text{im}\,\rho$.}
\label{fig:lattices}
\end{figure}

\subsubsection{The BRST cohomology}
\label{sec:VA-BRST}

We will use the orthogonal decomposition of $\C^n$ above (and its integral lattice) to decompose the Heisenberg algebras (and their Fock modules) appearing in free field realizations.

For example, given a basis $\{\phi^i\}_{i=1}^n$ for $\C^n$ with inner product $B(\phi^i,\phi^j) = \delta^{ij}$, the decomposition $\C^n = \rho(\C^r)\oplus \rho^\vee(\C^{n-r})$ corresponds to choosing a new basis
\be \big\{\phi_0^a\big\}_{a=1}^{r} \cup \big\{\phi_\perp^\alpha\big\}_{\alpha=1}^{n-r}\,, \qquad \phi_0^a:=(\rho^T \phi)^a\,,\quad \phi_\perp^\alpha:=(\rho^{\vee T}\phi)^\alpha\,, \ee
where the $\{\phi_0^a\}$ are a basis for $\C^r$ and the $\{\phi_\perp^\alpha\}$ are a basis for $\C^{n-r}$, with induced inner products $\rho^T\rho$ and $\rho^{\vee T}\rho^\vee$, respectively:
\be B(\phi_0^a,\phi_0^b) = (\rho^T\rho)_{ab}\,,\qquad B(\phi_0^\alpha,\phi_0^\beta) = (\rho^{\vee T}\rho^\vee)_{\alpha\beta}\,.\ee
Also, by orthogonality \eqref{split-SES2}, $B(\phi_0^a,\phi_\perp^\alpha)=0$.
The Heisenberg algebra thus decomposes as
\be H_\phi \,\cong\, H_{\phi_0}\otimes H_{\phi_\perp}\,. \ee
For a Fock module labelled by $\lambda\in \Z^n$ we can now write
\be \lambda =\rho(\lambda_0)+\rho^\vee(\lambda^\vee)\,,\qquad \lambda_0\in N\,,\quad \lambda^\vee\in N^\vee\,, \quad [\lambda_0]=[\lambda^\vee]\in H\,, \ee
and thus the Fock module decomposes as
\be \CF_{\lambda\cdot\phi}^\phi \,\cong\, \CF_{\lambda_0\cdot\phi_0}^{\phi_0} \otimes \CF_{\lambda^\vee\cdot\phi_\perp}^{\phi_\perp} \ee

In exactly the same way, we decompose Heisenberg algebras
\be H_\eta\,\cong\, H_{\eta_0}\otimes H_{\eta_\perp}\,,\qquad H_\theta\,\cong\, H_{\theta_0}\otimes H_{\theta_\perp}\,,\ee
along with their Fock modules, the only difference being that in the $\eta$ algebras the inner products of basis vectors all have opposite signs.
Altogether, the free field realization of the beta-gamma and bc VOA's in \eqref{bgbc-ff} becomes
\be \mc{V}_{\beta\gamma}^{T^*V}\otimes \mc{V}_{bc}^V \;\cong \hspace{-.3cm}\bigoplus_{\substack{ \mu,\lambda\in N;\;\; \mu^\vee,\lambda^\vee\in N^\vee \\[.1cm] \text{s.t.}\, [\mu]=[\mu^\vee],\; [\lambda]=[\lambda^\vee]}} \hspace{-.3cm}
  M_{(\rho(\mu)+\rho^\vee(\mu^\vee))\cdot \phi}^\phi \otimes \CF_{\mu\cdot \eta_0}^{\eta_0} \otimes \CF_{\mu^\vee\cdot\eta_\perp}^{\eta_\perp} \otimes \CF_{\lambda\cdot\theta_0}^{\theta_0} \otimes \CF_{\lambda^\vee\cdot\theta_\perp}^{\theta_\perp} \ee

In this free-field realization, the $\mathfrak{gl}(1)^r$ gauge currents for BRST reduction are given by
\begin{align}
  J_a &= \sum_i \rho_{ia} \left(-\normord{\beta^i\gamma^i}+\normord{b^ic^i} \right)=\sum\limits_i\rho_{ia}\partial \eta^i(z)+\sum\limits_i \rho_{ia}\partial \theta^i(z) \notag \\
  &=  \pd\eta^a_0(z) + \pd\theta^a_0(z)\,.
\end{align}
These commute with $M_{(\rho(\mu)+\rho^\vee(\mu^\vee))\cdot \phi}^\phi$ since they do not involve $\phi$ at all, and also commute with $\CF_{\mu^\vee\cdot\eta_\perp}^{\eta_\perp}$ and $ \CF_{\lambda^\vee\cdot\theta_\perp}^{\theta_\perp}$, thanks to our orthogonal decomposition. The only terms in the decomposition on which the currents act nontrivially are
\be \CF_{\mu\cdot \eta_0}^{\eta_0} \otimes \CF_{\lambda\cdot\theta_0}^{\theta_0}\,. \ee
Recalling from \eqref{BRST-triv0} of  Section \ref{sec:BRST} that 
\begin{equation} \label{BRST-triv}
H_{\rm BRST}\left(\mathfrak{gl}(1)^r,   \CF_{\mu\cdot \eta_0}^{\eta_0} \otimes \CF_{\lambda\cdot\theta_0}^{\theta_0}
\right)=\delta_{\lambda,\mu}\mathbb{C},
\end{equation}
the BRST cohomology of $\mc{V}_{\beta\gamma}^{T^*V}\otimes \mc{V}_{bc}^{V}$ evaluates to
\begin{equation}\label{eqffBRSTlatticeNq}
  H_{\rm BRST}\left(\mathfrak{gl}(1)^r, \mc{V}_{\beta\gamma}^{T^*V}\otimes \mc{V}_{bc}^{V}\right) \,\cong\hspace{-.3cm}\bigoplus_{\substack{ \mu\in N;\;\; \mu^\vee,\lambda^\vee\in N^\vee \\[.1cm] \text{s.t.}\, [\mu]=[\mu^\vee]=[\lambda^\vee]}} \hspace{-.3cm}
  M_{(\rho(\mu)+\rho^\vee(\mu^\vee))\cdot \phi}^\phi \otimes \CF_{\mu^\vee\cdot\eta_\perp}^{\eta_\perp} \otimes \CF_{\lambda^\vee\cdot\theta_\perp}^{\theta_\perp}\,.
\end{equation}

There is useful way to regroup this.
Note that
\be \CV_{\rho^\vee,0} := \bigoplus_{\lambda^\vee\in \Z^{n-r}} \CF^{\theta_\perp}_{\lambda^\vee\cdot\theta_\perp} \label{def-V0} \ee
is simply the lattice VOA corresponding to the full-rank lattice $\Z^{n-r}\in\C^{n-r}$ with inner product $\rho^{\vee T}\rho^\vee$\,. Its category of modules is semisimple, generated by the simple currents
\be \CV_{\rho^\vee,h} := \bigoplus_{\lambda^\vee\,\in\, \Z^{n-r}+h_{N^\vee}} \CF^{\theta_\perp}_{\lambda^\vee\cdot\theta_\perp}\,, \label{def-Vh} \ee
for each $h\in H$, where $h_{N^\vee}$ representing the equivalence class of $h\in N^\vee/\Z^{n-r}$, as in \eqref{Z-decomp}. These obey $\CV_{\rho^\vee,h}\times \CV_{\rho^\vee,h'}=\CV_{\rho^\vee,h+h'}$\,.

Similarly, the VOA
\be \CM_{\rho,0} := \bigoplus_{\substack{\lambda\,\in\,\Z^r \\
\lambda^\vee\,\in\,\Z^{n-r}}} M_{(\rho(\lambda)+\rho^\vee(\lambda^\vee)\cdot\phi}^\phi \otimes \CF_{\lambda^\vee\cdot\eta_\perp}^{\eta_\perp}\, \label{def-M0} \ee
has modules labeled by $h\in H$, given by
\be \CM_{\rho,h} := \bigoplus_{\substack{\lambda\,\in\,\Z^r+h_N \\
\lambda^\vee\,\in\,\Z^{n-r}+h_{N^\vee}}} M_{(\rho(\lambda)+\rho^\vee(\lambda^\vee))\cdot\phi}^\phi \otimes \CF_{\lambda^\vee\cdot\eta_\perp}^{\eta_\perp}\,. \label{def-Mh} \ee
From the relation \eqref{SF-ff-n} between singlet algebras and symplectic fermions, we see that $\CM_{\rho,0}$ can also be understood as an invariant subalgebra (\emph{a.k.a.} orbifold)
\be \CM_{\rho,0} \,\cong\, \big( \CV_{SF}^{T^*V}\otimes \CV_{\rho^\vee,0}^-\big)^{(\C^*)^{n-r}} \label{M0-orb} \ee
of $n$ symplectic fermions and a lattice VOA $\CV_{\rho^\vee,0}^- = \bigoplus_{\lambda^\vee\in \Z^{n-r}} \Fock_{\lambda^\vee\cdot\eta_\perp}^{\eta_\perp}$ corresponding to the lattice $\Z^{n-r}$ with negative bilinear form $-\rho^{\vee T}\rho^\vee$\,. The group $(\C^*)^{n-r}$ acts anti-diagonally, so that taking invariants ties together the weight spaces in the two factors, to reproduce \eqref{def-M0}.
Now the full module category of $\CM_{\rho,0}$ is more complicated, but by the work of \cite{McRae:2020} the $\CM_{\rho,h}$ are still simple $H$ currents, obeying $\CM_{\rho,h}\times \CM_{\rho,h'}=\CM_{\rho,h+h'}$\,.

Putting everything together, we arrive at:
\begin{Prop} \label{prop:A}
There is an isomorphism
\begin{align}
\mc{V}_{\rho}^A &= H_{\rm BRST}\left(\mathfrak{gl}(1)^r, \mc{V}_{\beta\gamma}^{T^*V} \otimes \mc{V}_{bc}^{V}\right) \notag \\
&\hspace{1in} \cong\, \bigoplus_{\substack{ \mu\in N;\;\; \mu^\vee,\lambda^\vee\in N^\vee \\[.1cm] \text{s.t.}\, [\mu]=[\mu^\vee]=[\lambda^\vee]}} \hspace{-.3cm}
  M_{(\rho(\mu)+\rho^\vee(\mu^\vee))\cdot \phi}^\phi \otimes \CF_{\mu^\vee\cdot\eta_\perp}^{\eta_\perp} \otimes \CF_{\lambda^\vee\cdot\theta_\perp}^{\theta_\perp} \notag \\
& \hspace{1in} \cong\, \bigoplus\limits_{h\in H} \CM_{\rho,h}\otimes \CV_{\rho,h} \label{eqBRSTbgq}
\end{align}
expressing $\CV_\rho^A$ as extension of the product of VOA's $\CM_{\rho,0}\otimes \CV_{\rho^\vee,0}$ by the $H$ simple currents defined in \eqref{def-Vh}--\eqref{def-Mh}. In particular, the BRST cohomology in the original definition of $\CV_\rho^A$ is concentrated in cohomological degree zero.
\end{Prop}

 One can use the fact that $\CV_\rho^A$ is a simple current extension of $\CM_{\rho,0}\otimes \CV_{\rho^\vee,0}$ to relate their categories of modules, applying the work of \cite{creutzig2017tensor}. This would be closely analogous to \cite{Creutzig:2022ugv} as well as \cite{Adamovic:2022oqj}.  We will describe the modules of $\CV_\rho^A$ in other ways later.

\subsubsection{Some examples}
\label{sec:VA-examples}

\begin{Exp}
Consider the case $\rho=(1)$, \emph{i.e.} $U(1)$ gauge theory with $V=\C$ of weight 1. Then $\Lambda^\perp=0$ and by \eqref{eqSFfreefield} we have
\be \CV^A_\rho=\bigoplus_{\lambda\in \mathbb{Z}} M_{\lambda\varphi}^\varphi \,\cong\, \CV_{SF}\,. \ee
Thus the boundary VOA is just symplectic fermions. The explicit relation between $\mc{V}_{\beta\gamma}\otimes \mc{V}_{bc}$ and symplectic fermions $\chi^\pm$ is
\begin{equation}
    \beta(z) b(z)\mapsto \chi^+(z),~-\gamma(z)c(z)\mapsto \chi^-(z). 
\end{equation}
\end{Exp}

\begin{Exp}
As a slight generalization, consider any simple abelian theory whose charge lattice is isomorphic to a trivial lattice, meaning that the rank-$r$ lattice with bilinear form $\rho^T\rho$ is isomorphic to the standard Euclidean lattice $\Z^r$ as in Proposition \ref{prop:triv}. In particular, $H=\{0\}$.
Then $\Z^n=\rho(\Z^r)\oplus\rho^\vee(\Z^{n-r})$ and each of the lattices $\rho(\Z^r)$ and $\rho^\vee(\Z^{n-r})$ is isomorphic to standard $\Z^r$ and $\Z^{n-r}$. We may decompose $M_{(\rho(\lambda)+\rho^\vee(\lambda^\vee))\cdot\phi}^\phi \,\cong\, M_{\lambda\cdot\phi_0}^{\phi_0} \otimes M_{\lambda^\vee\cdot\phi_\perp}^{\phi_\perp}$ and it follows from the free field realizations of symplectic fermions
\eqref{SF-ff-n} and and symplectic bosons \eqref{betagamma-MF0} that
\be \CM_{\rho,0} \,\cong\, \left( \bigoplus_{\lambda\in \Z^r} M_{\lambda\cdot\phi_0}^{\phi_0} \right) \otimes \left( \bigoplus_{\lambda^\vee\in \Z^{n-r}} M_{\lambda^\vee\cdot\phi_\perp}^{\phi_\perp} \otimes \CF_{\lambda^\vee\cdot\eta_\perp}^{\eta_\perp} \right) \,\cong\, \CV_{SF}^{\otimes r}\otimes \mc{V}_{\beta\gamma}^{\otimes {n-r}}\,. \ee
Similarly, by bosonization,
\be \CV_{\rho^\vee,0}\,\cong\, \bigoplus_{\lambda^\vee\in \Z^{n-r}}\CF_{\lambda^\vee\cdot\theta_\perp}^{\theta_\perp} \,\cong\, \mc{V}_{bc}^{\otimes n-r}\,, \ee
whence
\be \CV_\rho^A = \CM_{\rho,0}\otimes \CV_{\rho^\vee,0} \,\cong\, \CV_{SF}^{\otimes r}\otimes \mc{V}_{\beta\gamma}^{\otimes {n-r}} \otimes   \mc{V}_{bc}^{\otimes n-r}\,. \ee
This is consistent with Proposition \ref{prop:triv}, which states that the underlying physical theory $\CT_{\rho}$ factorizes as a product of $r$ SQED's (with boundary VOA $\CV_{SF}$ as in the previous example) and $n-r$ free hypermultiplets (with boundary VOA $\mc{V}_{\beta\gamma}\otimes \mc{V}_{bc}$). 
\end{Exp}

We also note that when $H$ is trivial, the entire VOA $\CV_\rho^A$ can be expressed as an orbifold of symplectic fermions and lattice VOA's. Indeed, from \eqref{M0-orb} together with $H=\{0\}$ we have
\be \CV_\rho^A\,\cong\, \big( \CV_{SF}^{T^*V}\otimes \CV_{\rho^\vee,0}^{-}\big)^{(\C^*)^{n-r}}\otimes \CV_{\rho^\vee,0}^+\,.\ee
Later in Section \ref{subsecgaugeextA}, we will generalize this idea.  We will show that, for \emph{any} $\rho$, $\CV_\rho^A$ is isomorphic to an orbifold of symplectic fermions modulo a Morita-trivial vertex algebra. 

\begin{Exp}
Consider the case when $\rho=(1,1)^{\mathsf{T}}$, \emph{i.e.} $G=U(1)$ and $V=\C^2$ with weights 1,1.
Then $\rho^{\vee}=(1,-1)^{\mathsf{T}}$\,, and recall that $H\simeq \Z_2$.
The Heisenberg VOA underlying the free field realization of $\CV_\rho^A$ in Prop. \ref{prop:A} is based on a vector space with basis $\{\phi^1,\phi^2,\eta_\perp,\theta_\perp\}$ and nonzero pairings
\be B(\phi^i,\phi^j)=\delta^{ij}\,,\quad B(\eta_\perp,\eta_\perp)=-2\,,\quad B(\theta_\perp,\theta_\perp)=2\,.\ee
Then the boundary VOA takes the form
\be \CV_\rho^A \;\cong\, \bigoplus_{\substack{a,b,c\,\in\, \Z \\ \text{s.t.}\; a+b+c\;\text{is even}}}  M_{a\phi^1+b\phi^2}^\phi\otimes \Fock_{\frac12(a-b)\eta_\perp}^{\eta_\perp}\otimes \Fock^{\theta_\perp}_{\frac12 c\,\theta_\perp} \ee
For small $a,b,c$ we can give the following correspondence between generators of the direct summands and fields in $\mc{V}_{\beta\gamma}^{\otimes 2}\otimes \mc{V}_{bc}^{\otimes 2}$ prior to taking the BRST cohomology:
\be \begin{array}{c@{\quad\leftrightarrow\quad}c}
\beta^1 b^1,\, \beta^1 b^2 &   M_{\phi^1}\otimes \Fock_{\frac{1}{2}\eta_\perp}\otimes \Fock_{\pm\frac{1}{2}\theta_\perp} \\
\beta^2 b^1,\, \beta^2 b^2  & M_{\phi^2}\otimes \Fock_{-\frac{1}{2}\eta_\perp}\otimes \Fock_{\pm\frac{1}{2}\theta_\perp} \\
\gamma^1 c^1,\, \gamma^1 c^2 & M_{-\phi^1}\otimes \Fock_{-\frac{1}{2}\eta_\perp}\otimes \Fock_{\pm\frac{1}{2}\theta_\perp} \\
\gamma^2 c^1,\, \gamma^2 c^2 & M_{-\phi^2}\otimes \Fock_{\frac{1}{2}\eta_\perp}\otimes \Fock_{\pm\frac{1}{2}\theta_\perp}  \\
b^1 c^2,\,b^2 c^1 & M_{0}\otimes \Fock_{0}\otimes \Fock_{\pm\theta_\perp}\end{array} \label{eq11VOAA}
\ee
\end{Exp}

\subsubsection{Symmetries in the free-field realization}
\label{sec:VA-sym-ff}

In Section \ref{sec:VA-sym} we saw that $\CV_\rho^A$ should have a cohomological $U(1)_H$ symmetry, as well as $T_F\times T_B\times T_{\rm top} \simeq U(1)^{n-r}\times U(1)^{n-r}\times U(1)^r$ flavor symmetry. In Proposition \ref{prop:A} we proved that the cohomological symmetry $U(1)_H$ acts trivially. The remaining symmetry appears in the free-field realization the following way.

Prior to BRST reduction, in the $\CV_{\beta}^{T^*V}\otimes \mc{V}_{bc}^V$ free-field realization \eqref{bgbc-ff}, the boundary $T_B$ symmetry has currents $J^B = \rho^{\vee T}\pd\theta$  (meaning $n-r$ currents $J^B_\alpha = \rho^\vee_{\alpha i}\pd\theta^i$). After decomposing the Heisenberg VOA's, these simply become $J^B_\alpha = \pd\theta_\perp$, which commute with the BRST reduction, and define $T_B$ currents in $\CV_\rho^A$, acting exclusively on the lattice VOA modules $\CV_{\rho^\vee,h}$. Under the zero-modes of $J^B$, the Fock modules therein have weight
\be T_B:\qquad J^B = \pd\theta_\perp\,,\qquad \text{weight}(\CF_{\lambda^\vee\cdot\theta_\perp}^{\theta_\perp}) = \rho^{\vee T}\rho^\vee \lambda^\vee  \in \Z^{n-r}\,.\label{A-JB}\ee

The bulk flavor symmetry $T_F$ acts prior to BRST reduction with currents $J^F = \sigma^T(\pd\eta+\pd\theta)$. The BRST currents for the gauge action are $J^G=\rho^T(\pd\eta+\pd\theta)$, and we can shift the $J^F$ by multiples of $J^G$ (which is exact in BRST cohomology) to ensure that they depend exclusively on $\eta_\perp$ and $\theta_\perp$. Put differently, we project onto the subspace normal to the gauge charge lattice. The modified flavor currents become $J^F=(\sigma^T-\sigma^T\rho(\rho^T\rho)^{-1}\rho^T)(\pd\eta+\pd\theta)$. With a bit of linear algebra, this gets written in terms of $\eta_\perp,\theta_\perp$ as
\be T_F:\qquad J^F = (\rho^{\vee T}\rho^\vee)^{-1} (\pd\eta_\perp+\pd\theta_\perp)\,. \label{A-JF} \ee
Under the zero-mode of $J^F$, the modules in the lattice decomposition have weights
\be \text{weight}\big(M_{(\rho(\mu)+\rho^\vee(\mu^\vee))\cdot \phi}^\phi \otimes \CF_{\mu^\vee\cdot\eta_\perp}^{\eta_\perp} \otimes \CF_{\lambda^\vee\cdot\theta_\perp}^{\theta_\perp}\big) = \lambda^\vee-\mu^\vee \in \Z^{n-r}  \ee
(Note both $\mu^\vee$ and $\lambda^\vee$ may be fractional, but their difference is integral precisely because $[\mu^\vee]=[\lambda^\vee]\in H$.)

The topological symmetry $T_{\rm top}$ is not enhanced to Kac-Moody. Before the BRST quotient, it acts on the modules $M_{\mu\cdot\phi}^\phi\otimes \Fock_{\mu\cdot\eta}^\eta \otimes \Fock_{\lambda\cdot\theta}^\theta$ in the free-field realization \eqref{bgbc-ff} with weights $\sigma^{\vee T}\lambda \in \Z^r$. After the BRST quotient, this is idenfitied with
\be T_{\rm top}:\qquad \text{weight}\big(M_{(\rho(\mu)+\rho^\vee(\mu^\vee))\cdot \phi}^\phi \otimes \CF_{\mu^\vee\cdot\eta_\perp}^{\eta_\perp} \otimes \CF_{\lambda^\vee\cdot\theta_\perp}^{\theta_\perp}\big) = \mu+\sigma^{\vee T}\rho^\vee\lambda^\vee \in \Z^r\,. \label{A-Jtop} \ee
This is again integral due to $[\mu]=[\lambda^\vee]\in H$. The lack of Kac-Moody enhancement is reflected in the fact that the putative currents $\pd\phi$ that would be required to measure the $\mu$ weight are not contained in $M^\phi_0$, as they are not in the kernel of the screening charges.

To write down the stress tensor in the free-field realization, let us first write the stress tensor of the BRST complex before taking the cohomology using the free-field realization: 
\be
L(z)=\frac{1}{2}\sum_i \left(\normord{\partial\phi^i\partial\phi_i}-\partial^2\phi^i-\normord{\partial \eta^i\partial \eta_i}+ \normord{\partial \theta^i\partial\theta_i}\right)+ \sum_a \normord{\partial c^a b_a}.
\ee
We can use the projection $\pi$ and $\pi^\perp$ to write $\partial \eta^i=\sum \pi_{ij}\partial \eta^j+\pi^\perp_{ij}\partial \eta^j$. The term $\sum \pi_{ij}\partial \eta^j$ is equal to:
\be
\sum \pi_{ij}\partial \eta^j=\sum_{j}\left(\rho(\rho^\T\rho)^{-1} \rho^\T\right)_{ij} \partial \eta^j=\sum_{a, j}\left(\rho(\rho^\T\rho)^{-1} \right)_{i}{}^a J^{\beta\gamma}_a.
\ee
Here the term $J^{\beta\gamma}_a$ is the generator of $U(1)^r$ Kac-Moody symmetry in $\mc{V}_{\beta\gamma}^{\otimes n}$.
We denote by $\partial \eta_{\alpha}$ the term $\sum \tau^\alpha{}_j\partial \eta^j$. Using these, we can rewrite the second term:
\be
\sum \pi_{ij}^\perp\partial \eta^j=\sum_{j}\left(\tau^\T(\tau\tau^\T)^{-1} \tau\right)_{ij} \partial \eta^j=\sum_{\alpha} \left(\tau^\T(\tau\tau^\T)^{-1}\right)_{i\alpha} \partial \theta^{\alpha}.  
\ee
For convenience, let us denote by $\Pi$ the matrix $\rho(\rho^\T\rho)^{-1}$ and $\Pi^\perp$ the matrix $\tau^\T(\tau\tau^\T)^{-1}$. We have:
\be
\sum\limits_i \normord{\partial\eta^i\partial\eta_i}=\sum_{i, a, b}\Pi_{i}{}^a\Pi^{ib}\normord{J_a^{\beta\gamma}J_b^{\beta\gamma}}+\sum_{i,\alpha, \beta} \Pi_{i}^{\perp}{}^\alpha\Pi^{\perp, i\beta}\normord{\partial \eta_{\alpha}\partial \eta_{\beta}}.
\ee
Here we used the fact that $\pi\cdot \pi^\perp=0$. Similarly, we have a decomposition for $\sum_i \normord{\partial \theta^i\partial\theta_i}$ as:
\be
\sum\limits_i \normord{\partial \theta^i\partial\theta_i}=\sum_{i, a, b}\Pi_{i}{}^a\Pi^{ib}\normord{J_a^{bc}J_b^{bc}}+\sum_{i,\alpha, \beta} \Pi_{i}^{\perp}{}^\alpha\Pi^{\perp, i\beta}\normord{\partial \theta_{\alpha}\partial \theta_{\beta}}.
\ee
Here $J_a^{bc}$ are the generators of the $U(1)^r$ symmetry in $\mc{V}_{bc}^{\otimes n}$ and $\partial\theta_\alpha=\sum \tau^\alpha{}_j\partial \theta^j$. In the cohomology, the term: 
\be
\sum_{i, a, b} 
\Pi_{i}{}^a\Pi^{ib}\normord{J_a^{\beta\gamma}J_b^{\beta\gamma}}+\sum_{i, a, b}\Pi_{i}{}^a\Pi^{ib}\normord{J_a^{bc}J_b^{bc}}+\sum_a \normord{\partial c^a b_a}
\ee
is exact, and we are left with the following expression of the conformal element:
\be\label{eq:L-A-ff-cohomology}
\frac{1}{2}\sum_i \left(\normord{\partial\phi^i\partial\phi_i}-\partial^2\phi^i-\sum_{ \alpha, \beta}\Pi_{i}^{\perp}{}^\alpha\Pi^{\perp, i\beta} \normord{\partial \eta_{ \alpha}\partial \eta_{\beta}}+\sum_{ \alpha, \beta}\Pi_{i}^{\perp}{}^\alpha\Pi^{\perp, i\beta} \normord{\partial \theta_{\alpha}\partial \theta_{\beta}}\right). 
\ee

\section{Boundary VOA for $\mc T_\rho^B$}
\label{secDirichletB}

In this section, we define the boundary VOA $\CV_\rho^B$ for the B twist $\mc T_\rho^B$. It was already discussed in \cite{CG} that the boundary VOA is a non-purterbative completion of the affine Lie superalgebra $\widehat{\grho}$. In this section, we identify boundary monopole operators as generators of simple currents of $\widehat{\grho}$, and define the non-purterbative completion $\mc V_\rho^B$. We justify our identification by giving a free-field realization of $\mc V_\rho^B$, and further, by proving the mirror symmetry statement of the boundary VOAs in Theorem \ref{thm:VOA}. 

The structure of this section is as follows:

\begin{itemize}
    \item In Section \ref{sec:Dbc}, we recall the Dirichlet boundary condition of $\mc T_\rho^B$ according to \cite{CG}. 

    \item In Section \ref{subsecperturbalg}, we recall the definition of $\widehat{\grho}$ and introduce monopole operators and the completion $\CV_\rho^B$.

    \item In Section \ref{sec:Bff}, we construct free-field realizations of $\widehat{\grho}$ and $\CV_\rho^B$, using lattice extension of Heisenberg VOAs.

    \item In Section \ref{secMirrorsymbdryVOA}, we prove mirror symmetry statement in terms of boundary VOAs in Theorem \ref{thm:VOA}. 
    
\end{itemize}

\subsection{Dirichlet boundary condition and its symmetry}
\label{sec:Dbc}

The holomorphic boundary condition compatible with the bulk B twist descends from a physical boundary condition that preserves 2d $\CN=(0,4)$ supersymmetry and is Dirichlet-like \cite{CG}. In particular, it fixes all values of hypermultiplet scalars, and fixes a trivialization of the gauge bundle, breaking gauge symmetry to a global symmetry.

Since gauge symmetry is explicitly broken, there is no boundary gauge anomaly to cancel. However, it is still useful to analyze the boundary 't Hooft anomalies for global symmetries, as they control OPE coefficients in the boundary VOA.

The boundary fields that survive on the Dirichlet boundary condition are, roughly speaking, the complement of the fields on the Neumann b.c. of Section \ref{sec:Nbc}. They include:
\begin{itemize}
    \item $N:=F_{zt}+i\pd_z\sigma$: $r$ boundary curvature components, modified by real gauge scalars $\sigma$; these particular combinations become Q-closed elements of the boundary VOA upon twisting
    \item $\varphi$: $r$ complex vectormultiplet scalars
    \item $\lambda_+,\lambda_+'$: $2r$ bulk gauginos, in representation $T^*\mathfrak g_\C \simeq \C^r\oplus \C^r$, with positive (right-handed) boundary chirality; they are superpartners of $N$ and $\varphi$
    \item $\psi_-^X,\psi_-^Y$: $2n$ bulk hypermultiplet fermions, in representations $V,V^*$, that have negative (left-handed) boundary chirality.
\end{itemize}
The hypermultiplet scalars $X,Y$ are set to zero.

At the boundary, the bulk $G=U(1)^r$ gauge symmetry becomes a global flavor symmetry. In addition, there is again bulk $SU(2)_R\times SU(2)_C$ R-symmetry, and bulk $T_F=U(1)^{n-r}$ and $T_{\rm top}=U(1)^{n-r}$ flavor symmetry. The above fields have charges
\be \label{Dcharges}
\begin{array}{c|c|c|c|c}
\text{symmetry} & \text{f.s.} & N,\varphi & \lambda_+, \lambda_+'  & \psi_-^X,\psi_-^Y \\\hline
G = U(1)^r & \mathbf{G} & 0,0 & 0, 0  & \rho,-\rho   \\
T_F = U(1)^{n-r} & \mathbf{F} & 0,0 & 0,0  & \sigma,-\sigma  \\
T_{\rm top} = U(1)^r & \mathbf{T} & 0,0 & 0,0 &0,0 \\
U(1)_H & \mathbf{H} & 0,0 & 1,1  & 0,0  \\
U(1)_C & \mathbf{C} & 0,2 & 1,-1  & -1,-1  
\end{array}
\ee
The boundary anomaly polynomial is computed using formula \eqref{anom-def}, and simply evaluates to
\begin{align} \label{anom-D}
\CI_\pd &= +2\mb G\cdot \mb T- \frac{r}{2}(\mb H+\mb C)^2 - \frac{r}{2}(\mb H-\mb C)^2 +\frac12|\rho\mb G+\sigma\mb F-\mb C|^2  +\frac12|-\rho\mb G-\sigma\mb F-\mb C|^2 \notag \\ 
&= -r\mb H^2+(n-r)\mb C^2 + 2\mb T\cdot \mb G+|\sigma\mb F+\rho\mb G|^2
\end{align}
Note that the contributions of fermions are exactly opposite to those on the Neumann boundary condition of Section \ref{sec:Nbc}, due to opposite chiralities. We also \emph{choose} the bulk Chern-Simons level to be $2\mb G\cdot \mb T$; this defines how the topological symmetry acts on monopole operators.

\subsubsection{3d mirror symmetry}

It is easy to see that the Neumann boundary anomaly \eqref{anom-N} and the Dirichlet boundary anomaly \eqref{anom-D} are identical upon swapping
\be \label{MS-sym} \begin{array}{c@{\quad\leftrightarrow\quad}c}
  r & n-r \\
  \rho,\rho^\vee & \rho^\vee,\rho \\
  \sigma,\sigma^\vee & \sigma^\vee,\sigma \\
  \mb C,\mb H & \mb H,\mb C \\
  \mb T,\mb F & \mb F, \mb T \\
  \mb B & \mb G
\end{array} \ee
This is evidence that the two boundary conditions may be mirror to each other, up to the corresponding redefinitions of global symmetries. One nontrivial aspect of this correspondence, which is not a standard part of the \emph{bulk} mirror symmetry dictionary, is that the boundary $U(1)_B$ symmetry carried by extra 2d matter on the Neumann b.c. corresponds to the global boundary $G$ symmetry on the Dirichlet b.c.

\subsubsection{Boundary monopole operators}
\label{sec:mon1}

The Dirichlet boundary condition supports boundary monopole operators. The analysis of boundary monopole operators in \cite{BDGH,dimofte2018dual} shows that exist `bare' BPS boundary monopoles $U_\mu$ labelled by all cocharacters $\mu$ of the gauge group, meaning
\be \mu \in \Z^r\,. \ee
Their charges under the various global symmetries are calculated by differentiating the anomaly polynomial with respect to $\mb G$. Monopole operators also acquire boundary spin, proportional to the product of their `magnetic' charge $\mu$ and their `electric' charge under $G$.
Altogether, we find
\be \label{mon-charges}
\begin{array}{c|ccccc|c} & G & U(1)_F & U(1)_{\rm top}  & U(1)_H & U(1)_C & \text{Spin}(2) \\\hline
\text{charges of $U_\mu$} & (\rho^T\rho)\mu & (\sigma^T\rho)\mu & \mu & 0 & 0 & \tfrac12\mu\cdot(\rho^T\rho)\mu = \tfrac12|\rho\mu|^2 \end{array}
\ee

\subsection{Perturbative and non-perturbative boundary VOA}
\label{subsecperturbalg}

Upon taking a topological B twist of the bulk, the boundary condition of the previous section (more accurately, a slight deformation thereof) becomes purely holomorphic, and supports a VOA \cite{CG}. Perturbatively, the VOA is generated by the bulk fields that survive at the boundary. Their OPE's are computed from one-loop bulk-boundary Feynman diagrams, as explained further in \cite{CDG-chiral, garner2022vertex}. 

The perturbative boundary VOA turns out to have the form of a Kac-Moody Lie superalgebra $\widehat\grho$. It is based on a finite Lie superalgebra $\grho$, defined as 
\begin{align} \label{def-grho} \grho &:= T^*(\mathfrak g_\C\oplus \Pi V) \;=\; \mathfrak g_\C \oplus \Pi V \oplus \Pi V^* \oplus \mathfrak g_\C^* \\
& \hspace{.8in} \text{with basis}\;\; (N_a,\;\;\; \psi^{i,+},\;\;\; \psi^{i,-},\;\;\; E^a)\,. \notag
\end{align}
for $a=1,...,r$ and $i=1,...,n$.
This is a $\Z$-graded Lie algebra. Its even/bosonic part is a copy of $\mathfrak g_\C=\C^r$ and a copy of $\mathfrak g_\C^*=\C^r$ (with basis elements $N_a,E^a$). Its odd/fermionic part, indicated by the parity shift `$\Pi$', is $V\oplus V^* = \C^{2n}$ (with basis $\psi^{i,\pm}$). The nonvanishing Lie brackets are given by
\be \label{g-brackets}
[N_a,\psi^{i, \pm}]=\pm \rho^i{}_a \psi^{i,\pm}\,,\qquad 
\{\psi^{i,+},\psi^{j,-}\}=\delta^{ij} \sum_a \rho^i{}_aE^a.
\ee

\begin{Rmk} The reason for this Lie algebra appearing is that the topological B-twist of 3d $\CN=4$ gauge theory is perturbatively equivalent to a Chern-Simons theory for $\grho$. Hypermultiplet fermions and vectormultiplet bosons can be re-interpreted as the components of a $\grho$ super-connection. The statement extends to general group $G$ and symplectic representation $T^*V$, using a Lie superalgebra given by $\mathfrak g_{G,V}:=(\mathfrak g_\C\oplus \Pi V)$. The generalization of \eqref{g-brackets} states that $[N,-]$ encodes the $\mathfrak g_\C$ action on all other elements, $E$ is central, and $\{\psi,\psi\}$ is the complex moment map for the $\mathfrak g_\C$ action on $T^*V$.

A Kac-Moody VOA arises because the Dirichlet boundary condition above generalizes the  Dirichlet-like boundary conditions used in classic analyses of Chern-Simons theory, \cf\ \cite{Witten-Jones, EMSS}, which perturbatively support Kac-Moody VOA's.
\end{Rmk}

To define the Kac-Moody Lie superalgebra, we need to choose an invariant, even, nondegenerate bilinear pairing $\kappa:\grho\times \grho\to\grho$, otherwise known as a level. There are many choices of pairings on $\grho$. The one encoded in the B-twisted bulk action is induced from the symplectic pairing on $T^*(\mathfrak g \oplus \Pi V)$, and shifted by a one-loop boundary correction, \cf~\cite{CDG-chiral, garner2022vertex}; it is given by
\begin{equation}\label{eqbilinearBper}
   \kappa(N_a,N_b)=(\rho^T\rho)_{ab}=\sum\limits_i \rho^i{}_{a}\rho_{ib}\,,\qquad \kappa(N_a,E^b)=\delta_a{}^b\,,\qquad\kappa(\psi^{i,+},\psi^{j,-})=\delta^{ij}\,.
\end{equation}
(The $\rho^T\rho$ term is the one-loop correction.)

The corresponding Kac-Moody VOA $\widehat\grho$ is strongly generated by fields $N_a(z),\psi^{i,\pm}(z),E^a(z)$, with OPE's of any two generators $x,y$ given universally by $x(z)y(w)\sim \frac{\kappa(x,y)}{(z-w)^2}+\frac{[x,y](w)}{z-w}$. Explicitly, this means
\begin{equation}\label{eqperturbVOABgenOPE}
\begin{aligned}
& N_a(z)E^b(w)\sim \frac{\delta_{a}{}^b}{(z-w)^2}\,,\qquad\;\; \quad N_a(z)N_b(w)\sim \frac{(\rho^T\rho)_{ab}}{(z-w)^2}\,,\\ &N_a(z)\psi^{i,+}(w)\sim \frac{ \rho^i{}_a\psi^{i,+}(w)}{(z-w)}\,,\qquad N_a(z)\psi^{i,-}(w)\sim \frac{- \rho^i{}_a\psi^{i,-}(w)}{(z-w)}\,,\\ &\hspace{.3in}\psi^{i,+}(z)\psi^{j,-}(w)\sim \frac{\delta^{ij}}{(z-w)^2}+\frac{\delta^{ij}\sum_a \rho^i{}_aE^a(w)}{z-w}\,.
\end{aligned}
\end{equation}
These VOA generators correspond to fields on the Dirichlet b.c. from Section \ref{sec:Dbc} as follows:
\begin{itemize}
    \item $N_a$ are the modified boundary field strengths $F_{zt}+i\pd_z\sigma$
    \item $E$ are the complex vectormultiplet scalars $\varphi$, dualized to elements of $\mathfrak g_\C^*$ rather than $\mathfrak g_\C$ by using the Killing form (normalized by the physical gauge couplings) on each $\mathfrak{gl}(1)$ factor in the gauge group
    \item $\psi^{i,+},\psi^{i,-}$ are the conjugates of the hypermultiplet fermions, $\psi^+=\overline{\psi_-^X}$ and $\psi^-=\overline{\psi_-^Y}$.
\end{itemize}

\begin{Exp}
When $\rho=(1)$, $\widehat\grho$ is the affine Lie superalgebra of $\mathfrak{gl}(1|1)$ (see for example \cite{creutzig2020tensor} and Section \ref{sec:gl11}. The notation $(N,\psi^+,\psi^-,E)$ that we use here is standard in the literature on  $\mathfrak{gl}(1|1)$. 
\end{Exp}

\begin{Exp}
When the gauge group $G$ is trivial $(r=0)$, the B-twisted VOA simply consists of $n$ copies of symplectic fermions $\psi^{i,\pm}(z)$. This is the full non-perturbative result, since when there is no gauge group, there are no monopole operators.

More generally, any $\widehat\grho$ (and its non-perturbative completion $\CV_{\rho}^B$ defined below) has a free-field realization via $n$ symplectic fermions $\CV_{SF}^{T^*V}$, in which $\psi^{i,\pm}(z)$ is given by coupling symplectic fermion generator $\chi^{i, \pm}(z)$ with vertex operators for degenerate bosons, see Section \ref{subsecrelatealt}.

\end{Exp}

\subsubsection{Symmetries of the VOA}
\label{sec:sym-B}

The various boundary global symmetries summarized in \eqref{Dcharges} should all induce symmetries of the boundary VOA.

The global $G=U(1)^r$ symmetry exists only at the boundary, since it is gauged in the bulk. Thus it should become an enhanced $\mathfrak{gl}(1)^r$ Kac-Moody symmetry of the VOA. It is easy to see that its currents are given by the fields $N^a(z)$. Its level is clearly $(\rho^T\rho)$ from \eqref{eqperturbVOABgenOPE}; this also agrees with the $|\rho\mb G|^2$ term in the boundary 't Hooft anomaly \eqref{anom-D}.

The topological $T_{\rm top}=U(1)^r$ symmetry has a bulk current that is exact in the B twist; thus it should also descend to a Kac-Moody $\mathfrak{gl}(1)^r$ symmetry of the boundary VOA. Moreover, the boundary 't Hooft anomaly \eqref{anom-D} does not contain a $\mb T^2$ term, and implies that this $\mathfrak{gl}(1)^r$ Kac-Moody algebra has level zero, \emph{i.e.} is totally commutative. There is an obvious level-zero abelian Kac-Moody algebra in $\widehat{\grho}$, generated by the currents $E^a(z)$, and we shall argue that these are indeed the $T_{\rm top}$ currents. They do not act on any the other fields in $\widehat{\grho}$. Indeed, $T_{\rm top}$ should measure charge of boundary monopole operators, which are a non-perturbative feature. 

The hypermultiplet flavor symmetry $T_F=U(1)^{n-r}$ has a nontrivial bulk current (it acts on bulk hypermultiplet operators $X,Y$ in the B-twist), and thus leads to finite $T_F$ symmetry of the boundary VOA, which does not have a Kac-Moody enhancement. This is the obvious symmetry acting on the $T^*V$-valued fermions $\psi^\pm$ with charges $\pm\sigma$.

The $U(1)_H$ R-symmetry becomes homological degree in the B twist. However, all the fields generating the boundary VOA have $U(1)_H$ charge zero. The addition of boundary monopole operators will not change this fact; thus, \emph{the B-twisted boundary VOA is manifestly supported in homological degree zero}.

The $U(1)_C$ R-symmetry mixes with $\text{Spin}(2)$ boundary spin to define the conformal grading $L_0$ on the boundary VOA. Explicitly, the charges are related by $L_0 = \frac12 C + J$ (where $J$ denotes boundary spin). All the generators of the Kac-Moody VOA turn out to have conformal weight one. Namely, $N\sim F_{zt}+i\pd_z\sigma$ has $C=0$ and $J=1$; $\psi^\pm$ have $C=1$ and $J=\frac12$; and $E$ has $C=2$ and $J=0$. Since the bulk stress tensor is Q-exact, there can exist a conserved boundary stress tensor. It is given by
\be \label{L-B}
    L(z)\,=\,\frac{1}{2}\Big(\sum_a\normord{N_a E^a}+\normord{E^a N_a}-\sum_i \normord{\psi^{i,+}\psi^{i,-}}+\sum_i \normord{\psi^{i,-}\psi^{i,+}}\Big)\,.
\ee
This expression is derived directly from bulk-boundary physics in \cite{garner2022vertex}.

In summary, the symmetries of $\widehat{\grho}$ and their charges are
\be \label{VOA-B-sym}
\begin{array}{c|c|cccc}
\text{symmetry} & \text{current} & N & E & \psi^+ & \psi^- \\\hline
G=U(1)^r & N(z) & 0 & 0 & \rho & -\rho \\
T_{\rm top} = U(1)^r & E(z) & 0 & 0 & 0 & 0 \\
T_F = U(1)^{n-r} & \text{n.a.} & 0 & 0 & \sigma & -\sigma \\
\text{homological}\;U(1)_C & \text{n.a.} & 0 & 0 & 0 & 0 \\
\text{conformal} & L(z) & 1 & 1 & 1 & 1 \end{array}
\ee

\begin{Rmk}
For a simple Lie algebra $\mathfrak k$, it is well known how to construct a Kac-Moody VOA~$\widehat{\mathfrak k}_\kappa$. The level $\kappa: \mathfrak k \times \mathfrak k \to \C$ is an invariant, nondegenerate, symmetric bilinear form, and thus specified uniquely by fixing a multiple of the Cartan-Killing form.
 The conformal element (stress tensor) is then determined by the Segal-Sugawara construction in terms of the inverse of the level, interpreted as a quadratic Casimir. Explicitly, if $J^a$ are generators of $\mathfrak k$, and $\kappa(J^a,J^b) = \kappa^{ab}$, then the stress tensor is --- up to overall normalization --- $L(z) = (\kappa^{-1})_{ab} \normord{J^a(z)J^b(z)}$. 
The change of normalization can be understood as shifting the level (by the dual Coxeter number) before inverting. 

For our Kac-Moody VOA $\widehat{\grho}$, the level \eqref{eqbilinearBper} and the quadratic Casimir appearing in the stress tensor \eqref{L-B} are related in a similar way. Namely, we must first remove the term  $\rho^T\rho$ in \eqref{eqbilinearBper} --- which was the one-loop level shift, encoded in the boundary 't Hooft anomaly --- before inverting to construct a Casimir. However, since the space of bilinear forms (equivalently, the space of quadratic Casimirs) is not just one-dimensional, this shift has a much more drastic effect than just changing an overall normalization.
 
In Appendix \ref{app:conformal}, we discuss the generalization of the Segal-Sugawara construction that applies to boundary Kac-Moody algebras in 3d $\CN=4$ theories with general gauge group $G$ (even non-abelian) and symplectic matter representation $T^*V$.
\end{Rmk}

\subsubsection{Non-perturbative completion}
\label{sec:mon2}

It is expected that the B-twisted boundary VOA is non-perturbatively corrected by boundary monopole operators. When $G$ is nonabelian, monopole operators have nonzero cohomological degree, and the correction is expected to involve nontrivial quotients, analogous to the simple quotients that appear in the classic constructions of WZW algebras from Kac-Moody algebras. These quotients are not well understood. In contrast, when $G$ is abelian, it is easy to see from the boundary anomaly polynomial \eqref{anom-D} that monopole operators have zero cohomological degree --- this follows because there are no $\mb H\cdot\mb G$ terms mixing the $U(1)_H$ R-symmetry and the boundary $G$ symmetry. Thus one might hope in the abelian that monopole operators simple extend the Kac-Moody algebra $\widehat{\grho}$. We give a proposal here for the appropriate extension.

Recall from Section \ref{sec:mon1} that boundary monopole operators $U_\mu$ are labelled by gauge cocharacters $\mu\in \Z^r$.
For each $\mu\in \Z^r$, we expect the monopole operator to take the form:
\be \CU_\mu(z) \,=\, \normord {e^{\mu\cdot \int N(z)}} = \normord{e^{\sum_a \mu^a \int N_a(z)}} \label{mon-vertex} \ee
where the right-hand side is the vertex operator given by the formal integral of $N$. (The same formal integral shows up in the construction of Heisenberg algebras, \emph{cf.} $v=\int J$ in Section~\ref{sec:Heis}. Here the $N_a$ generate a Heisenberg subalgebra with bilinear form $\rho^T\rho$.) This vertex operator satisfies 
\be \label{mon-OPE} \begin{array}{c} \ds N_a(z)\CU_\mu(w) \sim \frac{(\rho^T\rho\mu)_a}{z-w} \CU_\mu(w)\,,\qquad E^a(z)\CU_\mu(w)\sim \frac{\mu^a}{z-w}\CU_\mu(w) \\[.2cm]
\ds \psi^{i,\pm}(z) \CU_\mu(w) \sim (z-w)^{\pm (\rho\mu)^i} \normord{\psi^{i,\pm}(z)\CU_\mu(w)}\,, \end{array}\ee
where by $\normord{\psi^{i,\pm}(z)\CU_\mu(w)}$ we mean $\normord{\psi^{i,\pm}\CU_\mu}(w) + (z-w)\normord{\pd\psi^{i,\pm}\CU_\mu}(w) + \frac12(z-w)^2\normord{\pd^2\psi^{i,\pm}\CU_\mu}(w)+\ldots$, where the number of terms that contribute to the singular OPE depends on $(\mu\rho)^i$. The non-purterbative VOA should be generated by $\widehat{\grho}$ together with $\CU_\mu$ for all $\mu\in \Z^r$. 

To give a precise mathematical definition, we introduce the following spectral-flow automorphism of $\widehat{\grho}$. For each $ \mu, \nu\in \C^r$ such that $\rho(\nu)\in \Z^n$, the VOA $\widehat{\grho}$ has the following automorphism:
\be
\sigma_{\mu, \nu}N_a(z)= N_a(z)-\frac{\mu_a}{z},\qquad \sigma_{\mu, \nu}E_a(z)=E_a(z)-\frac{\nu_a}{z},\qquad \sigma_{\mu, \nu}\psi^{i, \pm}(z)=z^{\mp \rho (\nu)_i}\psi^{i, \pm}.
\ee
It is not difficult to see that the spectral-flow automorphism $\sigma_{\mu, \nu}$ is associated to the $\mathfrak{gl}(1)$ Kac-Moody symmetry action $(\mu-\rho^\T\rho\nu)\cdot E(z)+\nu \cdot N(z)$ (see Section \ref{sec:spectralflow}). The OPE of equation \eqref{mon-OPE} between $\widehat{\grho}$ and $\CU_\mu$ is reflected by the action of $\widehat{\grho}$ on $\sigma_{\rho^\T\rho\mu, \mu}\vert 0\rangle$, where $\vert 0\rangle$ is the vacuum vector of the vacuum module $\widehat{\grho}$:
\be\begin{array}{c}
N_a(z)\sigma_{\rho^\T\rho\mu, \mu}\vert 0\rangle\sim \frac{(\rho^\T\rho\mu)_a}{z}\sigma_{\rho^\T\rho\mu, \mu}\vert 0\rangle,\qquad E_a(z)\sigma_{\rho^\T\rho\mu, \mu}\vert 0\rangle\sim \frac{\mu_a}{z}\sigma_{\rho^\T\rho\mu, \mu}\vert 0\rangle\\  \psi^{i,\pm}(z) \sigma_{\rho^\T\rho\mu, \mu}\vert 0\rangle=z^{\pm \rho(\mu)_i}\sigma_{\rho^\T\rho\mu, \mu}\psi^{i,\pm}(z)\vert 0\rangle.
\end{array}
\ee
This suggests that for each $\mu\in \Z^r$, we should identify $\CU_\mu$ with the spectral flow of the vacuum vector:
\be
\CU_\mu\leftrightarrow \sigma_{\rho^\T\rho\mu, \mu}\vert 0\rangle.
\ee
This spectral flow is associated to the Kac-Moody symmetry $\mu\cdot N(z)$, exactly as equation \eqref{mon-vertex} suggests. 

Let us define $\mathbb U_\mu$ to be the module of $\widehat{\grho}$ generated by $\sigma_{\rho^\T\rho\mu, \mu}\vert 0\rangle$, or in other words, the spectral flow of the vacuum module:
\be
\mathbb U_\mu:= \sigma_{\rho^\T\rho\mu, \mu}\mathbb U_0=\sigma_{\rho^\T\rho\mu, \mu}\widehat{\grho}.
\ee
We will prove in Section \ref{sec:VOAlines}, Lemma \ref{Lem:grhosimple} that the vacuum module $\mathbb U_0$ is simple as its own module,\footnote{This is also a consequence of the free field realization of Section \ref{sec:Bff}, especially Proposition \ref{PropffpurtBgaugeq}.} and consequently all the $\mathbb U_\mu$ are simple. These are simple currents in the sense of Section \ref{sec:spectralflow}, and the work of \cite{Li:1997} implies the following fusion rule:
\be
\mathbb U_\mu\otimes_{\widehat{\grho}}\mathbb U_{\mu'}\cong \mathbb U_{\mu+\mu'},
\ee
which is realized by the OPE:
\be\label{eq:mons-OPE} \CU_\mu(z)\CU_{\mu'}(w) \sim (z-w)^{\mu\cdot \rho^T\rho\mu'}\CU_{\mu+\mu'}(w)\, \ee
In this expression, the vertex operator associated to $\sigma_{\rho^\T\rho\mu, \mu}\vert 0\rangle$ is precisely $\CU_\mu$. If we restrict to $\mu, \mu'\in \Z^r$, then the above OPE is local. Consequently, the direct sum:
\be
\bigoplus_{\mu \in \Z^r}\mathbb U_\mu
\ee
admit the structure of a vertex operator algebra, such that the vertex operator associated to $\sigma_{\rho^\T\rho\mu, \mu}\vert 0\rangle$ is simply $\CU_\mu$.
We propose that the non-perturbative boundary VOA is given by this extension:

\begin{Def} \label{def-VB}
Let $\mc{V}_\rho^B := \bigoplus_{\mu\in \Z^r} \mathbb U_\mu$ denote the extension of $\widehat{\grho}$ by all the modules $\mathbb U_\mu$. 
\end{Def} 

The strongest justification for $\mc{V}_\rho^B$ being the correct boundary VOA in the B twist will be given in Section 
 \eqref{secMirrorsymbdryVOA}, where we prove that there is an equivalence $\mc{V}_\rho^B \cong \CV_{\rho^\vee}^A$ of boundary VOA's for mirror theories. However, we can also give a more intuitive physical justification.
 
Recall that in a 3d $\CN=4$ $U(1)$ abelian gauge theory with no matter, half-BPS monopole operators have an explicit description in terms of the dual photon field $\gamma$ and the vectormultiplet scalar $\sigma$ \cite{SW-3d,IS},
\be U_\mu = e^{i\mu(\gamma+i\sigma)}\qquad (\mu\in \Z)\,. \ee
The dual photon satisfies $d\gamma = * dF$, and the $\pd_z$ derivative of $\gamma+ i\sigma$ is precisely the combination $F_{zt}+i\pd_z\sigma$ whose boundary value defines the field $N(z)$. This suggests that \eqref{mon-vertex} defines the monopole operators in the boundary VOA, at least in the absence of matter. 

We may also consider the action of global symmetries. From \eqref{VOA-B-sym} and \eqref{mon-OPE} it follows that the vertex operator $\CU_\mu$ has charge $\mu$ under the (putative) $T_{\rm top}$ topological symmetry with current $E$; and it has charge $\rho^T\rho\mu$ under the boundary $G$ symmetry. Moreover, using \eqref{L-B} we find that its conformal weight is $\frac{1}{2}|\rho\mu|^2$. This all agrees perfectly with the expected charges of monopole operators in \eqref{mon-charges}.

\subsection{Free field realization}
\label{sec:Bff}

Next, we introduce a free field realization of the Kac-Moody algebra $\widehat{\grho}$ and its extension $\CV_\rho^B$. The free field realization will make the extension much easier to analyze, as it will involve more standard Fock modules. Moreover, it will enable us to prove 3d mirror symmetry of boundary VOA's in the next section.

Let $H_{X,Y,Z}$ denote the Heisenberg VOA associated to the $2r+n$ dimensional vector space with basis $\{X_a,Y^a,Z^i\}$ for $a=1,...,r$ and $i=1,...,n$, and bilinear form
\begin{equation}
    B( X_a,Y^b)=\delta^b{}_a\,,\qquad B( Z^i,Z^j)=\delta^{ij}\,.
\end{equation}
Let
\be \CV_Z = \bigoplus_{\lambda\in \Z^n} \Fock_{\lambda\cdot Z}^{X,Y,Z} \ee
be the lattice VOA obtained by extending $H_{X,Y,Z}$ by Fock modules generated by $\normord{e^{\lambda Z}}$ for all $\lambda\in \Z^r$.
The assignments:
\begin{equation}\label{ffRealization}
    \begin{aligned}
     &N_a\mapsto \partial X_a +\sum_i \rho_{ia}\partial Z^i\,,\\
     &E^a\mapsto  \partial Y^a\,,\\
     &\psi^{i, +}(z)\mapsto \normord{e^{Z^i}}\,,\\
     &\psi^{i,-}\mapsto \normord{\sum_i \rho^i{}_{a}\partial Y^a e^{-Z^i}}+\normord{\partial e^{-Z^i}}
    \end{aligned}
\end{equation}
define an embedding of the Kac-Moody VOA $\widehat{\grho}$ into the lattice VOA $\CV_Z$.

To write the conformal element in terms of the free field algebra, we start with the terms involving $N$ and $E$:
\be
\frac{1}{2}\left(\sum\normord{N_aE^a}+\normord{E^aN_a}\right)=\frac{1}{2}\left(\sum_a\normord{\partial X_a\partial Y^a}+\normord{\partial Y^a\partial X_a}\right)+\sum_{ia}\rho_{ia}\partial Z^i\partial Y^a.
\ee
Now for $\psi^{i,\pm}$, we first have:
\be
\frac{1}{2} \sum_i \left(\normord{e^{Z^i}\partial e^{-Z^i}}-\normord{\partial e^{-Z^i} e^{Z^i}}\right)=-\frac{1}{2}\left(\sum_i \normord{\partial Z_i\partial Z^i}-\normord{\partial^2Z_i}\right)
\ee
which follows from the OPE:
\be
e^{Z(z)}e^{-Z(w)}=\frac{1}{z-w}+\partial Z(w)+(z-w)(\frac{1}{2}\partial^2 Z(w)+\frac{1}{2}\normord{\partial Z(w)\partial Z(w)})+\cdots .
\ee
Using the same OPE, we can compute:
\be
\frac{1}{2} \sum_{i,a}\left(\normord{e^{Z^i}\rho_{ia}\partial Y^ae^{-Z^i}}-\normord{\rho_{ia}\partial Y^ae^{-Z^i} e^{Z^i}}  \right)=\sum_{ia}\rho_{ia}\partial Z^i\partial Y^a-\sum_{i, a}\rho_{ia}\partial^2 Y^a.
\ee
Combining all these, one finds that the stress tensor of $\widehat{\grho}$ takes the following form:
\begin{equation} \label{L-B-ff}
\frac 1 2 \Big(\normord{ \sum_a (\partial X_a\partial Y^a + \partial Y^a \partial X_a) +\sum_i(\partial Z^i)(\partial Z_i)}\Big) + \frac 1 2 \sum_i \Big(\sum_{a}\rho_{ia}\partial^2Y^a-\partial^2Z_i\Big)\,.
\end{equation}
\begin{Exp}
When $\rho=(1)$, we find the free field realization:
\begin{equation}
\begin{aligned}
    & N\mapsto \partial X+\partial Z\,,\qquad E\mapsto \partial Y\\
    & \psi^{+}\mapsto  \normord{e^{Z}}\,,\qquad \psi^-\mapsto \normord{\partial Y e^{-Z}}+\normord{\partial e^{-Z}}.
\end{aligned}
\end{equation}
This is exactly the free field realization of affine $\mathfrak{gl}(1|1)$ used in \cite{creutzig2021duality}, if we shift $N$ by $-E/2$. 
\end{Exp}

We would like to identify this embedding as the kernel of screening operators. Recall that the lattice VOA $\CV_Z$  has Fock modules
\be \CV_{Z,\mu X+\nu Y} = \bigoplus_{\lambda\in \Z^n} \CF_{\mu \cdot X+\nu\cdot Y+\lambda\cdot Z}^{X,Y,Z} \ee
labelled by $\mu,\nu\in \C^r$, lifted from the Fock modules of the Heisenberg VOA $H_{X,Y,Z}$. When $\rho(\mu)\in \Z^n$, define intertwiners $S^i(z): \CV_{Z,\mu X+\nu Y} \to \CV_{Z,\mu X+\nu Y-(\rho Y)^i}(\!(z)\!)$ given by
\be S^i(z)=\normord{e^{Z^i(z)-\sum_a \rho^i{}_aY^a(z)}}\,, \label{def-B-screen} \ee
with corresponding screening charges $S^i_0 = \frac{1}{2\pi i}\oint S^i(z)\,dz$. Note that this is only well-defined when $\rho(\mu)\in \Z^n$, since this ensures that $S^i$ is integer moded. We claim that

\begin{Prop}\label{PropffpurtBgaugeq}
The image of the embedding $\widehat{\grho}\hookrightarrow \CV_Z$ is the kernel of the screening charges:
\begin{equation}
   \widehat{\grho}\,\cong\, \bigcap_{i=1}^n \mathrm{ker}\, S^i_0\big|_{\CV_Z}\,.
    \end{equation}
\end{Prop}

Assuming this result (which we will prove momentarily), it is easy to further identify the extension corresponding to the non-perturbative boundary VOA. The formal integrals of $N$ that appeared in \eqref{mon-vertex} are now identified as $X+\rho Z$, with corresponding vertex operators
\be \CU_\mu^{\rm free} \;=\; \normord{e^{\mu\cdot(X+\rho^T Z)}}\,. \label{Ufree}\ee
One can verify that this does satisfy the OPE of equation \eqref{mon-OPE} and \eqref{eq:mons-OPE}. The screening operators $S^i(z)$ commute with $\CU_\mu^{\rm free}$, since the inner product $B(X_a+(\rho^T Z)_a,Z^i-(\rho Y)^i) = -\rho_{ia}+\rho_{ia}=0$. Thus the $\CU_\mu^{\rm free}$ are in the kernel of the screening charges, and an extension of $\CV_Z$ by the modules generated by $\CU_\mu^{\rm free}$ will commute with taking kernels:
\be \begin{array}{c} \text{$\widehat{\grho}$ extended by $\CU_\mu$'s} \;=\; \text{(kernel of $S$'s on $\CV_Z$) extended by $\CU_\mu$'s} \\
 \cong\; \text{kernel of $S$'s on ($\CV_Z$ extended by $\CU_\mu^{\rm free}$'s)\,.} \end{array} \ee

Explicitly, let
\be \CV_{Z,\mu\cdot(X+\rho^T Z)} = \bigoplus_{\lambda\in \Z^n} \Fock_{\mu\cdot(X+\rho^T Z)+\lambda\cdot Z}^{X,Y,Z} \ee
be the $\CV_Z$ module generated by $\CU_\mu^{\rm free}$. Since $\bigoplus_{\lambda\in \Z^n} \Fock_{(\mu\cdot(X+\rho^T Z)+\lambda\cdot Z}^{X,Y,Z} \,\cong\, \bigoplus_{\lambda\in \Z^n} \Fock_{\mu\cdot X+\lambda\cdot Z}^{X,Y,Z}$, we have $ \CV_{Z,\mu\cdot(X+\rho^T Z)}\,,\cong\, \CV_{Z,\mu\cdot X}$. Denote
\be \CV_{X,Z} \,:=\, \bigoplus_{\mu\in \Z^r} \CV_{Z,\mu\cdot X} \,\cong\, \bigoplus_{\substack{\mu\in \Z^r \\ \lambda\in \Z^n}} \CF_{\mu\cdot X+\lambda\cdot Z}^{X,Y,Z} \ee
the full integral lattice extension of $H_{X,Y,Z}$ by the $X$ and $Z$ lattices. We obtain

\begin{Cor} \label{cor:VB-ff} The non-perturbative boundary VOA from Definition \ref{def-VB} has a free field realization
\be \mc{V}_\rho^B\,\cong\, \bigcap_{i=1}^n \mathrm{ker}\, S^i_0 \big|_{\CV_{X,Z}} \ee
\end{Cor}

\begin{proof}
    As we have shown above, by mapping $\CU_\mu$ to $\CU_\mu^{\mathrm{free}}$, we obtain a map from $\CV_\rho^B$ to $\CV_{X,Z}$ whose image is in the kernel of all $S^i_0$. It is injective since all $\mathbb U_\mu$ are simple as $\widehat{\grho}$ modules, and it is surjective since by \cite{adamovic2003classification, Creutzig2014coset},  the module:
    \be
\bigcap\limits_i \mathrm{Ker} S^i_0\big|_{\CV_{Z, \mu\cdot (X+\rho^\T Z)}}
    \ee
    is also a simple module of $\bigcap\limits_i \mathrm{Ker} S^i_0\big|_{\CV_{Z}}=\widehat{\grho}$. This completes the proof. 
    
\end{proof}

\subsubsection{Free field symmetries}
\label{sec:VB-sym-ff}

It is straightforward to translate the symmetries/currents of $\mc{V}_\rho^B$ to the free-field realization. The $G$ currents are of course $N = \pd X+\rho^T \pd Z$; and the (proposed) $T_{\rm top}$ currents are $E = \pd Y$.  The stress tensor was already identified in \eqref{L-B-ff}. 
The charges of $G$, $T_{\rm top}$, and $T_F$ on each Fock module $\CF_{\mu\cdot X+\lambda\cdot Z}$ are
\be G:\; \rho^T\lambda\,,\qquad T_{\rm top}:\; \mu\,,\qquad  T_F:\;\sigma^T\lambda\,. \label{mon-charges2} \ee
In particular, the charges of boundary monopole operators $\CU_\mu^{\rm free}$ are $\rho^T\rho\mu$, $\mu$, and $\sigma^T\rho\mu$, in agreement with \eqref{mon-charges}.

\subsubsection{Proof of Proposition \ref{PropffpurtBgaugeq}}

First of all, it is straightforward to check that the map $\widehat{\grho}  \to \CV_Z$ is contained in the kernel of $S^i_0$ for all $i$. The fact that this is an embedding follows from the fact that $\widehat{\grho}$ is simple as its own module, which will be proved in Section \ref{sec:VOAlines}, Lemma \ref{Lem:grhosimple}.  To show that it is an isomorphism, we will compare characters.

For any representation $R$ of $G=U(1)^r$ (with weights $N_{a,0}\in \Z$) and $\text{Spin}(2)$ (with conformal weights $L_0\in \frac12 \Z$) let
\be \chi(R) := \text{Tr}_R\; s^{N_0} q^{L_0} = \text{Tr}_R\;s_1^{N_{1,0}}s_2^{N_{2,0}}\cdots s_r^{N_{r,0}} q^{L_0} \ee
where $s_a$ and $q^{\frac12}$ are formal variables. It is clear that when $R=\widehat{\grho}$ or $R=\mc V_Z$, the weight of $L_0$ is bounded from below, and that each $L_0$ weight space is finite-dimensional, so that $\chi(R) \in \Z(\!(q^{\frac12})\!)[s,s^{-1}]$. This character is clearly positive, so if $R$ and $R'$ satisfy the boundedness and finiteness assumptions and $\chi(R)=\chi(R')$, then $R\simeq R'$ as vector spaces. Consequently, if ther is an embedding $R\hookrightarrow R'$, then it is an isomorphism of vector spaces.

Therefore, to finish the proof, we just need to show that the characters coincide:
\be
\chi (\widehat{\grho})=\chi (\mathrm{Ker} \bigcap\limits_i S^i_0 \big|_{\CV_{Z}}).
\ee

Define the Pochhammer symbols
\be \label{Poch} (\alpha;q)_\infty := \prod_{n=0}^\infty (1-q^n \alpha)\,,\qquad (\alpha_1,...,\alpha_k;q)_\infty := \prod_{i=1}^k (\alpha_i;q)_\infty\,. \ee
Then the character of $\widehat{\grho}$ is simply
\begin{equation}\label{eqindexgqB}
 \chi(\widehat{\grho}) \,=\, 
  \frac{1}{(q;q)_\infty^{2r}}\,\prod\limits_{i=1}^n \Big(-q\prod_as_a^{\rho^i{}_a},-q\prod_as_a^{-\rho^i{}_a};q\Big)_\infty\,.
\end{equation}

Next we'd like to compute the character of the kernel of $S^i_0$ on $\CV_Z$. We will do this (schematically) by introducing a filtration whose associated graded makes $X$ and $Y$ commutative, and leaves the character unchanged.

Let $U(\CV_Z)$ denote the mode algebra of $\CV_Z$. By the bose-fermi correspondence, it takes the form
\be U(\CV_Z) = \C[x_{a,k},y^a_k]_{1\leq a\leq n}^{k\in \Z} \otimes  U(\mc{V}_{bc}^{\otimes n})\,, \ee
where $x_{a,k}$ and $y^a_k$ are the modes of $\pd X(z)$ and $\pd Y(z)$, which are non-commutative; and the full lattice extension of $H_Z$ has been dualized to $n$ free fermions. Define the filtration
\be U(\CV_Z) = \bigcup_{d=0}^\infty F_d U(\CV_Z) \ee
where
\begin{equation}
  F_d U(\CV_Z)= \bigcup_{f\leq d}\text{Span}\{x_{a_{1}, k_1}\cdots x_{a_f, k_f}y_*^*\} U(\mc{V}_{bc}^{\otimes n})\,.
\end{equation}
In other words, $F_d U(\CV_Z)$ is generated by words that contain at most $d$ modes of $X$'s followed by an arbitrary (possibly zero) number of modes of $Y$'s.
The associated graded is
\begin{equation}
    \text{Gr}_*FU(\CV_Z)=  \mathbb{C}[x_{a,k},y_k^a]\otimes U(\mc{V}_{bc}^{\otimes n})\,.
\end{equation}
where all the $x$'s and $y$'s now commute. For each $\mu\in \Z^r$, the module $\CV_{Z, \mu\cdot Y}$ as a $U(\CV_Z)$ module is likewise filtered. In the associated graded, we have
\begin{equation}
    \text{Gr}_* F\CV_{Z, \mu\cdot Y}\cong \text{Gr}_* F\CV_{Z,0}=\mathbb{C}[x_{a,k},y_k^a]_{k<0} \otimes \mc{V}_{bc}^{\otimes n}\,,
\end{equation}
and it is clear that the natural map $\CV_{Z, \mu\cdot Y}\to \text{Gr}_* F\CV_{Z, \mu\cdot Y}$ is an isomorphism of vector spaces.

Moreover, since the $S^i_0$'s do not contain $X$ fields, their action preserves the filtrations on the mode algebra and on modules. Therefore, the diagram
\begin{equation}
    \begin{tikzcd}
        \CV_Z\arrow[r, "S^i"] \dar& V_{\mathcal Z, -(\rho Y)^i}\dar\\
        \text{Gr}_* F\CV_Z & \text{Gr}_* F\CV_Z
    \end{tikzcd}
\end{equation}
can be completed to a diagram
\begin{equation}
    \begin{tikzcd}
        \CV_Z\arrow[r, "S^i"] \dar& V_{\mathcal Z, -(\rho Y)^i}\dar\\
        \text{Gr}_* F\CV_Z \arrow[r, "\overline{S}^i"]& \text{Gr}_* F\CV_Z
    \end{tikzcd}
\end{equation}
One easily verifies that $\overline{S}^i$ has the expression:
\begin{equation}
    \overline{S}^i=\frac{1}{2\pi i}\oint\mathrm{d}z\, e^{-\sum_i \rho^i{}_aY^a}\normord{e^{Z^i}}
\end{equation}
where the modes of $Y$  now commute with modes of $X$. By the snake lemma, the induced map:
\begin{equation}
    \bigcap\limits_i \mathrm{Ker}(S^i)\to \bigcap\limits_i \mathrm{Ker}(\overline{S}^i)
\end{equation}
is an isomorphism of vector spaces.

The process of taking associated graded further preserves the conformal grading and the $G$ grading (by $N_{a, 0}$) since the embedding $ F_dV_{\mathcal Z, \mu}\subseteq F_{d+1}V_{\mathcal Z, \mu}$ is one of graded vector spaces. The kernel $\bigcap\limits_i \mathrm{Ker}(\overline{S}^i)$ can be identified easily:
\begin{equation}
    \mathbb{C}[x_{a,k},y_k^a]_{k<0} \otimes \CV_{SF}^{\otimes n}\,, \label{SF-char}
\end{equation}
since the kernel of $\overline{S}^i$ can be easily identified with the kernel of $\oint\mathrm{d}z\normord{e^{Z^i}}$. The piece $e^{-\sum_i \rho^i{}_aY^a}$ in the definition of $\overline{S}^i$ commutes with everything in the associated graded. The character of $\mathbb{C}[x_{a,k},y_k^a]_{k<0}$ is $(q;q)_\infty^{-2r}$, since the $x$'s and $y$'s are free commutative variables with conformal weight $-k$ and uncharged under $G$. The character of $n$ symplectic fermions $\CV_{SF}^{\otimes n}$, transforming in representation $T^*V$ of $G$, is $\prod\limits_{i=1}^n \Big(-q\prod_as_a^{\rho^i{}_a},-q\prod_as_a^{-\rho^i{}_a};q\Big)_\infty$. Thus the character of
\eqref{SF-char} coincides with equation \eqref{eqindexgqB}. This completes the proof.

\subsection{Mirror symmetry of boundary VOA's}
\label{secMirrorsymbdryVOA}

We can now use free field realizations to prove the first main result of the paper.

\begin{Thm}\label{thm:VOA}
Give a pair of simple abelian $\CN=4$ gauge theories $\CT_{\rho}$, $\CT_{\rho^\vee}$ that are 3d mirror, the A-twist boundary VOA of $\CT_{\rho}$ is isomorphic to the B-twist boundary VOA of $\CT_{\rho^\vee}$,
    \begin{equation}
        \CV_\rho^A \,\cong\, \CV_{\rho^\vee}^B\,.
    \end{equation}
Furthermore, this isomorphism preserves stress tensors (conformal elements), mapping \eqref{L-A} to \eqref{L-B}; maps the $T_F\times T_B = U(1)^{2(n-r)}$ Kac-Moody currents of $\CV_\rho^A$ to the $T_{\rm top}\times G$ Kac-Moody currents of $\CV_{\rho^\vee}^B$; and maps the finite $T_{\rm top}=U(1)^r$ symmetry of $\CT_\rho^A$ to the finite $T_F$ symmetry of $\CT_{\rho^\vee}^B$.
\end{Thm}

\noindent\emph{Proof.} \; Consider the free field realization of $\CT_{\rho^\vee}^B$ from Section \ref{sec:Bff}, having replaced $\rho$ with $\rho^\vee$. The free field realization begins with a Heisenberg algebra $H_{X,Y,Z}$ with basis vectors $\{X_\alpha,Y^\alpha,Z^i\}$ for $1\leq \alpha\leq n-r$, $1\leq i\leq n$ with bilinear form
\begin{equation}
    B( X_\alpha,Y^\beta) =\delta_\alpha{}^\beta\,,\qquad ( Z^i,Z^j) =\delta^{ij}\,.
\end{equation}
Note that this Heisenberg algebra is isomorphic to $H_{\phi,\eta,\theta}$ with basis vectors $\{\phi^i,\eta^\alpha,\theta^\alpha\}$ for $1\leq \alpha\leq n-r$, $1\leq i\leq n$ and bilinear form
\be  B( {\phi}^i,{\phi}^j)=\delta^{ij}\,,\qquad  B({\theta}^\alpha,{\theta}^\beta)=-B( {\eta}^\alpha,{\eta}^\beta)=(\rho^{\vee\T}\rho^\vee)^{\alpha\beta}\,, \ee
under the change of basis
\be \label{redef} \phi = Z-\rho^\vee Y\,,\qquad \theta = X+\rho^{\vee\T}Z\,,\qquad  \eta = -X + \rho^{\vee\T}\rho^\vee Y-\rho^{\vee\T}Z\,, \ee
whose inverse is
\be X = -\eta -\rho^{\vee\T}\phi\,,\qquad  Y = (\rho^{\vee\T}\rho^\vee)^{-1}(\theta+\eta)\,,\qquad Z = \phi+\rho^\vee (\rho^{\vee\T}\rho^\vee)^{-1}(\theta+\eta)\,. \label{inv-redef} \ee
Our plan now is to apply the steps used to construct the free-field realization of $\CT_{\rho^\vee}^B$ to $H_{\phi,\eta,\theta}$ rather than $H_{X,Y,Z}$, and manifestly recover the free-field realization of $\CT_\rho^A$.

The next step in the B-twist free field realization is to extend by the Fock modules $\CF_{\tilde \mu\cdot Z}$ for $\tilde \mu\in \Z^n$. Just as in Sections \ref{sec:VA-lattice} and \ref{sec:VA-BRST}, we use an orthogonal decomposition of this lattice, writing each element of $\Z^n$ uniquely as
\be \tilde\mu= \rho(\mu) + \rho^\vee(\mu^\vee) \in \Z^n \quad\text{for}\; \mu \in N\,,\quad \mu^\vee\in N^\vee\,,\quad [\mu]=[\mu^\vee]\in H\,, \ee
recalling that $N,N^\vee$ are rational refinements of $\Z^r$, $\Z^{n-r}$ with $N/\Z^r\simeq  H \simeq N^\vee/\Z^{n-r}$. From \eqref{inv-redef} together with orthogonality of $\rho$ and $\rho^\vee$, we have
\be \tilde \mu\cdot Z= (\rho(\mu) + \rho^\vee(\mu^\vee))\cdot Z = (\rho(\mu) + \rho^\vee(\mu^\vee))\cdot \phi + \mu^\vee\cdot(\theta+\eta)\,. \ee
Therefore,
\be \CV_Z \bigoplus_{\tilde \mu\in \Z^n} \CF_{\tilde \mu\cdot Z}^{X,Y,Z} \,\cong\, \bigoplus_{\substack{\mu \in N,\; \mu^\vee\in N^\vee \\ [\mu]=[\mu^\vee] }} \CF_{(\rho(\mu) + \rho^\vee(\mu^\vee))\cdot \phi}^\phi \otimes \CF_{\mu^\vee\cdot\eta}^\eta \otimes \CF_{\mu^\vee\cdot\theta}^\theta\,. \ee
In the new variables, the screening operators from  \eqref{def-B-screen} are simply $S^i = \normord{e^{(Z-\rho^\vee Y)^i}} = \normord{e^{\phi^i}}$. Taking their kernel and using Proposition \ref{PropffpurtBgaugeq}, we obtain
\be \widehat{\mathfrak g_{\rho^\vee}} \,\cong\, \bigoplus_{\substack{\mu \in N,\; \mu^\vee\in N^\vee \\ [\mu]=[\mu^\vee] }} M_{(\rho(\mu) + \rho^\vee(\mu^\vee))\cdot \phi}^\phi \otimes \CF_{\mu^\vee\cdot\eta}^\eta \otimes \CF_{\mu^\vee\cdot\theta}^\theta\,. \ee

Finally, we extend by monopole operators, which are the Fock modules $\CF_{m\cdot(X+\rho^{\vee\T}Z)} = \CF_{m\cdot\theta}$ for $m\in \Z^{n-r}$. By Corollary \ref{cor:VB-ff}, we have
\be \CV_{\rho^\vee}^B \,\cong\, \bigoplus_{\substack{\mu \in N,\; \mu^\vee\in N^\vee,\; m\in \Z^{n-r} \\ [\mu]=[\mu^\vee] }} M_{(\rho(\mu) + \rho^\vee(\mu^\vee))\cdot \phi}^\phi \otimes \CF_{\mu^\vee\cdot\eta}^\eta \otimes \CF_{(m+\mu^\vee)\cdot\theta}^\theta\,. \label{AB-m} \ee
Defining a shifted summation variable $\lambda^\vee=m+\mu^\vee$, we can rewrite this as
\be \CV_{\rho^\vee}^B \,\cong\, \bigoplus_{\substack{\mu \in N,\; \mu^\vee,\lambda^\vee\in N^\vee \\ [\mu]=[\mu^\vee]=[\lambda^\vee] }} M_{(\rho(\mu) + \rho^\vee(\mu^\vee))\cdot \phi}^\phi \otimes \CF_{\mu^\vee\cdot\eta}^\eta \otimes \CF_{\lambda^\vee\cdot\theta}^\theta\,, \ee
which is manifestly equivalent to the free-field realization of $\CV_\rho^A$ from Proposition \ref{prop:A} (up to relabeling $\eta,\theta\mapsto \eta_\perp,\theta_\perp$).

We must also argue that various symmetries of the A and B VOA's match as expected across 3d mirror symmetry. This is relatively straightforward, since we have already described these symmetries in both the A free-field realization (Section \ref{sec:VA-sym-ff}) and the B free-field realization (Section \ref{sec:VB-sym-ff}).

To match conformal element, we start by rewriting $X_\alpha, Y_\alpha$ and $Z^i$ using $\phi^i,\theta_{ \alpha}$ and $\eta_{ \alpha}$ as follows, using the matrices $\Pi$ and $\Pi^\perp$ from Section \ref{sec:VA-sym-ff}:
\be
\begin{aligned}
& Z_i=\phi_i+\sum\limits_\alpha \Pi^{\perp}_{i\alpha}(\theta^\alpha+\eta^\alpha)\\
& X_\alpha=-\eta^\alpha-\sum (\rho^\vee)_{i\alpha}\phi^i\\
& Y_\alpha=\sum\limits_{i, \beta}\Pi_{i\alpha}^{\perp}\Pi^{\perp, i}{}_{\beta} (\theta^\beta+\eta^\beta)
\end{aligned}
\ee
This coincides with equation \eqref{inv-redef} because $\Pi^\perp \Pi^{\perp, \T}$ is the inverse of $\rho^{\vee \T}\rho^\vee$. We can now rewrite the conformal element from equation \eqref{L-B-ff}. First, we have:
\be
\begin{aligned}
\frac{1}{2}\sum_\alpha\left(\normord{\partial X_\alpha \partial Y^\alpha}+\normord{\partial Y^\alpha \partial X_\alpha}\right)=&-\sum_{\alpha, \beta, i}\Pi_{i\alpha}^{\perp}\Pi^{\perp, i}{}_{\beta} \normord{\partial \eta^\alpha\partial\eta^\beta} -\sum_{\alpha, \beta, i}\Pi_{i\alpha}^{\perp}\Pi^{\perp, i}{}_{\beta} \normord{\partial \eta^\alpha\partial\theta^\beta}\\ &-\sum_{i,\alpha,j,\beta}(\rho^\vee)_{i}{}^\alpha\Pi_{j\alpha}^{\perp}\Pi^{\perp, j}{}_{\beta} \normord{\partial\phi^i(\partial\theta^\beta+\partial\eta^\beta)}
\end{aligned}
\ee
The second part is:
\be
\begin{aligned}
\frac{1}{2}\sum_i \normord{\partial Z_i\partial Z^i}=&\frac{1}{2}\sum_i \normord{\partial \phi_i\partial \phi^i} +\sum_i\Pi^\perp_{i\alpha}\normord{\partial \phi^i (\partial \theta^\alpha+\partial \eta^\alpha)}\\ &+\frac{1}{2}\sum_{i,\alpha,\beta}\Pi^\perp_{i\alpha}\Pi^{\perp, i}{}_\beta \normord{(\partial \theta^\alpha+\partial \eta^\alpha)(\partial \theta^\beta+\partial \eta^\beta)}.
\end{aligned}
\ee
If we add these terms together with $\frac{1}{2}\sum_{i,\alpha} \rho^\vee_{i\alpha}\partial^2Y^\alpha-\partial^2Z_i$, and use the fact that $\rho^\vee (\Pi^\perp)^\T\Pi^\perp=\Pi^\perp$, then we can cancel many terms in the sum, and arrive at the following expression: 
\be
\frac{1}{2}\sum_i \left(\normord{\partial\phi^i\partial\phi_i}-\partial^2\phi^i-\sum_{ \alpha, \beta}\Pi_{i}^{\perp}{}^\alpha\Pi^{\perp, i\beta} \normord{\partial \eta_{ \alpha}\partial \eta_{\beta}}+\sum_{ \alpha, \beta}\Pi_{i}^{\perp}{}^\alpha\Pi^{\perp, i\beta} \normord{\partial \theta_{\alpha}\partial \theta_{\beta}}\right). 
\ee
This matches perfectly with equation \eqref{eq:L-A-ff-cohomology}. Thus the isomorphism above preserves the stress tensor of the two VOA's. 

The boundary Kac-Moody symmetry $T_B=U(1)^{n-r}$ on the A side has current $\pd\theta$ from \eqref{A-JB}. Using \eqref{redef}, this corresponds to $\pd X+\rho^{\vee\T}\pd Z$ on the B side, which in turn is the correct current for $T_G$.

The Kac-Moody symmetry $T_F=U(1)^{n-r}$ on the A side has current $(\rho^{\vee\T}\rho^\vee)^{-1}(\pd\eta+\pd\theta)$ from \eqref{A-JF}. Using \eqref{inv-redef}, this becomes the current $\pd Y$ in the free-field realization on the B side, which is the current for $T_{\rm top}$.

Finally, the finite symmetry $T_{\rm top}=U(1)^r$ on the A side gives a $\Z^r$ grading to the VOA $\CV_\rho^A$; in the free-field realization of $\CV_\rho^A$, each summand has a $T_{\rm top}$ weight given by $\mu+\sigma^{\vee\T}\rho^\vee\lambda^\vee$ as in \eqref{A-Jtop}. In the shifted summation \eqref{AB-m}, with $\lambda^\vee=m+\mu^\vee$, this becomes
\be \mu+\sigma^{\vee\T}\rho^\vee(m+\mu^\vee) = \sigma^{\vee\T}(\rho\mu+\rho^\vee\mu^\vee) + \sigma^{\vee\T}\rho^\vee m = \sigma^{\vee\T}\tilde \mu + \sigma^{\vee\T}\rho^\vee m \,. \label{weight-top-AB} \ee
Let's compare this with \eqref{mon-charges2} on the B side, which states that the $T_F$ weight of a Fock module $\CF_{a\cdot X+b\cdot Z}$ in the free-field realization of $\CV_{\rho^\vee}^B$ should simply be $\sigma^{\vee\T}b$. There is a perfect match, noting that the weight ``$b$'' of $Z$ in the Fock modules of \eqref{AB-m} has two contributions: $\tilde \mu$ from the initial extension by $\CF_{\tilde\mu\cdot Z}$ modules; and $\rho^{\vee}m$ from the subsequent extension by monopoles $\CF_{m\cdot(X+\rho^{\vee\T}Z)}$. Put differently, the first contribution is the weight of $\psi^{i,\pm}$ fields in the perturbative VOA, and the second contribution comes from monopole operators having nontrivial flavor charge.  \hfill$\square$  \medskip

\begin{Exp}
When $\rho=0$ and $\rho^\vee=(1)$ (a free hypermultiplet), the A-twisted VOA $\CV_\rho^A$ is equal to $\mc{V}_{\beta\gamma}\otimes \mc{V}_{bc}$. The mirror theory has $\rho=(1)$ (SQED with one hypermultiplet), and the B-twisted VOA $\CV_{\rho^\vee}^B$ is a simple current extension of $\widehat{\mathfrak{gl}(1|1)}$. The extension is exactly the same one studied in \cite{creutzig2013w}. In terms of the free-field realization using $X,Y,Z$ in Section \ref{sec:Bff}, the extension corresponds to the Fock modules $\Fock_{nX}$ for $n\in \mathbb{Z}$. 
\end{Exp}

\begin{Exp}
Consider $\rho=(1,1)^\T$ and $\rho^\vee=(1,-1)^\T$. We described some of the generators of $\CV_\rho^A$ in \eqref{eq11VOAA}. Theorem \ref{thm:VOA} claims that that $\CV_\rho^A$ can also be realized as an extension of the Kac-Moody superalgebra~$\widehat{\mathfrak{g}_{\rho^\vee}}$. We can identify the explicit map of generators: 
\begin{equation} \label{example-11-VOA}
    \begin{aligned}
    & N\mapsto :b^1c^1:-:b^2c^2:\\
    &\psi^{1,+}\mapsto \beta^1b^1\,,\quad \psi^{2,+}\mapsto \beta^2b^2\\
    & \psi^{1,-}\mapsto -\gamma^1c^1\,,\quad \psi^{2,-}\mapsto -\gamma^2c^2\\
    & E\mapsto  (\normord{b^1c^1}-\normord{\beta^1\gamma^1})
    \end{aligned}
\end{equation}
Most of the relations in $\widehat{\mathfrak g_{\rho^\vee}}$ are preserved by this map, with one exception. We expect
\be \psi^{2,+}(z)\psi^{2,-}(w) \,\sim\, \frac{1}{(z-w)^2} - \frac{E(w)}{z-w}\,, \ee
but from the RHS of \eqref{example-11-VOA} find instead
\begin{equation}
    \psi^{2,+}(z)\psi^{2,-}(w)\,\sim\, \frac{1}{(z-w)^2}+\frac{\normord{b^2c^2}-\normord{\beta^2\gamma^2}}{z-w}\,.
\end{equation}
The difference in the first-order pole is simply the BRST current $J=-\normord{b^1c^1}+\normord{\beta^1\gamma^1}-\normord{b^2c^2}+\normord{\beta^2\gamma^2}$, which is $Q$-exact in $\CV_\rho^A$. We also note that the boundary monopole operators of charge $\pm1$ in $\CV_{\rho^\vee}^B$, which generate the non-perturbative extension, correspond to $\beta^1b^2$ and $\gamma^1c^2$ in $\CV_\rho^A$.

\end{Exp}

\section{Morita-equivalent VOA's from gauging and ungauging}
\label{secMorita}

In the remainder of the paper, we will begin to shift our focus from boundary VOA's themselves to their braided tensor categories of modules. A powerful technique for studying braided tensor cartegories stems from gauging and ungauging --- \emph{a.k.a.} equivariantizing and de-equivariantizing --- various symmetries that act on them. These ideas go back to work on CFT's and string theory in the late 1980's: gauging (orbifolding) a finite symmetry in a modular fusion category was developed in \cite{DVVV}; its inverse was the simple current extension of \cite{SY-simple}. These procedures were later formalized and generalized mathematically. For example, the orbifolding procedure for finite groups, acting on categories, is synthesized in \cite{egno} (with a later generalization in\cite{McRae:2020}); and VOA extensions at the level of categories are captured by the lifting functors of \cite{creutzig2017tensor, creutzig2020direct}.

We would like to apply established gauging and un-gauging operations to study modules of our 3d $\CN=4$ boundary VOA's. For example, it would be convenient to relate module categories in a gauge theory $\CT_\rho$ to module categories in a theory of $n$ free hypermultiplets (which are already understood) via orbifolds and simple current extensions. 
Unfortunately, it is not obvious how to do this directly, since $\CV_\rho^A,\CV_\rho^B$ are \emph{not} related to the respective free-hypermultiplet VOA's $\CV_{\beta\gamma}^{T^*V},\CV_{SF}^{T^*V}$ by standard orbifolds and extensions. Fortunately, there exists a excellent workaround.

In this section, we will introduce two ``alternative'' VOA's that are directly related to free-hypermultiplet VOA's by standard operations; schematically:
\be \label{altAB-intro} \begin{array}{l} \wt\CV_\rho^A := \CV_{\beta\gamma}^{T^*V} \text{\;extended by $\Z^r$ simple currents} \\[.2cm]
\wt\CV_\rho^B := \text{$(\C^*)^r$ orbifold of $\CV_{SF}^{T^*V}\otimes \CV_{\Z^n}^-$}\,,\end{array}
\ee
where the $\Z^r$ currents and orbifold action depend on the charge matrix $\rho$, and $\CV_{\Z^n}^-$ is a complete Euclidean lattice VOA of negative-definite bilinear form. We will then prove in Theorem \ref{ThmMirSymbdVOA} that these alternative VOA's are related to the actual ones by tensoring with additional self-dual lattice VOA's. Combined with 3d mirror symmetry (Theorem \ref{thm:VOA}), we have
\be \label{Morita-intro} \wt\CV_\rho^A \otimes \mc{V}_{bc}^{\otimes n} \;\cong\; \wt\CV_{\rho^\vee}^B \otimes \mc{V}_{bc}^{\otimes n} \;\cong\; \CV_\rho^A\otimes \CV_{\Z^{2r}}^{+-}\;\cong\; \CV_{\rho^\vee}^B \otimes  \CV_{\Z^{2r}}^{+-}\,.\ee

Now, self-dual lattice VOA's have trivial module categories. Therefore, \eqref{Morita-intro} turns out to imply that the module categories of $\wt\CV_\rho^A,\wt\CV_\rho^B$ are equivalent to those of $\CV_\rho^A,\CV_\rho^B$. This is known as Morita equivalence. It ultimately allows us to replace the actual VOA's with the alternative ones in order to study module categories.

It might seem surprising that a $(\C^*)^r$ orbifold did not appear in both lines of \eqref{altAB-intro}. The physical origin of this is explained later, in Section \ref{sec:highersym}. In brief, the $U(1)^n$ flavor symmetry of $n$ hypermultiplets, part of which we gauge to obtain $\CT_\rho$, behaves very differently depending on which topological twist we take. In the B twist, it is complexified to a $\C^*$ symmetry, leading to a standard orbifold. In the A twist, the bulk $U(1)^n$ symmetry is instead promoted to a $\Z^n$ one-form symmetry. The generators of the one-form symmetry are ``flavor vortex'' line operators. Gauging part of the $U(1)^n$ flavor symmetry in the A twist effectively acts by gauging the associated one-form symmetry. It screens bulk line operators by flavor vortices that have become dynamical. On the boundary, it induces a simple current extension.

We already saw a reflection of this phenomenon in Sections \ref{sec:VA-sym} and \ref{sec:VB-sym-ff}, when we discussed symmetries of boundary VOA's. For $n$ free hypermultiplets in the B twist, $U(1)^n$ remains a finite symmetry of the boundary VOA $\CV_{SF}^{T^V}$ (so it makes sense to orbifold by it). For $n$ free hypermultiplets in the A twist, $U(1)^n$ is promoted to a Kac-Moody symmetry on the boundary. The generators of the bulk $\Z^n$ one-form symmetry are ``flavor vortex'
lines, whose endpoints on the boundary produce spectral-flow modules for the Kac-Moody symmetry. We will review the notion of spectral-flow modules in Section \ref{subsecgaugeextA}.

\subsection{$\wt{\mc{V}}_{\rho}^A$ by extending $\mc{V}_{\beta\gamma}^{T^*V}$}\label{subsecgaugeextA}

Recall from \eqref{nbg-ff} that $n$ copies of the beta-gamma VOA has a free-field realization
\be \mc{V}_{\beta\gamma}^{T^*V} \;\cong\; \bigoplus_{\lambda\in \Z^n} M_{\lambda\cdot\phi}^\phi \otimes \Fock_{\lambda\cdot\eta}^\eta\,, \label{betagamma-ff2} \ee
where $M_{\lambda\cdot\phi}^\phi\subset \Fock_{\lambda\cdot\phi}^\phi$ are singlet modules embedded in the Fock modules of a rank-$n$ Heisenberg algebra with $B(\phi^i,\phi^j)=\delta^{ij}$; and $\Fock_{\lambda\cdot\eta}^\eta$ are Fock modules for a rank-$n$ Heisenberg algebra with $B(\eta^i,\eta^j)=-\delta^{ij}$. This is the boundary VOA for free $T^*V$ hypermultiplets in the A twist --- modulo free fermions $\CV_{bc}^V$, which have a trivial module category.

As explained further in Section \ref{sec:highersym}, free hypermultiplets in the A twist have a $\Z^n$ one-form symmetry, generated by flavor vortex lines.
On the boundary, the symmetry becomes a $\Z^n$ spectral-flow automorphism of $\CV_{\beta\gamma}^{T^*V}$, associated to its $\mathfrak{gl}(1)^n$ Kac-Moody symmetry. In particular, the endpoints of bulk flavor vortices are spectral-flow modules
\be \V_m := \Sigma_m(\CV_{\beta\gamma}^{T^*V})\,. \ee
By the general theory reviewed in Section \ref{sec:spectralflow} (based on \cite{Li:1997}), the $\V_m$ are simple currents for $\CV_{\beta\gamma}^{T^*V}$. In particular, their fusion rules are
\be \V_m\times\V_{m'} = \V_{m+m'} \ee
and they have trivial monodromy around each other.
In the free-field realization, where the $\mathfrak{gl}(1)^n$ currents is $J_F^i=\partial\eta^i$, the vertex operator generating each spectral-flow module is $\normord{e^{\mu\cdot\eta(z)}}$. 
The module itself is therefore
\be \V_\mu \,\simeq\, \bigoplus_{\lambda\in\Z^n} M_{\lambda\cdot\phi}^\phi \otimes \Fock_{(\lambda+\mu)\cdot\eta}^\eta\,.\ee

Now we are in a position to analyze the effect of screening vortex lines. Gauging a $U(1)^r$ subgroup of the bulk symmetry with charges $\rho$ screens flavor vortices labelled by elements of the sublattice $m\in \rho(\Z^r)\subseteq \Z^n$.
Thus suggests
\begin{Def} \label{def:altA}
For a simple abelian gauge theory with charge matrix $\rho$, the alternative A-twisted boundary VOA is
\begin{align}  \wt\CV_\rho^A &\,:=\,  \bigoplus_{\mu\in \Z^r} \V_{\rho(\mu)} \qquad \text{(as $\mc{V}_{\beta\gamma}^{T^*V}$ modules)} \\
&\,\cong\, \bigoplus_{\substack{\lambda\in \Z^n \\ \mu\in \Z^r}} M_{\lambda\cdot\phi}^\phi \otimes \Fock_{(\lambda+\rho(\mu))\cdot\eta}^\eta\,. \end{align}
\end{Def}

We can simplify this a bit further, similar to how we simplified the free-field realization of $\CV_\rho^A$ in Section \ref{sec:VA-BRST}. Splitting $\lambda=\rho(\lambda_0)+\rho^\vee(\lambda^\vee)$ for $\lambda_0\in N\supset \Z^r$, $\lambda^\vee\in N^\vee\supset \Z^{n-r}$ with $[\lambda_0]=[\lambda^\vee]\in H$, and decomposing the second Heisenberg VOA as
\be \CH_\eta \,\cong\, \CH_{\eta_0}\otimes \CH_{\eta^\perp}\,,\qquad \eta_0=\rho^\T\eta\,,\quad \eta_\perp=\rho^{\vee\T}\eta\,, \ee
such that the $\eta_0$ fields have bilinear form $\rho^T\rho$ and the $\eta_\perp$ fields have bilinear form $\rho^{\vee\T}\rho^\vee$, we obtain
\be \wt\CV_\rho^A\;\cong\; \bigoplus_{\substack{\lambda_0\in N,\;\lambda^\vee\in N^\vee \\ \mu\in \Z^r \\ [\lambda]=[\lambda^\vee]\in H}} M_{(\rho(\lambda_0)+\rho^\vee(\lambda^\vee))\cdot \phi}^\phi \otimes \Fock_{\lambda^\vee\cdot\eta_\perp}^{\eta_\perp}\otimes \Fock_{(\lambda+\mu)\cdot\eta_0}^{\eta_0}\,.\ee
Further redefining $\mu\mapsto \mu-\lambda$ this becomes
\be \label{VA-alt-decomp} \wt\CV_\rho^A\;\cong\; \bigoplus_{\substack{\lambda_0,\mu\in N,\;\lambda^\vee\in N^\vee \\ [\lambda]=[\lambda^\vee]=[\mu]\in H}} M_{(\rho(\lambda_0)+\rho^\vee(\lambda^\vee))\cdot \phi}^\phi \otimes \Fock_{\lambda^\vee\cdot\eta_\perp}^{\eta_\perp}\otimes \Fock_{\mu\cdot\eta_0}^{\eta_0}\,.\ee

We recognize in the first two factors the VOA $\CM_{\rho,0}$ from \eqref{def-M0} (and its modules $\CM_{\rho,h}$ labelled by $h\in H$). The third factor can be summed in each $H$ equivalence class to give a negative-definite analogue of the lattice VOA \eqref{def-V0} and its simple $H$ currents \eqref{def-Vh}, namely
\be \CV_{\rho^\vee,h}^{-} \;:=\; \bigoplus_{\lambda^\vee\in \Z^{n-r}+h_{N^\vee}} \CF_{\lambda^\vee\cdot\eta_\perp}^{\eta_\perp}\,.\ee
Altogether, $\wt \CV_\rho^A$ may be expressed as an $H$ simple current extension
\be \wt\CV_\rho^A \;\cong\; \bigoplus_{h\in H} \CM_{\rho,h}\otimes \CV_{\rho^\vee,h}^{-}\,. \ee
This is however \emph{not} (in any obvious way) equivalent as the expression for $\CV_\rho^A$ in \eqref{eqBRSTbgq}, due to the negative norm of the lattice VOA.

\subsection{$\wt{\mc{V}}_\rho^B$ as an orbifold}
\label{subsecgaugeBalt}

Next, consider $n$ symplectic fermions $\CV_{SF}^{T^*V}$, the boundary VOA for free $T^*V$ hypermultiplets in the B twist. This has the free-field realization
\be \CV_{SF}^{T^*V} \;\cong\; \bigoplus_{\lambda\in \Z^n} M_{\lambda\cdot\phi}^\phi\,, \ee
where each singlet module $M_{\lambda\cdot\phi}^\phi\subset \Fock_{\lambda\cdot\phi}^\phi$ is embedded in a Fock module for a rank-$n$ Heisenberg VOA with $B(\phi^i,\phi^j)=\delta^{ij}$. We tensor this with a complete negative-definite lattice VOA
\be \CV_{\Z^n}^- \;:=\; \bigoplus_{\mu\in\Z^n}\Fock_{\mu\cdot \eta}^\eta\,, \ee
built from a rank-$n$ Heisenberg VOA with $B(\eta^i,\eta^j)=-\delta^{ij}$.

Note that tensoring with $\CV_{\Z^n}^-$ does not affect the category of modules, since $\CV_{\Z^n}^-$ has a trivial module category. The reason for tensoring with $\CV_{\Z^n}^-$ is the following.

Consider the diagonal $(\C^*)^r$ action on  $\CV_{SF}^{T^*V}\otimes \CV_{\Z^n}^- \,\cong\, \bigoplus_{\lambda,\mu\in\Z^n} M_{\lambda\cdot\phi}^\phi\otimes  \Fock_{\mu\cdot \eta}^\eta$ under which the subspace $M_{\lambda\cdot\phi}^\phi\otimes  \Fock_{\mu\cdot \eta}^\eta$ has weight $\rho^\T(\lambda-\mu)$.
\begin{Def} \label{def:altB}
Let $\wt{\mc{V}}_\rho^B$ be the $(\C^*)^r$ orbifold (i.e. invariant subspace) of $\CV_{SF}^{T^*V}\otimes \CV_{\Z^n}^-$ under the diagonal action above,
\be \wt{\mc{V}}_\rho^B \;:=\; \big(\CV_{SF}^{T^*V}\otimes \CV_{\Z^n}^-\big)^{(\C^*)^r}\,. \ee
\end{Def}
\noindent Then we have that
\begin{Lem} \label{lem:AB-alt}
There is an isomorphism $\wt\CV_\rho^A\,\cong\,\wt\CV_{\rho^\vee}^B$ for any 3d-mirror pair of simple abelian theories. \end{Lem}

\begin{proof}
By definition
\be \CV_{\rho^\vee}^B\;\cong\; \Big[  \bigoplus_{\lambda,\mu\in\Z^n} M_{\lambda\cdot\phi}^\phi\otimes  \Fock_{\mu\cdot \eta}^\eta \Big]^{(\C^*)^{n-r}} \ee
where $(\C^*)^{n-r}$ acts with weights $\rho^{\vee\T}(\lambda-\mu)$ on the $(\lambda,\mu)$ subspace. As usual, we decompose $\lambda=\rho(\lambda_0)+\rho^\vee(\lambda^\vee)$ and $\mu=\rho(\mu_0)+\rho^\vee(\mu^\vee)$ with $\lambda_0,\mu_0\in N$, $\lambda^\vee,\mu^\vee\in N^\vee$, and $[\lambda_0]=[\lambda^\vee]\in H$, $[\mu_0]=[\mu^\vee]\in H$. The $(\C^*)^{n-r}$ weights become $\rho^{\vee\T}\rho^\vee(\lambda^\vee-\mu^\vee)$, and taking invariants sets $\lambda^\vee=\mu^\vee$. We are left with
\begin{align} \CV_{\rho^\vee}^B &\;\cong\; \bigoplus_{\substack{\lambda_0,\mu_0\in N,\;\lambda^\vee\in N^\vee \\ [\lambda_0]=[\lambda^\vee]=[\mu_0]\in H}} M_{(\rho(\lambda_0)+\rho^\vee(\lambda^\vee))\cdot\phi}^\phi\otimes  \Fock_{(\rho(\mu_0)+\rho^\vee(\lambda^\vee))\cdot \eta}^\eta \\
&= \bigoplus_{\substack{\lambda_0,\mu_0\in N,\;\lambda^\vee\in N^\vee \\ [\lambda_0]=[\lambda^\vee]=[\mu_0]\in H}} M_{(\rho(\lambda_0)+\rho^\vee(\lambda^\vee))\cdot\phi}^\phi\otimes \Fock_{\lambda^\vee\cdot\eta^\perp}^{\eta_\perp} \otimes \Fock_{\mu_0\cdot\eta_0}^{\eta_0}\,, \label{VB-alt-decomp}
\end{align}
where in the last line we decomposed the $\eta$ Heisenberg VOA as $ \CH_\eta \,\cong\, \CH_{\eta_0}\otimes \CH_{\eta^\perp}$ with $\eta_0=\rho^\T\eta$ and $\eta_\perp=\rho^{\vee\T}\eta$. The final expression \eqref{VB-alt-decomp} perfectly matches \eqref{VA-alt-decomp}.

\end{proof}

\subsection{Morita-equivalent VOA's}
\label{subsecrelatealt}

To prove that the simplified VOA's of the previous sections are Morita-equivalent to the actual boundary VOA's of 3d $\CN=4$ theories, we now define yet another VOA and relate it to all the ones appearing above.

Consider the 3d mirror of the construction from Section \ref{subsecgaugeextA}. There, we begin with $n$ free hypermultiplets in the A twist, with boundary VOA $\mc{V}_{\beta\gamma}^{T^*V}$; then we extended the VOA by spectral-flow modules corresponding to bulk vortex lines that would be screened up on gauging. Now we'll begin with the 3d mirror, $n$ copies of $\text{SQED}_1$ in the B twist. Its boundary VOA $\big(\CV_{(1)}^B\big)^{\otimes n}$ is the extension of $\widehat{\mathfrak{gl}(1|1)}^{\otimes n}$ by boundary monopole operators. Explicitly, from Section \ref{subsecperturbalg}, $\widehat{\mathfrak{gl}(1|1)}^{\otimes n}$ has generators $N_i(z),\psi^{i,\pm}(z),E^i(z)$, with
\be \begin{array}{c} N_i(z)E^j(z) \sim \frac{\delta_i{}^j}{(z-w)^2}\,,\qquad N_i(z)N_j(w) \sim\frac{\delta_{ij}}{(z-w)^2} \\[.2cm]
 N_i(z)\psi^{j,\pm}(w) \sim \pm \frac{\delta_i{}^j\psi^{j,\pm}}{z-w}\,,\qquad  \psi^{i,+}(z)\psi^{j,-}(w) \sim \frac{\delta^{ij}}{(z-w)^2} + \frac{\delta^{ij} E^j(w)}{z-w}\,. \end{array} \ee
Then $\big(\CV_{(1)}^B\big)^{\otimes n}$ is obtained from $\widehat{\mathfrak{gl}(1|1)}^{\otimes n}$ by extending by modules generated by $\normord{e^{\nu\cdot \int N}}$ for all $\nu\in \Z^n$. We note from Theorem \ref{thm:VOA} that
\be \big(\CV_{(1)}^B\big)^{\otimes n}\;\cong\; \mc{V}_{\beta\gamma}^{T^*V} \otimes \mc{V}_{bc}^V\,, \label{n-hypers-SQED} \ee
with extra $\mc{V}_{bc}$ factors due to Definition \ref{def-VA} of canonical A-twisted VOA's.

In order to obtain a mirror of $\wt\CV_\rho^A$ from Definition \ref{def:altA}, we should \emph{further} extend by spectral flow modules, associated to the mirror of the Kac-Moody currents from the A side. The mirror of the current $\pd\eta^i$ are just $E^i$. This motivates
\begin{Def}
Given a charge matrix $\rho$ for a simple abelian theory, define
\begin{align} \CV_{\rho}^{\star} &\,:=\, \text{$\big(\CV_{(1)}^B\big)^{\otimes n}$ extended by $\normord{e^{\rho(\mu)\cdot\int E}}\;\; \forall\, \mu\in \Z^r$}    \label{def-V*} \\
&\;\cong\, \text{$\widehat{\mathfrak{gl}(1|1)}^{\otimes n}$ extended by $\normord{e^{\nu\cdot\int N}}$ and $\normord{e^{\rho(\mu)\cdot\int E}}$  \;\;$\forall\,\nu\in \Z^n,\;\mu\in \Z^r$} \notag
\end{align}
where in both cases we mean extended by modules corresponding to the vertex operators.
\end{Def}
\noindent In the bulk physics, the extension by $\normord{e^{\rho(\mu)\cdot\int E}}$ can be interpreted as screening (removing) Wilson line operator operators corresponding to part of the gauge symmetry of $(\text{SQED}_1)^{\otimes n}$ that should be un-gauged to get $\CT_{\rho^\vee}$.

The relation between $\CV_\rho^\star$ and $\wt \CV_\rho^A$ follows almost by construction:
\begin{Lem} \label{lem:V*B}
There is an isomorphism of VOA's
\be \CV_\rho^\star\,\cong\, \wt\CV_\rho^A \otimes \mc{V}_{bc}^V \ee
\end{Lem}

\begin{proof}
Consider the equivalence $\big(\CV_{(1)}^B\big)^{\otimes n}\;\cong\; \mc{V}_{\beta\gamma}^{T^*V} \otimes \mc{V}_{bc}^V$ of \eqref{n-hypers-SQED}. Theorem \ref{thm:VOA} identifies the $T_{\rm top}=U(1)^n$ Kac-Moody symmetry on the left (with currents $E^i$) with the $T_F=U(1)^n$ Kac-Moody symmetry on the right (with currents $\normord{-\beta^i\gamma^i+b^ic^i}$). Thus, an extension of the LHS by all spectral-flow modules for a subgroup $G\subset U(1)^n$ with charges $\rho$ is isomorphic to an extension of the RHS the corresponding spectral-flow modules. Moreover, on the RHS, extending by spectral-flow modules commutes with tensoring by $\mc{V}_{bc}^V$, because the spectral-flow modules for the factor $\mc{V}_{bc}^V$ alone are trivial (they are isomorphic to vacuum modules). Thus, upon extension of both sides, $\big(\CV_{(1)}^B\big)^{\otimes n}\;\cong\; \mc{V}_{\beta\gamma}^{T^*V} \otimes \mc{V}_{bc}^V$ implies $\CV_\rho^\star\,\cong\, \wt\CV_\rho^A \otimes \mc{V}_{bc}^V.$
\end{proof}

It also follows from Lemma \ref{lem:AB-alt} that $\CV_\rho^\star \,\cong\, \wt\CV_{\rho^\vee}^B\otimes \mc{V}_{bc}^V$.
We'd now like to relate $\CV_\rho^\star$ and $\CV_{\rho^\vee}^B$. To do so, we use free-field realizations of both, following Section \ref{sec:Bff}.

In the case of $\CV_\rho^\star$, we introduce a rank-$3n$ Heisenberg algebra $\CH_{\wt X,\wt Y,\wt Z}$ with basis $\{\wt X_i,\wt Y^i,\wt Z^i\}_{i=1}^n$ and bilinear form $B(\wt Z^i,\wt Z^j)=B(\wt X_i,\wt Y^j)=\delta^{ij}$.
Then
\be\label{gl11ff} \widehat{\mathfrak{gl}(1|1)}^{\otimes n} \,\cong\,  \bigcap_{i=1}^n\text{ker} \, S_0^i\Big| \bigoplus_{\lambda\in \Z^n} \CF_{\lambda\cdot\wt Z}^{\wt X,\wt Y,\wt Z}\,,\qquad  S^i  \;=\; \normord{e^{\wt Z^i-\wt Y^i}}\,, \ee
where $N\mapsto \pd \wt X+\pd\wt Z$, $E\mapsto \pd \wt Y$, $\psi^+\mapsto \normord{e^{\wt Z}}$, and $\psi^-\mapsto \normord{(\pd Y-\pd Z)e^Z}$. Correspondingly,
\be \big(\CV_{(1)}^B\big)^{\otimes n} \,\cong\, \bigcap_{i=1}^n\text{ker} \, S_0^i\Big| \bigoplus_{\lambda,\mu\in \Z^n} \CF_{\lambda\cdot\wt Z+\mu\cdot\wt X}^{\wt X,\wt Y,\wt Z}\,, \ee
\be  \CV_\rho^\star \,\cong\, \bigcap_{i=1}^n\text{ker} \, S_0^i\Big| \bigoplus_{\substack{\lambda,\mu\in \Z^n \\ \nu\in \Z^r}} \CF_{\lambda\cdot\wt Z+\mu\cdot\wt X+\rho(\nu)\cdot\wt Y}^{\wt X,\wt Y,\wt Z}\,. \label{V*-ff} \ee

Now we'll change basis in the Heisenberg algebra, in a way that allows us to factor off a complete rank $2r$ lattice VOA.
We choose (as in Section \ref{sec:global}) a splitting $\sigma:\Z^{n-r}\to\Z^n$ and co-splitting $\sigma^\vee:\Z^r\to \Z^n$ of the maps in the exact sequence associated to our 3d theory, whose maps satisfy
\be \rho^\T\sigma^\vee=\text{id}_{\Z^r}\,,\qquad \rho^{\vee\T}\sigma = \text{id}_{\Z^{n-r}}\,,\qquad \rho^{\vee\T}\rho = \sigma^{\vee\T}\sigma = 0\,, \ee
or simply  $(\rho^\vee,\sigma^\vee)=(\rho,\sigma)^{-1,\T}$\,. Then we relate $\{\wt X_i,\wt Y^i,\wt Z^i\}_{i=1}^n$ to $\{X_a',X_\alpha,Y'{}^a,Y^\alpha,Z^i\}$ with $1\leq a \leq r$, $1\leq \alpha\leq n-r$, $1\leq i\leq n$ using the transformation
\begin{subequations}
\label{V*-redef}
\be  Z=\wt Z-\sigma^\vee\rho^\T \wt Y\,,\qquad \begin{array}{c}   X'=\sigma^{\vee\T}(\wt X+\wt Z)\,, \\[.1cm] X = \rho^{\vee\T} \wt X\,, \end{array} \qquad \begin{array}{c} Y' = \rho^\T \wt Y\,, \\[.1cm] Y = \sigma^\T \wt Y\,,\end{array} \ee
whose inverse is
\be \wt Z = Z+\sigma^\vee Y'\,,\qquad \wt X = \rho X'+  \sigma X  - \rho\sigma^{\vee\T}(Z+\sigma^\vee Y')\,,\qquad \wt Y = \sigma^\vee Y' + \rho^\vee Y\,.\ee
\end{subequations}
We obtain an isomorphic Heisenberg algebra $\CH_{\wt X,\wt Y,\wt Z}\,\cong,\CH_{X,X',Y,Y',Z}$ with bilinear form
\be B(Z^i,Z^j) = \delta^{ij}\,,\qquad B(X_\alpha,Y^\beta) = \delta_\alpha{}^\beta\,,\qquad B(X_a',Y'{}^b) = \delta_a{}^b \ee
(and all other pairings of basis elements vanishing).

Under this transformation, the screening operator becomes $S^i = \normord{e^{Z^i - (\rho^\vee Y)^i}}$. Replacing the $\mu$ summation index in \eqref{V*-ff} with $\sigma^\vee\mu'+\rho^\vee\mu$ where $\mu'\in \Z^r$ and $\mu\in \Z^{n-r}$, we then obtain
\begin{align} \CV_\rho^\star &\,\cong\, \bigcap_{i=1}^n \text{ker}\,S_0^i\Big| \bigoplus_{\substack{\lambda\in \Z^n,\; \mu\in \Z^{n-r} \\ \mu',\nu\in \Z^r }} \CF^{X,X',Y,Y',Z}_{(\lambda-\sigma^\vee(\mu'))\cdot Z + \mu\cdot X+\mu'\cdot X' +(\nu+\sigma^{\vee\T}(\lambda)-\sigma^{\vee\T}\sigma^\vee(\mu'))\cdot Y'}  \\
 &\,=\, \bigcap_{i=1}^n \text{ker}\,S_0^i\Big| \bigoplus_{\substack{\lambda\in \Z^n,\; \mu\in \Z^{n-r} \\ \mu',\nu\in \Z^r }} \CF^{X,X',Y,Y',Z}_{\lambda\cdot Z+\mu\cdot X+\mu'\cdot X'+\nu \cdot Y'}\,,
 \end{align}
where in the second step we have shifted summation variables $\lambda\mapsto \lambda+\sigma^\vee(\mu')$ and (subsequently) $\nu\mapsto \nu-\sigma^{\vee\T}(\lambda)$. The Fock module then factors, and we find
\begin{align} \CV_\rho^\star &\,\cong\, \bigg[\bigcap_{i=1}^n \text{ker}\,S_0^i\Big| \bigoplus_{\lambda\in \Z^n,\;\mu\in \Z^{n-r}} \Fock^{X,Y,Z}_{\lambda\cdot Z+\mu\cdot X}\bigg] \;\otimes\; \bigoplus_{\mu',\nu\in \Z^r} \Fock^{X',Y'}_{\mu'\cdot X'+\nu\cdot Y'}\,,   \label{V*-factor}
\end{align}
where the screening operator $S^i=\normord{e^{Z^i-(\rho^\vee Y)^i}}$ only acts on the first factor.
Comparing with the free-field realization in Corollary \ref{cor:VB-ff}, with the screening operators in \eqref{def-B-screen}, we identify the first factor in \eqref{V*-factor} as $\CV_{\rho^\vee}^B$. The second factor is a complete rank-$2r$ lattice VOA, in an indefinite lattice with bilinear form $\left(\begin{smallmatrix} 0 & 1 \\ 1 & (\sigma^\vee)^\T\sigma^\vee \end{smallmatrix}\right)$. We denote it $\CV_{\Z^{2r}}^{+-}$. Thus we have shown that $\CV_\rho^\star\,\cong\, \CV_{\rho^\vee}^B \otimes \CV_{\Z^{2r}}^{+-}$. Combining this with Theorem \ref{thm:VOA} on 3d mirror symmetry of boundary VOA's, as well as Lemmas \ref{lem:AB-alt} and \ref{lem:V*B} above, we arrive at

\begin{Thm}\label{ThmMirSymbdVOA}
Given any pair of 3d-mirror simple abelian theories, with charge matrices $\rho,\rho^\vee$, there are isomorphisms of VOA's 
\be \CV_\rho^A\otimes \CV_{\Z^{2r}}^{+-} \;\cong\; \CV_{\rho^\vee}^B\otimes \CV_{\Z^{2r}}^{+-}  \;\cong\; \CV_\rho^\star \;\cong\; \wt\CV_\rho^A\otimes \mc{V}_{bc}^V \;\cong\; \wt \CV_{\rho^\vee}^B\otimes \mc{V}_{bc}^V\,. \ee
\end{Thm}

\begin{Cor}
The VOA's $\CV_\rho^A,\CV_{\rho^\vee}^B,\wt\CV_\rho^A,\wt\CV_{\rho^\vee}^B$ are all Morita-equivalent, meaning that their abelian braided tensor categories of modules are equivalent.
\end{Cor}

\begin{proof}
This follows from the Theorem once we establish that $\mc{V}_{bc}^V$ and $\CV_{\Z^{2r}}^{+-}$ have trivial module categories. This in turn follows from the fact that that both $\mc{V}_{bc}^V$ and $\CV_{\Z^{2r}}^{+-}$ are lattice VOA's corresponding to \emph{self-dual} lattices, which is easy to check. Self-dual lattice VOA's are Morita trivial \cite{DLM}.
\end{proof}

\section{Interlude: topological one-form symmetry}
\label{sec:highersym}

We'd now like to make some general remarks on the structure of symmetries in the bulk topological A and B twists. In particular, we'd like to explain the appearance of higher one-form symmetries \cite{GKSW} that act on categories of line operators, controlling operations of gauging and ungauging. These symmetries manifest as spectral flow in boundary VOA's, which we've already seen in Section \ref{subsecperturbalg} and Section \ref{secMorita}.

The one-form symmetries discussed here are exclusively features of topological twists. Indeed, we remarked back in Section \ref{sec:global} that a simple abelian theory $\CT_\rho$ cannot (essentially by definition) have a physical one-form symmetry. What happens in the topological twists is that part of the zero-form flavor symmetry ($T_F$ in the A twist and $T_{\rm top}$ in the B twist) becomes cohomologically trivial, and is promoted to an infinite one-form symmetry. We explain this, in general terms, in Section \ref{sec:highersym}. The argument is a particular manifestation of symmetries acting topologically --- via their homotopy type --- in TQFT; see \cite{Teleman-ICM} for a discussion in a mathematical-physics context.

The remaining flavor symmetry ($T_{\rm top}$ in the A twist and $T_F$ in the B twist) is also interesting: it gets complexified, and allows the bulk theory to couple to flat complex connections. Preliminary investigations of this structure appeared in \cite{CDGG}.

Altogether, the zero-form flavor symmetry of $\CT_\rho$ induces the following symmetry in the topological twists:
\be \begin{array}{c|c|c}
 & \CT_\rho^A & \CT_\rho^B \\\hline
T_F & \pi_1(T_F)\simeq \Z^r\;\text{one-form} & (\C^*)^r\;\text{zero-form} \\
T_{\rm top} &(\C^*)^{n-r}\;\text{zero-form} & \pi_1(T_{\rm top})\simeq \Z^{n-r}\;\text{one-form}  
\end{array}
\ee
This may usefully be compared with symmetries in the boundary VOA's from Sec.~\ref{sec:VA-sym},~\ref{sec:sym-B}:
\be \begin{array}{c|c|c}
 & \CV_\rho^A & \CV_\rho^B \\\hline
T_F &  \text{$\mathfrak{gl}(1)^r$ Kac-Moody $\leadsto$ $\Z^r$ spectral flow} & \text{finite $(\C^*)^r$} \\
T_{\rm top} & \text{finite $(\C^*)^{n-r}$} &  \text{$\mathfrak{gl}(1)^{n-r}$ Kac-Moody $\leadsto$ $\Z^{n-r}$ spectral flow} 
\end{array}
\ee

In the remainder of this section, we will develop some general physical arguments, without any reference to particular theories or twists. We'll explain the origin of topological symmetry enhancements, then focus on the case of $\Z$ one-form symmetries that is ultimately most relevant for us. We'll discuss their implications for categories of line operators, and their relation to boundary spectral flow.

\subsection{The origin of topological actions}
\label{sec:oneform}

Consider any 3d TQFT $\CT^Q$ of cohomological type, obtained as the twist of an underlying non-topological theory $\CT$. Suppose further that $\CT$ has a non-anomalous global symmetry $G$, where $G$ is a \emph{continuous}, connected group with Lie algebra $\mathfrak g$; and suppose that the current $J$ for this symmetry is $Q$-exact, satisfying
\be J = Q(S)\,. \label{J-exact} \ee
(Recall that the current is a spacetime one-form values in $\mathfrak g^*$.)

Given an element $g=e^\alpha\in G$ and a codimension-one submanifold $\Sigma$ in spacetime, one would construct a  defect implementing a symmetry transformation by $g$ across $\Sigma$ by integrating the current:
\be  \CO[\Sigma,g] := e^{\int_\Sigma \langle \alpha,*J\rangle}\,. \ee
(The exponential is to be interpreted with a point-splitting regularization, determined by a co-orientation of $\Sigma$.)
This defect is topological even in the physical theory $\CT$ due to current conservation $d*J=0$.

In the topologically twisted theory $\CT^Q$, we can go further, and consider symmetry transformations with \emph{non-constant} parameters. For any continuous map $\hat g:\Sigma\to G$, represented by a (possibly discontinuous) map $\hat \alpha:\Sigma\to\mathfrak g$, consider the operator
\be \CO[\Sigma,\hat g] := e^{\int_\Sigma \langle \hat\alpha,*J\rangle}\,. \ee
This is typically no longer a topological defect in the physical theory. However, it is $Q$-closed, and so still topological in $\CT^Q$. Moreover, any small, continuous deformations of the map $\hat g$ are $Q$-exact: replacing $\hat g$ with $\hat g e^{\epsilon \hat \beta}$, where $\beta:\Sigma\to \mathfrak g$ is continuous and $\epsilon$ is small, we have
\be \CO[\Sigma,\hat g e^{\hat \beta}] \approx \CO[\Sigma,\hat g ]\Big(1+ \epsilon \int_\Sigma \langle\hat\beta,*J\rangle\Big)
 = \CO[\Sigma,\hat g] + Q\Big(  \CO[\Sigma,\hat g] \int_\Sigma \langle\hat\beta,*S\rangle \Big)\,. \ee
Thus, in $\CT^Q$, we expect $\CO[\Sigma,\hat g]$ to depend only on the homotopy class of the map $\hat g:\Sigma\to G$.

A boring consequence is that any constant symmetry transformation should act trivially on the twisted theory,
\be \CO[\Sigma,g] = 1 + Q(...) \qquad \text{for $g$ constant.} \ee
In particular, there can be no charged local operators in $\CT^Q$; all $Q$-closed charged local operators are automatically $Q$-exact. A more interesting consequence is that there may be new symmetry defects corresponding to homotopically nontrivial maps $\hat g:\Sigma\to G$, \emph{a.k.a.} ``large'' or ``winding'' global symmetry transformations. The product structure induced from collision of neighboring surfaces $\Sigma,\Sigma'$ (where $\Sigma'$ is a slight deformation of $\Sigma$, defined by a choice of co-orientation on $\Sigma$) is that of pointwise multiplication in the space of maps $\text{Maps}(\Sigma,G)$,
\be \CO[\Sigma,\hat g]\CO[\Sigma',\hat g'] \simeq \CO[\Sigma,\hat g\hat g']\,,\qquad \text{where} (\hat g\hat g')(x) = \hat g(x)\cdot \hat g'(x) \ee

\subsection{One-form symmetry and the category of line operators}
\label{sec:sym-lines}

Now let's specialize the preceding construction to the case that $\Sigma=\S^1\times \R$ is a cylinder and $G=U(1)$. We expect $\CT^Q$ to have symmetry defects labelled by
\be \label{one} \text{homotopy classes of maps:}\; S^1\times \R\to U(1)\quad\simeq\;\pi_1(U(1))=\Z\,. \ee
The group structure, coming from pointwise multiplication on maps from $S^1$, coincides with the concatenation product on $\pi_1$.

The non-trivial winding transformations \eqref{one} have direct consequences for the category of line operators $\widehat \CC$ of the twisted theory $\CT^Q$.

First, there must be a $\Z$ action on the category, a set of auto-equivalences
\be  \CF_m:\wh\CC \to \wh\CC \ee
satisfying $\CF_m\circ\CF_{m'} \simeq \CF_{m+m'}$. The action of the functor $\CF_m$ on an object $\ell$ is defined by wrapping $\ell$ with a surface $\Sigma\times \R$ with a flavor transformation \eqref{one} of winding number $m$.

Second, this action must be internal, in the sense that the category actually contains a collection $\{\LL_m\}_{m\in \Z}$ of line operators with tensor algebra
\be \label{L-tensor} \LL_m\otimes \LL_{m'}\simeq \LL_{m+m'}\,,\qquad  \LL_0 \simeq \id\,, \ee
and
\be \CF_m(\ell) \simeq \LL_m\otimes \ell\qquad \text{for any $\ell\in \wh\CC$}\,. \label{internal} \ee
Each object $\LL_m$ is simply defined by wrapping a flavor transformation \eqref{one} of winding number $m$ around an \emph{empty} cylinder, \emph{i.e.} around the identity line. In other words, $\LL_m:=\CF_m(\id)$. The property \eqref{internal} follows by noticing that a winding transformation can be deformed to be the identity on half of $S^1$, and then pulled off of a given line operator, as shown in Figure~\ref{fig:homotope}.

\begin{figure}[htb]
\centering
\includegraphics[width=5.8in]{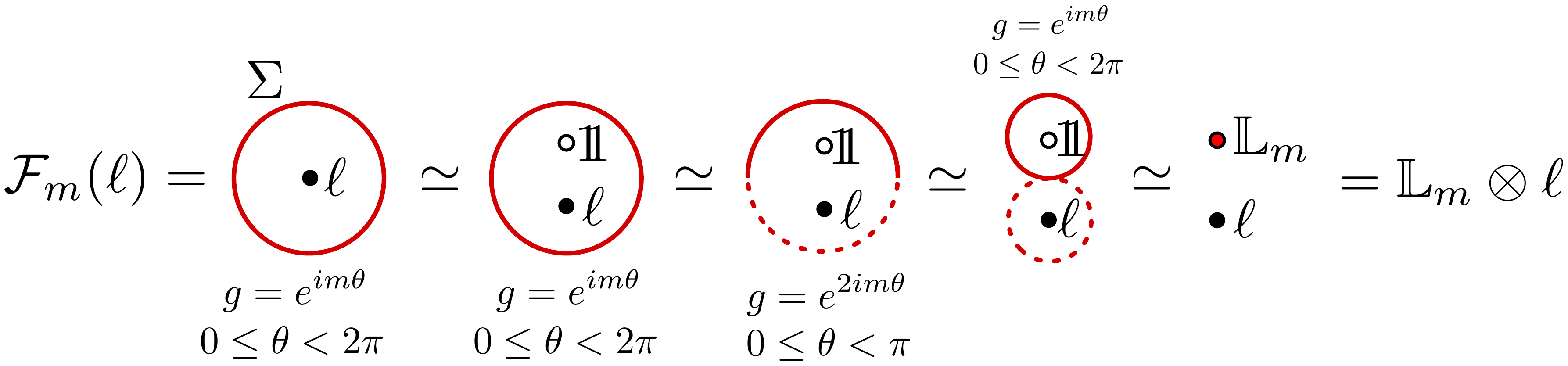}
\caption{Pictorial proof that $\CF_m(\ell)\simeq \LL_m\otimes \ell$, showing slices in the plane transverse to line operators and the symmetry surface $\Sigma\simeq S^1\times \R$.}
\label{fig:homotope}
\end{figure}

In addition, the $\Z$ one-form symmetry is non-anomalous --- this follows from  the original $G=U(1)$ symmetry being non-anomalous. Therefore, the symmetry defects $\LL_m$ can cross each other freely. Formally, this implies that that the monodromy (or double-braiding) of $\LL_m$'s is trivial:
\be \label{L-mon} \mu(\L_m,\L_{m'}) = R_{\L_{m'},\L_{m}} \circ R_{\L_{m},\L_{m'}} = \text{id}\,.\ee

Finally, we expect the category of line operators to decompose according to characters, or ``charges'' $\text{char}(\Z)\simeq\C/\Z$ of the one-form symmetry. To put this more concretely, let us assume that the $\LL_m$ braid semisimply with the remainder of the category%
\footnote{Since in general $\wh\CC$ is not a semisimple category, it is not obvious that this would be the case.}, %
implying that the monodromy of $\LL_m$ and any indecomposable object $\ell$ is just a constant multiple of the identity endomorphism
\be \mu(\L_m,\ell) = \xi\,\text{id}\,,\qquad \xi\in \C^*\,. \ee
Then the category has a block decomposition
\be \label{C-block1} \wh \CC = \bigoplus_{\upsilon\in \C/\Z} \wh\CC_{\upsilon}\,,\ee
where each ``block'' $\wh\CC_{\upsilon}$ is a subcategory such that
\begin{itemize}

\item Every object of $\wh\CC$ decomposes as a direct sum of objects that belong to individual blocks.

\item Within each block, the monodromy is
\be \mu(\LL_m,\ell) = e^{2\pi im\upsilon}\,\text{id}\,,\qquad \ell\in \wh\CC_\upsilon\,, \ee
\item There are no morphisms among objects in different blocks. (This amounts to charge conservation for the $\Z$ symmetry: lines labelled by different $\upsilon\in \C/Z$ cannot have junctions with each other.)

\item The tensor product among blocks respects the group structure of their charges:
\be \otimes: \wh\CC_{\upsilon} \boxtimes \wh\CC_{\upsilon'} \to \wh\CC_{\upsilon+\upsilon'}\,. \ee
This follows from wrapping one-form generators around a product of two lines:
\be \raisebox{-.7in}{\includegraphics[width=5in]{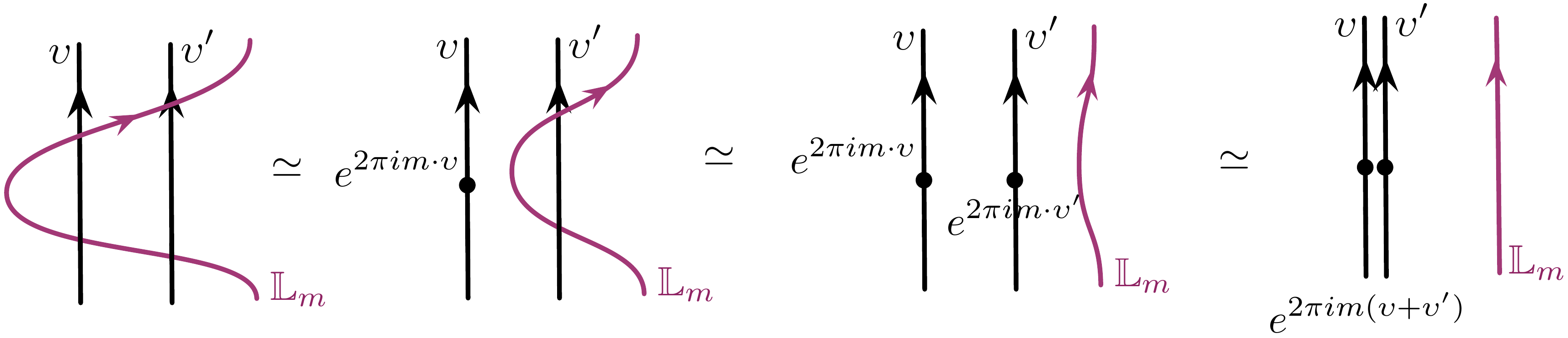}} \ee

\end{itemize}

Putting everything together, we arrive at:
\begin{PThm} \label{thm:sym} If a cohomological 3d TQFT $\CT^Q$ comes from a QFT $\CT$ with non-anomalous $U(1)$ zero-form symmetry whose current is $Q$-exact, then $\CT^Q$ has a one-form $\Z$ symmetry. In particular, the braided tensor category $\wh\CC$ of line operators has a collection of objects $\{\LL_m\}_{m\in \Z}$ that satisfy
\be \LL_0 \simeq \id\,,\qquad \LL_m\otimes \LL_{m'}\simeq \LL_{m+m'}\,,\qquad \mu(\LL_m,\LL_{m'})=\text{id}\,. \ee
If braiding with the $\LL_m$ is semisimple, then $\wh\CC$ decomposes into blocks labelled by the character group $\C/\Z$, respecting its group structure.
\end{PThm}

\subsection{Gauging}
\label{sec:Zgauge}

Next, suppose we would like to gauge a $G=U(1)$ symmetry of the physical theory $\CT$, whose current is exact in a topological twist, as in the preceding discussion. In the topologically twisted theory $\CT^Q$, gauging $G$ will induce a gauging of \emph{all} higher symmetries associated to the homotopy type of $G$. In particular, it will induce a gauging of the $\pi_1(G)=\Z$ one-form symmetry. It is reasonable to expect that this will be the \emph{only} consequence of gauging in $\CT^Q$, because the homotopy type of $G=U(1)$ is completely captured by its fundamental group.

The gauging of a one-form symmetry amounts physically to making the one-form generators, summing over all possible insertions/configurations of them in the path integral \cite{GKSW}. For the braided tensor category $ \wh\CC$ of line operators in $\CT^Q$, gauging the $\Z$ symmetry should be implemented by
\begin{itemize}
    \item Restricting to objects in the category that have trivial $\Z$ charge. In other words, restricting to the block $\wh\CC_{\upsilon=0}$ in \eqref{C-block1}.
    \item Imposing an equivalence relation 
    \be  \LL_m\otimes \ell \simeq \ell \qquad \forall\;m\in \Z \ee
    among all the remaining objects in the category. Alternatively, one may tensor all objects with the formal sum $\sum_{m\in \Z}\L_k$.
\end{itemize}
The gauged theory $\CT^Q/\Z$ has a new $\C^*$ zero-form symmetry, \cf\ \cite{GKSW,HLS-gauging} for in the case of finite-group gauging.
Roughly speaking, local operators of charge $q$ for this $\C^*$ come from $(\L_m,\L_{m'})$ junctions in $\CT_\rho^A$ with $m'-m=q$.

We will implement a rigorous version of this $\Z$ one-form symmetry gauging in Section \ref{sec:VOAlines}, using the lifting-functor technology of \cite{creutzig2017tensor, creutzig2020direct}.

\subsection{Relation to spectral flow}
\label{sec:sym-spectral}

We'll complete this circle of ideas by giving an elementary explanation of the relation between topological one-form symmetry and boundary spectral flow.

Suppose our physical theory $\CT$ with non-anomalous $U(1)$ global symmetry is placed on a half-space $\C_{z,\bar z}\times \R_{t\geq 0}$, with a boundary condition at $t=0$ that preserves the $U(1)$ symmetry. This means that in addition to the bulk current there is a boundary current $J^\pd$, satisfying bulk-boundary conservation equations
\be \label{bb-cons} d*J = 0\,,\qquad d_\pd *_\pd J^\pd = J_t\,, \ee
where $d_\pd$ and $*_\pd$ denote the two-dimensional exterior derivative and Hodge-star operator, and $J_t$ is the normal component of the bulk current.

Symmetry transformations with constant parameter $g=e^\alpha$ are implemented by choosing a two-dimensional submanifold $\Sigma$ of $\C\times \R_{\geq 0}$, which may end on the boundary, and inserting a defect
\be \CO[\Sigma,g]= \exp\Big(\alpha\int_\Sigma  *J + \alpha\oint_{\pd \Sigma} *J^\pd\Big)\,. \ee
The bulk-boundary conservation equation \eqref{bb-cons} ensures that $\CO[\Sigma,g]$ is invariant under continuous variations of the surface $\Sigma$ \emph{and} of its boundary.

Now suppose that the bulk-boundary system has a BRST symmetry $Q$, whose cohomology is topological in the bulk but \emph{holomorphic} on the boundary, so that the boundary supports a VOA $\CV$.
Moreover, suppose as in the preceding discussion that the bulk $U(1)$ symmetry is infinitesimally trivialized in the topological twist, $J=Q(S)$. Then the second equation in \eqref{bb-cons} becomes $d_\pd *_\pd J^\pd=Q(...)$. Thus the boundary current is conserved on its own in the topological twist. Holomorphicity of the boundary theory then implies holomorphic current conservation
\be \pd_{\bar z}J_z^\pd = 0\,,\qquad \text{on the boundary of in $\CT^Q$}\,, \ee
so that the boundary VOA must contain a $\mathfrak{gl}(1)$ Kac-Moody algebra generated by $J_z^\pd$\,.

\begin{Rmk}
    The well-known enhancement of finite global symmetry to Kac-Moody symmetry in 2d holomorphic QFT is a very close analogue of our arguments in Section \ref{sec:oneform} for one-form symmetry enhancement in a topological 3d bulk. Both of these deal with a physical symmetry that has been partially trivialized --- either made holomorphic or topological.
\end{Rmk}

In the topologically twisted bulk-boundary system $\CT^Q$, we can now define non-constant symmetry defects associated to any surface $\Sigma$, even ending on the boundary. Choose a smooth map $\hat g=e^{\hat\alpha}:\Sigma\to U(1)$, and consider 
\be \CO[\Sigma,\hat g] = \exp\Big(\int_\Sigma \hat \alpha * J + \frac{1}{2\pi}\oint_{\pd\Sigma} \hat \alpha J_z\,dz\Big)\,, \ee
where in the second term $\hat \alpha$ has been analytically continued to a holomorphic function in a neighborhood of $\pd\Sigma\subset \C_{z,\bar z}\times\{t=0\}$. This defect depends topologically on both the surface $\Sigma$ (including its boundary) and the map $\hat g$; it will be interesting when the map has nontrivial homotopy type.

In Section \ref{sec:sym-lines}, we defined generators $\CL_m$ of the bulk $\pi_1(U(1))=\Z$ one-form symmetry, by taking $\Sigma$ to be a cylinder $\S^1\times \R$ with a winding symmetry transformation $\hat g=e^{im\theta}$. The line operator $\L_m$ is obtained by shrinking the cylinder to a line. Now consider this operation in the presence of a boundary, as on the top of Figure \ref{fig:cylinder}. Taking $\Sigma=S^1\times \R_{t\geq 0}$, we have a defect
\begin{align} \CO[\Sigma,e^{im\theta}]  &= \exp\Big(im\int_{S^1\times \R_{t\geq 0}} \theta *J +\frac{im}{2\pi} \oint_{S^1\times {t=0}} \log z\, J_z\,dz\Big) \notag \\
&= \exp\Big(im\int_{S^1\times \R_{t\geq 0}} \theta *J\Big) \cdot \exp\Big( -\frac{m}{2\pi i} \oint_{S^1\times \{t=0\}} \log z\, J_z\,dz\Big) \notag \\
&\simeq \LL_m\big|_{t\geq 0} \cdot \exp\Big( \frac{m}{2\pi i}\oint_{S^1\times\{t=0\}}\frac{1}{z} v(z)\,dz\Big)\notag  \\
&= \LL_m\big|_{t\geq 0} \cdot e^{m v_0}|0\rangle\,.  \label{L-spec} 
\end{align}
Here we've used the formal anti-derivative $J_z=\pd v$ of an abelian current, \emph{a.k.a.} Heisenberg field, as in Section \ref{sec:Heis}. 

\begin{figure}[htb]
\centering
\bigskip
\includegraphics[width=5in]{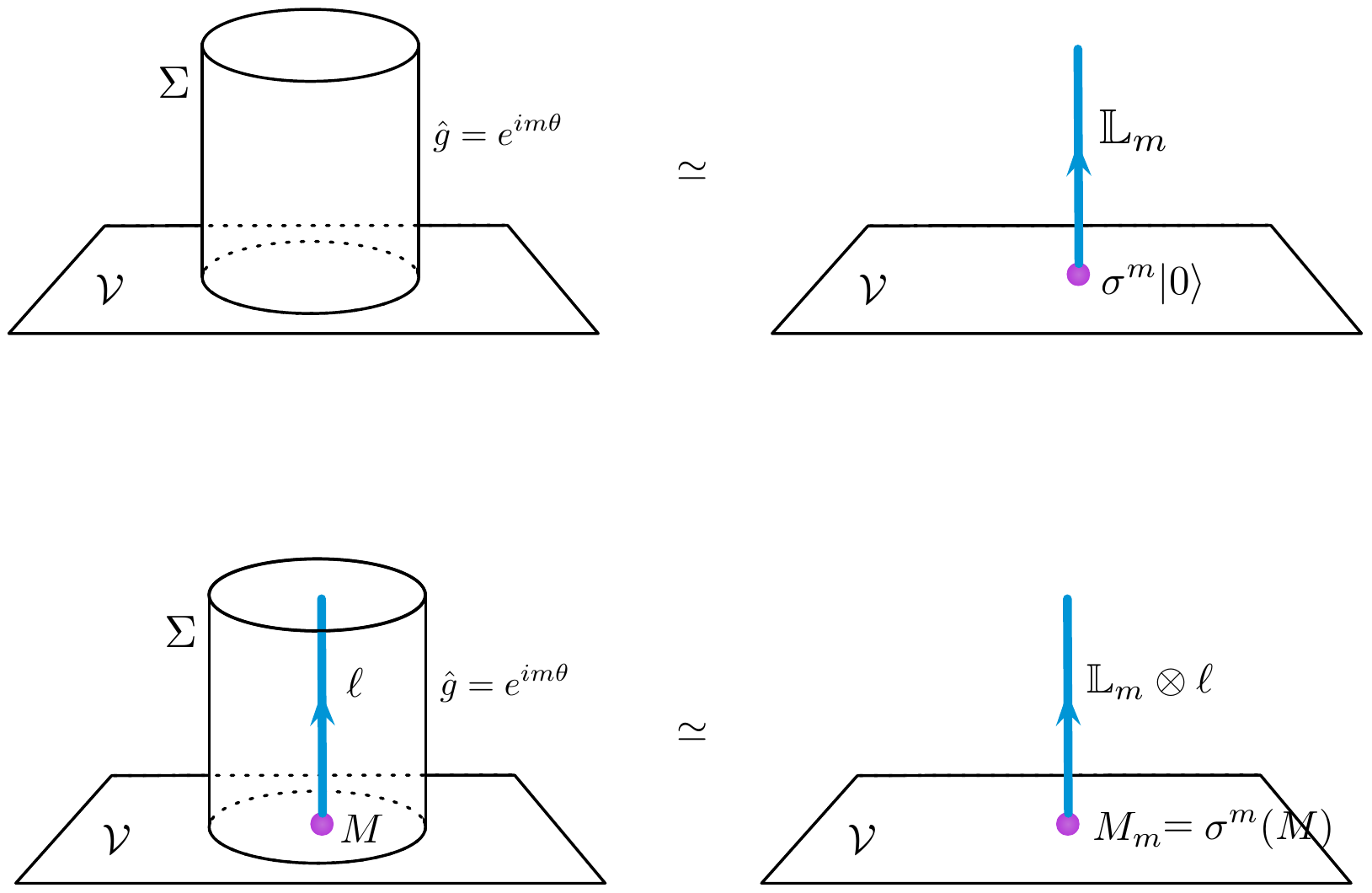}
\caption{Shrinking cylindrical winding defects in a twisted theory $\CT^Q$ to generate spectral-flow automorphisms of a boundary VOA.}
\label{fig:cylinder}
\end{figure}

The operator $e^{mv_0}$ in \eqref{L-spec} is the spectral-flow automorphism `$\sigma^m$' of the boundary VOA $\CV$ associated to its $\mathfrak{gl}(1)$ Kac-Moody symmetry. It acts on the current $J_{z,0}$ precisely by shifting its zero-mode
\be  e^{m v_0} J_{z,n} = J_{z,n}+ \delta_{n,0}\,m\,, \ee
exactly as in \eqref{specflow-zero} of Section \ref{sec:spectralflow}. In \eqref{L-spec}, it is acting on the vacuum state, since we did not put any insertions inside the cylinder $\Sigma$. If we insert other states in $\CV$, we will get the spectral flow of the entire vacuum module. More generally, inserting states in other VOA modules at the endpoints of bulk lines in the cylinder, as on the bottom of Figure \ref{fig:cylinder}, leads simultaneously to the action of the $\Z$ one-form symmetry in the bulk, and spectral flow on the boundary. In summary:

\begin{PProp}
Given a 3d bulk theory with $U(1)$ symmetry that is enhanced to a $\pi_1(U(1))=\Z$ one-form symmetry in a topological twist, and a boundary condition that becomes holomorphic in the twist, the endpoints of one-form symmetry defects $\LL_m$ are spectral flows $\sigma^m\CV$ of the vacuum module. Acting with the one-form symmetry in the bulk (by tensor product with $\LL_m$ is accompanied by a simultaneous spectral-flow action on the boundary.
\end{PProp}  

In light of this correspondence, many of the properties of bulk one-form symmetry and the structure of bulk line operators from Section \ref{sec:sym-lines} could be rephrased as familiar properties of VOA's with spectral-flow automorphisms.

\section{The category of line operators: bulk considerations}
\label{sec:bulklines}

We now turn our attention fully to line operators in the topological twists of simple abelian 3d $\CN=4$ theories.
Line operators in the topological twists of a 3d $\CN=4$ theory form a DG category, which is expected to have a braided tensor structure \cite{RW,KRS}. There turn out to be many subtle choices in determining how large the categories of line operators are.
In this section, we will describe ``minimal'' versions $\wh\CC_\rho^A$,  $\wh\CC_\rho^B$ of the categories in the A and B twists of simple abelian theories $\CT_\rho$. The discussion will be entirely from the perspective of bulk physics. Its main purpose is to provide a physical foundation for the rigorous but more abstract VOA constructions of Section \ref{sec:VOAlines}. 

Our categories $\wh\CC_\rho^A$,  $\wh\CC_\rho^B$ will be defined as the categories generated by a certain minimal collection of BPS line operators compatible with either the A or B twists, which we would always expect to find in a physical theory. Our notion of ``minimal'' is driven by the structure imposed by topological one-form symmetry, as described in Section \ref{sec:oneform}. The categories must contain at least the one-form symmetry generators themselves: $\pi_1(T_F)=\Z^{n-r}$ in the A twist and $\pi_1(T_{\rm top}) = \Z^r$ in the B twist. We will further identify ``charged objects'' within every block of the expected decompositions by characters of the one-form symmetry,
\be \label{AB-decomp} \wh\CC_\rho^A = \bigoplus_{\upsilon\in (\C/\Z)^{n-r}}\wh\CC_{\rho,\upsilon}^A\,,\qquad \wh\CC_\rho^B = \bigoplus_{\upsilon\in (\C/\Z)^{r}}\wh\CC_{\rho,\upsilon}^B\,.\ee
closed under the tensor-action of one-form symmetry.

In the B twist, the generators of the one-form symmetry are Wilson lines $\{\W_m\}_{m\in \Z^r}$, and the charged objects are mixed monodromy-Wilson lines
\be \W_m^{(\upsilon)}\,,\qquad m\in \Z^r\,,\; \upsilon\in (\C/\Z)^r\,, \ee
labelled by a gauge charge $m\in\text{char}(G)=\Z^r$ and a monodromy defect around which the complexified gauge connection has holonomy $e^{2\pi i\upsilon}\in G_\C$. Fully defining the monodromy-Wilson lines turns out to be surprisingly subtle, as they need to be supplemented by additional constraints on hypermultiplets, which we ultimately implement by forming certain bound states of Wilson lines.

In the A twist, the generators of the one-form symmetry are flavor vortices $\{\V_m\}_{m\in \Z^{n-r}}$, and the charged objects are generalized vortex defects, which specify a meromorphic singularity in the hypermultiplet scalars roughly of the form  
\be \V_m^{(\upsilon)}:\quad \begin{cases} X(z) \in z^{\sigma(\m)}\C[\![z]\!] \\ Y(z) \in z^{-\sigma(\m)-1} \C[\![z]\!]\end{cases},\quad \mu_F(X,Y) \in \frac{\upsilon}{z}+\C[\![z]\!]\,,\qquad m\in \Z^{n-r},\,\upsilon\in (\C/\Z)^{n-r}\,. \ee
Here $\sigma:\Z^{n-r}\to\Z^n$ is a splitting of our usual exact sequence (defining flavor charges), and $\mu_F$ are moment maps for the flavor symmetry. We justify physically that the $\V_m^{(\upsilon)}$ are 3d mirrors of the $\W_m^{(\upsilon)}$. 

The monodromy-Wilson lines $\W_\m^{(\upsilon)}$ are three-dimensional counterparts of the 't Hooft-Wilson line of 4d gauge theories \cite{Kapustin-WtH}.
This description of line operators is also analogous to the construction of surface operators in 4d $\CN=4$ Yang-Mills in work of Gukov and Witten \cite{GukovWitten}. It generalizes the constructions of 3d $\CN=4$ line operators of Assel-Gomis \cite{AsselGomis} (which did not consider vortex lines with more than first-order singularities) and the subsequent classification in Dimofte-Garner-Geracie-Hilburn \cite{lineops} (which briefly mentioned but did not study B-type monodromy defects, or generalized vortex defects). Monodromy-Wilson lines and generalized vortex defects have appeared in the work of Hilburn and Raskin \cite{hilburn2021tate}, in their geometric models of sheaves on derived loop spaces.

We will describe $\W_\m^{(\upsilon)}$ and $\V_\m^{(\upsilon)}$ in greater detail, and construct their \emph{morphism} spaces, \emph{a.k.a.} the spaces of local operators at their junctions. This information is, in principle, enough to determine the DG categories $\wh\CC_\rho^B$, $\wh\CC_\rho^A$ that they generate. The categories themselves contain many more objects, which are complexes formed from the $\W_\m^{(\upsilon)}$ and $\V_\m^{(\upsilon)}$; these can be thought of physically as quarter-BPS bound states of the $\W_\m^{(\upsilon)}, \V_\m^{(\upsilon)}$

Further physical data is required to determine the braided tensor structure on $\wh\CC_\rho^B$, $\wh\CC_\rho^A$. Some tensor products and braidings of individual $\W_\m^{(\upsilon)},\V_\m^{(\upsilon)}$ follow immediately from one-form symmetry considerations.

However, the braiding of complexes of these lines is highly nontrivial. We will ultimately access the full braided tensor structure via boundary VOA's in Section \ref{sec:VOAlines}.

\subsection{B twist}
\label{sec:Blines}

\subsubsection{Wilson lines}

Consider a simple abelian gauge theory $\CT_\rho$ with $G=U(1)^r$ and $N=\C^n$. In the topological B twist, the gauge connections  naturally combine with triples of adjoint vectormultiplet scalars (which have been twisted into a spacetime 1-form) to form complexified $GL(1,\C)^r$ connections $\CA^a$, $1\leq a \leq r$. One basic family of line operators in the B twist are complexified Wilson lines, labelled by characters $\m$ of the gauge group, and a curve $\gamma$ in 3d spacetime
\be \W_\m(\gamma) := \exp\, i\int_\gamma \m\cdot \CA\,,\qquad \m\in \Z^r\,. \ee

The Wilson lines are generators of the topological one-form symmetry $\pi_1(T_{\rm top})$ in the B twist. They can be constructed as winding defects, along the lines of Section \ref{sec:oneform}. A winding defect for the topological symmetry $T_{\rm top}$ is equivalent to a Wilson line, due to the mixed Chern-Simons coupling between $T_{\rm top}$ and the $U(1)^r$ gauge symmetry. (A classic version of this argument appears in \cite{Witten-topology}.)

Each $\W_\m$ should define an object in the category of line operators. The object itself is independent of the curve label `$\gamma$'; it should be thought of as the label of a line operator supported on an arbitrary, infinitesimal, line segment. The  morphism space among two Wilson lines $\W_\m,\W_{\m'}$ is the vector space of gauge-invariant local operators at their junction. In the B twist such operators come from polynomials in complex hypermultiplet scalars. We expect that
\be \label{Hom-W-phys} \text{Hom}^\bullet (\W_\m,\W_{\m'}) = \Big[\C[X^i,Y^i]_{i=1}^n\big/\big(\mu^a\big)_{a=1}^r\Big]_{\text{weight}\,\m'-\m}\,. \ee
In this expression, we first form the polynomial ring $\C[X^i,Y^i]_{i=1}^n$ in variables $X^i,Y^i$, which is $\Z_{\geq 0}\times \Z^n$ graded by 1) integer-valued cohomological degree, coming from $U(1)_H$ symmetry, under which all $X$'s and $Y$'s have degree +1; and 2) weights for the $U(1)^n$ symmetry under which $X^i$ has weight $(0,...0,\underset{(i)}{1},0,...,0)^T\in \Z^n$ and $Y^i$ has weight $(0,...0,\underset{(i)}{-1},0,...,0)^T\in \Z^n$. Then we quotient out by the ideal generated by complex moment maps for the gauge symmetry,
\be \mu^a = \rho_i{}^aX^iY^i\,, \ee
since these complex moment maps are $Q$-exact in the B twist. Finally, to preserve gauge invariance at the junction, we restrict to the subspace of gauge charge $\m'-\m$; this is the subspace generated by monomials whose weight $\lambda\in \Z^n$ satisfies
\be \rho^T \lambda = \m'-\m\,. \ee

The space of local operators on the RHS of \eqref{Hom-W-phys} still has its cohomological $\Z_{\geq 0}$ grading, with $|X|=|Y|=1$. This should correspond to cohomological grading of the derived morphism space on the LHS.

It will be useful for us to give an equivalent characterization of the morphism space \eqref{Hom-W-phys}, in the language of DG algebra, as follows.. The moment maps $\mu^a$ are the images of certain bulk gauginos $\lambda^a$ (denoted $\lambda_-^a$ in \eqref{Ncharges}) under the B-twist supercharge,
\be Q(\lambda^a) = \mu^a\,. \label{gaugino-mu} \ee
These gauginos are twisted to spacetime scalars, of odd fermion number and of cohomological degree +1. Including them in the construction of the morphism space, we find
\be  \text{Hom}^\bullet (\W_\m,\W_{\m'}) = H^\bullet\big(\C[X^i,Y^i,\lambda^a],Q\big)\Big|_{\text{weight}\,\m'-\m}\,, \label{Wilson-Hom} \ee
where the moment-map constraints $\mu^a=0$ are now imposed cohomologically.

\subsubsection{Monodromy defects}

The B twist also admits disorder operators, labelled by classes of singular solutions to the equations of motion in the neighborhood of a line. In particular, the complexified flat connections $\CA^a$ may acquire monodromy $e^{2\pi i \upsilon^a}$ for some $\upsilon\in (\C/\Z)^r$. In local coordinates on $\C_{z,\bar z}\times \R_t$, with a line operator at $z=0$, the connection takes the form
\be \label{A-mon} \CA \,\sim\, \upsilon\,d\theta = \frac{\upsilon}{2i} \Big(\frac{dz}{z}-\frac{d\bar z}{\bar z}\Big) \qquad \text{near $z=0$}, \ee  
modulo gauge transformations with at-most power-law growth as $z\to 0$. Note that winding gauge transformations $e^{2\pi i \m\theta}$ with $\m\in \Z^r$ do shift $\nu\mapsto \nu+\m$, whence the periodicity in the parameter $\upsilon$.

Such a monodromy defect must be accompanied by setting some or all of the hypermultiplet scalars to zero near $z=0$. This is because the equations of motion in the B twist also require the scalars to be covariantly constant, $d_\CA X = d_\CA Y = 0$, which implies that they are invariant under the monodromy $e^{2\pi i \upsilon}\in GL(1,\C)^r$. This leads to an important subtlety: to fully define a monodromy line, we must supplement \eqref{A-mon} with a constraint on the hypermultiplets, ensuring that $e^{2\pi i\rho(\upsilon)}X=X$ and $e^{2\pi i\rho(\upsilon)}Y=Y$; but there is not a unique, minimal way to do this. We'll explain this feature from one perspective here. It will reappear in Section \ref{subsec:LineB}.

As a simple example, consider $\text{SQED}_1$, namely $G=U(1)$ and $V=\C$, with charge matrix $\rho=(1)$. There is a single pair of hypermultiplet scalars $(X,Y)$, of gauge charges $(1,-1)$. We may analyze what happens near a line operator at $z=0$ by compactifying the 3d theory along the angular $\theta$ direction, to obtain a 2d TQFT on a half-space $\R_{|z|>0}\times \R_t$, with the line operator at its boundary. The effective 2d theory is a B model with target $\C^*\times T^*V$ and gauge group $U(1)$, where the `$\C^*$' in the target is the holonomy $g=\exp(i\oint \CA)$ of the 3d connection along the compactification circle. The effective 2d B model has a superpotential given by (\emph{cf.} \cite{BDGH,CDGG})
\be W_{2d} = (g-1)XY\,, \label{W2d} \ee
whose critical locus imposes invariance under the holonomy, $gX=X$ and $g^{-1}Y=Y$. 

The presence of a singular connection $\CA\sim \upsilon\,d\theta$ in 3d translates to setting $g=e^{2\pi i \upsilon}$ at $|z|=0$ in the 2d model. When $\upsilon\notin \Z$, boundary conditions for the 2d model then require a further ``matrix factorization'' of the remaining superpotential $(e^{2\pi i\upsilon}-1)XY$ \cite{KapustinLi}. This amounts to coupling to 1d quantum mechanics at $z=0$ with a curved differential $Q_{1d}$ that squares to $(e^{2\pi i\upsilon}-1)XY$ rather than zero. A minimal matrix factorization (known mathematically as the Koszul matrix factorization) uses a quantum mechanics with a two-dimensional Hilbert space $\CH=\C^2$, and differential $Q_{1d}:\CH\to \CH$ given by
\be Q_{1d} = \begin{pmatrix} 0 & (e^{2\pi i\upsilon}-1)X \\ Y & 0 \end{pmatrix}\,, \qquad Q_{1d}^2 = (e^{2\pi i\upsilon}-1)XY\,\text{id}_{\C^2}\,. \ee
Note that upon rescaling basis vectors in $\CH$, this is equivalent to $Q_{1d} = \sm{ 0 & X \\ (e^{-2\pi i\upsilon}-1)Y & 0 }$. As both $X,Y$ are in the image of $Q_{1d}$, they are both \emph{effectively} set to zero at $z=0$, as we would want in order to achieve $e^{2\pi i\upsilon}X=X$ and $e^{-2\pi i\upsilon}Y=Y$.

The important subtlety arises from remembering that we have a gauge theory. We expect gauge symmetry to unbroken by the monodromy line, and thus must choose gauge charges for the two states in $\CH$, compatible with the charges $|X|=1$ and $|Y|=-1$.
There is no canonical way to do this. The two simplest options are to give the states in $\CH$ charges $(1,0)$, or charges $(0,-1)$. More generally, we can have charges $(m,m-1)$ for any $m\in \Z$. 

What we have found can be translated back to the 3d bulk in the following way. There is no canonical monodromy defect in the B twist of $\text{SQED}_1$. Rather, there is a family of monodromy defects $\W_m^{(\upsilon)}$, labeled by integers $m\in \Z$, which combine a singular connection with a \emph{bound state of two consecutive Wilson lines}
\be \W_m^{(\upsilon)}:\quad \CA \sim \upsilon\,d\theta\quad\text{superposed with}\quad \W_m\overset{(e^{2\pi i \upsilon}-1)X}{\underset{Y}{\rightleftarrows}} \W_{m+1}\qquad (\upsilon \in (\C\backslash\Z)/\Z)\,. \ee
The bound state deforms the direct sum of Wilson lines $\W_m$ and $\W_{m+1}$. (We'll see more examples of bound states in Section \ref{sec:W-DG}.)
Different monodromy defects are related by tensoring (colliding) with pure Wilson lines:  $\W_{m'}\otimes \W_m^{(\upsilon)} \simeq \W_{m+m'}^{(\upsilon)}$.

We can now generalize to the B twist of any simple abelian theory $\CT_\rho$. Let $\upsilon \in (\C/\Z)^r$. 
The $X^i$ and $Y^i$ with $\rho(\upsilon)^i \notin \Z$ are the fields that we must (effectively) set to zero in order to get $e^{2\pi i\rho(\upsilon)}X=X$ and $e^{-2\pi i\rho(\upsilon)}Y=Y$. To accomplish this, we superpose the monodromy defect with the following bound state of Wilson lines
\be \label{def-monW} \W_\m^{(\upsilon)}:= \quad \CA\sim \upsilon\,d\theta\quad\text{with}\quad \W_\m\otimes \bigotimes_{\text{$i$ s.t. $\rho(\upsilon)^i \notin \Z$}} \bigg[ \W_0\overset{(e^{2\pi i \rho(\upsilon)^i}-1)X^i}{\underset{Y^i}{\rightleftarrows}} \W_{(0,...,\underset{i}{1},...,0)} \bigg]\,.  \ee
Here the tensor product of complexes is taken following the usual rules of homological algebra, and the tensor product of all individual Wilson lines is assumed to satisfy $\W_\m\otimes \W_{\m'} = \W_{\m+\m'}$. (We will justify this in Section \ref{sec:highersym}.) The total curved differential $Q_{1d}$ in the bound state satisfies $Q_{1d}^2 = Y\cdot (e^{2\pi i \rho(\upsilon)}-1)X$, as required to obtain a matrix factorization of an effective 2d superpotential.

The $\W_\m^{(\upsilon)}$ are hybrid ``monodromy-Wilson lines'' of $\CT_\rho^B$. They are analogues of hybrid Wilson-'t Hooft lines of 4d gauge theories \cite{Kapustin-WtH}. Note that $\W_\m^{(0)}=\W_\m$ are just the pure Wilson lines.
The entirely family of monodromy-Wilson lines is related by tensoring with pure Wilson lines,
\be \W_{\m'}\otimes \W_\m^{(\upsilon)} \simeq \W_{\m+\m'}^{(\upsilon)}\,. \ee

\begin{Exp}
Consider $\text{SQED}_2$, with $G=U(1)$, $V=\C^2$, and $\rho=\sm{ 1\\1}$. The hypermultiplet scalars are $(X^1,Y^1)$ and $(X^2,Y^2)$, both of gauge charges $(+1,-1)$. When $\upsilon =0$ (mod $\Z$), we simply have $\W_m^{(0)}=\W_m$. When $\upsilon\notin \Z$, and $m\in \Z$ the monodromy-Wilson line $\W_m^{(\upsilon)}$ superposes a monodromy singularity with a bound state
\be \W_{m} \overset{\sm{ \xi X^1 \\ \xi X^2}}{\underset{\sm{ Y^1 & Y^2}}{\rightleftarrows}} \W_{m+1}\oplus \W_{m+1} \overset{\sm{ \xi X^2 & -\xi X^1}}{\underset{\sm{Y^2\\-Y^1}}{\rightleftarrows}} \W_{m+2}\,, \ee
where $\xi = e^{2\pi i \upsilon}-1$.
Note that the differential here squares to $W_{2d}=(e^{2\pi i \upsilon}-1)(X^1Y^1+X^2Y^2)$ on each Wilson-line summand of the bound state.
\end{Exp}

\subsubsection{Morphisms and the DG category}
\label{sec:W-DG}

We computed morphisms (local operators) among Wilson lines in \eqref{Wilson-Hom}, and we can now generalize to the hybrid monodromy-Wilson lines. Since the B twist localizes to \emph{flat} complex connections, the monodromy $\upsilon \in (\C/\Z)^{r}$ cannot change across a junction of line operators. This means that
\be \text{Hom}^\bullet(\W_\m^{\upsilon},\W_{\m'}^{\upsilon'}) = 0\quad \text{unless}\quad \upsilon=\upsilon' \;\text{(mod $\Z^r$)}\,. \label{Hom-mon-0} \ee
This is a first indication that $\W_m^\upsilon$ belong in different blocks of the B-twist category, as in \eqref{AB-decomp}.

When $\upsilon=\upsilon'$, the local operators at a junction between $\W_\m^\upsilon$ and $\W_{\m'}^{\upsilon}$ still get a contribution from polynomials in the scalars $X,Y$, of gauge charge $\m'-\m$, modulo the gauge moment maps constraints $\mu^a=0$. However, the bound states decorating these lines effectively impose the conditions $e^{2\pi i \rho(\upsilon)}X=X$ and $e^{-2\pi i \rho(\upsilon)}Y=Y$. This removes some of the scalars. In addition, setting $e^{2\pi i \rho(\upsilon)}X=X$ and $e^{-2\pi i \rho(\upsilon)}Y=Y$ may automatically cause some of the gauge moment maps $\mu^a$ to vanish, rendering the corresponding gauginos $\lambda^a$ in \eqref{gaugino-mu} $Q$-closed. These gauginos then contribute to the morphism space as well.

The final expected answer for the morphism space can be succinctly described as follows. Let $L_\upsilon \subseteq T^*V$ be the subspace
\be \label{def-La} L_\upsilon:= \{e^{2\pi i\rho(\upsilon)}X=X,e^{-2\pi i\rho(\upsilon)}Y=Y\} = \{X_i=Y_i=0\}_{\rho(\upsilon)^i \notin \Z}\,,\ee
containing the non-vanishing $X$'s and $Y$'s. Let $\C[X^i,Y^i,\lambda^a]\big|_{L_\upsilon}$ be the algebra of polynomials in the hypermultiplet scalars and gauginos restricted to this subspace. It has a differential $Q$ induced from the bulk differential $Q(\lambda^a)=\mu^a$ in \eqref{gaugino-mu}. Then
\be \label{HomW-gen}  \text{Hom}^{\bullet}(\W^{(\upsilon)}_\m,\W^{(\upsilon')}_{\m'}) = \delta^{\upsilon\upsilon'} H^\bullet\big( \C[X^i,Y^i,\lambda^a]\big|_{L_\upsilon},Q\big) \Big|_{\text{weight}\,\m'-\m}\,. \ee

\begin{Rmk}
    This answer could have also been obtained from a more rigorous computation using a compactified 2d B-model with superpotential as in \eqref{W2d}, though we won't describe the details here. Part of what makes the rather complicated-looking bound states from \eqref{def-monW} ``minimal'' is that that they guarantee there are no additional contributions to spaces of local operators at junctions, aside from the expected \eqref{HomW-gen}.
\end{Rmk}

\begin{Exp}
Consider $\text{SQED}_1$, with $G=U(1)$ and scalars $(X,Y)$, of gauge charges $(1,-1)$ and cohomological degree $(1,1)$. The bulk gaugino $\lambda$ has gauge charge 0 and cohomological degree 1, and satisfies
\be Q(\lambda) = XY\,. \label{SQED-mom} \ee

The charges of Wilson lines are labelled by $\m\in \Z$,
and the possible monodromy defects are labelled by $\upsilon\in \C/\Z$. The subspace $L_\upsilon$ is given by
\be L_\upsilon = \begin{cases} T^*\C & \upsilon=0\;\text{(mod $\Z$)} \\
X=Y=0 & \upsilon\neq 0\;\text{(mod $\Z$)} \end{cases}\,.\ee
In particular, when $\upsilon\neq 0$ (mod $\Z$), we have $Q(\lambda)=0$, so $\lambda$ will contribute to morphism spaces. The morphism spaces are
\be \text{Hom}^\bullet(\W_\m^{(\upsilon)},\W_{\m'}^{(\upsilon')}) = \delta^{\upsilon\upsilon'} \begin{cases} \C\langle X^{\m'-\m}\rangle & \upsilon=0,\m'>\m \\
\C\langle 1\rangle & \upsilon=0,\m'=\m \\
\C\langle Y^{\m-\m'} \rangle & \upsilon = 0, \m'<\m \\
\C[\lambda]=\C\langle 1,\lambda\rangle & \upsilon\neq 0,\m'=\m\\
0 & \upsilon \neq 0,\m'\neq \m \end{cases}
\ee
Here by $\C\langle v_1,v_2,...\rangle$ we mean the vector space spanned by $v_1,v_2$,...\,.
\end{Exp}

We'd now like to take the ``closure'' of the set of $\W_\m^{(\upsilon)}$ line operators under some natural physical operations. The first of these, which we focus on here, is a standard way of combining one-dimensional defects in any cohomological TQFT, and leads to the notion of a DG category generated by the $\W_\m^{(\upsilon)}$. (See \emph{e.g.} \cite{Douglas-Dbranes,Dbranes} for analogous descriptions of boundary conditions in 2d A and B models.)

Namely, given a collection $\ell_1,\ell_2,...,\ell_k$ of line operators, one can
tensor each with a vector space $V_1,V_2,...,V_k$ (possibly of nontrivial cohomological degree) and then introduce a nontrivial differential on the direct sum $(V_1\otimes \ell_1)\oplus (V_2\oplus \ell_2)\oplus\cdots\oplus (V_k\otimes \ell_k)$. In more physical terms, this corresponds to introducing 1d quantum mechanical systems (also called Chan-Paton bundles) on the collection of lines, and deforming the collection with a nontrivial coupling. More mathematically, the result of introducing a differential on a sum $V_1\otimes \ell_1\oplus V_2\oplus \ell_2\oplus\cdots\oplus V_k\otimes \ell_k$ can also be represented by a complex, whose terms are sums of $\ell_i$'s. In the end, taking a closure under such operations leads to the DG category generated by $\W_\m^{(\upsilon)}$'s. This motivates

\begin{Def} \label{def:CB-bulk}
Let $\wh\CC_\rho^B$ be the DG category generated by the objects $\W_\m^{(\upsilon)}$, with morphisms given by \eqref{HomW-gen}. In other words, the objects of $\wh\CC_\rho^B$ are complexes whose terms are direct sums of $\W_\m^{(\upsilon)}$'s and cohomological shifts thereof, and whose maps are morphisms among the $\W_\m^{(\upsilon)}$'s. Two complexes are deemed equivalent if they are quasi-isomorphic; this implies that they represent the same line operator in the B-twisted TQFT.
\end{Def}

In terms of the underlying $\CN=4$ theory, the pure Wilson lines $\W_\m$ come from half-BPS line operators compatible with the B twist. The monodromy-Wilson lines $\W_\m^{(\upsilon)}$ appear to only be quarter-BPS, due to the bound states that must be used to decorate them. More general complexes of the $\W_\m$'s and $\W_\m^{(\upsilon)}$'s all correspond to quarter-BPS line operators.

\begin{Exp}
    Let's continue the $\text{SQED}_1$ example above. We'll describe a few of the `quarter-BPS' objects in the DG category $\wh\CC_{(1)}^B$.
    
    The trivial line operator $\W_0=\id$ the Wilson line $\W_{-1}$ have a morphism $X\in \text{Hom}(\W_{-1},\W_0)$, so we can form a complex
    \be \W_{-1} \overset{X}\longrightarrow \W_0 \label{complex0} \ee
    Note that `$X$' naturally has cohomological degree 1, as required for a differential.
    
    The complex \eqref{complex0} can be interpreted in a few different ways. It represents a bound state of the trivial line operator and and $\W_{-1}$ --- a deformation of their sum $\W_0\oplus \W_{-1}$ by the differential $\left(\begin{smallmatrix} 0& X\\0&0\end{smallmatrix}\right)$\,. This is otherwise known as an extension of $W_0$ by $W_{-1}$, or the cone of the map $X:W_{-1}\to W_{0}$\,.

    Alternatively, noting that quotienting out by the image of the differential in \eqref{complex0} sets $X=0$, we may think of \eqref{complex0} roughly as a single line operator that imposes $X=0$.

    Further bound states/extensions of Wilson lines can be described in similar ways. There is a bound state
    \be \W_{1} \overset{Y}{\longrightarrow} \W_0\,, \label{complex1} \ee
    that effectively sets $Y=0$. There is a higher bound state combining \eqref{complex0} and \eqref{complex1}
    \be \W_0 \overset{\sm{Y\\X}}\longrightarrow \W_{-1}\oplus \W_1 \overset{\sm{-X & Y}}\longrightarrow \W_0\,, \label{complex2} \ee
    which (roughly) sets both $X=0$ and $Y=0$,
    as well as iterated self-extensions of \eqref{complex0} of any length, \emph{e.g.}
    \be \W_{-1}\overset{X}\longrightarrow\W_0\overset{Y}\longrightarrow\W_{-1}\overset{X}\longrightarrow \cdots \overset{Y}\longrightarrow \W_{-1} \overset{X}\longrightarrow \W_0\,. \ee
    (Note that each composition of differentials $X\cdot Y$ and $Y\cdot X$ vanishes because the moment map \eqref{SQED-mom} is set to zero.)

    For monodromy defects of generic $\upsilon\neq 0$ (mod $\Z$), we find self-extensions of arbitrary length of the form
    \be \W^{(\upsilon)} \overset\lambda\longrightarrow \W^{(\upsilon)} \overset\lambda\longrightarrow \cdots \overset\lambda\longrightarrow  \W^{(\upsilon)} \overset\lambda\longrightarrow \W^{(\upsilon)}\,. \ee
\end{Exp}

\subsubsection{Elementary tensor products and braiding}
\label{sec:Wgen-tensor}

The second operation under which we would like the category $\wh\CC_\rho^B$ to be closed is the tensor product, induced physically by collision of parallel line operators. Ultimately, we will argue that $\wh\CC_\rho^B$ is closed under tensor products by exhibiting it as the derived category of an abelian category of VOA modules that in turn has a well-defined \emph{exact} tensor product. Here, we simply remark that it is possible to guess some (but not all) of the tensor products among the $\W_\m^{(\upsilon)}$.

Consider, for example, pure Wilson lines $\W_\m$. Since they generate the topological one-form symmetry, we must have a simple fusion rule
\be \W_m\otimes \W_{m'} \simeq \W_{m+m'}\,. \label{eq:Wfuse} \ee
Moreover, it is clear from the bound states defining monodromy-Wilson lines that
\be \W_m\otimes \W_{m'}^{(\upsilon')} \simeq \W_{m+m'}^{(\upsilon')}\,.\label{fusionna} \ee

Collision of monodromy defects is trickier, however. Abelian monodromy around nearby parallel line operators must add, so in general we must have
\be \W_{\m}^{(\upsilon)} \otimes \W_{\m'}^{(\upsilon')} = \text{some complex of $\W_{\m''}^{(\upsilon+\upsilon')}$, for various $\m''$}\,. \label{monodromy-fusion} \ee
We will eventually see, using fusion of boundary VOA modules, that \eqref{monodromy-fusion} typically contains several different summands and/or complexes. The precise structure is controlled by which hypermultiplet scalars are effectively set to zero for monodromies $\upsilon,\upsilon'$, and $\upsilon+\upsilon'$.

\begin{Exp}
    Consider $\text{SQED}_1$, and monodromy lines $\W_0^{(\upsilon)}$, $\W_0^{(\upsilon')}$ such that none of $\upsilon,\upsilon'$, or $\upsilon+\upsilon'$ vanish (mod $\Z$). Each of $\W_0^{(\upsilon)}$, $\W_0^{(\upsilon')}$ contains a bound state of pure Wilson lines $\W_0$ and $\W_1$, so their tensor product should contain a bound state of $(\W_0\oplus \W_1)^{\otimes 2} = \W_0\oplus \W_1\oplus \W_1\oplus \W_2$. The unique way to regroup this into $\W_m^{(\upsilon+\upsilon')}$ lines is
\be \W_0^{(\upsilon)}\otimes \W_0^{(\upsilon')} \simeq \W_0^{(\upsilon+\upsilon')}\oplus \W_1^{(\upsilon+\upsilon')}\,. \ee
 \end{Exp}

This example can be generalized to any simple abelian theory. So long as $\upsilon,\upsilon',\upsilon+\upsilon'$ all require the same set of hypermultiplets to vanish, we expect
\be \W_\m^{\upsilon}\otimes \W_{\m'}^{\upsilon'} \simeq \Big[\bigotimes_{\text{$i$ s.t. $\rho(\upsilon)^i\notin\Z$}} \big(\id \oplus \W_{(0,...,\underset{i}1,...,0)}\big)\Big]\otimes \W_{\m+\m'}^{(\upsilon+\upsilon')}\,. \label{gen-W-tensor} \ee
When different hypermultiplets vanish, the RHS can be more interesting.

We may also guess some elementary braiding morphisms directly from bulk considerations. Recall that the braiding of lines $\ell_1,\ell_2$ is an isomorphism $c_{\ell_1,\ell_2} \in \text{Hom}(\ell_1\otimes \ell_2,\ell_2\otimes\ell_1)$ that corresponds physically to the  local operator defined by the braided configuration
\be   \raisebox{-.5in}{\includegraphics[width=1.8in]{braidR.pdf}}  \ee

As there are no bulk Chern-Simons terms for the gauge connection, we expect Wilson lines to braid trivially with each other,
\be R_{\W_\m,\W_{\m'}} = \text{id} \in \text{Hom}_{\wh\CC_\rho^B}(\W_{\m}\otimes\W_{\m'},\W_{\m'}\otimes\W_\m) = \text{End}_{\wh\CC_\rho^B}(\W_{\m+\m'})\,. \ee
This turns out to be true up to a sign.
Moreover, wrapping a Wilson line all the way around a monodromy defect should measure its monodromy. This is expressed as
\be \mu(\W_{m'},\W_{\m}^{(\upsilon)}) = e^{2\pi im\cdot \upsilon}\cdot\text{id} \in \text{End}_{\wh\CC_\rho^B}(\W_{m+m'}^{(\upsilon)})\,, \ee
where in general the monodromy, or double braiding, is defined as
\be \mu(\ell,\ell') = R_{\ell',\ell}\circ R_{\ell,\ell'} \in \text{End}(\ell\otimes \ell')\,.\ee
The braiding of other monodromy-Wilson lines, and more general bound states, can be remarkably complicated!

\subsection{A twist}
\label{sec:A-lines}

We'll next describe the A-twist mirrors of the basic line operators in the B twist. To this end, it is helpful to recall a geometric classification of half-BPS lines, from \cite{lineops}.

\subsubsection{Review of line operators in the A twist}
\label{sec:A-review}

Half-BPS line operators compatible with the A twist of a 3d $\CN=4$ gauge theory were constructed in \cite{lineops}, motivated by characterizations of surface operators in 4d supersymmetric gauge theory (beginning with \cite{GukovWitten}) and mathematical work of Webster \cite{Webster-tilting} and Hilburn-Yoo. These line operators are generalized vortex lines, associated with \emph{Lagrangian singularities} in the gauge and matter fields. 

Specifically, taking 3d spacetime to be $\C_{z,\bar z}\times \R_t$ with line operators at $z=\bar z=0$, a basic set of half-BPS line operators are classified by
\begin{itemize}
    \item[1)] A choice of a meromorphic boundary condition for hypermultiplet scalars $\big(X(z),Y(z)\big)\in \CL$ near $z=0$ near $z=0$, defined in terms of a subvariety $\CL \subset V(\!(z)\!) \times V^*(\!(z)\!)$ such that
    \begin{itemize}
        \item $\CL$ is holomorphic Lagrangian with respect to the holomorphic symplectic form   $\Omega=\oint dz\, \langle\delta  Y(z)\wedge \delta X(z)\rangle$ --- in the sense that it is a maximal subvariety on which $\Omega$ vanishes; and
        \item $\CL$ is compatible with the vanishing of holomorphic moment maps $\mu^a=\rho_i{}^aX^i(z)Y^i(z)$ for the gauge symmetry, in the sense that $\{\mu^a=0\}\cap\CL\neq \oslash$.
    \end{itemize}
    \item[2)] A choice of subgroup $\Gamma\subset G_\C(\!(z)\!)$ of complexified, meromorphic gauge transformations near $z=0$ that preserves the boundary condition in (1).
\end{itemize}
This geometric description is obtained after complexifying the gauge group and passing to a holomorphic gauge. 

Here we use the notation $\C[\![z]\!]$ to denote formal power series (holomorphic functions in an infintesimal neighborhood of $z=0$) and $\C(\!(z)\!)$ to denote formal Laurent series (meromorphic functions in an infinitesimal neighborhood of $z=0$). Coordinates on $V(\!(z)\!)\times V^*(\!(z)\!)$ are given by the coefficients of the Laurent series
\be X^i(z)=\sum_{k\in \Z} x_k^i z^k\,,\qquad Y^i(z)=\sum_{k\in \Z} y_k^i z^k\,,\ee
in terms of which the holomorphic symplectic form is $\Omega = 2\pi i \sum_{k\in \Z}\sum_{i=1}^n \delta y_k^i\wedge \delta x_{-k-1}$.

\begin{Exp} The trivial/identity line operator has
\be \id:\quad \CL_{\id} = V[\![z]\!]\times V^*[\![z]\!]\,,\qquad \Gamma_{\id}=G_\C[\![z]\!]\,, \label{id-vortex} \ee
which simply says that hypermultiplet scalars and gauge transformations are regular near $z=0$. The space $\CL_{\id}$ is cut out by the constraints $x_k=y_k=0$\; $\forall\, k < 0$; this is a maximal subspace on which $\Omega\big|_{\CL_{\id}}=0$. 
\end{Exp}

We will further impose a natural finiteness condition on vortex lines, by requiring the singularity data $(\CL,\Gamma)$ to be finitely far from that of the identity line \eqref{id-vortex}. Namely, we ask that both $(\CL\cup \CL_{\id})\backslash(\CL\cap\CL_{\id})$ and $(\Gamma\cup \Gamma_{\id})\backslash(\Gamma\cap\Gamma_{\id})$ have bounded dimension. If one thinks physically of defining line operators by coupling the bulk 3d theory to 1d quantum mechanics along a line, the finiteness condition amounts to asking that the quantum mechanics has (in the topological twist) a finite-dimensional Hilbert space.

We'll now describe two classes of vortex lines for simple abelian gauge theories that have%
\be \Gamma = \Gamma_{\id} = GL(1,\C[\![z]\!])^r\,, \ee
just like the identity line \eqref{id-vortex}. 
(Here
$U(1)_\C[\![z]\!]=GL(1,\C[\![z]\!])$ is the group of invertible power series, \emph{i.e.} power series whose constant term is nonzero.)

\subsubsection{Flavor vortices}
\label{sec:A-flavor}

The first class of abelian vortex lines preserved by $\Gamma_{\id}$ was studied in detail in \cite{lineops}. They are known as flavor vortices. (They had also appeared in \cite{BDGH, BFN-lines} and much earlier work on supersymmetric line operators in 3d $\CN=2$ theories. See \cite{lineops} for a more complete list of references.)

For any $\lambda\in \Z^n$, consider the condition
\be X^i\in z^{\lambda^i} \C[\![z]\!]\,,\quad  Y^i\in z^{-\lambda^i}\C[\![z]\!]\,, 
\label{XY-vortex}\ee
\emph{i.e.} $X,Y\in \CL_{\lambda}$ with $\CL_{\lambda}=\big\{x^i_{k+\lambda^i}=0,\,y^i_{k-\lambda^i}=0\big\}_{k<0}$.
This allows each $X^i$ to have a pole of some order while simultaneously forcing $Y^i$ to have a zero of the same order, or vice versa. The pairing up of poles and zeroes ensures that $\CL_\lambda$ is Lagrangian.

Each $\CL_\lambda$ as above defines a vortex line. However, these vortex lines are not all independent, because they can be screened by dynamical vortex particles. Indeed, a (complexified) gauge transformation by  $g(z) = z^\m$ (for any $\m\in \Z^r$) in the neighborhood of a line operator sends $(X,Y)\mapsto (z^{\rho(\m)}X,z^{-\rho(\m)}Y)$.
This renders equivalent any pair of vortex lines labelled by $\lambda,\lambda'\in \Z^n$ such that $\lambda-\lambda'\in \rho(\Z^r)\,.$ After choosing a splitting $\sigma$ of the exact sequence of gauge/flavor charges as in \eqref{split-SES2}, we can label inequivalent line operators by cocharacters of the flavor symmetry $\mv\in \text{cochar}(T_F)\simeq  \Z^{n-r}$,
\be \label{def-Vm} \V_\mv:\quad X(z)\in z^{\sigma(\mv)}\C[\![z]\!]\,,\quad Y(z)\in z^{-\sigma(\mv)}\C[\![z]\!]\,,\qquad \mv \in \Z^{n-r}\,.\ee

The $\V_\mv$ are called ``flavor vortices'' because they can be created by applying a non-constant $T_F$ flavor-symmetry transformation (by $z^\m$) in the punctured neighborhood of $z=0$. This immediately identifies them as the $\pi_1(T_F)=\Z^{n-r}$ one-form symmetry defects of $\CT_\rho^A$ expected from Section \ref{sec:highersym}. Thus they are also the 3d mirrors of Wilson lines $\W_m$ in $\CT_{\rho^\vee}^B$.

We note that the constraints \eqref{def-Vm} defining $\V_{\mv}$ preserve the cohomological $U(1)_C$ symmetry of the A twist (under which $X,Y$ are uncharged). 
They also preserve twisted spin in the A twist: the diagonal subgroup of $U(1)_H$ and rotations in the $z$ plane, which becomes the conformal grading in a boundary VOA, \cf \eqref{conf-A}. Recall that the twisted spin of $X,Y$ and the coordinate $z$ is
\be |X|=|Y|=\tfrac12\,,\qquad |z|=-1\,. \label{spin-A} \ee

\subsubsection{Flavor defects}
\label{sec:mirror-monodromy}

There is a second class of vortex lines with regular gauge symmetry $\Gamma_{\id}$ that also preserve $U(1)_C$ cohomological symmetry and twisted spin. We expect them to be the 3d mirrors of monodromy defects.

Let $\uv\in \C^{n-r}$. A first guess at constructing the mirror of a monodromy defect is to fix the residue of the complex moment maps for the $T_F$ flavor symmetry to constant values
\be \mu_F \sim \frac{\uv}{z}+\text{regular} \qquad \text{near $z=0$}\,, \label{mu-residue} \ee
where $\mu_F = \sigma^T(XY)$, or in components
\be \mu_F^\alpha = \sum_{i=1}^n \sigma_i{}^\alpha X^iY^i\,. \ee

This is reasonable for several reasons. It induces a singular, meromorphic profile for $X,Y$ near $z=0$, fitting the general geometric characterization of vortex lines from Section \ref{sec:A-review}. Moreover, it preserves twisted spin: since $\mu_F$ is bilinear in $X,Y$, it has twisted spin 1, so its residue is dimensionless and can be set to a nonzero constant value.

An even stronger motivation for considering \eqref{mu-residue} comes from 3d mirror symmetry. The complex flavor moment maps $mu_F$ in  $\CT_\rho$ are the 3d mirrors of the complex vectormultiplet scalars $\varphi$ in $\CT_{\rho^\vee}$. In the B twist, the latter complexify the gauge connection in the $\C_{z,\bar z}$ plane, forming the combination $\CA=A+i(\varphi\,dz+\ol\varphi\,d\bar z)$ that gets a pole near a monodromy defect \eqref{A-mon}.

In terms of the underlying 3d $\CN=4$ theories, the literal mirror of a monodromy $\CA=\frac{\uv}{2i}\Big(\frac{dz}{z}-\frac{d\bar z}{\bar z}\Big)$ is not quite \eqref{mu-residue}. Rather, the mirror map sets $\mu_F \sim \text{Im}(\uv)/z$ while simultaneously creating a singularity in the $T_F$ flavor currents near $z=0$. To see this, consider the gauge-invariant curvature of the complexified connection $\CA$ near a monodromy singularity,
\be \CF=d\CA = 2\pi \uv \delta^{(2)}(z,\bar z) \label{F-sing} \ee
The real curvature is Hodge-dual to the $T_{\rm top}$ flavor symmetry current in $\CT_{\rho^\vee}$ so the complexified curvature may be written $\CF=*J_{\rm top}+id(\varphi\,dz+\ol\varphi\,d\bar z)$. The 3d mirror of \eqref{F-sing} in $\CT_\rho$ is therefore
\be *J_F +i(\pd_{\bar z}\mu_F-\pd_z\bar\mu_F) = 2\pi \uv \delta^{(2)}(z,\bar z)\,. \label{realvortex} \ee
The imaginary part of $\uv$ specifies a singuarlity in $\mu_F$, while the real part specifies a singularity in $*J_F$. 

The identification \eqref{realvortex} explains why the true physical mirror of a monodromy line should be labelled by a \emph{periodic} parameter $\uv\in (\C/\Z)^{n-r}$. The physical theory $\CT_\rho$ has dynamical particles of all integral $T_F$ flavor charges, which themselves source delta-function singularities in the flavor current.
In the presence of flavor-charged particles, it only possible to fix $*J_F = 2\pi \text{Re}(\uv)\delta^{(2)}(z,\bar z)$ for an equivalence class $\text{Re}(\uv)\in (\R/\Z)^{n-r}$.

In the A twist of $\CT_\rho$, the constraint \eqref{mu-residue} is equivalent to \eqref{realvortex}, up to $Q$-exact terms. (In the A twist, the flavor current itself is exact, and there are no longer any $T_F$-charged particles.) Periodicity in the parameter $\uv$ is restored in the category of line operators, as lines related by integer shifts of $\uv$ all become equivalent objects in the category. This periodicity will be manifest for the VOA modules of Sections \ref{subsec:baby}--\ref{subsec:LineA}. It is also manifest from the description of line operators in the A twist as D-modules on a loop space, \cf\ \cite{CCG,hilburn2021tate}, though we will not pursue this perspective further here.

There is an important additional subtlety to overcome in fully defining the mirrors of monodromy line. It parallels the additional choices required to carefully defining the B-twisted monodromy lines themselves. From an A-twist perspective, the problem is that the constraints \eqref{mu-residue} define a subvariety of $V(\!(z)\!)\times V^*(\!(z)\!)$ that is infinitely far away from the Lagrangian $\CL_\id = V[\![z]\!]\times V^*[\![z]\!]$ for the identity line. In particular, \eqref{mu-residue} allows infinitely many negative (polar) modes of $X$ and $Y$ fields to vary freely. To satisfy the finiteness condition of Section \ref{sec:A-review}, we must specify a Lagrangian in $V(\!(z)\!)\times V^*(\!(z)\!)$ that is finitely far away from $\CL_\id$ and still allows residues \eqref{mu-residue}. There is no unique minimal way to do this. There are multiple choices, all related to each other by tensoring with flavor vortices $\V_{\mv}$.

One minimal choice is the following. We can ``solve'' the residue constraints $\mu_F = \sigma^T(XY) = \frac{\uv}{z}+...$ by keeping all the $X$ fields regular
\be X^i(z) = x_0^i+x_1^i z + \ldots \ee
and giving the $Y$ fields a first-order pole
\be Y^i(z) = \frac{y_{-1}^i}{z}+ y_0^i+y_1^i z + \ldots \quad\text{such that} \quad y_{-1}^i = \frac{\rho^\vee(\uv)^i}{x_0^i}\,. \label{solve-y} \ee
Here $\rho^\vee:\Z^{n-r}\to \Z^n$ is the second map in the exact sequence \eqref{split-SES2}; it satisfies $\sigma^T\rho^\vee = \text{id}_{\Z^{n-r}}$, whence \eqref{solve-y} implies $\mu_F=\frac{\uv}{z}+...$.  The subvariety of $V(\!(z)\!)\times V^*(\!(z)\!)$ associated with these constraints is
\be \label{def-Lu} \CL^{\uv,Y} := \begin{cases} x_k = 0 & k<0 \\
y_k = 0 & k<-1 \\
y_{-1}^i = 0 & \rho^\vee(\uv)^i = 0 \\
y_{-1}^ix_0^i = \rho^\vee(\uv)^i & \rho^\vee(\uv)^i \neq 0 \end{cases} \ee

The variety $\CL^{\uv,Y}$ is now just finitely far from $\CL_{\id}$. It can also be checked directly that it is Lagrangian. In particular, $\CL^{\uv,Y}$ can be constructed from $\CL_\id$ as the graph of the generating function
\be f^{\uv}: \CL_\id \to \C\,,\qquad f^{\uv}(\{x_k,y_k\}_{k\geq 0}) := \rho^\vee(\uv)\cdot \log x_0\,. \ee
The exponential of this generating function $\exp(f^{\uv}) = x_0^{\rho^\vee(\uv)}$ will play a direct role in the definition of boundary VOA modules, \cf\ \eqref{bg-W}. 

The choice to giving only $Y$'s a pole --- as opposed to giving some other combination of $X$'s and $Y$'s a pole --- in order to solve $\mu_F = \uv/z+...$ was entirely arbitrary. However, all other choices are related to \eqref{def-Lu} by shifts of the modes of $X$'s and $Y$'s, \emph{i.e.} by replacing $(X,Y)\mapsto (z^\lambda X,z^{-\lambda}Y)$ for some $\lambda\in \Z^n$. Recall from Section \ref{sec:A-flavor} that shifts by gauge charges lead to equivalent line operators. The distinct, minimal ways to solve $\mu_F = \uv/z+...$ are simply classified by combining \eqref{def-Lu} with all possible shifts by flavor charges. This leads to an entire family of vortex lines,
\be \V_\mv^{(\uv)}:=\quad \begin{cases} X \in z^{\sigma(\mv)}\C[\![z]\!] \\ Y \in z^{-\sigma(\mv)-1}\C[\![z]\!] \end{cases} \quad\text{and}\quad 
y_{-\sigma(\mv)^i-1}^i = \frac{\rho^\vee(\uv)^i}{x_{\sigma(\mv)^i}^i}\,. \label{def-monV} \ee
We propose these to be the 3d mirrors of the monodromy-Wilson lines \eqref{def-monW}. We will refer to the entire family as \emph{flavor defects}.

When $\uv=0$ a flavor defect becomes a standard flavor vortex, $\V_\mv^{(0)}\simeq \V_\mv$\,. We also expect equivalences 
\be \V_\mv^{(\uv)}\simeq \V_\mv^{(\uv')} \quad\text{when}\quad \uv-\uv'\in \Z^r\,, \ee
which we ultimately justify in Sections \ref{subsec:baby}--\ref{subsec:LineA}.

\subsubsection{Morphisms, fusion, braiding, and the category}
\label{sec:A-cat}

In the B twist, we computed all morphism spaces \eqref{HomW-gen} among monodromy-Wilson lines $\W_\m^{\upsilon}$. If the flavor defects \eqref{def-monV} are indeed 3d mirrors of the monodromy-Wilson lines, then there should be a match
\be \text{Hom}^\bullet(\V_{\mv}^{\uv},\V_{\mv'}^{\uv'})\;\;\text{in $\CT_\rho^A$} \;\;\simeq\;\; \text{Hom}^\bullet(\W_{\mv}^{\uv},\W_{\mv'}^{\uv'})\;\;\text{in $\CT_{\rho^\vee}^B$} \label{match-WV} \ee
for all $\mv\in \Z^{n-r}$ and $\uv\in (\C/\Z)^{n-r}$.

When $\uv=\uv'=0$, the morphism spaces in the A twist are generated by monopole operators. They were computed in \cite{lineops} using a version of a state-operator correspondence that lead to cohomology of moduli spaces of solutions to vortex equations. The result was proven to satisfy \eqref{match-WV}. (The computation can also be done using methods that Braverman-Finkelberg-Nakajima introduced to construct Coulomb branches \cite{BFN-lines,Webster-Coulomb,Webster-tilting}.) Similar methods should yield a computation of morphism spaces for more general $\uv,\uv'$ as well, though we won't pursue that here.

Just like in the B twist, line operators in the A twist can be further combined --- added and/or coupled to 1d quantum mechanics, and collided with each other. They should form a DG braided tensor category. The smallest category that contains all the flavor defects we have just constructed is:
\begin{Def} \label{def:CA-bulk}
    Let $\wh\CC_\rho^A$ be the DG category generated by the flavor defects $\V_\mv^\uv$, in the A-twisted theory $\CT_\rho^A$. In other words, the objects of $\wh\CC_\rho^A$ are represented by complexes of the $\V_\mv^\uv$ and cohomological shifts thereof, modulo quasi-isomorphism.
\end{Def}

\noindent In terms of the underlying physical 3d $\CN=4$ theory, the flavor vortices $\V_\mv$ represent half-BPS line operators \cite{lineops}. More general flavor defects $\V_\mv^{(\uv)}$ would seem to represent half-BPS line operators as well, as they correspond to holomorphic Lagrangians (as opposed to real Lagrangians) in the loop-space characterization. This seems at odds with monodromy-Wilson lines $\W_\mv^{(\uv)}$ being quarter-BPS, and poses a small puzzle to solve.
General complexes of the $\V_\mv$ and $\V_\mv^\uv$ would all be quarter-BPS.

In Section \ref{subsec:LineA} we will construct an abelian category of VOA modules whose simple objects are precisely the $\wh\CC_\rho^A$, with derived morphism spaces that match all expectations from 3d mirror symmetry \eqref{match-WV}. 
The identification of bulk vortex defects with modules for the boundary VOA is straightforward, as the boundary VOA is manifestly built from the modes of the $X$ and $Y$ bulk fields. We will further establish a braided tensor structure on the category of VOA modules.

A very small part of this braided tensor structure can be inferred directly from symmetry considerations. Namely, as the $\V_\mv$ are $\Z^r$ one-form generators, we expect
\be \V_\mv \otimes \V_{\mv'}^{\uv'}\simeq \V_{\mv+\mv'}^{\uv'}\,. \ee
Moreover, from the identification of $\V_m^{(\upsilon)}$ as mirrors of monodromy defects, we expect a monodromy
\be \mu(\V_{\mv'}, \V_\mv^{(\uv)}) = e^{2\pi i \mv'\cdot\uv'}\cdot\text{id}\,.\ee
Obtaining any further information on tensor products and braiding directly in the A-twist  seems extremely difficult.

\section{The category of line operators from boundary VOA's}
\label{sec:VOAlines}

In this section, we define the braided tensor category of line operators for $\CT_\rho^A$ and $\CT_\rho^B$ using boundary vertex operator algebras. 

More specifically, for the theory $\CT_\rho^A$, using the isomorphism $\CV_\rho^A\otimes \CV_{\Z^{2r}}^{+-}\cong \wt{\CV}_\rho^A\otimes \CV_{bc}^V$, and that $\CV_{\Z^{2r}}^{+-}$ and $\CV_{bc}^V$ have trivial category of modules, we define an abelian category $\mc C_\rho^A$ in Definition \ref{def:lineopeA}, as a full subcategory of generalized modules of $\wt{\CV}_\rho^A$. We will show that this category has the structure of a braided tensor category defined by logarithmic intertwining operators, using the idea of simple current extensions, to be recalled in Section \ref{subsecsimplecurrent}. 

For the theory $\CT_\rho^B$, one could simply use the mirror symmetry statement in terms of boundary VOA, namely Theorem \ref{thm:VOA} and Theorem \ref{ThmMirSymbdVOA} to transport the category $\mc C_\rho^A$ to the B side. We will, however, use a different route. We will define an abelian category $\mc C_\rho^B$ in Definition \ref{Def:lineopeB}, as a full subcategory of modules of $\mc V_\rho^B$, and show that it has the structure of a braided tensor category, using the fact that $\mc V_\rho^B$ is a simple current extension of $\widehat{\grho}$.

We propose the mathematical definition of the category of line operators $\widehat{\mc C}_\rho^A$ and $\wh{\mc C}_\rho^B$ to be the bounded derived category:
\be
\widehat{\mc C}_\rho^A\simeq D^b \mc C_\rho^A,\qquad \widehat{\mc C}_\rho^B\simeq D^b \mc C_\rho^B.
\ee
Tensor structure and braiding of $\widehat{\mc C}_\rho^A$ and $\widehat{\mc C}_\rho^B$ will be induced from the corresponding structure of the abelian categories, thanks to the fact that tensor product is exact on the abelian categories in question. Under this proposal, the space of junctions between line operators, derived in Section \ref{sec:bulklines} by physical analysis, is the Yoneda extension group between objects in the bounded derived category. 

We give the following justifications of this proposal:
\begin{enumerate}
    \item We carefully analyze the content of $\mc C_\rho^A$ and $\mc C_\rho^B$ and show that simple objects in these categories are in one-to-one correspondence with the basic physical line operators of Section \ref{sec:bulklines}.

    \item We analyze examples of tensor product and monodromy between simple objects and show that these match the result from physical analysis of Section \ref{sec:bulklines}. 

    \item We compute the derived extension group between simple Wilson lines, and show that these reproduce the results of Section \ref{sec:bulklines}. 

    \item More importantly, we prove mirror symmetry statement, which is an equivalence of BTC:
    \be
\mc C_\rho^A\simeq \mc C_{\rho^\vee}^B,\qquad \mc C_\rho^B\simeq \mc C_{\rho^\vee}^A.
\ee
We show that under this equivalence, the simple objects match perfectly as expected from physical arguments.
    
\end{enumerate} 

The structure of this section is as follows:

\begin{itemize}
    \item In Section \ref{subsecsimplecurrent}, we recall the machinery of VOA extensions \cite{creutzig2017tensor, creutzig2020direct}, in particular simple current extensions \cite{creutzig2020simple}   which is the machinery we will use in defining the braided tensor categories $\mc C_\rho^A$ and $\mc C_\rho^B$. 

    \item In Section \ref{subsec:baby}, we recall the most fundamental example of abelian mirror symmetry from boundary VOA, based on \cite{BN22}. This will be a building block for the content of the rest of the section. 

    \item In Section \ref{subsec:LineA}, we define the abelian category $\mc C_\rho^A$ and show it has the structure of a braided tensor category. We also classify simple objects and compute their fusion rules and examples of monodromy. 

    \item In Section \ref{subsec:LineB}, we define the abelian category $\mc C_\rho^B$, and show it has the structure of a braided tensor category. We also classify simple objects and compute their fusion rules and examples of monodromy.

    \item In Section \ref{subsec:MirrorLines}, we show the equivalence of BTC $\mc C_\rho^A\simeq \mc C_{\rho^\vee}^B$. We also show how the simple objects in the categories match under this equivalence. 

    \item In Section \ref{subsec:HomLines}, we compute derived extension groups in the derived category $D^b \mc C_\rho^B$, and show that they match the physical computations in Section \ref{sec:bulklines}.

\end{itemize}

\subsection{Simple Current extensions and intertwining operators}\label{subsecsimplecurrent}

Before giving the definition of the category of line operators for $\mc T_{\rho}^A$ and $\mc T_\rho^B$, we will first recall the basic properties of simple current extensions, at the level of generality needed to tackle our braided tensor categories of VOA representations. 

Let $\CV$ be a VOA and $\mc C$ be a category of generalized modules of $\CV$. An object $\mc M$ in $\mc C$ is a smooth module of $U(\CV)$, the universal enveloping algebra of $\CV$. By the work of Huang-Lepowsky-Zhang \cite{huang2014logarithmic1,huang2010logarithmic2, huang2010logarithmic3, huang2010logarithmic4, huang2010logarithmic5, huang2010logarithmic6, huang2011logarithmic7, huang2011logarithmic8}, under certain conditions on $\CV$ and $\mc C$, there is a ``tensor-product" structure on $\mc C$, where tensor product of two modules is defined as the universal object of logarithmic intertwining operators.

Specifically, given three modules $\mc M, \mc N$ and $\mc P$, a logarithmic intertwining operator is a map:
\be \label{Y-int}
Y: \mc M\otimes \mc N\to \mc P\{z\}[\log z]
\ee
that satisfies a similar set of properties the state-operator map of a VOA needs to satisfy. 
The curly bracket means a formal power series with any possible complex  exponents, \ie\ sums of the type
\[
\sum_{\substack{y \in \mathbb C \\ n \in \mathbb N}} P_{y, n}\, z^y\, (\log(z))^n, \qquad  P_{y, n} \in\mc P.
\]
For a full definition, see \cite{creutzig2017tensor}. Let $\mathrm{Int}(\mc M\otimes\mc N,\mc P)$ be the set of $\mc M, \mc N, \mc P$ intertwining operators \eqref{Y-int}. For fixed $\mc M$ and $\mc N$, 
 a universal object for the space of intertwining operators is a module denoted $\mc M\otimes_\CV \mc N$, together with an intertwining operator $Y:\mc M\otimes\mc N\to \mc M\otimes_\CV\mc N \{z \}[\log z]$ that induces an isomorphism
\be
\mathrm{Int}(\mc M\otimes\mc N, P)\cong \mathrm{Hom}(\mc M\otimes_\CV\mc N, \mc P),
\ee
for any other module $\mc P$.

 It is difficult in general to verify the existence of such a universal object and in particular the existence of tensor category. Usually one can show this existence if the category $\mathcal C$ satisfies certain finiteness conditions \cite[Theorem 3.3.4]{creutzig2021tensor}. 
The best known examples are first the Kazhdan-Lusztig category for an affine Lie algebra associated to a semi-simple Lie algebra, which  at most levels satisfies the necessary finiteness conditions and has the structure of a braided tensor category; and second the Virasoro algebra at any central charge has a category of modules satisfying all the necessary finiteness conditions and hence has the structure of a braided tensor category \cite{Creutzig:2020zvv}. In general it is a quite difficult problem to establish  the structure of a braided tensor category on a category of modules of a VOA. 

In practice, especially in the VOA that we study in this paper, the category $\mc C$ will \emph{not} satisfy the conditions that ensure the existence of vertex tensor category directly. In particular
objects in the category $\mc C$ will fail to have bounded-from-below conformal degree. The works \cite{creutzig2017tensor, creutzig2020direct}  resolve this problem in particular in the following situation. Suppose that there is a vertex operator subalgebra $\CW\hookrightarrow \CV$ such that $\CV$ as a module of $\CW$ is decomposed into $\CV=\bigoplus_{c\in C}\CW_c$, where $\CW_c$ are simple modules of $\CW$, labelled by an abelian group $C$, and OPE $\CW_c\otimes \CW_{c'}\to \CW_{cc'}(\!(z)\!)$. In many examples, including the ones considered in this paper, this occurs where there is an action of another abelian group $C^\vee$ on $\CV$, such that $\CW=\CV^{C^\vee}$, the invariant part, and $\CW_c=\CV^{C^\vee, c}$ where $c$ denotes an element in the character lattice of $C^\vee$. In this case, even though modules of $\CV$ might not have  conformal degrees bounded from below, if modules of $\CW$ do, then one can still obtain braided tensor structure from intertwining operators for $\CV$. 

More specifically, suppose $\mc D$ is a category of modules of $\CW$ that is shown to have braided tensor structure via logarithmic intertwining operators, and that $\CV$ is an object in $\mc D$. Then we can view $\CV$ as a commutative algebra object $\mc A$ in $\mc D$, such that the state-operator correspondence of $\CV$ is the multiplication map: $\CV\otimes_\CW \CV\to \CV$, and the locality is the commutativity condition. In other words, $\CV$ is viewed as a tensor functor from $\mathrm{Vect}_C$, the category of vector spaces graded by $C$, to $\mc D$. We call $\CV$ a simple current extension of $\CW$ by $\{\CW_c\}_{c\in C}$ in $\mc D$. Let $\mc M$ be an object in $\mc D$ that has the structure of a left $\mc A$-module, meaning there is a morphism $\mc A\otimes_\CW \mc M\to \mc M$ that induces the module structure. Then one can interprete this morphism as a intertwining operator $\CV\otimes \mc M\to \mc M \{z\}[\log z]$. This is not quite a module for $\CV$ due to possible non-locality of the intertwining operator. 
Non-locality is measured by the monodromy, and $\mc M$ is a $\CV$-module if 
the monodromy is the identity.

``Monodromy'' here refers to the square of the braiding. Given modules $\mc M,\mc N$, and braiding isomorphisms
\be \label{def-RMN} R_{\mc M,\mc N}:\mc M\otimes_\CW \mc N\to \mc N\otimes_\CW \mc M\,,\qquad R_{\mc N,\mc M}:\mc N\otimes_\CW \mc M\to \mc M\otimes_\CW\mc N\,, \ee
the monodromy between $\mc M$ and $\mc N$ is
\be \label{def-monMN} \mu(\mc M,\mc N) := R_{\mc N,\mc M}\circ R_{\mc M,\mc N} \in \text{End}(\mc M\otimes_{\mc W} \mc N)\,.\ee
The locality condition above states that $\mc M$ is a $\CV$-module if $\mu(\mc A,\mc M)=\text{id}$\,.
Thus, an $\mc A$ module object in $\mc D$ that is local to $\mc A$ is precisely a module of the VOA $\CV$. Let us denote by $\mc A\!-\!\mathrm{Mod}_{\mathrm{loc}}(\mc D)$ the category of local $\mc A$ modules in $\mc D$. This is a full subcategory of $\CV$ modules, and the work of \cite{creutzig2017tensor} states the following, which essentially translates the universal object for logarithmic intertwiners from $\CW$ to $\CV$. 

\begin{Thm}
$\mc A\!-\!\mathrm{Mod}_{\mathrm{loc}}(\mc D)$ has the structure of a braided tensor category defined by the universal property of logarithmic intertwining operators. 
\end{Thm}

In the applications to follow, we will simply refer to this category as $\CV\!-\!\mathrm{Mod}_{\mathrm{loc}}(\mc D)$, to stress the fact that this is a category of generalized modules of $\CV$. Given any module $\mc M$ of $\CW$ in $\mc D$, one can obtain a left module of $\mc A$ by $\mc A\otimes_\CW \mc M$.    One can think of this as induction from modules of $\mathbbm{1}$, the identity object (which is a commutative algebra object), to $\mc A$ modules. One can ask whether $\mc A\otimes_\CW \mc M$ is local with respect to $\mc A$, namely whether it gives rise to a local module of $\CV$. Similar to the above locality, this happens precisely when the monodromy is trivial:
\begin{equation}
  R_{\mc M,\mc A}\circ R_{\mc A,\mc M}=\mathrm{Id}.
\end{equation}
Thus, if we let $\mc D^{[0]}$ be the full subcategory of $\mc D$ whose objects have trivial monodromy with $\mc A$, then the assignment:
\begin{equation}
    \mathcal{L}(\mc M):=\mc A\otimes_\CW \mc M
\end{equation}
gives a functor:
\begin{equation}
    \mathcal L: \mc D^{[0]}\longrightarrow \mc V\!-\!\mathrm{Mod}_{\mathrm{loc}}(\mc D)
\end{equation}
The theorem of \cite{creutzig2017tensor} states:
\begin{Thm}\label{Lexact}
The functor $\mathcal L$ is a braided tensor functor. If tensor product in $\mathcal D$ is exact, and there are no simple-current fixed points, then $\mathcal L$ is a full, surjective functor that maps simples to simples and preserves composition series \cite[Proposition 3.2]{Creutzig:2022ugv}.
\end{Thm}
A simple-current fixed point is a simple module $\mc M$, such that $\CW_c \otimes_\CW \mc M \cong \mc M$. All the categories in this work are fixed-point free.  The above theorem means that with the assumptions of $\CV$ and $\CW$, namely that $\CV$ is a simple current extension of $\CW$ in $\mc D$, we can view the category $\mc V\!-\!\mathrm{Mod}_{\mathrm{loc}}(\mc D)$ as the de-equivariantization of $\mc D^{[0]}$ by $C^\vee$. This is also denote by $D^{[0]}_{C^\vee}$ in \cite{egno}. The effect of the de-equivariantization in this case is simple: it identifies an object $\mc M$ of $\mc D^{[0]}$ with the object $\CW_c\otimes_\CW \mc M$ for all $c$, so we will also suggestively write it as $\mc D^{[0]}/C$. Thus, if the category $\mc C$ we would like to study happens to be $\mc D^{[0]}/C$, then $\mc C$ will have the structure of a braided tensor category. Moreover, tensor product will be an exact functor since it is an exact functor on $\mc D$ and that $\mc L$ preserves composition series. 

We also comment that in the applications of this paper, the object $\mc A$ is not technically an object of $\mc D$, since for us $C$ will be a product of \emph{infinite} cyclic groups, $C=\Z^\#$. This is resolved in \cite{creutzig2020direct} by constructing a suitable completion $\mathrm{Ind}(\mc D)$ that allows infinite direct sum, and that the braided tensor category structure will simply be extended to this completion. Then all the formulation above is applied to $\mathrm{Ind}(\mc D)$ and one can obtain a braided tensor category of $\CV$ modules even when modules of $\CV$ don't have conformal degrees bounded from below. This was successfully applied in \cite{BN22} to the case of the symplectic boson VOA $\mc{V}_{\beta\gamma}$, to obtain a large category of modules of $\mc{V}_{\beta\gamma}$ that has the structure of a braided tensor category. This category greatly extended the one studied in \cite{allen2020bosonic}, and is predicted in \cite{BN22} to be the category of line operators for the A twist of a free hypermultiplet.

We proceed now to review the results of \cite{BN22}, as they not just a basic example, but a fundamental building block for our more general constructions.

\subsection{Line operators and abelian mirror symmetry: first example}\label{subsec:baby}

In this section, we review the simplest example of 3d mirror symmetry for categories of line operators, corresponding to the A twist of a free hypermultiplet and the B twist of $\text{SQED}_1$. This mirror pair of 3d theories was introduced back in Section \ref{sec:eg1}.

The boundary VOA for a hypermultiplet $\CV_{\rm hyper}^A=\mc{V}_{\bg}\otimes \mc{V}_{bc}$,
which contains the affine VOA $\widehat{\mathfrak{gl}(1|1)}$ as a sub-VOA. In fact, according to \cite{creutzig2013w}, $\mc{V}_{\bg}\otimes \mc{V}_{bc}$ as a module of $\widehat{\mathfrak{gl}(1|1)}$ is a direct sum of $\Z$ simple currents. 
Using this, \cite{BN22} defined a category of modules of $\mc{V}_{\bg}\otimes \mc{V}_{bc}$, which was denoted $\cbg$. Using the technique of simple current extension, it was shown that  $\cbg$ is a braided tensor category and that there is an equivalence:
\be
\cbg\simeq KL^{[0]}/\Z
\ee
of braided tensor categories. This is an instance of mirror symmetry statement for the category of line operators. In the next two subsections, we will recall the details of the category $\cbg$ and the equivalence above. 

\begin{Rmk}
Recall that the free-fermion VOA $\CV_{bc}$ has a trivial module category. Any module category for $\CV_{\beta\gamma}\otimes \CV_{bc}$ is equivalent to one for $\CV_{\beta\gamma}$ alone, under the equivalence
\be M\in \CV_{\beta\gamma}\text{-Mod}\;\mapsto\; M\otimes \CV_{bc} \in \CV_{\beta\gamma}\otimes\CV_{bc}\text{-Mod}\,. \label{iso-bc} \ee
\end{Rmk}

\subsubsection{The category of modules of $\mc{V}_{\beta\gamma}$}

We recall facts about the VOA $\mc{V}_{\beta\gamma}$ and its category of modules from \cite{BN22}. Recall from Section \ref{sec:betagamma}, the VOA $\mc V_{\bg}$ is generated by fields $\beta(z)$ and $\gamma(z)$ with OPE:
\be
\beta(z)\gamma (w)\sim \frac{-1}{z-w}.
\ee
We make the following choice of conformal element:
\be
L(z)=\frac{1}{2}\left( \normord{\partial \beta\gamma}-\normord{\beta\partial \gamma}\right),
\ee
such that the conformal weights are $|\beta|=|\gamma|=\frac{1}{2}$, so that we can write mode expansion:
\be
\beta(z)=\sum_{n\in \Z} \beta_nz^{-n-\frac{1}{2}},\qquad \gamma(z)=\sum_{n\in \Z} \gamma_nz^{-n-\frac{1}{2}}. 
\ee
The OPE leads to the following commutation relation of the modes:
\be
\beta(z)\gamma (w)\sim \frac{-1}{z-w}\rightsquigarrow [\gamma_n,\beta_{-n}]=1.
\ee
For all $n$, the commutation relation of $\gamma_n$ and $\beta_{-n}$ is that of Weyl algebra, where we recognize $\gamma_n$ as $\partial_{\beta_{-n}}$. Therefore, one can (intuitively) recognize the universal enveloping algebra $U(\mc V_{\bg})$ as an infinite product of Weyl algebras:
\be
\begin{array}{c @{\qquad} @{\cdots \quad \otimes \quad}  c @{\quad \otimes \quad} c @{\quad \otimes \quad} c @{\quad \otimes \quad} c @{ \quad \cdots}}
U(\mc{V}_{\beta\gamma})\sim
	& \C[\beta_{-1}, \gamma_1]
	& \C[\beta_0, \gamma_0]
	& \C[\beta_1, \gamma_{-1}]
	&\C[\beta_2, \gamma_{-2}]
\end{array}
\ee
Strictly speaking, the universal enveloping algebra $U(\mc{V}_{\beta\gamma})$ is a topological algebra, in the sense that a module $\mc M$ of the VOA $\mc{V}_{\beta\gamma}$ is a module of the universal enveloping algebra where for each $m\in \mc M$, there exists $N>0$ such that $\beta_nm=\gamma_nm=0$ for $n>N$. These are called smooth modules of $U(\mc{V}_{\beta\gamma})$. This picture, which was called the column picture in \cite{BN22}, was helpful in visualizing the modules and there extensions. 

Let us now introduce simple modules of $\mc V_{\bg}$ as modules of the mode algebra, represented in column pictures. The simple vacuum module $\mc V_0$ is given by:
\be
\begin{array}{c @{\qquad} @{\cdots \quad \otimes \quad}  c @{\quad \otimes \quad} c @{\quad \otimes \quad} c @{\quad \otimes \quad} c @{ \quad \cdots}}
U(\mc{V}_{\beta\gamma})
	& \C[\beta_{-1}, \gamma_1]
	& \C[\beta_0, \gamma_0]
	& \C[\beta_1, \gamma_{-1}]
	&\C[\beta_2, \gamma_{-2}]\\
\mc V_0
	& \C[\beta_{-1}]
	& \C[\gamma_0]
	& \C[\gamma_{-1}]
	& \C[\gamma_{-2}]
\end{array}
\ee
Here, in writing $\C[\beta_n]$, we mean the module of the Weyl algebra $\C[\beta_n, \gamma_{-n}]$ where $\beta_n$ acts freely and $\gamma_{-n}$ acts as $\pd_{\beta_n}$. Similarly, in writing $\C[\gamma_n]$, we mean the module of the Weyl algebra $\C[\gamma_n, \beta_{-n}]$ where $\gamma_n$ acts freely and $\beta_{-n}$ acts as $-\pd_{\gamma_n}$. Note that positive modes act locally nilpotently on $\mc V_0$. This simple modules can be twisted by spectral-flow automorphism $\Sigma$, which was used in Section \ref{subsecgaugeextA} to define the simplified A-twist boundary VOA $\wt{V}_\rho^A$. Recall that it acts on the fields as:
\be
\Sigma \beta=z\beta,\qquad \Sigma \gamma=z^{-1}\gamma.
\ee
 This can be used to twist the simple module to obtain $\mc V_n:=\Sigma^n \mc V_0$. We can write this in column picture as:
\be
\begin{array}{c @{\qquad} @{\cdots \quad \otimes \quad}  c @{\quad \otimes \quad} c @{\quad \otimes \quad} c @{\quad \otimes \quad} c @{ \quad \cdots}}
U(\mc{V}_{\beta\gamma})
	& \C[\beta_{n-1}, \gamma_{-n+1}]
	& \C[\beta_{n}, \gamma_{-n}]
	& \C[\beta_{n+1}, \gamma_{-n-1}]
	&\C[\beta_{n+2}, \gamma_{-n-2}]\\
\mc V_n
	& \C[\beta_{n-1}]
	& \C[\gamma_{-n}]
	& \C[\gamma_{-n-1}]
	& \C[\gamma_{-n-2}]
\end{array}
\ee
For each $\upsilon\in \C/\Z$, there is a module $\mc W_{\upsilon}$, whose representation in column picture is as follows:
\be \label{bg-W}
\begin{array}{c @{\qquad} @{\cdots \quad \otimes \quad}  c @{\quad \otimes \quad} c @{\quad \otimes \quad} c @{\quad \otimes \quad} c @{ \quad \cdots}}
U(\mc V_{\bg})
	& \C[\beta_{-1}, \gamma_1]
	& \C[\beta_0, \gamma_0]
	& \C[\beta_1, \gamma_{-1}]
	&\C[\beta_2, \gamma_{-2}]\\
\mc W_\upsilon 
	& \C[\beta_{-1}]
	&\gamma_0^{\upsilon} \C[\gamma_0, \gamma_0^{-1}]
	& \C[\gamma_{-1}]
	& \C[\gamma_{-2}]
\end{array}
\ee
Here we abuse notation and use $\upsilon$ also as a representative in $\C$. One can alternatively define $\mc W_\upsilon$ to be generated by a vector $w_\upsilon$ such that positive modes of $\beta$ and $\gamma$ act trivially, negative modes act freely, and $\beta_0\gamma_0w_\upsilon=-(\upsilon+1)w_\upsilon$.

We emphasize that $\mc W_\upsilon$ depends only on the equivalence class of $\upsilon\in \C/\Z$. There is an obvious isomorphism among modules for which $\upsilon$ differs by an integer, since $\gamma_0^{\upsilon} \C[\gamma_0, \gamma_0^{-1}] = \oplus_{k\in \Z} \C\langle \gamma_0^{\upsilon+k}\rangle$.

When $\upsilon\ne 0\in \C/\Z$, $\mc W_\upsilon$ is irreducible. Otherwise, it fits into the following short exact sequence:
\begin{equation}\label{atyploewy+}
\begin{tikzcd}
 0\rar &  \mc{V}_0\rar & \mc W_0\rar &   \mc{V}_{1}\rar & 0
\end{tikzcd}
\end{equation}
We will denote $\mc W_0$ by $\mc W_{0}^+$, and introduce $\mc W_{0}^-$ to be the module whose column presentation is 
\be
\begin{array}{c @{\qquad} @{\cdots \quad \otimes \quad} c @{\quad \otimes \quad} c @{\quad \otimes \quad} c @{\quad \otimes \quad} c @{ \quad \cdots}}
U(\mc V_{\bg})
	& \C[\beta_{-1}, \gamma_1]
	& \C[\beta_0, \gamma_0]
	& \C[\beta_1, \gamma_{-1}]
	&\C[\beta_2, \gamma_{-2}]\\
\mc W_{0}^- 
	& \C[\beta_{-1}]
	&  \C[\beta_0, \beta_0^{-1}]
	& \C[\gamma_{-1}]
	& \C[\gamma_{-2}]
\end{array}
\ee
This fits in the following short exact sequence:
\begin{equation}\label{atyploewy-}
\begin{tikzcd}
 0\rar &  \mc{V}_{1}\rar & \mc W_{0}^-\rar &    \mc{V}_{0} \rar & 0
\end{tikzcd}
\end{equation}
We can also take the spectral flows $\Sigma^n \mc W_\upsilon$, which are still simple when $\upsilon\ne 0\in \C/\Z$. Altogether, we define:

\begin{Def}
The category $\cbg$ is the smallest abelian full subcategory of $\mc{V}_{\bg}$ modules that are finite length whose composition factors are $\Sigma^n\mc W_\upsilon$ and $\mc V_n$ for $\upsilon\ne 0 \text{ mod }\Z$  and $n\in \Z$.
\end{Def}
 This category admits a decomposition:
\be
\cbg=\bigoplus_{\upsilon\in \C/\Z} \mc C_{\bg, \upsilon}
\ee
where $\mc C_{\bg, \upsilon}$ is the subcategory generated by $\Sigma^k \mc W_\upsilon$, and can be alternatively characterized by requiring the generalized eigenvalue of $J_0^{\bg}$ to be contained in $\upsilon$, where recall $J^{\bg}(z)=-\normord{\beta(z)\gamma(z)}$.

The identification of the \emph{simple} modules in $\cbg$ with physical vortex lines from Section \ref{sec:A-lines} is straightforward. The fields $\beta,\gamma$ generating $\CV_{\beta\gamma}$ are the boundary values of bulk hypermultiplet scalars $Y(z),X(z)$. Comparing with the description of flavor vortices and flavor defects in Section \ref{sec:A-lines}, this lets us match
\be  \text{bulk}\;\; \V_m^{(\upsilon)} \;\leftrightarrow\;\begin{cases} \CV_m & \upsilon = 0\;\;\text{mod $\Z$} \\
\Sigma^m \mc W_\upsilon & \upsilon \neq 0\;\;\text{mod $\Z$}
\end{cases} \ee
In particular, the spectral flows of the vacuum module give a zero of some order to $X$ and allow a correspinding pole of the same order in $Y$ (or vice versa), matching flavor vortices $\V_m$ in \eqref{def-Vm}. The modules with $\upsilon\neq 0$ fix a nonzero residue for the boundary current, which is precisely the the boundary value of the bulk \emph{moment map} operator \eqref{mu-residue}. Moreover, the constraint on the residue is solved in a minimal way, by giving a simple pole to $Y$, as in \eqref{def-Lu}. 

The following theorem is the main result of \cite{BN22}.

\begin{Thm}\label{Thm:cbgBTC}
The category $\cbg$ is a braided tensor category whose tensor structure is defined by logarithmic intertwining operators. 
\end{Thm}

The proof of this statement uses exactly the idea of simple current extensions, and the relation between $\mc{V}_{\bg}$ and the affine Lie superalgebra $\widehat{\mathfrak{gl}(1|1)}$ \cite{creutzig2013w}. However, before going into this part of the stody, we recall some facts about indecomposable modules in $\cbg$ and their free-field realizations. These facts are crucial in the proof of the above theorem, and will be important in the computation of monodromy in the following sections.

The column picture allows one to compute extension groups between objects in $\cbg$, which then can be used to characterize indecomposable objects. 
This was done in detail in \cite{BN22}, which we recall the results here. When $\upsilon\notin\mathbb{Z}$, for each $k>0$, there is a unique $k$-th self extension of $\Sigma^n \mc W_{\upsilon}$, namely a indecomposable module whose all composition factors are $\Sigma^n \mc W_{\upsilon}$. We will denote this by $\Sigma^n \mc W_{\upsilon}^k$. The module $\mc W_{\upsilon}^k$ has the following presentation: 
\be
\begin{array}{c @{\qquad} @{\cdots \quad \otimes \quad} c @{\quad \otimes \quad} c @{\quad \otimes \quad} c @{\quad \otimes \quad \cdots}}
U(\mc V_{\bg})
	& \C[\beta_{-1}, \gamma_1]
	& \C[\beta_0, \gamma_0]
	& \C[\beta_1, \gamma_{-1}]
	\\
\mc W_{\upsilon}^k 
	& \C[\beta_{-1}]
	&\gamma_0^\upsilon \log(\gamma_0)^{k-1}\C[\gamma_0, \gamma_0^{-1}]
	& \C[\gamma_{-1}]
\end{array}
\ee
Similarly, there are indecomposable self extensions $\Sigma^n\mc W_0^{\pm, k}$ of $\mc W_0^\pm$. One can show that these span the category, in the following sense:

\begin{Thm}\label{ThmCbgindecom}
Any object in $\cbg$ is a subquotient of a finite direct sum of $\Sigma^n\mc W_\upsilon^{k}$ and $\Sigma^n\mc W_0^{-, k}$. 
\end{Thm}

This allows one to study the category $\cbg$ using free field realizations, since $\Sigma^n \mc W_0^\pm$ and $\Sigma^n \mc W_{\upsilon}$ and their extensions can be induced from Fock modules. In Section \ref{subsecBRSTAff}, we have introduced the free field realizaion of $\mc{V}_{\beta\gamma}$ using the Heisenberg VOA corresponding to $ \phi$ and $\eta$, with pairing $(\phi,\phi)=-( \eta,\eta)=1$. The embedding:
\begin{equation}
    \mc{V}_{\beta\gamma}\longrightarrow V_L:=\bigoplus\limits_n \Fock_{n\phi+n\eta}
\end{equation}
identifies Fock modules of the free field VOA with modules of $\mc{V}_{\beta\gamma}$. For each $\mu=a\phi+b\eta$ for $a-b\in \mathbb{Z}$, the Fock module $\Fock_{a\phi+b\eta}$ can be lifted to a Fock module of  $V_L$ given by:
\begin{equation}
   V_{L, \mu} =\bigoplus\limits_n \Fock_{(a+n)\phi+(b+n)\eta}.
\end{equation}
We have the following proposition, proved in \cite[Proposition 4.7]{wood2020admissible} and \cite[Proposition 4.1]{adamovic2019fusion}:

\begin{Prop}
Let $\mu=a\phi+b\eta$. When $a\notin \mathbb{Z}$, there is an isomorphism of $\mc{V}_{\bg}$ modules $V_{L, \mu}\cong \Sigma^{b-a-1} \mc W_{-a}$. When $a\in \mathbb{Z}$, there is an isomorphism of $\mc{V}_{\bg}$ modules $V_{L, \mu}\cong \Sigma^{b-a-1} \mc W_{0}^{-}$.

\end{Prop}

For each $k>0$, there is a self extension of $\Fock_{a\phi+b\eta}$, generated by $(\phi+\eta)^{k-1}\vert \mu \rangle$, and it can be lifted to a module of $V_L$, as the action of $\exp (\phi(z)+\eta(z))$ is non-logarithmic. We denote the lift by $V_{L,\mu}^k$. The above identification can be generalized into the following:

\begin{Prop}
Let $\mu=a\phi+b\eta$. When $a\notin \mathbb{Z}$, there is an isomorphism of $\mc{V}_{\bg}$ modules $V_{L, \mu}^k\cong \Sigma^{b-a-1} \mc W_{-a}^k$. When $a\in \mathbb{Z}$, there is an isomorphism of $\mc{V}_{\bg}$ modules $V_{L, \mu}^k\cong \Sigma^{b-a-1} \mc W_{0}^{-, k}$.
\end{Prop}

\subsubsection{Kazhdan-Lusztig category of affine $\mathfrak{gl}(1|1)$}

Let $\widehat{\mathfrak{gl}(1|1)}$ be the affine Lie superalgebra associated to the Lie superalgebra $\mathfrak{gl}(1|1)$. This is the perturbative boundary VOA of the B twist of $\mathrm{SQED}_1$, the theory with $\rho=(1)$.
A generalized $\widehat{\mathfrak{gl}(1|1)}$ module $W$ is called \emph{finite-length} if it has a finite composition series of irreducible $\widehat{\mathfrak{gl}(1|1)}$ modules. $W$ is called \emph{grading restricted} if it is graded by generalized conformal weights (the generalized eigenvalues of $L_0$), the weight spaces are all finite-dimensional, and the generalized conformal weights are bounded from below. For more details, see \cite{creutzig2017tensor}. 

\begin{Def}
The Kazhdan-Lusztig category $KL$ is defined as the category of finite-length grading-restricted generalized $\widehat{\mathfrak{gl}(1|1)}$ modules. 
\end{Def}

The category $KL$ is studied in detail in \cite{creutzig2020tensor}, in which it is shown that it is a rigid braided tensor category. Moreover, it admits a decomposition:
\be
KL=\bigoplus\limits_{e\in \C}KL_e
\ee
where $KL_e$ is the full subcategory of modules where the action of $E_0$ has generalized eigenvalue $e$. The fusion product respects this decomposition in the sense that $KL_e\times KL_{e'}\mapsto KL_{e+e'}$. 

Let $M$ be a finite-dimensional module of $\mathfrak{gl}(1|1)$, then one can obtain a module of $\widehat{\mathfrak{gl}(1|1)}$ in the following way. Let $\widehat{\mathfrak{gl}(1|1)}_{\geq 0}$ be the Lie subalgebra of $\widehat{\mathfrak{gl}(1|1)}$ generated by non-negative modes. One can view $M$ as a module of $\widehat{\mathfrak{gl}(1|1)}_{\geq 0}$ via the map $\widehat{\mathfrak{gl}(1|1)}_{\geq 0}\to \widehat{\mathfrak{gl}(1|1)}$. Define the module $\widehat{M}$ as the induction:
\be
\widehat{M}=U(\widehat{\mathfrak{gl}(1|1)})\otimes_{U(\widehat{\mathfrak{gl}(1|1)}_{\geq 0})}M.\ee
Then $\widehat{M}$ is an object in $KL$ and by the definition of $KL$, any simple object is a quotient of $\widehat{M}$ for some $M$. 

The Lie algebra $\mathfrak{gl}(1|1)$ has the following set of simple modules:
\begin{enumerate}

\item $A_{n,0}$, where $N$ acts with weight $n$ and all other modes act as zero. This module is one-dimensional.

\item $V_{n, e}$ where $N$ acts with weight $n\pm \frac 12$, and $E$ acts with weight $e\ne 0$. This module is two-dimensional. 

\end{enumerate}

From the above, any simple module in $KL$ is a quotient of $\widehat{V}_{n,e}$ or $\widehat{A}_{n,0}$. The following is shown in \cite{Creutzig:2011cu}:

\begin{itemize}

\item $\widehat{V}_{n,e}$ is irreducible iff $e\notin\mathbb{Z}$.

\item When $e\in \Z$ but $e\ne 0$, $\wh{V}_{n, e}$ is reducible, and fits into the following short exact sequence:
\begin{equation}
\begin{tikzcd}
0\rar &  \widehat{A}_{n+1,e}\rar & \widehat{V}_{n-e/2,e}
\rar & \widehat{A}_{n,e}\rar & 0 & (e>0)\\
 0\rar & \widehat{A}_{n-1,e}\rar & \widehat{V}_{n-e/2,e}
\rar & \widehat{A}_{n,e}\rar & 0 & (e<0)
\end{tikzcd}
\end{equation}

\end{itemize}

The modules $\wh{A}_{n,e}$ are simple currents of $\widehat{\mathfrak{gl}(1|1)}$ as they can be defined as the image of the following spectral-flow automorphisms $\sigma_{n,l}$:
\be
	\sigma_{n,l}(N) = N - \frac{n}{z} \qquad \sigma_{n,l}(E) = E - \frac{l}{z} \qquad \sigma_{n,l}(\psi_\pm) = z^{\mp l} \psi_\pm.
\ee
Introduce a function $\epsilon(l)$ on $\mathbb{Z}$ given by:
\begin{equation}
\epsilon(l) = \left\{
\begin{array}{rl}
-\frac{1}{2} & \text{if } l < 0,\\
0 & \text{if } l = 0,\\
\frac{1}{2} & \text{if } l > 0.
\end{array} \right.
\end{equation}
Define $\epsilon(l,l')=\epsilon(l)+\epsilon(l')-\epsilon(l+l')$. The simple currents $\wh{A}_{n, e}$ have the following simple fusion rules:
\be
\wh{A}_{n, l}\otimes_{\widehat{\mathfrak{gl}(1|1)}} \wh{A}_{n', l'}\cong \wh{A}_{n+n'-\epsilon(l, l'), l+l'}. 
\ee

As we have done in Section \ref{sec:mon1}, the full boundary VOA is the extension of $\vgl$ by monopole operators, which can be identified with the modules $\widehat{A}_{-m/2+\epsilon(m),m}$. The following, which was derived in \cite{creutzig2013w},  is an example of Theorem \ref{thm:VOA}:
\be
\mc{V}_{\bg}\otimes \mc{V}_{bc}=\bigoplus_{m} \widehat{A}_{-m/2+\epsilon(m),m}.
\ee
Namely, $\mc{V}_{\bg}\otimes \mc{V}_{bc}$ is a simple current extension of $\vgl$. Dually, there is a $\C^*$ action on $\mc{V}_{\bg}\otimes \mc{V}_{bc}$, induced by the action of $E_0$, such that $\vgl$ is the $\C^*$ invariant subspace. The technique of simple current extensions implies that there is a braided tensor category:
\be
\mc{V}_{\bg}\otimes \mc{V}_{bc}\mathrm{-Mod}_{\mathrm{loc}}(\text{Ind}(KL))
\ee
together with a full and surjective tensor functor: 
\be
\mc L: \text{Ind}(KL)^{[0]}\longrightarrow \mc{V}_{\bg}\otimes \mc{V}_{bc}\mathrm{-Mod}_{\mathrm{loc}}(\text{Ind}(KL)).
\ee
Here $\text{Ind}(KL)^{[0]}$ is the subcategory of $\text{Ind}(KL)$ that has trivial monodromy with $\mc{V}_{\bg}\otimes \mc{V}_{bc}$. Another main result of \cite{BN22} is the following:

\begin{Prop}
The image of $KL^{[0]}$ under $\mc L$ is identified with $\cbg$, and so $\cbg$ has the structure of a braided tensor category.
\end{Prop}

The category $KL^{[0]}$ can be identified with the category $KL^N$, the subcategory of $KL$ consisting of objects on which the action of $N_0$ is semi-simple with integer eigenvalues, and the lifting functor identifies $\cbg$ with the de-equivariantization $KL^{N}_{\C^*}$. Since the de-equivariantization identifies an object $\mc M$ with $\widehat{A}_{-m/2+\epsilon(m),m}\otimes_{\widehat{\mathfrak{gl}(1|1)}} \mc M$, we can also denote this suggestively as $KL^N/\Z$. The lifting functors respects the decomposition of both $KL$ and $\cbg$, since it maps objects in $KL_e$ to objects in $\mc C_{\bg, \upsilon}$ for any representative $e\in \upsilon$. We can match modules of the VOA across this equivalence as follows. 

\begin{enumerate}

\item The image of $\wh{A}_{n,0}$ for $n\in \Z$ under $\mc L$ is identified with the spectral flow $\mc{V}_{n}$. From the bulk perspective of Section \ref{sec:Blines}, these are the Wilson lines $\W_n$.

\item The image of $\wh{V}_{n,e}$ for $n+\frac{1}{2}\in \Z$ is identified with $\Sigma^{n-1/2}\mc W_{e}$. From the bulk perspective, these are the monodromy defects $\W_{n-1/2}^{(e)}$. 

\end{enumerate}

In the following sections, we will generalize this baby example, and define braided tensor categories. In particular:

\begin{itemize}
    \item The category $\mc C_\rho^A$ will be defined as a de-equivariantization of $\cbg^{\boxtimes n}$, using the fact that $\wt{\mc V}_\rho^A$ is a simple current extension of $\mc V_{\bg}^{V}$. 

    \item The category $\mc C_\rho^B$ will be defined as a de-equivariantization of $KL_\rho$, the Kazhdan-Lusztig category of $\widehat{\grho}$, since $\mc V_\rho^B$ is a simple current extension of $\widehat{\grho}$. 
    
\end{itemize}

Objects in $\mc C_\rho^A$ are lifted from $\cbg^{\boxtimes n}$, which allows us to compare these objects with A twist bulk line operators in Section \ref{sec:A-lines}. Objects in $\mc C_\rho^B$ are lifted from $KL_\rho$, which allows us to compare these objects with B twist bulk line operators in Section \ref{sec:Blines}. The proof of the main result, Theorem \ref{ThmCACB}, is also based on this example, and Theorem \ref{ThmMirSymbdVOA}.

\subsection{The category of line operators for the A twist}\label{subsec:LineA}

\subsubsection{The definition of the category}

In Section \ref{subsecNeumannAVOA} we defined the VOA $\mc V_\rho^A$ and studied its free-field realizations. We'd now like to construct an abelian category of modules $\CC_\rho^A$ for $\CV_\rho^A$ that has a braided tensor structure defined via intertwining operators, and whose simple objects correspond to the flavor vortices and flavor defects of Section \ref{sec:A-lines}.

Recall that in Section \ref{subsecgaugeextA} we constructed an ``alternative'' vertex algebra $\wt\CV_\rho^A$ as the $\Z^{r}$ simple current extension of $n$ copies of $\CV_{\beta\gamma}$.
We proved in Theorem \ref{ThmMirSymbdVOA} that $\wt\CV_\rho^A$ is Morita-equivalent to $\CV_\rho^A$, as we have an isomorphism:
\be
\wt{\CV}_\rho^A\otimes \mc V_{bc}^V\cong \mc V_\rho^A\otimes \mc V_{\Z^{2r}}^{+-},
\ee
where $\mc V_{bc}^V$ and $\mc V_{\Z^{2r}}^{+-}$ are lattice VOA associated to self-dual lattices, and therefore are Morita-trivial. 
We will analyze the simple current extension (and thus modules of $\wt\CV_\rho^A$) at the level of categories, rather than constructing modules of $\CV_\rho^A$ directly.

From the perspective of Section \ref{sec:highersym}, our strategy is to start with the category of line operators for $n$ free hypermultiplets in the A twist --- now with braided tensor structure --- and to obtain the category for $\CT_\rho^A$ by gauging a $\Z^r$ one-form symmetry.

The $n$-th Deligne product $\cbg^{\boxtimes n}$ is a category of modules of the VOA $\mc{V}_{\bg}^{\otimes n}$, and can be alternatively defined as the smallest full abelian subcategory of $\mc{V}_{\bg}^{\otimes n}$ modules that are finite length, whose composition series are tensor products of $\Sigma^n \mc W_{\upsilon}$ and $\CV_n$. As we have reviewed in the previous section, this category has the structure of a braided tensor category via logarithmic intertwining operators. It also has a decomposition:
\be \label{Cbg-ndecomp}
\cbg^{\boxtimes n}=\bigoplus\limits_{\upsilon\in (\C/\Z)^n} \overset{n}{\underset{i=1}{\boxtimes}} \mc C_{\beta\gamma, \upsilon_i}.
\ee
Each direct summand in the definition of $\wt{\mc{V}}_{\rho}^A$ is a simple current of $\mc{V}_{\bg}^{\otimes n}$, and thus an object in $\cbg^{\boxtimes n}$. In fact, we can recognize the direct summands as follows:
\be
\V_{\rho(\mu)}\longleftrightarrow \bigotimes\limits_{1\leq i\leq n} \CV_{\rho(\mu)_i}, \text{ for all } \mu\in \Z^r,
\ee
such that the OPE of $\wt{\mc V}_\rho^A$ induces the fusion rule:
\be
\V_{\rho(\mu)}\otimes_{\mc V_{\bg}^{\otimes n}} \V_{\rho (\mu')}\cong \V_{\rho (\mu+\mu')}.
\ee
Therefore, the VOA $\wt{\mc{V}}_{\rho}^A$ can be viewed as an algebra object in $\mathrm{Ind}(\cbg^{\boxtimes n})$. Following Section~\ref{subsecsimplecurrent}, there is a lifting functor of braided tensor categories:
\be
\mc L_A: \cbg^{\boxtimes n, \rho,[0]}\longrightarrow \wt{\mc{V}}_{\rho}^A\mathrm{-Mod}_{\text{loc}} \left(\mathrm{Ind}(\cbg^{\boxtimes n}) \right).
\ee
Here $\mc C_{\beta\gamma}^{\boxtimes n,\rho, [0]}$ is the subcategory of $\mathcal{C}_{\beta\gamma}^{\boxtimes n}$ whose objects have trivial monodromy with $\wt{\mc{V}}_{\rho}^A$.

\begin{Def}\label{def:lineopeA}
Let us denote by $\mc C_\rho^A$ the image of $\mathcal{C}_{\beta\gamma}^{\boxtimes n, \rho, [0]}$ under $\mathcal L_A$:
\begin{equation}
    \mc C_\rho^A:= \mathcal{L}_A \left(\mathcal{C}_{\beta\gamma}^{\boxtimes n, \rho, [0]}\right).
\end{equation}
This is a braided tensor category via logarithmic intertwiners.
\end{Def}

Note that the symplectic boson VOA $\mc{V}_{\bg}^{\otimes n}$ is an orbifold of $\wt{\mc{V}}_{\rho}^A$ by $(\C^*)^r$ if we assign the weight of $\mathbb{V}_{\rho (\mu)}$ to be $\mu$. Indeed, in this case, the only invariant part of $\wt{\mc{V}}_{\rho}^A$ is $\mathbb{V}_0$ which is simply $\mc{V}_{\bg}^{\otimes n}$. We conclude that $\mc{C}_\rho^A$ is a de-equivariantization $\mathcal{C}_{\beta\gamma, (\C^*)^r}^{\boxtimes n, \rho, [0]}$, or $\mathcal{C}_{\beta\gamma}^{\boxtimes n, \rho, [0]}/\Z^r$.

The advantage of this definition is that we immediately obtain a braided tensor category, but a disadvantage is that it involves a seemingly complicated requirement of trivial monodromy. In the next section we show that this monodromy condition is something very explicit, and amounts to the integrability of an abelian group action. 

\subsubsection{Locality requirements}

In order to classify objects in the category $\mc C_\rho^A$, we must understand the category $\mathcal{C}_{\beta\gamma}^{\boxtimes n, \rho, [0]}$, that is, the sub-category whose objects have trivial monodromy with $\wt{V}_\rho^A$. Recall the $\mathfrak{gl}(1)$ currents in $\CV_{\bg}^{\otimes n}$ given by:
\be
J^i(z)=-\normord{\beta^i\gamma^i}.
\ee
We will show, in the following proposition, that this monodromy condition is an integrability condition. More specifically, we show that a module $\mc M$ has trivial monodromy with $\wt{V}_\rho^A$ if and only if the action of $  \sum \rho_{ia}J^i_{0}$ induces a $(\mathbb{C}^*)^r$ action on $\mc M$. 

\begin{Prop}\label{Propcbgsscbgmono}
   An object $\mc M$ belongs to $\mathcal{C}_{\beta\gamma}^{\boxtimes n, \rho, [0]}$ if and only if $\sum \rho_{ia}J^i_0$ acts semi-simply with integer eigenvalues.
\end{Prop}

The proof uses free field realizations. Recall the embedding $\mc{V}_{\bg}^{\otimes n}\to V_L^{\otimes n}$, and the embeddings of simple modules $\mathbb{V}_{\mu}\to V_{L,\mu}^{\otimes n}$ for $\mu\in \Z^n$. Let $\mc M$ be any module of the VOA  $V_{L}^{\otimes n}$, then the fusion product over free field algebra gives an isomorphism:
\begin{equation}
\mathbb{V}_{\mu}\otimes_{\mc{V}_{\beta\gamma}^{\otimes n}}\mc M\cong V_{L,\mu}^{\otimes n}\otimes_{V_{L}^{\otimes n}}\mc M.
\end{equation}
This is true because fusion product is exact for both symplectic bosons and Heisenberg VOA. Let $\mc M$ be a module of $V_{L}^{\otimes n}$ generated by:
\begin{equation}
    (\phi^{i_1}+\eta^{i_1})\cdots (\phi^{i_k}+\eta^{i_k}) \vert v \rangle,
\end{equation}
where $v=\sum_i t_i\phi^i+s_i\eta^i$ and $t_i-s_i\in \Z$. The action of $\phi_0^i$ is defined by the commutation relation $[\phi_0^i,\phi^i]=1$, and similarly, we define the action for $\eta_0^i$. These are modules of the Heisenberg VOA $\mc H_{\phi,\eta}$ that can be lifted to $V_L^{\otimes n}$, and from Section \ref{subsec:baby}, as modules of $\mc{V}_{\bg}^{\otimes n}$, $\mc M$ can be identified with tensor products of $\Sigma^n W_{\upsilon}^k$ and $\Sigma^n W_0^{-, k}$.  

\begin{Lem}\label{LemmonoAq}
 For each $\mc M$ as above the monodromy:
    \begin{equation}
\mathbb{V}_{\mu}\otimes_{\mc{V}_{\beta\gamma}^{\otimes n}}\mc M\longrightarrow \mathbb{V}_{\mu}\otimes_{\mc{V}_{\beta\gamma}^{\otimes n}}\mc M
    \end{equation}
    is given by
    \begin{equation}
    \mathrm{Id}\otimes e^{-2\pi i \sum_{i} \mu_iJ_0^i},     
    \end{equation}
    where $J_0^i$ is the zero mode of $J^i=-\normord{\beta^i\gamma^i}$. 
\end{Lem}

\begin{proof}
As commented above, we may use the intertwining operator of the free field algebra $V_{L}^{\otimes n}$ of the modules:
\begin{equation}
  \mc Y: V_{L,\mu}^{\otimes n}\otimes \mc M\to V_{L,\mu}^{\otimes n}\otimes_{V_L}\mc M\{z\}[\log (z)].
\end{equation}
By assumption, $\mc M$ is generated by a vector of the form:
\begin{equation}
   w= (\phi^{i_1}+\eta^{i_1})\cdots (\phi^{i_k}+\eta^{i_k})\vert \nu\rangle,
\end{equation}
with $\nu=\sum t_i\phi^i+s_i\eta^i$ where $t_i-s_i\in \mathbb{Z}$. The fusion rule is given by the logarithmic intertwiner $\mc{Y}_{\mu,\nu}$ defined by the following formula on the generators:
\begin{equation}
    \mc{Y}_{\mu,\nu}(\vert\mu\rangle, z) w=\normord{e^{\sum_{i} \mu_i\phi^i}} w.
\end{equation}
In this formula, the logarithmic part comes from:
\begin{equation}
    e^{\sum_{i} \mu_i \phi^i_0\text{log}(z)} (\phi^{i_1}+\eta^{i_1})\cdots (\phi^{i_k}+\eta^{i_k})\vert \nu\rangle.
\end{equation}
Since $[\phi_0^i,\phi^j+\eta^j]=\delta^{ij}$, the above is given by:
\begin{equation}
(\phi^{i_1}+\eta^{i_1}+ \mu^{i_1}\text{log}(z))\cdots (\phi^{i_k}+\eta^{i_k}+\mu^{i_k}\text{log}(z)) e^{\sum_{i} \mu^it_i\text{log}(z)}\vert\nu\rangle.
\end{equation}
To compute the monodromy, we rotate the $z$ coordinate by $z\mapsto e^{2\pi i}z$, which results in $\text{log}(z)\mapsto \text{log}(z)+2\pi i$. The contribution of the above comes from two parts, where the first part is:
\begin{equation}
(\phi^{i_1}+\eta^{i_1}+\mu^{i_1}\text{log}(z)+2\pi i\mu^{i_1})\cdots (\phi^{i_k}+\eta^{i_k}+\mu^{i_k}\text{log}(z)+2\pi i\mu^{i_k}),
\end{equation}
and the second part is:
\begin{equation}
    e^{\sum_{i} \mu^it_i\text{log}(z)+2\pi i\sum_{i} \mu^it_i}\vert\nu\rangle.
\end{equation}
The contribution of the first part can be written compactly as:
\begin{equation}
    e^{2\pi i\sum_{i} \mu_i\frac{\partial}{\partial \eta^i}},
\end{equation}
while the second part as:
\begin{equation}
  e^{2\pi i \sum_{i} \mu^it_i }=e^{2\pi i \sum_{i} \mu^is_i }=e^{-2\pi i( \sum_i\mu_i\eta^i, \nu)}.
\end{equation}
One can verify that the morphism corresponding to:
\begin{equation}
    e^{\sum \mu_i\frac{\partial}{\partial \eta^i}}e^{-2\pi i( \sum_i\mu_i\eta^i, \nu)}
\end{equation}
is nothing but:
\begin{equation}
    e^{-2\pi i \sum_i \mu_i\eta_0^i}=e^{-2\pi i \sum_{ i} \mu_iJ_0^i}.
\end{equation}
We have, in conclusion
\begin{equation}
    \mc{Y}_{\mu,\nu}(\vert\mu\rangle, e^{2\pi i}z) w=\mc{Y}_{\mu,\nu}(\vert\mu\rangle, z)e^{-2\pi i \sum \mu_iJ_0^i}w.
\end{equation}
This means that the monodromy is given by:
\begin{equation}
    \text{Id}\otimes e^{-2\pi i \sum \mu_iJ_0^i}. 
\end{equation}

\end{proof}

We can now finish the proof of Proposition \ref{Propcbgsscbgmono} using Lemma \ref{LemmonoAq}. Indeed, by Theorem \ref{ThmCbgindecom}, any object in $\cbg^{\boxtimes n}$ is a subquotient of a direct sum of such module $\mc M$ restricted from the free field algebra $V_L^{\otimes n}$. By Lemma \ref{LemmonoAq}, for each $\mu\in \Z^r$, the action of monodromy:
\be
\mu(\V_{\rho(\mu)}, \mc M): \mathbb{V}_{\rho (\mu)}\otimes_{\mc V_{\bg}^{\otimes n}} \mc M\longrightarrow \mathbb{V}_{\rho (\mu)}\otimes_{\mc V_{\bg}^{\otimes n}}\mc M
\ee
is given by $\text{Id}\otimes e^{-2\pi i \sum \mu^a\rho_{ia}J_0^i}$. Since monodromy is functorial with respect to sub-object and quotient object, it follows that the action of monodromy is given by the same expression for all objects in $\cbg^{\boxtimes n}$. For this to be identity, one requires that $\sum \mu^a\rho_{ia}J_0^i$ acts as integers for all $\mu\in \Z^r$, which is the statement of Proposition \ref{Propcbgsscbgmono}. As a consequence, we have:

\begin{Cor}
Objects in the subcategory $\boxtimes_i \mc C_{\beta\gamma, \upsilon_i}$ gives rise to local modules of $\wt{\mc{V}}_{\rho}^A$ only if $\sum \rho_{ia}\upsilon^i=0$ mod $\Z$ for all $a$. The statement is if and only if for simple objects, since on these objects $J_0^i$ acts semi-simply. 

\end{Cor}

This makes the study of the category $\mc C_\rho^A$ much more explicit. Simple objects in this category comes from lifts of simple objects in $\boxtimes_i \mc C_{\beta\gamma, \upsilon_i}$ for $\sum \rho_{ia}\upsilon^i=0 \text{ mod }\Z$, and a simple object $\mc M$ in $\boxtimes_i \mc C_{\beta\gamma, \upsilon_i}$ is identified with $\mathbb{V}_{\rho(\mu)}\otimes_{\mc{V}_{\bg}^{\otimes n}}\mc M\cong \Sigma^{\rho(\mu)}\mc M$ for all $\mu\in\Z^r$.  This implies also that $\mc C_\rho^A$ admits a decomposition:
\be
\mc C_\rho^A=\bigoplus\limits_{\upsilon\in (\C/\Z)^n, \rho(\mu)\cdot \upsilon=0}\mc{C}_{\rho, \upsilon}^A.
\ee
To write this more efficiently, we use the short exact sequence  in equation \eqref{SES}. The short exact sequence (the conjugate version) induces a short exact sequence of groups:
\be
\begin{tikzcd}
(\C/\Z)^{n-r}\rar{\rho^\vee} & (\C/\Z)^n\rar{\rho^\T} & (\C/\Z)^r
\end{tikzcd}
\ee
and $\rho(\mu)\cdot \upsilon=0 \text{ mod }\Z$ for all $\mu$ precisely when $\upsilon\in \mathrm{Ker}(\rho^\T)=\mathrm{Im}(\rho^\vee)$. We can alternatively write:
\be
\mc C_\rho^A=\bigoplus\limits_{\upsilon\in (\C/\Z)^{n-r}}\mc{C}_{\rho, \upsilon}^A
\ee
such that $\mc{C}_{\rho, \upsilon}^A$ is the image of $\mc C_{\beta\gamma, \rho^\vee(\upsilon)}^{\boxtimes n, \rho, [0]}$ under $\mc L_A$. The above decomposition of the abelian category induces a decomposition:
\be
\widehat{\mc{C}}^A_{\rho}=\bigoplus\limits_{\upsilon\in (\C/\Z)^{n-r}}\widehat{\mc{C}}_{\rho, \upsilon}^A
\ee
for the derived category. 

\subsubsection{Line operators as simple objects}

In this section, we compare the simple objects in the category $\mc C_\rho^A$ with the line operators expected from bulk physics in Section \ref{sec:bulklines}. This justifies our definition of $\CC_\rho^A$ as a category of VOA modules whose derived category captures a minimal set of bulk lines.
We will also show that the lifting functor coincides with the BRST functor on the simple objects in $\mc{C}_{\beta\gamma}^{\boxtimes n, \rho, [0]}$. 

For an object to be in $\mc{C}_{\beta\gamma}^{\boxtimes n, \rho, [0]}$, the action of $\sum \rho_{ia}J_a^i$ must be semi-simple with integer eigenvalues. In particular, the simple objects in $\boxtimes_i \mc{C}_{\bg, [0]}$ all satisfy this.  Simple objects in this sector of the category are given by $\mathbb{V}_\lambda=\bigotimes\limits_i \CV_{\lambda_i}$ for $\lambda\in \Z^n$, the spectral flow of vacuum module. Upon lifting, however, $\mc L_A\mathbb{V}_{\lambda}\cong \mc L_A\mathbb{V}_{\mu}$ if and only if $\lambda-\mu\in \mathrm{Im}(\rho)$, which is the equivalence relation in Section \ref{secMorita}. By the splitting of the short exact sequence \eqref{split-SES}, we can choose a representative of this equivalence class to be $\mathbb{V}_{\sigma (m)}$ for $m\in \Z^{n-r}$. These corresponds to the vortex lines of the abelian gauge theory, as they change the pole of $\beta^i$ and $\gamma^i$ according to $\sigma(m)$. 

The general vortex line operators sets $J^i(z)\sim \frac{\upsilon^i}{z}$, and they are compatible with gauging precisely when $\sum\rho_{ia}\upsilon^i\in \Z$. We identify these general vortex line operators with the simple objects in $\boxtimes_i \mc{C}_{\bg, \upsilon_i}$. They are simply of the form $\otimes_i \mc M_i$, such that $\mc M_i=\mc V_{n_i}$ when $\upsilon_i=0 \text{ mod }\Z$ and $\mc M_i=\Sigma^{n_i}\mc W_{\upsilon_i}$ when $\upsilon_i\ne 0 \text{ mod }\Z$. 

We summarize this in the following proposition:

\begin{Prop}
Simple objects in $\mc C_\rho^A$ are labelled by $\Z^{n-r}\times (\C/\Z)^{n-r}$. For $(m,\upsilon)\in \Z^{n-r}\times (\C/\Z)^{n-r}$, we denote by $\V_m^\upsilon$ the corresponding simple object, in accordance with Section \ref{sec:bulklines}. The module $\V_m^\upsilon$ is the image under the lifting functor $\mc L_A$ of the module of $\mc{V}_{\bg}^{\otimes n}$ of the form $\otimes_i \mc M_i$, where:
\be
\mc M_i=\begin{cases}
\mc V_{\sigma(m)_i} & \text{ if } \rho^\vee(\upsilon)_i=0 \text{ mod }\Z \\
\Sigma^{\sigma(m)_i}\mc W_{[\rho^\vee(\upsilon)_i]} & \text{ if } \rho^\vee(\upsilon)_i\ne 0 \text{ mod }\Z
\end{cases}
\ee

\end{Prop}

\begin{Rmk}
We comment that the category $\mc C_\rho^A$ can be also defined as the full subcategory of $\wt{\mc{V}}_{\rho}^A$ modules that are finite length and whose composition factors are given by $\V_m^\upsilon$ for $(m,\upsilon)\in \Z^{n-r}\times (\C/\Z)^{n-r}$. 
\end{Rmk}

\begin{Exp}
    As an example, consider the case when $\rho=1$, namely A twist of $\mathrm{SQED}_1$. In this case $\wt{\mc V}_\rho^A$ is:
    \be
\wt{\mc V}_\rho^A\cong \bigoplus_{n\in \Z}\sigma^n \mc V_{\bg}\cong \mc V_{SF}\otimes \mc V_{\Z}^{-}.
    \ee
    Here $\mc V_{\Z}^{-}$ is a lattice VOA of a self-dual lattice of negative signature. 
    The only simple modules of $\mc V_{\bg}$ that satisfy the monodromy condition are the spectral flow of vacuum $\CV_n$, and are identified upon the lift. Consequently, there is only one simple module in $\mc C_\rho^A$, which is compatible with the fact that $\mc V_{SF}$ has a unique simple module (the vacuum module). 
    
\end{Exp}

In Section \ref{secMorita}, we showed that there is an isomorphism of VOA:
\be
\mc{V}_\rho^A\otimes \mc{V}_{\Z^{2r}}^{+-}\cong \wt{\mc V}_\rho^A\otimes \mc{V}_{bc}^V
\ee
Given a module $\mc M$ of $\mc{V}_{\bg}^{\otimes n}$, we can obtain a module of $\mc{V}_\rho^A$ by $H_{BRST}^*\left(\mc M\otimes \mc{V}_{bc}^{\otimes n}\right)$, and thus a module of $\mc{V}_\rho^A\otimes \mc{V}_{\Z^{2r}}^{+-}$ given by:
\be
H_{BRST}^*\left(\mc M\otimes \mc{V}_{bc}^{\otimes n}\right)\otimes \mc{V}_{\Z^{2r}}^{+-}.
\ee
On the other hand, when $M$ has trivial monodromy with $\wt{\mc V}_\rho^A$, we obtain a module $\mc{L}_A(\mc M)$ of $\wt{\mc V}_\rho^A$, and thus a module of $\wt{\mc V}_\rho^A\otimes \mc{V}_{bc}^{V}$ given by:
\be
\mc{L}_A(\mc M)\otimes \mc{V}_{bc}^{V}.
\ee
We further justify the use of $\wt{\mc V}_\rho^A$ for the definition of $\mc{C}_{\rho}^A$ by the following:

\begin{Thm}
Let $\mc M$ be a simple module in $\boxtimes_i \mc C_{\bg, \upsilon_i}$. If $\sum \rho_{ia}\upsilon^i\ne 0 \text{ mod }\Z$ for some $a$, then $H_{BRST}^*\left(\mc M\otimes \mc{V}_{bc}^{\otimes n}\right)$ is zero. Otherwise we have an isomorphism of VOA modules:
\be
H_{BRST}^*\left(\mc M\otimes \mc{V}_{bc}^{\otimes n}\right)\otimes \mc{V}_{\Z^{2r}}^{+-}\cong \mc{L}_A(\mc M)\otimes \mc{V}_{bc}^{V}.
\ee

\end{Thm}

\begin{proof}
Consider $V_{L,s\phi+t\eta}^{\otimes n}$, the module of the lattice VOA $V_L^{\otimes n}$ generated by $\sum_i s_i\phi^i+t_i\eta^i$ with $s_i-t_i\in \Z$. For each $s_i$ with $s_i\in \Z$, there is a screening operator $S^i=\oint\mathrm{d}z\exp (\phi^i(z))$ acting on this module, and the kernel:
\be
\mc M=\bigcap_{i, s_i\in \Z}\mathrm{Ker}(S^i)\Big\vert V_{L,s\phi+t\eta}^{\otimes n}
\ee
is a simple module of $\mc{V}_{\bg}^{\otimes n}$ and all the simple modules are obtained this way. The free field realization of $\mc M\otimes \mc{V}_{bc}^{\otimes n}$ is of the form:
\be
\mc M\otimes \mc{V}_{bc}^{\otimes n}\cong \bigcap_{i, s_i\in \Z}\mathrm{Ker}(S^i)\Big\vert \bigoplus_{\lambda,\mu\in \Z^n}\CF_{(\mu+s)\phi}\otimes \Fock_{(\mu+t)\eta}\otimes \Fock_{\lambda\theta}.
\ee
Since the BRST cohomology only acts on the factor $\Fock_{(\mu+t)\eta}\otimes \Fock_{\lambda\theta}$, and that the screening operator only acts on $\CF_{(\mu+s)\phi}$, they commute with each other and we have:
\be
H^*_{BRST}\left(\mc M\otimes \mc{V}_{bc}^{\otimes n}\right)\cong \bigcap_{i, s_i\in \Z}\mathrm{Ker}(S^i)\Big\vert \bigoplus_{\lambda,\mu\in \Z^n}\CF_{(\mu+s)\phi}\otimes H^*_{BRST}\left(\Fock_{(\mu+t)\eta}\otimes \Fock_{\lambda\theta}\right).
\ee
Splitting the vectors $\mu+s=\rho(\mu_0+s_0)+\rho^\vee (\mu^\vee+s^\vee)$, $\mu+t=\rho(\mu_0+t_0)+\rho^\vee (\mu^\vee+t^\vee)$ and $\lambda=\rho(\lambda_0)+\rho^\vee (\lambda^\vee)$ where $\lambda_0,\mu_0, s_0, t_0$ are spanned by $N$ and $\lambda^\vee, \mu^\vee, s^\vee, t^\vee$ spanned by $N^\vee$ such that $[\lambda_0]=[\lambda^\vee]\in H$ and $[\mu_0]=[\mu^\vee]\in H$, the BRST cohomology can be rewritten as:
\be
\bigcap_{i, s_i\in \Z}\mathrm{Ker}(S^i)\Big\vert \hspace{-20pt}\bigoplus_{\substack{\lambda_0,\mu_0\in N,~ \lambda^\vee, \mu^\vee \in N^\vee \\ [\lambda_0]=[\lambda^\vee]\in H,~ [\mu_0]=[\mu^\vee]\in H }} \hspace{-20pt}\CF_{(\mu+s)\phi}\otimes \CF_{\rho^\vee(\mu^\vee+t^\vee)\eta}\otimes \CF_{\rho^\vee\lambda^\vee\theta}\otimes H_{BRST}^*\left(\CF_{\rho(\mu_0+t_0)\eta}\otimes \CF_{\rho\lambda_0\theta}\right).
\ee
If $\sum_i t_i\rho^{ia}\notin \Z$ for some $a$, then $t_0$ is not in $N$, there is no $\lambda_0, \mu_0$  such that $\rho(\mu_0+t_0)=\rho\lambda_0$, and the above cohomology is trivial. This is the first part of the statement. For the second part, assuming now that $t_0\in N$, then the above becomes:
\be
\bigcap_{i, s_i\in \Z}\mathrm{Ker}(S^i)\Big\vert \hspace{-20pt}\bigoplus_{\substack{\lambda_0\in N,~ \lambda^\vee, \mu^\vee \in N^\vee \\ [\lambda_0]=[\lambda^\vee]= [\mu^\vee]+[t_0]\in H }} \hspace{-20pt}\CF_{\rho(\lambda_0+s_0-t_0)\cdot \phi+\rho^\vee(\mu^\vee+s^\vee)\cdot \phi}\otimes \CF_{\rho^\vee(\mu^\vee+t^\vee)\eta}\otimes \CF_{\rho^\vee\lambda^\vee\theta}.
\ee
Let us now tensor this with the VOA $\mc{V}_{\Z^{2r}}^{+-}$ and use the field re-definition as in equation \eqref{redef} to rewrite the above into:
\be
\left[\bigcap\limits_{i, s_i\in \Z} \mathrm{Ker}(S_0^i) \Big\vert \bigoplus_{\lambda\in \Z^n,~ \mu\in [t_0]+\Z^{n-r} } \CF_{\lambda\cdot Z+(\mu-t^\vee)\cdot X+\tau (t-s)\cdot Y}^{X,Y,Z}\right]\otimes \mc{V}_{\Z^{2r}}^{+-}.
\ee
We then use the re-definition in equation \eqref{V*-redef} to obtain:
\be
\bigcap\limits_{i, s_i\in \Z} \mathrm{Ker}(S_0^i)\Big\vert \bigoplus_{\lambda\in \Z^n, \mu\in \Z^n, \nu\in \Z^r}\CF^{\wt{X}, \wt{Y}, \wt{Z}}_{\lambda\cdot \wt{Z}+(\mu-t)\wt{X}+(t-s)+\rho(\nu)\cdot \wt{Y}}.  
\ee
By \cite[Proposition 4.8]{BN22}, this is precisely the free-field realization of $\mc L_A(\mc M)\otimes \mc{V}_{bc}^V$ for the simple module $\mc M=\bigcap_{s_i\in\Z}\mathrm{Ker}(S^i)\Big\vert V_{L, s\phi+t\eta}^{\otimes n}$. This completes the proof. 

\end{proof}

Therefore, the lifting functor and BRST functor coincides on simple modules satisfying the integrability condition $\sum \rho_{ia}J^i_0\in \Z$, and the BRST cohomology simply sends those simple modules that don't satisfy this to a trivial module. The abelian category $\mc C_\rho^A$ can be alternatively defined to be the abelian category of finite-length modules of $\mc V_\rho^A$ whose composition factors are given by BRST cohomology of $\mc M\otimes \mc{V}_{bc}^{\otimes n}$ for simple objects $\mc M$ in $\cbg^{\boxtimes n}$. This means that our definition of the category $\mc C_\rho^A$ using the lifting functor is reasonable, and it allows us to use the idea of simple current extension to define braided tensor structure. 

We can now derive fusion rules of $\V_m^\upsilon$ using the above identification and fusion rules of $\mc V_{\bg}$ from \cite[Corollary 4.14]{BN22}. 

\begin{Cor}
    The simple modules $\V_m^\upsilon$ satisfy the following fusion rules: 
    \begin{itemize}
        \item $\V_m^{[0]}\otimes_{\wt{V}_\rho^A} \V_{m'}^{[0]}\cong \V_{m+m'}^{[0]}$, exactly matching equation \eqref{eq:Wfuse}. 

        \item $\V_m^{[0]}\otimes_{\wt{V}_\rho^A} \V_{m'}^{\upsilon}\cong \V_{m+m'}^{\upsilon}$ exactly matching equation \eqref{fusionna}.

        \item If $\rho^\vee(\upsilon)_i, \rho^\vee(\upsilon')_i$ and $\rho^\vee(\upsilon)_i+\rho^\vee(\upsilon')_i\ne 0 \text{ mod } \Z$ for all $i$, then $\V_m^{\upsilon}\otimes_{\wt{V}_\rho^A} \V_{m'}^{\upsilon'}\cong\bigoplus_{\epsilon\in \{0,1\}^{n}}\V^{\upsilon+\upsilon'}_{m+m'+\tau(\epsilon)}$, where $\tau: \Z^n\to \Z^{n-r}$ is the transpose of $\rho^\vee$, matching (and generalizing) equation \eqref{monodromy-fusion}.  

        \item The fusion product $\V_m^{\upsilon}\otimes_{\wt{V}_\rho^A}\V_{m'}^{-\upsilon}$ is an indecomposable module whose length is given by $4\cdot \#\{i, \rho^\vee(\upsilon)_i\ne 0 \text{ mod }\Z\}$, and whose simple factors are given by $\V_{m+m'\pm \tau(\epsilon)}$ where $\epsilon_i=0$ or $1$ and it can be equal to $1$ only when $\rho^\vee(\upsilon)_i\ne 0 \text{ mod }\Z$.  
        
    \end{itemize}
\end{Cor}

The other fusion rules are more complicated, and we won't record them all here. Recall the action of monodromy:
\be
\mu(\mc M,\mc N): \mc M\otimes_{\wt{\mc V}_\rho^A} \mc N\stackrel{R_{\mc M, \mc N}}{\longrightarrow} \mc N\otimes_{\wt{\mc V}_\rho^A} \mc M\stackrel{R_{\mc N, \mc M}}{\longrightarrow} \mc M\otimes_{\wt{\mc V}_\rho^A} \mc N.
\ee
Using free-field fusion rule, in a method similar to Lemma \ref{LemmonoAq}, we can derive the action of monodromy for these simple modules. (Recall the definition of monodromy from \eqref{def-monMN}.)

\begin{Cor}
    The monodromy between simple modules are given by:
    \begin{itemize}
        \item The monodromy $\V_m^{0}\otimes_{\wt{V}_\rho^A} \V_{m'}^{\upsilon}\longrightarrow \V_m^{0}\otimes_{\wt{V}_\rho^A} \V_{m'}^{\upsilon}$ is equal to $e^{2\pi i m\cdot \upsilon}$.

        \item If $\rho^\vee(\upsilon)_i, \rho^\vee(\upsilon')_i$ and $\rho^\vee(\upsilon)_i+\rho^\vee(\upsilon')_i\ne 0 \text{ mod } \Z$ for all $i$, then the monodromy $\V_m^{\upsilon}\otimes_{\wt{V}_\rho^A} \V_{m'}^{\upsilon'}\longrightarrow \V_m^{\upsilon}\otimes_{\wt{V}_\rho^A} \V_{m'}^{\upsilon'}$ is given by 
        $$\sum_{\epsilon \in \{ 0, 1\}^n}e^{2\pi i \left(\upsilon \cdot m'+\upsilon'\cdot m+\tau(\epsilon)\cdot (\upsilon+\upsilon')\right)} \textup{Id}_{\V^{\upsilon+\upsilon'}_{m+m'+\tau(\epsilon)}} .$$ 

        \item The monodromy on $\V_m^{\upsilon}\otimes_{\wt{V}_\rho^A}\V_{m'}^{-\upsilon}$ is not a phase, and has $\#\{i, \rho^\vee(\upsilon)_i\ne 0 \text{ mod } \Z\}$ many Jordan blocks each of size $2$. 
    \end{itemize}
\end{Cor}

\begin{Exp}
    Consider again the A twist of free hypermultiplet. In this case, the category $\mc C_\rho^A$ is simply $\cbg$, and the simple modules are given by $\CV_n$ and $\Sigma^n \mc W_{\upsilon}$ for $\upsilon\ne 0\text{ mod }\Z$. The fusion rule of these are given by:
    \begin{itemize}
        \item $\CV_n\otimes_{\mc V_{\bg}}\Sigma^{n'} \mc W_{\upsilon}=\Sigma^{n+n'}\mc W_{\upsilon}$. 

        \item $\Sigma^{n} \mc W_{\upsilon}\otimes_{\mc V_{\bg}}  \Sigma^{n'} \mc W_{\upsilon'}=\Sigma^{n+n'}\mc W_{\upsilon+\upsilon'}\oplus \Sigma^{n+n'+1}\mc W_{\upsilon+\upsilon'}$, if $\upsilon+\upsilon'\ne 0\text{ mod } \Z$. 

        \item $\Sigma^n \mc W_{\upsilon}\otimes_{\mc V_{\bg}}  \Sigma^{n'} \mc W_{-\upsilon}=\Sigma^{m+n+1}\mc P$, where $\mc P$ is an indecomposable module of length $4$. 
    \end{itemize}
\end{Exp}

\begin{Exp} 
    Consider the A twist of $\mathrm{SQED}_1$, with $\rho=(1)$. In this case, we have seen that the category $\mc C_\rho^A$ is the category of modules of $\mc V_{SF}$, and has a unique simple object. The object $\mc P$ from the previous example can be lifted to $\mc C_\rho^A$ and is a projective object in $\mc C_\rho^A$. The fusion rule of this is:
    \be
\mc L_A (\mc P)\otimes_{\wt{\mc V}_\rho^A} \mc L_A (\mc P)\cong \C^4\otimes \mc L_A(\mc P).
    \ee

\end{Exp}

\subsection{The category of line operators for the B twist}\label{subsec:LineB}

\subsubsection{The Kazhdan-Lusztig category of $\widehat{\grho}$}

Since we already have an isomorphism $\mc V_\rho^A\cong\mc{V}_\rho^B$, one may ask why we still need to define the category $\mc C_{\rho}^B$. After all, we must have equivalence of derived categories. The reason is that we would like to see the category defined from the perspective of the affine Lie superalgebra $\widehat{\grho}$, so that we can compare the objects with physical line operators. The VOA $\mc{V}_{\rho}^B$ is a simple current extension of $\widehat{\grho}$, we would like to use the Kazhdan-Lusztig category of $\widehat{\grho}$ to define the category, since objects in this category are directly comparable to physical line operators. 

\begin{Def} \label{def:KLrho}
We define $KL_\rho$, the Kazhdan-Lusztig category of $\widehat{\grho}$, to be the category of finite-length, grading-restricted modules of $\widehat{\grho}$.
\end{Def}

We would like to define $\mc C_{\rho}^B$ as lifts of modules from $KL_\rho$, but in order to do so, we first need to show that $KL_\rho$ is a braided tensor category. We claim:

\begin{Thm}\label{ThmKLqbrten}
$KL_\rho$ is a braided tensor category defined by logarithmic intertwining operators. Moreover, tensor product is exact functor $KL_\rho$. 
\end{Thm}

Let $\mathbb{U}_\mu$ be the module of $\widehat{\grho}$ generated by the monopole operator $\mc U_\mu$. 

\begin{Prop}\label{prop:monopolesimple}
    The modules $\mathbb{U}_\mu$ are simple currents in $KL_\rho$ that satisfy the fusion rules:
    \be
\mathbb{U}_\mu\otimes_{\widehat{\grho}} \mathbb{U}_\nu=\mathbb{U}_{\mu+\nu}.
    \ee
    The algebra object $\bigoplus \mathbb{U}_\mu$ in $KL_\rho$ can be identified with $\CV_\rho^B$. 
\end{Prop}

We will prove the two statements later, but they, together with the technique of simple current extensions, implies the existence of the following lifting functors between braided tensor categories:
\be
\mc{L}_{B}: KL_\rho^{[0]}\longrightarrow \mc{V}_\rho^B\mathrm{-Mod}_{\mathrm{loc}}\left(\mathrm{Ind}(KL_\rho)\right).
\ee
Similar to Definition \ref{def:lineopeA}, we define:

\begin{Def}\label{Def:lineopeB}
    Let us denote by $\mc{C}_{\rho}^B$ to be the image of $KL_\rho^{[0]}$ under $\mc{L}_{B}$:
    \be
\mc{C}_{\rho}^B:=\mc{L}_{B}(KL_\rho^{[0]}).
    \ee
    This is a braided tensor category via logarithmic intertwing operators. 
\end{Def}

As in the case of $\wt{\mc{V}}_{\rho}^A$, there is an action of $(\C^*)^r$ on $\mc{V}_\rho^B$, such that the weight of the monopole operator $\mc U_\mu$ is given by $\mu$, and the weight of other generators are zero. By definition, we clearly have $(\mc{V}_\rho^B)^{(\C^*)^r}=\widehat{\grho}$. We conclude that $\mc{C}_{\rho}^B$ is the de-equivariantization $KL_{\rho, (\C^*)^r}^{[0]}$, or $KL_{\rho}^{[0]}/ \Z^r$.

We would like to understand the category $KL_{\rho}$, and especially those that have trivial monodromy with $\mc{V}_\rho^B$, in order to compare with physical line operators. For this purpose, and for the proof of Theorem \ref{ThmKLqbrten}, we need to first look at Verma modules of $\widehat{\grho}$. 

\subsubsection{The Verma modules}

For each finite-dimensional module of $\grho$, say $M$, one can define a module for $\widehat{\grho}$, denoted by $\widehat{M}$, in the following way. The mode algebra of $\widehat{\grho}$, which we denote by $U\left(\widehat{\grho}\right)$, has a subalgebra $U_{\geq 0}$ generated by all non-negative modes. The Lie algebra $\grho$ is a quotient of $U_{\geq 0}$. We can then view $M$ as a $U_{\geq 0}$ module through the quotient map $U_{\geq 0}\to \grho$. In other words, we let the positive modes act trivially. The module $\widehat{M}$ is then the induction of $M$:
\begin{equation}
    \widehat{M}:= M\otimes_{U_{\geq 0}} U\left(\widehat{\grho}\right).
\end{equation}

Any simple elememt in $KL_\rho$ is generated by its lowest conformal weight space, and so is a quotient of a Verma module. Therefore, to understand simple modules in $KL_\rho$, we need to understand Verma modules. 

\begin{Lem}\label{Lem:grhosimple}
    Let $\upsilon\in \C^r$ and $M$ be a simple module of $\grho$ such that $E^a$ acts with eigenvalue $\upsilon^a$. When $\sum_i\rho_{ia}\upsilon^a\notin \Z$ or $\sum_i\rho_{ia}\upsilon^a=0$ for all $i$, the module $\widehat{M}$ is simple.
\end{Lem}

\begin{proof}
    We only need to show that any $w\in \widehat{M}$ generates the entire module. Note that an element in $\widehat{M}$ is always of the form:
    \be
\sum N^a_*\psi^{i,+}_*\psi^{i,-}_*E^a_* v
    \ee
where $v$ is in the lowest conformal weight space of $\widehat{M}$, which is simply $M$. The subscripts are all negative integers. We give a lexicographic order to $\widehat{M}$ such that $N>\psi^+>\psi^->E$, and
\be
N^a_{-n}>N^a_{-n+1}>\cdots> N^{a+1}_{-n}>N^{a+1}_{-n+1}
\ee
and similarly for $\psi^\pm$ and $E$. Given an $w$, let $w=w_0+w'$ where $w_0=N^a_*\psi^{i,+}_*\psi^{i,-}_*E^a_* v$ is a  homogeneous vector that is the biggest in the lexicographic order, such that $w_0>w'$. Denote by $\mc W$ the sub-representation generated by $w$. 

We perform the following procedure. If the expression of $w_0$ involves $(N^a_{-n})^k$, we will apply to $w$ $(E^a_n)^k$. Since $[E^a_n, N^a_{-n}]=n$, the vector $(E^a_n)^k w_0$ will have no $N^a_{-n}$ in its expression. Moreover, each time applying $E^a_{n}$, we obtain a non-zero vector. We repeate this process until all the $N^a_*$ in the expression of $w_0$ is killed. 

The next step we apply $\psi^{i,-}_{n}$ until there is no $\psi^{i,+}_*$ in the expression of $w_0$. Note that this step, we need the assumption that $\sum_i\rho_{ia}\upsilon^a\notin \Z$ or equal to $0$, since the commutator $\{\psi^{i,-}_n, \psi^{i,+}_{-n}\}=n+\sum_i\rho_{a}{}^iE^{a}_{0}$ acts non-trivially on $w_0$ when $n+\sum_i\rho_{a}{}^iE^{a}_{0}\ne 0$. By definition, $\sum_i\rho_{ia}E^{a}_{0}$ acts as $\sum_i\rho_{ia}\upsilon^a$, which is not an integer or equal to $0$ by assumption. Therefore the action of $\psi^{i,-}_n$ is nontrial on $w_0$. We can now safely keep the procedure until all the $\psi^{i,+}_*$ in the expression of $w_0$ is annihilated. 

We can repeat this process for all $\psi^{j,-}_*$ and $E^a_*$, until all the negative modes in the expression of $w_0$ is annihilated. Since $w_0>w'$, this process must annihilate $w'$ entirely, and we are left with a nonzero vector in the lowest conformal weight space. Consequently, $\mc W$ contains a nonzero vector $v$ in $M$. We can now conclude the proof since $M$ is assumed to be simple, which means $\mc W$ contains $M$, and consequently $\widehat{M}$.  
    
\end{proof}

\begin{Rmk}
    In particular, the vacuum module of $\widehat{\grho}$ is simple, since the vacuum module is defined as $\widehat{\C}$ where $\C$ is the trivial (and therefore simple) $\grho$ module. This confirms the discussion of Section \ref{sec:mon2}, especially Proposition \ref{PropffpurtBgaugeq} and Corollary \ref{cor:VB-ff}. 
\end{Rmk}

Consider a simple module $M$ of $\grho$. Since $N^a_0$ and $E^a_0$ (zero mode of $N^a(z),E^a(z)$) commutes with each other, there must be at least one simultaneous eigenvector for all of them. By acting on $\psi^{i,-}_0$, we may assume that this eigenvector $v$ is annihilated by $\psi^{i,-}_0$ for all $i$. Let $(m_a,\upsilon_a)$ be its eigenvalues under $N^a_0$ and $E^a_0$. It is clear then that the module $M$ is spanned by vectors of the form:
\begin{equation}\label{eqspansimplegq}
    \psi^{i_1,+}_0\cdots \psi^{i_k,+}_0v,\qquad i_1< i_2< \cdots < i_k, k\leq n.
\end{equation}
Each of this vector is an eigenvector of $N^a_0,E^a_0$ with eigenvalues:
\begin{equation}
    \left(m_a+\sum\limits_{1\leq s\leq k} \rho_{i_s a}, ~ \upsilon_a\right). 
\end{equation}
Let us define module $V_{(m,\upsilon)}$ to be the module generated by vectors of the form in equation \eqref{eqspansimplegq}. Then any simple module is a quotient of $V_{(m,\upsilon)}$ for some ${(m,\upsilon)}\in \C^r\times \C^r$. Denote by $\widehat{V}_{(m,\upsilon)}$ the corresponding Verma module. 

\begin{Lem}
    The module $V_{(m,\upsilon)}$ is a simple module of $\grho$ when $\sum \rho_{ia}\upsilon^a\ne 0$ for all $i$. 
\end{Lem}

\begin{proof}
    The idea of the proof is very similar to the previous lemma. We define a lexicographic order on $V_{(m,\upsilon)}$ such that $\psi^{1,+}>\cdots >\psi^{n,+}$. Let $w$ be any vector in $V_{(m,\upsilon)}$, and $M$ the module generated by $w$. Let $w=w_0+w'$ where $w_0>w'$ and $w_0$ is homogeneous. Then we can apply $\psi^{i,-}$ consequtively on $w$ to remove $\psi^{i,+}$. Beccause $\sum \rho_{ia}\upsilon^a$ is not zero, each of this action is non-trivial on $w_0$. After annihilating all the $\psi^{i,+}$ in the expression of $w_0$, the vector $w'$ is annihilated since $w_0>w'$, and we are left with a highest-weight vector. Therefore $M=V_{(m,\upsilon)}$. 
    
\end{proof}

As a consequence of the above two lemmas, the module $\widehat{V}_{(\mu, \lambda)}$ is simple when $\sum_i\rho_{ia}\lambda^a\notin \Z$ for all $i$, and is finite length for $\sum_i\rho_{ia}\lambda^a\notin \Z$ or $\sum_i\rho_{ia}\lambda^a=0$ for all $i$. In fact, we have the following proposition, whose proof is a word-to-word translation of the proof of \cite[Proposition 3.2]{BN22}. Let $\CG_\rho$ be the category of finite-dimensional modules of $\grho$, and $\CG_{\rho, \upsilon}$ the subcategory where generalized eigenvalue of $E^a$ is $\upsilon^a$. 

\begin{Prop}\label{Prop:KLGrho}
    When $\sum_i\rho_{ia}\upsilon^a\notin \Z$ or $\sum_i\rho_{ia}\upsilon^a=0$ for all $i$, the induction functor $M\mapsto \widehat{M}$ is an equivalence of abelian categories:
    \be
\CG_{\rho, \upsilon}\simeq KL_{\rho, \upsilon}.
    \ee
\end{Prop}

How do we deal with $KL_{\rho, \upsilon}$ when $\upsilon$ does not satisfy the above? The answer is the spectral-flow automorphism. The VOA $\widehat{\grho}$ has the following spectral-flow automorphism:
\be
\sigma_{\mu, \nu}N^a(z)=N^a(z)-\frac{\mu^a}{z},\qquad \sigma_{\mu, \nu}E^a(z)=E^a(z)-\frac{\nu^a}{z}\qquad \sigma_{\mu, \nu}\psi^{i,\pm}(z)=z^{\mp \sum_a \rho_{ia}\nu^a} \psi^{i,\pm}.
\ee
In this, $\mu\in \C^r$ and $\nu\in \C^r$ such that $\sum_a \rho_{ia}\nu^a\in \Z$ for all $i$, or in other words, $\nu\in \rho^{-1}(\Z^n)$. The spectral flow gives an equivalence:
\be
\sigma_{\mu, \nu}: KL_{\rho, \upsilon}\simeq KL_{\rho, \upsilon+\nu}.
\ee
For each $\upsilon\in \C^r$, using the splitting of the short exact sequence \eqref{split-SES}, there must be a $\nu\in\rho^{-1}(\Z^n)$, such that $\upsilon+\nu$ satisfies that $\sum_i\rho_{ia}(\upsilon^a+\nu^a)\notin \Z$ or $\sum_i\rho_{ia}(\upsilon^a+\nu^a)=0$ for all $i$. We are now ready to prove Theorem \ref{ThmKLqbrten}.

\begin{proof}[Proof of Theorem \ref{ThmKLqbrten}]

In \cite[Corollary 3.4.6]{creutzig2021tensor}, the authors showed that if all the grading-restricted generalized Verma modules of $\widehat{\grho}$ are of finite length, then $KL_\rho$ has the structure of a braided tensor category by logarithmic intertwining operators. Let $M$ be a simple module of $\grho$ in $\CG_{\rho, \upsilon}$. If $\sum_i\rho_{a}{}^i\upsilon^a\notin \Z$ or $\sum_i\rho_{a}{}^i\upsilon^a=0$ for all $i$ then we are done, since $\widehat{M}$ is simple and therefore is of finite-length. 

Suppose $\sum_i\rho_{a}{}^i\upsilon^a$ is a non-zero integer for some $i$. Choose $\nu\in \rho^{-1}(\Z^n)$ such that $\upsilon+\nu$ satisfy that entries of $\rho (\upsilon+\nu)$ is either zero or a non-integer. We just need to show that $\sigma_{0, \nu}\widehat{M}$ is of finite-length. Choose $m\in M$, let $I$ be the subset of all $i$ where $\sum \rho_{ai}\nu^a>0$ and $J$ where $\sum \rho_{ai}\nu^a<0$. Consider the sub-representation of $\sigma_{0, \nu}\widehat{M}$ given by:
\be
N:=\bigoplus\limits_{i\in I, 0\leq n< \rho_{ai}\nu^a}<\psi^{i,-}_n m>\oplus \bigoplus\limits_{j\in J, 0<n\leq -\rho_{aj}\nu^a}<\psi^{j,+}_n m>
\ee
Here $\langle v\rangle$ for $v\in \widehat{M}$ denotes the submodule of $\widehat{\grho}$ generated by the vector $v$. By definition of $N$, the quotient $\sigma_{0, \nu}\widehat{M}/N$ is generated by a single element $\overline{m}$ (the image of $m$ in the quotient) and that all the negative modes of $\widehat{\grho}$ acts trivially on $\overline{m}$. Let $M'$ be the $\grho$ submodule of $\sigma_{0, \nu}\widehat{M}/N$ generated by $\overline{m}$, which must be finite-dimensional. By universal property of the induction functor, we have a surjection $\widehat{M'}\to \sigma_{0, \nu}\widehat{M}/N$, and consequently, $\sigma_{0, \nu}\widehat{M}/N$ is finite length. 

We can now repeat this argument for all the summands in $N$, and this process must terminate because there are only finitely many positive modes that act non-trivially. We therefore obtain a finite filtration $\CF_*\sigma_{0, \nu}\widehat{M}$ such that each graded piece is a quotient of a finite-length Verma module. Therefore, $\sigma_{0, \nu}\widehat{M}$ is finite-length, and so must be $\widehat{M}$. This completes the proof. 

\end{proof}

We have now shown that $KL_\rho$ has the structure of a braided tensor category. This is, however, not enough since we need fusion to be exact, and we need to verify the fusion rule of the simple modules $\mathbb{U}_\mu$. This can be done in many different ways, we will do this via the relation between $\widehat{\grho}$ and $\widehat{\mathfrak{gl}(1|1)}$, derived in Section \ref{subsecrelatealt}. In there, it is shown that $\widehat{\grho}\otimes \mc{V}_{\mathbb{Z}^{2(n-r)}}^{+-}$ is a simple current extension of $\widehat{\mathfrak{gl}(1|1)}^{\otimes n}$, or in other words, it is an algebra object in $\mathrm{Ind}(KL^{\boxtimes n})$. Let $KL^{\boxtimes n, \rho, [0]}$ be the subcategory whose objects have trivial monodromy with $\widehat{\grho}\otimes \mc{V}_{\mathbb{Z}^{2(n-r)}}^{+-}$. Then there is a lifting functor:
    \be
\mc L_{\rho}^{\mathrm{ungauge}}: KL^{\boxtimes n, \rho, [0]}\longrightarrow \widehat{\grho}\otimes \mc{V}_{\mathbb{Z}^{2(n-r)}}^{+-}\mathrm{-Mod}_{\mathrm{loc}}\left(\mathrm{Ind}(KL^{\boxtimes n})\right).
    \ee

We prove:

\begin{Thm}\label{Thm:KLKLrho}
    The image of $KL^{\boxtimes n, \rho, [0]}$ under $\mc L_\rho^{\mathrm{ungauge}}$ is precisely  the category $KL_\rho$. Consequently, there is a braided tensor equivalence between $KL_\rho$ and the de-equivariantization of $KL^{\boxtimes n, \rho, [0]}$ by the lattice of simple modules defining $\widehat{\grho}\otimes \mc{V}_{\mathbb{Z}^{2(n-r)}}^{+-}$. 
\end{Thm}

The last statement of Theorem \ref{ThmKLqbrten} follows now from this theorem. To prove this, we use the free-field realizations.  

Consider the free field realization of Section \ref{sec:Bff}. For each $\upsilon\cdot X+m \cdot Y$, one has a $V_{\mathcal Z}$ module (and therefore a module of $\widehat{\grho}$) $V_{\mathcal Z, \upsilon\cdot X+m \cdot Y}$, generated by the vector $v:=\vert\upsilon\cdot X+m \cdot Y\rangle$. Let $b^i(z)=\normord{e^{Z^i(z)}}$ and $c^i(z)=\normord{e^{-Z^i(z)}}$, and consider the following grading $\Delta$ on $V_{\mathcal Z, \upsilon\cdot X+m \cdot Y}$ by:
\be
\Delta (v)=\Delta (c_{-1}^iv)=0,\qquad \Delta (b_{-1}^iv)=\Delta (X_{-1}^av)=\Delta (Y_{-1}^av)=1.
\ee
Moreover, $\Delta (A_{n}B)=\Delta (A)+\Delta (B)-n-1$. With this grading, $V_{\mathcal Z, \upsilon\cdot X+m \cdot Y}$ is positively graded, and the minimal degree part of this is spanned by vectors of the form:
\be
\prod\limits_{i_1<i_2<\cdots<i_k}c^{i_1}_{-1}\ldots c^{i_k}_{-1}v.
\ee
The vector $c_{-1}^1c_{-1}^2\cdots c_{-1}^n v$ has weight $(\upsilon_a)$ under the action of $E^a_0$ and weights:
\begin{equation}
    m_a-\sum_i \rho_{ia}.
\end{equation}
under the action of $N^a_0$. Moreover, this vector is annihilated by positive modes of $\widehat{\grho}$ by degree considerations as $\Delta(x)=1$ for $x\in \grho$, and moreover, it is killed by $\psi^{i,-}_0$ by considering the weights of $N_0$. Therefore, this vector, under the action of $\psi_0^{i, +}$, generates a copy of $V_{(m-\rho, \upsilon)}$, where $m-\rho=(\mu_a-\sum_i \rho_{ia})$. By universal property of induction functor, there is an induced morphism:
\begin{equation}\label{eqmorphvermaff}
    \widehat{V}_{(m-\rho, \upsilon)}\longrightarrow V_{\mathcal Z, \upsilon\cdot X+m \cdot Y}.
\end{equation}

\begin{Prop}\label{Propvermaffequiv}
When $\sum_i\rho_{ia}\upsilon^a\notin \Z$ or $\sum_i\rho_{ia}\upsilon^a=0$ for all $i$, the morphism in equation \eqref{eqmorphvermaff} is an isomorphism.
\end{Prop}

\begin{proof}
    We first show that this is an embedding. To do so, by Proposition \ref{Prop:KLGrho}, we only need to show that it is non-zero on any submodule of $V_{(m-\rho, \upsilon)}$. In fact, the $\grho$ module $V_{(m-\rho, \upsilon)}$ has a unique simple submodule generated by:
    \be
\prod\limits_{i, \sum \rho_{ia}\upsilon^a=0} \psi^{i,+}_0v. 
    \ee
It is very clear then that the above map $\widehat{V}_{(m-\rho, \upsilon)}\to V_{\mathcal Z, \upsilon\cdot X+m \cdot Y}$ is nonzero when restricted to this unique simple. 

To show that this is an isomorphism, we just need to define a positive grading on $\widehat{V}_{(m-\rho, \upsilon)}$ such that $\widehat{V}_{(m-\rho, \upsilon)}\to V_{\mathcal Z, \upsilon\cdot X+m \cdot Y}$ is graded and that they have the same grading. We define the grading on $\widehat{V}_{(m-\rho, \upsilon)}$ such that $\Delta (V_{(m-\rho, \upsilon)})=0$ and $\Delta (x_{-1}v)=1$ for all $x\in \grho$, and that $\Delta (A_{n}B)=\Delta (A)+\Delta (B)-n-1$. It is clear that the map $\widehat{V}_{(m-\rho, \upsilon)}\to V_{\mathcal Z, \upsilon\cdot X+m \cdot Y}$ is a map of positively graded vector spaces. It is straightforward that the graded character $\mathrm{Tr}(q^\Delta)$ agrees on the two modules, and so they must be isomorphic. 
    
\end{proof}

\begin{Cor}\label{Cor:VZKLrho}
    For any $\upsilon,m\in \C^r$, the module $V_{\CZ, \upsilon\cdot X+m\cdot Y}$ is an object in $KL_\rho$.
\end{Cor}

\begin{proof}
    We have proved this if $\sum \rho_{ia}\upsilon^a=0$ or $\sum \rho_{ia}\upsilon^a\notin \Z$ for all $i$. When this is not satisfied, then there is a $\nu\in \rho^{-1}(\Z^n)$ such that $\upsilon+\nu$ satisfies this, and therefore $V_{\CZ, (\upsilon+\nu)\cdot X+m\cdot Y}$ is an object in $KL_\rho$. We now can finish the proof since there is clearly an isomorphism:
    \be
V_{\CZ, \upsilon\cdot X+m\cdot Y}\cong \sigma_{\rho^\perp\rho\nu,\nu} V_{\CZ, (\upsilon+\nu)\cdot X+m\cdot Y},
    \ee
given by mapping $\sigma_{\rho^\perp\rho\nu,\nu}\vert (\upsilon+\nu)\cdot X+m\cdot Y\rangle$ to $\vert \upsilon\cdot X+m\cdot Y-\rho(\nu)\cdot Z\rangle$ (the free field realization of equation \eqref{ffRealization}). Since the spectral flow $\sigma$ preserves $KL_\rho$, the object $V_{\CZ, \upsilon\cdot X+m\cdot Y}$ is in $KL_\rho$, and the proof is complete. 
    
\end{proof}

We can now prove Theorem \ref{Thm:KLKLrho}:

\begin{proof}[Proof of Theorem \ref{Thm:KLKLrho}]
Recall the field redefinition of equation \eqref{V*-redef}. For each $\upsilon, m\in \C^r$, the module $V_{\CZ,\upsilon\cdot X+m\cdot Y}\otimes \mc{V}_{\mathbb{Z}^{2(n-r)}}^{+-}$ of 
 $\widehat{\grho}\otimes \mc{V}_{\mathbb{Z}^{2(n-r)}}^{+-}$ is a lift of the $\widehat{\mathfrak{gl}(1|1)}^{\otimes n}$ module $\bigoplus_{\nu\in \Z^n} \CF^{\wt{X},\wt{Y},\wt{Z}}_{\nu\cdot \wt{Z}+\rho(\upsilon)\cdot\wt{X}+\sigma^\vee(m)\cdot \wt{Y}}$ (restricted from the free-field algebra). This, together with Proposition \ref{Propvermaffequiv} and the proof of Corollary \ref{Cor:VZKLrho}, shows that $\CL_\rho^{\mathrm{ungauge}}$ is essentially surjective. Indeed, Proposition \ref{Propvermaffequiv} and the proof of Corollary \ref{Cor:VZKLrho} combined shows that any simple objects in $KL_\rho$ is a quotient of $V_{\CZ,\upsilon\cdot X+m\cdot Y}$ for some $\upsilon,m\in \C^r$, and therefore is in the image of $\CL_\rho^{\mathrm{ungauge}}$. By definition of $KL_\rho$ and $KL^{\boxtimes n}$, the functor $\CL_\rho^{\mathrm{ungauge}}$ is essentially surjective.

We now show that $\CL_\rho^{\mathrm{ungauge}}$ maps into $KL_\rho$. Let $\mc M$ be a simple module of $\widehat{\mathfrak{gl}(1|1)}^{\otimes n}$ that has trivial monodromy with $\widehat{\grho}\otimes \mc{V}_{\mathbb{Z}^{2(n-r)}}^{+-}$. By \cite[Proposition 4.3]{BN22}, every such $\mc M$ embeds uniquely into a module of the form $\bigoplus_{\nu \in \Z^n} \CF^{\wt{X},\wt{Y},\wt{Z}}_{\nu\cdot\wt{Z}+\upsilon'\cdot\wt{X}+m'\cdot\wt{Y}}$ for some $\upsilon',m'\in \C^n$. By the same proof of Lemma  \ref{LemmonoAq}, such a module can be lifted only when $\rho^\vee\upsilon'\in \Z^{n-r}$ and $\sigma^\T m'\in \Z^{n-r}$. When this is satisfied, the lift of such a module can be identified with $V_{\CZ,(\sigma^\vee)^\T (\upsilon')\cdot X+\rho^\T (m')\cdot Y}\otimes \mc{V}_{\mathbb{Z}^{2(n-r)}}^{+-}$, which by Corollary \ref{Cor:VZKLrho} is an object in $KL_\rho$. Therefore, the image of any simple module under $\CL_\rho^{\mathrm{ungauge}}$ lies in $KL_\rho$. This completes the proof.

\end{proof}

We now prove Proposition \ref{prop:monopolesimple}:

\begin{proof}[Proof of Proposition \ref{prop:monopolesimple}]
For each $\mu\in \Z^r$, the module $\mathbb{U}_\mu\otimes \mc{V}_{\mathbb{Z}^{2(n-r)}}^{+-}$ is the kernel of the screening operator in $\CV_{\CZ, \mu\cdot X}\otimes \mc{V}_{\mathbb{Z}^{2(n-r)}}^{+-}$. By Theorem \ref{Thm:KLKLrho}, this is the image of the kernel of the screening operators in  $\bigoplus\limits_{\nu\in \Z^n}\CF^{\wt{X},\wt{Y},\wt{Z}}_{\nu\cdot\wt{Z}+\rho (\mu)\cdot \wt{X}}$ under $\CL_\rho^{\mathrm{ungauge}}$. The fact that these are simple follows from \cite[Proposition 4.3]{BN22} and their fusion rules now follow from the fusion rules in $KL^{\boxtimes n}$.

To finish the proof, we just need to comment that the fusion product
\be
\mathbb{U}_\mu\otimes_{\widehat{\grho}}\mathbb{U}_\nu\longrightarrow \mathbb{U}_{\mu+\nu}
\ee
can be induced from the OPE of $\mc V_\rho^B$, which is a consequence of the free-field realization of $\mathbb{U}_\mu$. 
\end{proof}

\subsubsection{Simple objects and physical line operators}

We have now justified our definition of $\mc{C}_{\rho}^B$. However, we still need to understand how monodromy acts on objects in $KL_\rho$ in order to study the objects in $\mc{C}_{\rho}^B$. It is convenient to have a statement similar to that of Proposition \ref{Propcbgsscbgmono}.
We claim the following proposition, which shows that the monodromy condition is equivalent to the requirement that the aciton of $N_0^a$ gives rise to an action of $(\mathbb{C}^*)^r$ on a module. The proof of this will follow from mirror symmetry statement, and will be presented in the next section. 

\begin{Prop}\label{prop:intBside}
    A module $\mc M$ of $\widehat{\grho}$ belong to $KL_\rho^{[0]}$ if and only if the action of $N_0^a$ is semi-simple with integer eigenvalues. 
\end{Prop}

Using this, let us consider the structure of the category $\mc{C}_{\rho}^B$. Since $KL_\rho$ has a decomposition:
\be
KL_{\rho}=\bigoplus\limits_{\upsilon\in \C^r} KL_{\rho, \upsilon}
\ee
where $KL_{\rho, \upsilon}$ is the full-subcategory where the generalized eigenvalue of $E_0^a$ is $\lambda^a$. Consequently, there is a decomposition:
\be
KL_{\rho}^{[0]}=\bigoplus\limits_{\upsilon\in \C^r} KL_{\rho, \upsilon}^{[0]}.
\ee
The lifting functor identifies an object $\mc M$ with $\mathbb{U}_\mu\otimes_{\widehat{\grho}} \mc M$, and it is clear that $\mathbb{U}_\mu$ induces an equivalence:
\be
\mathbb{U}_\mu\otimes_{\widehat{\grho}} -: KL_{\rho, \upsilon}^{[0]}\stackrel{\simeq}{\longrightarrow}KL_{\rho, \upsilon+\mu}^{[0]}.
\ee
This leads to the following Corollary, in which we abuse the notation $\upsilon$ to mean both an element in $(\C/\Z)^r$ and a representative in $\C^r$. 

\begin{Cor}
    The category $\mc{C}_{\rho}^B$ admits a decomposition that is compatible with fusion product:
    \be
\mc{C}_{\rho}^B=\bigoplus\limits_{\upsilon\in (\C/\Z)^r}\mc{C}_{\rho, \upsilon}^B
    \ee
    where $\mc{C}_{\rho, \upsilon}^B$ is the image of $KL_{\rho, \upsilon}^{[0]}$ for any representative $\upsilon$. 
\end{Cor}

Now to understand simple modules, we only need to understand simple modules of $KL_{\rho, \upsilon}^{[0]}$. From the previous discussions, we see that any simple module of $KL_{\rho, \upsilon}$ is a quotient of $\CV_{\CZ, \upsilon\cdot X+m\cdot Y}$.

\begin{Prop}
    For each pair $(m,\upsilon)\in \C^r\times \C^r$, the module $\CV_{\CZ, \upsilon\cdot X+(m+\rho)\cdot Y}$ has a unique simple quotient, which we denote by $\mathbb{W}_{m, \upsilon}$. Here $m+\rho$ is the vector $(m_a+\sum_i\rho_{ia})$. 
\end{Prop}

    \begin{proof}
        By the proof of Theorem  \ref{Thm:KLKLrho}, the module $\CV_{\CZ, \upsilon\cdot X+(m+\rho)\cdot Y}\otimes \mc{V}_{\mathbb{Z}^{2(n-r)}}^{+-}$ is the lift of the module
        \be
        \bigoplus_{\nu\in \Z^n} \CF^{\wt{X},\wt{Y},\wt{Z}}_{\nu\cdot \wt{Z}+\rho(\upsilon)\cdot\wt{X}+\sigma^\vee (m+\rho)\cdot \wt{Y}}
        ,\ee under the functor $\mc L_{\mathrm{ungauge}}$, which by \cite{creutzig2021duality}, is identified with a tensor product of indecomposable Verma modules of $\widehat{\mathfrak{gl}(1|1)}^{\otimes n}$, and has a unique simple quotient, as each of the indecomposable module in the tensor product has a unique simple quotient. Since $\mc L_\rho^{\mathrm{ungauge}}$ preserves composition series, the statement follows. 
    \end{proof}

In order for $\mathbb{W}_{m, \upsilon}$ to be in $KL_\rho^{[0]}$ (namely that it satisfies monodromy condition), $m\in \Z^r$. Upon lifting, $\mathbb{W}_{m, \upsilon}$ is identified with $\mathbb{U}_\nu\otimes_{\widehat{\grho}}\mathbb{W}_{m, \upsilon}$, which is $\mathbb{W}_{m, \upsilon+\nu}$. This is a complete list of simples. We summarize this into:

\begin{Cor}
    The simple modules of $\mc{C}_{\rho}^B$ are labelled by $\Z^r\times (\C/\Z)^r$. For each $m\in \Z^r$ and $\upsilon\in (\C/\Z)^r$, we denote by $\mathbb{W}_m^\upsilon$ the corresponding simple object, in accordance with Section \ref{sec:bulklines}. The module $\mathbb{W}_m^\upsilon$ is the image under the lifting functor $\mc L_B$ of the module $\mathbb{W}_{m,\upsilon}$ for any representative $\upsilon$. 
\end{Cor}

In the special case when $\upsilon=0$, one can construct the simple module alternatively as follows. Consider the module $A_{m}$ of $\grho$ spanned by a single vector that has weight $m_a$ under the action of $N_a$ and annihilated by all other elements. It is simply from definition that $\mathbb{W}_m^{0}=\mc{L}_B\widehat{A_m}$. We identify these objects as the Wilson lines operators, since $\widehat{A_m}$ is defined by the representation of the gauge group $(\C^*)^r$ determined by the character $m$. When $\upsilon\ne 0$, we identify $\mathbb{W}_m^\upsilon$ as the general gauge vortex, since on such a module, the action of $E(z)$ is, up to gauge fixing, $E(z)\sim \frac{\upsilon}{z}$.

\begin{Exp}
    Consider the case when $\rho=0$, namely the B twist of free hyper. In this case, $\mc V_\rho^B=\mc V_{SF}$, since there is no monopole operator.  The category $KL_\rho\simeq \mc C_\rho^B$ is equvialent to $\C[\psi^+,\psi^-]\mathrm{-Mod}_{fin}$, the category of finite-dimensional modules of the exterior algebra $\C[\psi^+,\psi^-]$, as abelian category. This category has a unique simple object $\mc V_{SF}$, the vacuum module, which cprresponds to the trivial module $\C$ of $\C[\psi^+,\psi^-]$. It also has a projective object $\mc L_A (\mc P)$ as in Section \ref{subsec:LineA}, which corresponds to the free module $\C[\psi^+,\psi^-]$. 
\end{Exp}

\begin{Exp}
    Consider the case when $\rho=1$, namely the B twist of $\mathrm{SQED}_1$. In this case, the VOA $\mc V_\rho^B\cong \mc V_{\bg}\otimes \mc V_{bc}$. The category $KL_\rho\simeq KL$ is the Kazhdan-Lusztig category for $\widehat{\mathfrak{gl}(1|1)}$, and the simple objects that satisfy the monodromy condition are $\widehat{A}_{n, e}$ for $n+e/2\in \Z$ and $\widehat{V}_{n,e}$ for $n+\frac{1}{2}\in \Z$ and $e\notin \Z$. We can identify $\mathbb{W}_n^{[0]}$ with $\mc L_B(\widehat{A}_{n,0})$ and $\mathbb{W}_n^{[e]}$ for $e\notin \Z$ with $\mc L_B(\widehat{V}_{n+\frac{1}{2},e})$. These are simple modules of $\mc V_{\bg}\otimes \mc V_{bc}$, as explained in Section \ref{subsec:baby}.
\end{Exp}

\subsection{Mirror symmetry for line operators}\label{subsec:MirrorLines}

In this section, we will prove:
\begin{Thm}\label{ThmCACB}
    There is an equivalence of braided tensor categories:
    \begin{equation}
        \mathcal{C}_{\rho}^A\simeq \mathcal{C}_{\rho^\vee}^B.
    \end{equation}
\end{Thm}

\begin{proof}
 The strategy of the proof is to use the relation between $\wt{\mc V}_\rho^A$ and $\mc{V}_\rho^B$  derived in Section \ref{subsecrelatealt}:
\begin{equation}
    \wt{\mc V}_\rho^A\otimes  \mc{V}_{bc}^{V}\cong \mc V_{\rho^\vee}^B\otimes \mc{V}_{\mathbb{Z}^{2(n-r)}}^{+-},
\end{equation}
and the fact that there is a commutative diagram of simple current extensions:
\be
\begin{tikzcd}
      &  \wt{\mc V}_\rho^A\otimes \mc{V}_{bc}^{V}\cong \mc V_{\rho^\vee}^B\otimes \mc{V}_{\mathbb{Z}^{2(n-r)}}^{+-} & \\
      \mc{V}_{\beta\gamma}^{V}\otimes \mc{V}_{bc}^{V} \arrow[ur] & & \widehat{\grho}\otimes \mc{V}_{\mathbb{Z}^{2(n-r)}}^{+-} \arrow[ul]\\
       & \widehat{\mathfrak{gl}(1|1)}^{\otimes n} \arrow[ul]\arrow[ur]&
    \end{tikzcd}
\ee
which induces a diagram of lifting functors:
\be
\begin{tikzcd}
      &   \wt{V}_\rho^A\otimes \mc V_{bc}^V\mathrm{-Mod}_{\mathrm{loc}}\left(\mathrm{Ind}(KL^{\boxtimes n})\right) & \\
      \cbg^{\boxtimes n} \arrow[ur, "\mc L_A"] & & KL_{\rho^\vee} \arrow[ul, swap, "\mc L_B"]\\
       & KL^{\boxtimes n} \arrow[Rightarrow, shorten >=0.5cm,shorten <=0.5cm, uu] \arrow[ul, "\mc L_{\mathrm{tot}}^{\mathrm{ungauge}}"]\arrow[ur, swap, "\mc L_{\rho^\vee}^{\mathrm{ungauge}}"]&
\end{tikzcd}
\ee
This is a commutative diagram:
The domain of each functor is the subcategory that has trivial monodromy with the extension. Note that both $\mc C_\rho^A$ and $\mc C_{\rho^\vee}^B$ are full subcategories of $\wt{V}_\rho^A\otimes \mc V_{bc}^V\mathrm{-Mod}_{\mathrm{loc}}\left(\mathrm{Ind}(KL^{\boxtimes n})\right)$. The fact that the image of $\mc L_{\mathrm{tot}}^{\mathrm{ungauge}}$ is $\cbg^{\boxtimes n}$ is by \cite{BN22}, as we have recalled, and that the image of $\mc L_{\rho^\vee}^{\mathrm{ungauge}}$ is $KL_{\rho^\vee}$ is the content of Theorem \ref{Thm:KLKLrho}.  Since every extension involved is a simple current extension, by \cite{creutzig2021uprolling}, the order of such extensions do not matter. Consequently, the above diagram is a commutative diagram of lifting functors. Let $KL^{\boxtimes n, [0]}$ be the full subcategory of $KL^{\boxtimes n}$ that has trivial monodromy with $\wt{\mc V}_\rho^A\otimes \mc{V}_{bc}^{V}\cong \mc V_{\rho^\vee}^B\otimes \mc{V}_{\mathbb{Z}^{2(n-r)}}^{+-}$. The composition of the lifting functor:
\be
\mc L_A\circ \mc L_{\mathrm{tot}}^{\mathrm{ungauge}}: KL^{\boxtimes n, [0]}
\longrightarrow \mc C_\rho^A
\ee
is surjective, and gives an equivalence of braided tensor categories $\mc C_\rho^A\simeq KL^{\boxtimes n, [0]}/\Z^{n+r}$, de-equivariantization with respect to $(\C^*)^{n+r}$ determined by the extension $\wt{\mc V}_\rho^A\otimes \mc{V}_{bc}^{V}$. On the other hand, the lifting functor:
\be
\mc L_B\circ \mc L_{\rho^\vee}^{\mathrm{ungauge}}: KL^{\boxtimes n, [0]}
\longrightarrow \mc{C}_{\rho^\vee}^B
\ee
is also surjective and gives an equivalence $\mc{C}_{\rho^\vee}^B\simeq KL^{\boxtimes n, [0]}/\Z^{n+r}$, de-equivariantization with respect to the $(\C^*)^{n+r}$ determined by the same extension. Combining the two, there is an equivalence of braided tensor categories:
\be
\mc C_\rho^A\simeq KL^{\boxtimes n, [0]}/\Z^{n+r}\simeq \mc{C}_{\rho^\vee}^B,
\ee
and the proof is complete. 

\end{proof}

With this, we give the proof of Proposition \ref{prop:intBside}:

\begin{proof}[Proof of Proposition \ref{prop:intBside}]
By the proof of Theorem \ref{ThmCACB}, we only need to understand the category $KL^{\boxtimes n, [0]}$ and its image under $\mc L_{\rho^\vee}^{\mathrm{ungauge}}$. In \cite{BN22}, we have proven that the domain of the lifting functor $\mc L_{\mathrm{tot}}^{\mathrm{ungauge}}$ is the subcategory where the action of $N_0^i$ is semi-simple with integer eigenvalues for all $i$. On the other hand, Proposition \ref{Propcbgsscbgmono} implies that the domain of $\mc L_A$ is the subcategory where the action of $\sum\rho_{ai}J_0^i$ is semi-simple with integer eigenvalues. Up to an integer value, the action of $\sum\rho_{ai}J_0^i$ is identified with $-\sum \rho_{ai}E_0^i$, the domain of $\mc L_A\circ \mc L_{\mathrm{tot}}^{\mathrm{ungauge}}$, namely $KL^{\boxtimes n, [0]}$ is the subcategory of $KL^{\boxtimes n}$ such that $N_0^i$ and $\sum \rho_{ai}E_0^i$ acts semi-simply with integer eigenvalues. The image of this under the lifting functor $\mc L_{\rho^\vee}^{\mathrm{ungauge}}$ is clearly the subcategory such that the action of $N_0^a$ is semi-simple with integer eigenvalues, and the proof is complete. 

\end{proof}

Under the equivalence of Theorem  \ref{ThmCACB}, we can in fact identify the image of simple modules on both sides, which fulfills the physical expectations from Section \ref{sec:bulklines}. 

\begin{Cor}\label{cor:Mirrorlines}
    Under the equivalence of Theorem \ref{ThmCACB}, the simple objects $\mathbb{V}_m^{\upsilon}$ corresponds to $\mathbb{W}_m^\upsilon$. In particular, when $\upsilon=0 \text{ mod }\Z$, we find that the Wilson lines $\mathbb{W}_m^{[0]}$ corresponds to vortex lines $\mathbb{V}_m^{[0]}$. 
    
\end{Cor}

\begin{proof}
    By definition, $\mathbb{W}_m^\upsilon$ is the lift under $\mc L_B$ of the unique simple quotient of $\CV_{\CZ, \upsilon\cdot X+(m+\rho^\vee)\cdot Y}$ for a representative $\upsilon$. This is a module of the free field algebra, and is the lift of 
    \be
    \bigoplus_{\mu\in \Z^n} \CF^{\wt{X},\wt{Y},\wt{Z}}_{\mu\cdot\wt{Z}+\sigma(m+\rho^\vee)\cdot\wt{Y}+\rho^\vee(\upsilon)\cdot \wt{X}}, \ee
    where $m+\rho^\vee$ is the vector $\nu_\alpha+\sum_i \rho^\vee_{i\alpha}$. By \cite[Proposition 4.6]{BN22}, the image of this module under $\mc L_{\mathrm{tot}}^{\mathrm{ungauge}}$ is:
    \be
\bigotimes_i \Sigma_i^{b_i-1}\mc W_{ \rho^\vee(\upsilon)_i},
    \ee
    where $b_i=\sum_\alpha \sigma_{i\alpha}m^\alpha+\sum_{\alpha, j}\sigma_{i\alpha}\rho^\vee_{j\alpha}$. Each tensor factor is simple when $\sum_\alpha \rho_{i\alpha}^\vee\upsilon^\alpha\ne 0 \text{ mod } \Z$, or has a unique simple quotient when $\sum_\alpha \rho_{i\alpha}^\vee\upsilon^\alpha=0\text{ mod } \Z$, in which case the quotient is $\mc V_{-b_i+1}$. Using the splitting of the short exact sequence \eqref{split-SES}, there is an identify:
    \be
\sum_{\alpha,j}\sigma_{i}{}^\alpha\rho^\vee_{j\alpha}+\sum_{a, j}\rho_{i}{}^a\sigma^\vee_{aj}=1,
    \ee
   and so if we define $\rho(\sigma^\vee)$ to be the vector $\sum_{a, j}\rho_{i}{}^a\sigma^\vee_{aj}$, then we have the following fusion rule:
    \be\label{eq:Aquotient}
\mathbb V_{\rho (\sigma^\vee)}\otimes_{\CV_{\bg}^V}  \left(\bigotimes_i \Sigma_i^{b_i-1}\mc W_{ \rho^\vee(\upsilon)_i}\right)\cong  \bigotimes_i \Sigma_i^{\sigma(m)_i}\mc W_{\rho^\vee(\upsilon)_i}.
    \ee
    The simple quotient of the right-hand-side of equation \eqref{eq:Aquotient} is by definition the pre-image of $\mathbb{V}_m^\upsilon$ under the lifting functor $\mc L_A$. Since $\V_{\rho(\sigma^\vee)}$ is a direct-summand of $\wt{\mc V}_\rho^A$, the lifting functor identifies the two sides of equation \eqref{eq:Aquotient}. Therefore, the image under $\mc L_A$ of the simple quotient of the left-hand-side of equation \eqref{eq:Aquotient} is also $\mathbb{V}_m^\upsilon$. This completes the proof.
    
\end{proof}

\begin{Exp}
    Consider the case when $\rho=0$ and $\rho^\vee=1$. In this case, $\mc C_\rho^A$ is $\cbg$, and $\mc C_\rho^B$ is $KL^N/\Z$. This equivalence is the main content of \cite{BN22}. 

    On the other hand, $\mc C_\rho^B$ is the category of modules of $\mc V_{SF}$, and $\mc C_{\rho^\vee}^A$ is the de-equivariantization $\cbg^{[0]}/\Z$. The only simple in $\cbg^{[0]}/\Z$ is the spectral flow of the vacuum, which maps to the vacuum module of $\mc V_{SF}$ under this equivalence. 
    
\end{Exp}

\subsection{Computation of derived endomorphism in the category of line operators}\label{subsec:HomLines}

In the previous sections, we have shown that objects in the braided tensor categories $\mc C_\rho^A$ and $\mc C_{\rho}^B$ correspond to physical line operators, and that these objects behave nicely under the mirror symmetry equivalence. In this section, we will comment on how this category can be used to obtain the space of local operators at junction of lines. As a consequence, we re-derive the result of \cite{CCG}, that Higgs and Coulomb branches can be obtained from vertex operator algebras. This also provides the final piece of justification that the derived categories $D^b \mc C_\rho^A$ and $D^b \mc C_\rho^B$ agree with the bulk categories constructed in Section \ref{sec:bulklines}, as the derived extension groups of Wilson lines in $D^b \mc C_\rho^B$ reproduces the results in Section \ref{sec:Blines}.

The idea of the computation is that although the category $\mc C_\rho^A$ and $\mc C_{\rho}^B$ are defined as modules of a vertex operator algebra, they are in fact equivalent to representations of finite-dimensional Lie algebras, where Hom spaces can be computed using Chevalley-Eilenberg cochain complex. We comment that this is essentially a B-side computation since the Lie algebra that will be used for this computation is the Lie superalgebra $\grho$. It is possible however, to perform the computations solely using the A side category $\mc C_\rho^A$, without passing to the Lie superalgebra, since although the category $\mc C_\rho^A$ does not have projective or injective objects, there is a projective system of modules that in some sense behave like projective objects \cite[Theorem 2.3, Theorem 2.4]{BN22}, which can be used to form resolutions. In the following, we will choose the simpler method of utilizing the Lie superalgebras. 

Let $KL^{\boxtimes n, N, \rho(E)}$ the subcategory of $KL^{\boxtimes n}$ such that $N_{i,0}$ and $\sum \rho_{ia}E^i_0$ acts semi-simply with integer eigenvalues. Then from the proof of Theorem \ref{ThmCACB}, we see that $KL^{\boxtimes n, N, \rho(E)}$ is identified with $KL^{\boxtimes n, [0]}$ and that the lifting functor $\mc L_A$ express $\mc C_\rho^A$ as a de-equivariantization of $KL^{\boxtimes n, N, \rho(E)}$. For each $\upsilon \in (\C/\Z)^{n-r}$, let $\mc{C}_{\rho, \upsilon}^A$ be the subcategory labelled by $\upsilon$, and with an abuse of notation we use $\upsilon$ to be a representative, then the lifting functor restricts to:
\be
\mc L_A: KL^{\boxtimes n, N, \rho(E)}_{\rho^\vee \upsilon}\longrightarrow \mc{C}_{\rho, \upsilon}^A.
\ee
This is how we make contact with finite-dimensional algebras. Let $\mathfrak{gl}(1|1)^{\oplus n}\mathrm{-Mod}^{N,\rho(E)}$ be the category of representations of $\mathfrak{gl}(1|1)^{\oplus n}$ that $N_i$ acts semi-simply with integer eigenvalues and $\sum \rho_{ia}E^i$ also acts semi-simply with integer eigenvalues. For each $\upsilon$, denote by $\overline{\rho^\vee\upsilon}$ the vector such that $(\overline{\rho^\vee\upsilon})_i=\sum \rho^{\vee}_{i\alpha}\upsilon^\alpha$ if $\sum \rho^{\vee}_{i\alpha}\upsilon^\alpha$ is not an integer, otherwise $(\overline{\rho^\vee\upsilon})_i=0$. The following is proved in \cite[Proposition 3.2]{BN22}:

\begin{Lem}
    The induction functor induces an equivalence:
    \be
\mathrm{Ind}: \mathfrak{gl}(1|1)^{\oplus n}\mathrm{-Mod}^{N,\rho(E)}_{\overline{\rho^\vee\upsilon}}\simeq KL^{\boxtimes n, N, \rho(E)}_{\overline{\rho^\vee \upsilon}}.
    \ee
    where the subscript on the left hand side is the subcategory where the action of $E_i$ has eigenvalue given by $\overline{\rho^\vee\upsilon}$.  
\end{Lem}

The composition $\mc{L}_A\circ \mathrm{Ind}$ then is a surjective functor from representations of $\mathfrak{gl}(1|1)^{\oplus n}$ to the category $\mc{C}_{\rho, \upsilon}^A$. Moreover, since this functor maps tensor product of modules to fusion product (under certain conditions that is satisfied here), which is proved in \cite[Corollary 3.6]{BN22}, this functor exhibits $\mc{C}_{\rho, \upsilon}^A$ as a de-equivariantization of $\mathfrak{gl}(1|1)^{\oplus n}\mathrm{-Mod}^{N,\rho(E)}_{\overline{\rho^\vee\upsilon}}$ by $(\mathbb{C}^*)^r$. In terms of the finite-dimensional Lie algebra, this $(\mathbb{C}^*)^r$ action is induced by the lattice of  $\mathfrak{gl}(1|1)^{\oplus n}$ given by $\mathbb{C}_{\rho(\lambda)}$ for $\lambda\in \Z^r$, where $\mathbb{C}_{\rho(\lambda)}$ is the sub-module on which the action of $N_i$ is $\sum \rho_{ai}\lambda^a$ and the action of other generators are zero. The functor $\mc{L}_A\circ \mathrm{Ind}$ identifies a module $M$ in $\mathfrak{gl}(1|1)^{\oplus n}\mathrm{-Mod}^{N,\rho(E)}_{\overline{\rho^\vee\upsilon}}$ with $\mathbb{C}_{\rho(\lambda)}\otimes M$. So at least as abelian categories, we can identify $\mc C_{\rho, \upsilon}^A$ as the de-equivariantization of $\mathfrak{gl}(1|1)^{\oplus n}\mathrm{-Mod}^{N,\rho(E)}_{\overline{\rho^\vee\upsilon}}$ by this lattice. 

Let us now define the Lie sub-algebra $\mathfrak{h}$ of $\mathfrak{gl}(1|1)^{\oplus n}$ generated by $\wt{N_\alpha}:=\sum \rho^\vee_{i\alpha}N^i, \psi^{i,\pm}$ and $E^i$, then we have a restriction functor:
\be
\mathrm{Res}: \mathfrak{gl}(1|1)^{\oplus n}\mathrm{-Mod}^{N,\rho(E)}_{\overline{\rho^\vee\upsilon}}\longrightarrow \mathfrak{h}\mathrm{-Mod}^{\wt{N},\rho(E)}_{\overline{\rho^\vee\upsilon}}.
\ee

\begin{Lem}
    There kernel of the restriction functor is the tensor subcategory generated by $\mathbb{C}_{\rho(\lambda)}$. Moreover, the induced functor:
    \be
\mc{C}_{\rho, \upsilon}^A\longrightarrow \mathfrak{h}\mathrm{-Mod}^{\wt{N},\rho(E)}_{\overline{\rho^\vee\upsilon}}
    \ee
    is an equivalence of abelian categories.
\end{Lem}

\begin{proof}
  Let $M$ be a module of $\mathfrak{gl}(1|1)^{\oplus n}$, then it is clear that $\mathrm{Res}(M)\cong \mathrm{Res}(M\otimes \mathbb{C}_{\rho(\lambda)})$ for any $\lambda\in \Z^r$. Thus $\mathrm{Res}$ induces a functor from the de-equivariantization. To show that it is fully-faithful, we compute:
  \be
\mathrm{Hom}(\mathrm{Res}(M),\mathrm{Res}(N))=\mathrm{Hom}_{\mathfrak h}(\mathbb{C}, \mathrm{Res}(M^*\otimes N)).
  \ee
The restriction functor forgets the semi-simple action of $\rho(N)$, and so the right adjoint of the restriction functor is:
\be
\mathrm{Hom}_{\mathfrak h}(\mathbb{C}, \mathrm{Res}(M^*\otimes N))\cong \bigoplus\limits_{\lambda\in \Z^r}\mathrm{Hom}_{\mathfrak{gl}(1|1)^{\oplus n}}(\mathbb{C}_{\rho(\lambda)}, M^*\otimes N)\cong \bigoplus\limits_{\lambda\in \Z^r}\mathrm{Hom}_{\mathfrak{gl}(1|1)^{\oplus n}}(\mathbb{C}_{\rho(\lambda)}\otimes M, N).
\ee
By definition, the right hand side is the Hom between $M$ and $N$ in the de-equivariantization. 

We now show that it is essentially surjective. Let $M$ be an $\mathfrak{h}$ module that is say generated by a single element $m$, we obtain a module for $\mathfrak{gl}(1|1)^{\oplus n}$ by:
\be
\wt{M}=\bigoplus M\otimes \mathbb{C}_{\rho(\lambda)}
\ee
where the action of $\psi^{\pm, i}$ maps $M$ to $M\otimes \mathbb{C}_{\rho(\lambda+\rho^\T e_i)}$, where $e_i$ is the standard basis of $\Z^n$, having $1$ at the $i$-th entry and zero everywhere else. Let $N$ be the submodule generated by $m\in M\otimes \C_0$. It is finite-dimensional since $M$ is and the action of $\psi^{i,\pm}$ is nilpotent. We claim that the map $N\to M$ induced by $m\to m$ is in fact an isomorphism of modules of $\mathfrak h$. Indeed, it is well-defined because if $a\cdot m=0$ in $N$ for some $a\in \mathfrak{h}$, then it must be zero in $M$, so the map is a well-defined module map. It is injective for the same reason, and clearly it is surjective because any element in $M$ is of the form $a\cdot m$ for some $a\in U(\mathfrak{h})$. Thus we find that $\mathrm{Res}$ is surjective. This completes the proof. 

\end{proof}

This means that we can compute Homs of objects in $\mathfrak{h}\mathrm{-Mod}^{\wt{N},\rho(E)}_{\overline{\rho^\vee\upsilon}}$ instead. Since $E$ is central, when $\overline{\rho^\vee\upsilon}\ne \overline{\rho^\vee\upsilon'}$, then there is no Hom between objects in $\mathfrak{h}\mathrm{-Mod}^{\wt{N},\rho(E)}_{\overline{\rho^\vee\upsilon}}$ and $\mathfrak{h}\mathrm{-Mod}^{\wt{N},\rho(E)}_{\overline{\rho^\vee\upsilon'}}$. Let us focus on the case when $M, N$ both belong to $\mathfrak{h}\mathrm{-Mod}^{\wt{N},\rho(E)}_{\overline{\rho^\vee\upsilon}}$. In this case, we will use adjunction:
\be
\mathrm{Hom}_{\mathfrak h}(M,N)=\mathrm{Hom}_{\mathfrak h}(\C, M^*\otimes N)
\ee
which is a well-defined adjunction in the category $\mathfrak{h}\mathrm{-Mod}^{\wt{N},\rho(E)}$. Both $\C$ and $M^*\otimes N$ are in the subcategory where the generalized eigenvalue of $E_i$ is $0$. This, combined with the requirement that $\rho(E)$ acts semi-simply, implies that $\rho(E)$ acts trivially. Thus a module in the category $\mathfrak{h}\mathrm{-Mod}^{\wt{N},\rho(E)}_0$ is a module of $\mathfrak{h}/(\rho(E))$, and one clearly has an isomorphism:
\be
\mathfrak{h}/(\rho(E))\cong \mathfrak{g}_*(\rho^\vee)
\ee
as Lie superalgebras. This is the reason that this computation is essentially B side. 

Let $\mathfrak l$ be the Lie subalgebra of $\mathfrak h$ generated by $\psi^{i,\pm}$ and $E_i$, and by $\overline{\mathfrak l}$ the quotient by $\rho(E)$, which is a subalgebra of $\mathfrak{g}_*(\rho^\vee)$. If we treat the action of $\wt{N_a}$ as inducing a $(\mathbb{C}^*)^r$ action, then there is an equivalence:
\be
\mathfrak{h}\mathrm{-Mod}^{\wt{N},\rho(E)}_0\simeq \overline{\mathfrak l}\mathrm{-Mod}_0^{(\mathbb{C}^*)^r}.
\ee
We will use the right hand side to do the computation. Denote by $\mathrm{CE}(\overline{\mathfrak l})$ the Chevalley-Eilenberg chain complex, which is the vector space:
\be
\mathrm{CE}(\overline{\mathfrak l})=U(\overline{\mathfrak l})\otimes \mathrm{Sym}\left(\overline{\mathfrak l}[1]\right)
\ee
with a differential given by the Chevalley-Eilenberg differential:
\be
\begin{aligned}
\mathrm{d} (u\otimes \mathrm{Sym}(x_1[1], \cdots, x_p[1]))&=\sum_i (-1)^{\pm} ux_i\otimes  \mathrm{Sym}(x_1[1], \cdots , \widehat{x_i[1]}, \cdots , x_p[1])\\ &+\sum_{i<j}(-1)^{\pm}u\otimes \mathrm{Sym}([x_i, x_j][1], x_1[1], \cdots , \widehat{x_i[1]},\cdots, \widehat{x_j[1]},\cdots, x_p[1]). 
\end{aligned}
\ee
This complex is a complex of free $\overline{\mathfrak l}$ modules and has an augmentation map to the trivial module $\C$. It is a standard statement in Lie algebra theory that the augmentation map $\mathrm{CE}(\overline{\mathfrak l})\to \C$ is a free resolution of $\C$. We would like to argue that we can use this resolution to compute the above derived Hom. The problem is that this resolution is not an object in $\overline{\mathfrak l}\mathrm{-Mod}_0^{(\mathbb{C}^*)^r}$, we will resolve this by showing that some finite cut-off of this is. 

\begin{Lem}
    Let $P^*$ be a $(\mathbb{C}^*)^r$-equivariant free $U(\overline{\mathfrak l})$ resolution of $\C$. Let $V_*$ be any other $(\mathbb{C}^*)^r$-equivariant resolution, then there exists a $(\mathbb{C}^*)^r$-equivariant map from $P^*$ to $V_*$ inducing a quasi-isomorphism. 
\end{Lem}

This is a standard textbook result, and can be found in, for example, \cite{lang2012algebra}. Thus for any bounded resolution $V_*$ of $\C$ in $\overline{\mathfrak l}\mathrm{-Mod}_0^{(\mathbb{C}^*)^r}$, there is a quasi-isomorphism $\mathrm{CE}(\overline{\mathfrak l})\to V_*$. This map of course factors through a finite cut-off in homological degree:
\be
\mathrm{CE}(\overline{\mathfrak l})_{\geq -k}:=\mathrm{Ker}(\mathrm{CE}_{-k})\to \mathrm{CE}_{-k}\to \cdots \to \mathrm{CE}_0,
\ee
and moreover factors through $\mathrm{CE}(\overline{\mathfrak l})_{\geq -k}/(E_\alpha^N)$ for some $N$ large enough, since $E_\alpha$ acts nilpotently on $V_*$. Now we have a finite complex of finite-dimensional modules, and the projective system of resolutions $\mathrm{CE}(\overline{\mathfrak l})_{\geq -k}/(E_\alpha^N)$ is final in the category of all resolutions of $\C$. By the definition of Yoneda extension group, we have:
\be
\mathrm{RHom}(\C, M^*\otimes N)\cong \varinjlim_{N, k}\mathrm{Hom}\left(\mathrm{CE}(\overline{\mathfrak l})_{\geq -k}/(E_\alpha^N), M^*\otimes N\right)=\left(M^*\otimes N\otimes \mathrm{Sym}(\overline{\mathfrak l}^*[-1])\right)^{(\mathbb{C}^*)^r}. 
\ee
Here the right hand side $M^*\otimes N\otimes \mathrm{Sym}(\overline{\mathfrak l}^*[-1])$ is the Chevalley-Eilenberg cochain complex of $M^*\otimes N$ as a module of $\overline{\mathfrak l}$, and the supscript $(\mathbb{C}^*)^r$ means taking the invariants. We summarize this into the following:

\begin{Prop}
    Let $\mc M$ and $\mc N$ be objects in $\mc C_{\rho,\upsilon}^A$ and $M$ and $N$ be objects in $\overline{\mathfrak l}\mathrm{-Mod}_{\overline{\rho^\vee\upsilon}}^{(\mathbb{C}^*)^r}$ whose image under $\mc{L}_A\circ \mathrm{Ind}$ are $\mc M$ and $\mc N$. Then there is a quasi-isomorphism:
    \be
\mathrm{RHom}_{D^b\mc C_{\rho}^A}\left(\mc M, \mc N\right)\cong \left(M^*\otimes N\otimes \mathrm{Sym}(\overline{\mathfrak l}^*[-1])\right)^{(\mathbb{C}^*)^r}.
    \ee
\end{Prop}

As a particular example, consider $\mc M=\mathbb{V}_{m}$ and $\mc N=\mathbb{V}_{m'}$ for $m,m'\in \Z^{n-r}$, namely the simple vortex lines. The objects corresponding to them are $\C_{m}$ and $\C_{m'}$, the one-dimensional module of $\overline{\mathfrak l}$ that is trivial but has $(\mathbb{C}^*)^r$-equivariant structure given by $m$ and $m'$. The above implies: 
\be
\mathrm{RHom}_{\mc C_{\rho}^A}\left(\V_m, \V_{m'}\right)\cong \mathrm{Sym}\left(\overline{\mathfrak l}^*[-1])\right)^{(\mathbb{C}^*)^r, m'-m}
\ee
namely the $m'-m$ semi-invariant subspace of $\mathrm{Sym}(\overline{\mathfrak l}^*[-1])$. 

\begin{Lem}
    The algebra $\mathrm{Sym}(\overline{\mathfrak l}^*[-1])$ is quasi-isomorphic to the algebra $\C[X_i,Y_i]/(\mu_\alpha)$ as in Section \ref{sec:Blines}.
\end{Lem}

\begin{proof}
    The dual of $\psi^{\pm, i}$ gives variables $X_i$ and $Y_i$ and the dual of $\mathfrak{g}^*$ gives the fermionic variable $\lambda_\alpha$. The Chevalley-Eilenberg differential can be computed explicitly as:
    \be
d(\lambda_\alpha)=\sum \rho_{i\alpha}^\vee X^i Y^i,
    \ee
    and this is precisely quasi-isomorphic to $\C[X_i,Y_i]/(\mu_\alpha)$. 
\end{proof}

Consequently, we obtain:

\begin{Prop}
    There is a quasi-isomorphism: 
    \be
\mathrm{RHom}_{\mc C_{\rho}^A}\left(\V_m, \V_{m'}\right)\cong\left(\C[X_i,Y_i]/(\mu_\alpha)\right)_{m'-m}
    \ee
    where the right hand side is the subspace of semi-invariants. 
    
\end{Prop}

This perfectly matches the physical predictions from equation \eqref{Hom-W-phys}, if we use the equivalence between $\mc C_{\rho}^A$ with $\mc C_{\rho^\vee}^B$, and the identification between $\V_m$ with $\W_m$.

\appendix

\section{The conformal element and the anomaly}
\label{app:conformal}

In this appendix, we compute the conformal element for $\widehat{\mathfrak{g}_\rho}$. In particular, we will show that the conformal element is the quadratic  Casimir associated to a shifted quadratic form of equation \eqref{eqbilinearBper}. Moreover, we show that the shift is equivalent to the anomaly. This means that the anomaly is not only present for the physical theory, it is also relevant mathematically for the modified Sugawara construction of the conformal element. 

Let $\mathfrak{g}$ be the Lie algebra of a reductive Lie group $G$, and $V$ a complex representation of $\mathfrak g$. The purterbative VOA on the boundary of a B twist theory with gauge group $G$ and hypermultiplets valued in $V$ is an affine Lie superalgebra associated to the following Lie superalgebra: 
\begin{equation}\label{eq:LiesupBtwist}
 \mathfrak{g}_*=\mathfrak{g}[0]\oplus V[-1]\oplus V^*[-1]\oplus \mathfrak{g}^*[-2] \ni (N_a,\psi^{i,+},\psi^{i,-},E^a)
\end{equation}
together with a quadratic Casimir determined by physics: 
\be\label{eq:kappaBtwist}
\kappa (N_a, N_b)=(T_V-T_{\mathfrak g})\kappa_\mathfrak{g}(N_a, N_b),\qquad \kappa(N_a, E^b)=\delta^b{}_a,\qquad \kappa(\psi^{i,+},\psi^{j, -})=\delta^{ij}.
\ee
Here $\kappa_{\mathfrak g}$ is the standard killing form on $\mathfrak g$, extended trivially to $\mathfrak{g}_*$. We prove:

\begin{Prop}
The vertex algebra $V_{\kappa}(\mathfrak{g}_*)$ is a vertex operator algebra whose conformal element is given by the quadratic Casimir associated to $\kappa-(T_V-T_{\mathfrak g})\kappa_{\mathfrak g}$. 

\end{Prop}

Let us start the computation in a more general setting. Assume that we are given a Lie superalgebra $\mathfrak{g}_*=\mathfrak{g}_0\oplus \mathfrak{g}_1$ together with an adjoint invariant non-degenerate even bilinear form $\kappa$. Let $x_i$ be a set of basis and $x^i$ be a set of dual basis under $\kappa$. The corresponding quadratic Casimir for $\kappa$ is given by
\begin{equation}
    \Omega_{\kappa}=\sum_i (-1)^{|x_i|}x_ix^i,
\end{equation} 
where $|x_i|$ is the parity of $x_i$. Note that this only depends on the choice of $\kappa$ but not the basis, thus the notation. 

Let $V_k (\mathfrak g_*)$ be the affine Lie superalgebra associated to $\mathfrak g_*$ at level $k$. This vertex algebra is generated by $x_i(z)$ with OPE:
\begin{equation}
    x_i(z)x_j(w)\sim \frac{ k\kappa(x_i,x_j)}{(z-w)^2}+\frac{[x_i,x_j](w)}{z-w}.
\end{equation}
The vertex operator corresponding to the quadratic Casimir is given by:
\begin{equation}
\Omega_\kappa(z)=\sum_i(-1)^{|x_i|}\normord{x_i(z)x^i(z)}.
\end{equation}
Following the computation of \cite{FBZ}, the OPE between the quadratic Casimir element with $x_i$ is given by:
\begin{equation}\label{eqOPECas}
    \Omega_\kappa(z)x_i(w)\sim \frac{2kx_i(w)+[\Omega_\kappa, x_i](w)}{(z-w)^2}+\frac{2k \partial x_i(w)}{z-w}+\frac{\partial [
    \Omega_\kappa, x_i](w)}{z-w}.
\end{equation}
Here $[\Omega_\kappa,x_i]$ denotes the action of the quadratic Casimir on the Lie algebra as its own module:
\begin{equation}
[\Omega_\kappa,x_i]=\sum_j(-1)^{|x_j|}[x_j,[x^j,x_i]].
\end{equation}
When $\mathfrak g$ is simple, and $\kappa$ is chosen to be the standard killing form, the action of $\Omega_\kappa$ is given by multiplication by $2h^\vee$, where $h^\vee$ is the dual coxeter number of $\mathfrak{g}$. In this case the OPE above becomes:
\begin{equation}
    \Omega_\kappa(z)x_i(w)\sim \frac{2(k+h^\vee)x_i(w)}{(z-w)^2}+\frac{2(k+h^\vee) \partial x_i(w)}{z-w},
\end{equation}
and so $\frac{1}{2(k+h^\vee)}\Omega_\kappa (z)$ is a conformal element. This is the usual Segal-Sugawara construction. 

We now specify to the case of $\mathfrak{g}_*$ having the form of equation \eqref{eq:LiesupBtwist} and $\kappa$ as in equation \eqref{eq:kappaBtwist}. Under this bilinear form, the dual basis of $\{N_a, \psi^{i,+}, \psi^{i,-}, E^a\}$ is given by:
\begin{equation}
\{E^a, \psi^{i,-}, -\psi^{i,+}, N_a-(T_V-T_\mathfrak{g})\sum_b \kappa_{\mathfrak{g}}(N_a, N_b)E^b\}
\end{equation}
The quadratic Casimir is thus:
\begin{equation}
    \Omega_\kappa=\sum_a N_a E^a+E^aN_a-\sum_i\psi^{i,+}\psi_i^{-}+\sum_i\psi^{-}_i\psi^{i,+}-(T_V-T_{\mathfrak{g}})\sum_{a,b}E^aE^b\kappa_{\mathfrak{g}}(N_a,N_b).
\end{equation}
For convenience, let $N^a$ be the dual basis of $N_a$ under $\kappa_{\mathfrak{g}}$ and $E_a$ the dual of $N^a$ in $\mathfrak{g}^*$. Since:
\begin{equation}
    \sum_{a,b} E^aE^b\kappa_{\mathfrak{g}}(N_a,N_b)=\sum_a E_aE^a,
\end{equation}
we can simplify the quadratic Casimir:
\begin{equation}
    \Omega_\kappa=\sum_a N_a E^a+E^aN_a-\sum_i\psi^{i,+}\psi^{-}_i+\sum_i \psi^{-}_i\psi^{i,+}-(T_V-T_{\mathfrak{g}})\sum_a E_aE^a.
\end{equation}
By equation \eqref{eqOPECas}, whether this can be used as a conformal element depends on how $\Omega_\kappa$ acts on $\mathfrak{g}_*$, which is what we compute next:

\begin{Prop}
The action of $\Omega_\kappa$ on $\mathfrak{g}_*$ is nilpotent, with the only nontrivial action given by:
\begin{equation}
    \Omega_\kappa(N_a)=2(T_\mathfrak{g}-T_V)E_a.
\end{equation}
\end{Prop}

\begin{proof}
By considering degrees, $\Omega_\kappa$ acts nontrivially only on $N_a$, which we will compute. Let us start with the term:
\begin{equation}
    \sum_b [N_b,[E^b, N_a]].
\end{equation}
Denote by $f_{ab}^c$ the structural constants of $\mathfrak{g}$ associated with the basis $N_a$. Namely, $[N_a,N_b]=\sum f_{ab}^c N_c$. Since
\begin{equation}
\kappa([E^b,N_a],N_c)=f_{ac}^b=\kappa([N^b,E_a], N_c),
\end{equation}
and because $\kappa$ is non-degenerate, we find:
\begin{equation}
    \sum_b [N_b,[E^b, N_a]]=\sum_b [N_b,[N^b, E_a]]=2h^\vee E_a.
\end{equation}
Similarly 
\begin{equation}
    \sum_b [E^b,[N_b,N_a]]=2h^\vee E_a.
\end{equation}
Consider now:
\begin{equation}
    \sum_i\{\psi^{i,+},[\psi^{-}_i, N_a]\}.
\end{equation}
Computing the pairing:
\begin{equation}
    \kappa(\sum_i\{\psi^{i,+},[\psi^{-}_i, N_a]\},N^b)=\sum_i\kappa([ N_a,\psi^{-}_i],[ N^b,\psi^{i,+}])=\text{Tr}_V(N_aN^b).
\end{equation}
By our choice of normalization, this is zero unless $a=b$ in which case it is $T_V$, and thus:
\begin{equation}
    \sum_i\{\psi^{i,+},[\psi^{-}_i, N_a]\}=T_VE_a.
\end{equation}
Similarly, the $\psi_i^-\psi^{i,+}$ also gives the same contribution. Finally, the action of $E_aE^a$ is trivial by degree considerations, and we find:
\begin{equation}
    \Omega_\kappa(N_a)=2(T_{\mathfrak g}-T_V)E_a.
\end{equation}
as desired. 

\end{proof}

This proof shows that the quadratic Casimir fails in being a conformal element precisely because of the anomaly. To correct this, we need to use the additional fields in the vertex algebra to cancel this nilpotent part of the action. Consider the OPE:
\begin{equation}
    \sum_b \normord{E_bE^b}(z)N_a(w)\sim \frac{2k E_a(w)}{(z-w)^2}+\frac{2k \partial E_a(w)}{z-w}.
\end{equation}
Indeed, from the following computation, we have:
\be
\sum_b N_a(w) E_{b, -1}E^b_{-1}\vert 0\rangle \sim w^{-2}\sum_b N_{a, 1} E_{b, -1}E^b_{-1}\vert 0\rangle+w^{-1}\sum_b N_{a, 0}E_{b, -1}E^b_{-1}\vert 0\rangle=w^{-2} 2k E_{a,-1}\vert 0\rangle,
\ee
from which one finds:
\be
 \sum_b \normord{E_bE^b}(z)N_a(w)\sim \frac{2 k E_a(z)}{(w-z)^2}\sim \frac{2E_a(w)}{(z-w)^2}+\frac{2k\partial E_a(w)}{z-w}.
\ee
Consider the following element:
\begin{equation}
  \widetilde{\Omega}_\kappa=  \frac{1}{2}\Omega_{\kappa}+\frac{T_V-T_\mathfrak{g}}{2}\sum_a E_aE^a,
\end{equation}
whose OPE with $x(w)$ is given by:
\be
\widetilde{\Omega}_\kappa(z)x(w)=\frac{2k x(w)}{(z-w)^2}+\frac{2k \partial x(w)}{z-w}.
\ee
As a consequence, $\frac{1}{2k}\widetilde{\Omega}_\kappa(z)$ is a conformal element. This element can be written as: 
\begin{equation}
\widetilde{\Omega}_\kappa=\frac{1}{2}\left(\sum_a N_aE^a+E^aN_a-\sum_i \psi^{i,+}\psi_i^-+\sum_a \psi_i^-\psi^{i,+}\right),
\end{equation}
from which one can see that it is nothing but the quadratic Casimir associated to the quadratic form $\kappa'=\kappa-(T_V-T_{\mathfrak g})\kappa_{\mathfrak g}$. 

In fact, one can start with any bilinear form $\tilde{\kappa}=\kappa+c\kappa_{\mathfrak g}$ for some $c$. The quadratic Casimir $\Omega_{\tilde{\kappa}}$ will always act nilpotently with the same factor. Moreover, in the affine Lie algebra $V_k(\mathfrak{g}_*)$, the OPE between $N_a$ and $\sum E_b E^b$ will remain the same since the paring between $N_a$ and $E_a$ remain the same. As a consequence, the conformal element, in any choice $\tilde{\kappa}$, is always the quadratic Casimir associated to $\tilde{\kappa}-(T_V-T_{\mathfrak g})\kappa_{\mathfrak g}$, shifting by the anomaly.

\section{Half-index identities}
\label{app:index}

In this appendix we compute the graded characters of VOA's $\CV_\rho^A$ and $ \CV_{\rho^\vee}^B$, recovering known formulas for half-indices from the physics literature. 
Theorem \ref{thm:VOA} relating 3d-mirror VOA's
\be \CV_\rho^A \;\cong\; \CV_{\rho^\vee}^B \label{app-AB}  \ee
immediately implies an equality of graded characters, which generalizes half-index identities from \cite{Okazaki-04} to arbitrary $\rho,\rho^\vee$.

In the physics literature, the characters of boundary vertex algebras are known as half-indices. They are most generally defined for boundary conditions of 3d $\CN=2$ theories that preserve 2d $\CN=(0,2)$ SUSY \cite{Gadde:2013sca,SugishitaTerashima,YoshidaSugiyama,dimofte2018dual}, generalizing bulk 3d index computations of \cite{KW-index,ImamuraYokoyama}. Viewing a 3d $\CN=4$ theory as a 2d $\CN=2$ theory, the half-index can be defined as
\be \mathds{I}^{\CN=2}_\CV(\eta,q^{\frac12},x) \,=\, \text{Tr}_\CV\, (-1)^{2J} \eta^{H-C} q^{\frac12(2J+H)} x^F\,, \label{II-2}  \ee
where $H,C$ are the charges of the $U(1)_H\times U(1)_C \subset SU(2)_H\times SU(2)_C$ R-symmetry, $J \in \frac12\Z$ is boundary spin (with $2J$ mod 2 being fermion number), and $F$ is a vector of charges for all flavor symmetries that may be present.
Formula \eqref{II-2} can be further specialized to the character of a boundary VOA corresponding to either the A or B topological twist of the 3d $\CN=4$ bulk (so long as the boundary condition is compatible with either the A or B twist). For the A twist, the character is
\be \label{A-index} \mathds{I}^{A}_\CV(q^{\frac12},x) := \text{Tr}_\CV\, (-1)^C q^{L_0} x^F \,=\, \mathds{I}^{\CN=2}_\CV(-1,-q^{\frac12},x)\,, \ee
noting that $L_0 = J+H/2$ in the A twist.
For the B twist, the character is
\be \label{B-index} \mathds{I}^{B}_\CV(q^{\frac12},x) := \text{Tr}_\CV\, (-1)^H q^{L_0} x^F \,=\, \mathds{I}^{\CN=2}_\CV(-q^{\frac12},-q^{\frac12},x)\,, \ee
noting that $L_0 = J+C/2$ in the B twist.

\subsection{A twist}

In the A twist, we have $\CV_\rho^A = H_{\rm BRST}(\mathfrak{gl}(1)^r,\mc{V}_{\beta\gamma}^{T^*V}\otimes \mc{V}_{bc}^V) = H^\bullet\big((\mc{V}_{bc}^{\mathfrak g} \otimes \mc{V}_{\beta\gamma}^{T^*V}\otimes \mc{V}_{bc}^V)^{\rm rel}, Q_{\rm BRST}\big)$. The differential $Q_{\rm BRST}$ has homological degree $+1$ and is otherwise uncharged; taking cohomology does not change the character, so
\be \mathds{I}_{\CV_\rho^A} = \mathds{I}_{(\mc{V}_{bc}^{\mathfrak g} \otimes \mc{V}_{\beta\gamma}^{T^*V}\otimes \mc{V}_{bc}^V)^{\rm rel}} \ee
Recall from Section \ref{sec:BRST} that the relative complex $(\mc{V}_{bc}^{\mathfrak g} \otimes \mc{V}_{\beta\gamma}^{T^*V}\otimes \mc{V}_{bc}^V)^{\rm rel}$ is defined by starting with $\mc{V}_{bc}^{\mathfrak g} \otimes \mc{V}_{\beta\gamma}^{T^*V}\otimes \mc{V}_{bc}^V$, removing the zero-mode of the $\mathfrak g$-values $c$ fields, and restricting by hand to invariants under the finite $\mathfrak{gl}(1)^r$ symmetry. We can replicate these operations at the level of characters.

Let $\mc{V}_{bc}^{\mathfrak g}{}'$ denote $\mc{V}_{bc}^{\mathfrak g}$ (as a vector space) with zero-modes of $c$ removed. 
The space $\mc{V}_{bc}^{\mathfrak g}{}' \otimes \mc{V}_{\beta\gamma}^{T^*V}\otimes \mc{V}_{bc}^V$ has finite symmetry $U(1)_H\times G\times T_F\times T_B\times T_{\rm top}$, as described in Section \eqref{sec:VA-sym}. We compute a graded character
\be \mathds{I}_{\mc{V}_{bc}^{\mathfrak g}{}' \otimes \mc{V}_{\beta\gamma}^{T^*V}\otimes \mc{V}_{bc}^V}(q^{\frac12},s,f,b,t) = \text{Tr}_{\mc{V}_{bc}^{\mathfrak g}{}' \otimes \mc{V}_{\beta\gamma}^{T^*V}\otimes \mc{V}_{bc}^V}\, (-1)^H q^{L_0} s^G f^F b^B t^T\,, \ee
where $s=(s_1,...,s_r)$ are fugacities for the $U(1)^r$ gauge action (and $G=(G_1,...,G_r)$) are the gauge charges; $f=(f_1,...,f_{n-r})$ are $T_F$ flavor fugacities; $b=(b_1,...,b_{n-r})$ are $T_B$ boundary flavor fugacities, and $t=(t_1,...,t_r)$ are topological $T_{\rm top}$ flavor fugacities.
The character factorizes $\text{Tr}_{\mc{V}_{bc}^{\mathfrak g}{}' \otimes \mc{V}_{\beta\gamma}^{T^*V}\otimes \mc{V}_{bc}^V}=\text{Tr}_{\mc{V}_{bc}^{\mathfrak g}{}'}\text{Tr}_{\mc{V}_{\beta\gamma}^{T^*V}}\text{Tr}_{\mc{V}_{bc}^V}$ with
\be \text{Tr}_{\mc{V}_{bc}^{\mathfrak g}{}'} = (q;q)_\infty^{2r}\,,\qquad
 \text{Tr}_{\mc{V}_{\beta\gamma}^{T^*V}} = \prod_{i=1}^n \big(q^{\frac12}s^{\rho_i}f^{\sigma_i}, q^{\frac12}s^{-\rho_i}f^{-\sigma_i};q\big)_\infty^{-1}\,, \ee 
\be\notag \text{Tr}_{\mc{V}_{bc}^V} = \prod_{i=1}^n \big(-q^{\frac12}s^{\rho_i}f^{\sigma_i}b^{\rho^\vee_i}t^{\sigma^\vee_i}, -q^{\frac12}s^{-\rho_i}f^{-\sigma_i}b^{-\rho^\vee_i}t^{-\sigma^\vee_i};q\big)_\infty\,,
 \ee
with Pochhammer symbols as defined in \eqref{Poch}.
Taking invariants under the gauge action is easily implemented by expanding and taking the constant term in all the $s$ variables, or equivalently by taking a contour integral $\prod_a \oint \frac{ds_a}{2\pi i}$. Therefore, the character of the A-twisted VOA is 
\be \label{VA-index}  \mathds{I}_{\CV_\rho^A}(q^{\frac12},f,b,t) =  (q;q)_\infty^{2r} \prod_{a=1}^r \oint \frac{ds_a}{2\pi i} \prod_{i=1}^n \frac{ 
 \big(-q^{\frac12}s^{\rho_i}f^{\sigma_i}b^{\rho^\vee_i}t^{\sigma^\vee_i}, -q^{\frac12}s^{-\rho_i}f^{-\sigma_i}b^{-\rho^\vee_i}t^{-\sigma^\vee_i};q\big)_\infty }
 {     \big(q^{\frac12}s^{\rho_i}f^{\sigma_i}, q^{\frac12}s^{-\rho_i}f^{-\sigma_i};q\big)_\infty } \,.
\ee

\subsection{B twist}

In the B twist, we will compute
\be \mathds{I}_{\mc{V}_\rho^B}(q^{\frac12},f,g,t) = \text{Tr}_{\mc{V}_\rho^B} (-1)^C q^{L_0} f^F g^G t^T\,,\ee
where $f=(f_1,...,f_{n-r})$, $g=(g_1,...,g_r)$, and $t=(t_1,...,t_r)$ are fugacities for the $T_F$, $G$, and $T_{\rm top}$ flavor symmetries described in Section \ref{sec:sym-B}\,. Recall from Section \ref{sec:mon2} that the B-twisted VOA is a sum of monopole sectors
\be \mc{V}_\rho^B = \bigoplus_{\mu\in \Z^r} \CF_\mu^N\,, \ee
where each $\CF_\mu^N$ is a module for the perturbative VOA $\CF_0^N = \widehat\grho$, generated by a state $|\mu\rangle$ corresponding to the monopole vertex operator $\CU_\mu$ from \eqref{mon-vertex}.

The perturbative VOA has a PBW basis given by words in the non-positive modes of $N(z),E(z),\psi^+(z),\psi^-(z)$, written in this same order. Using the charges in \eqref{VOA-B-sym}, the character of the perturbative VOA is
\be \mathds{I}_{\widehat\grho} = \frac{1}{(q;q)_\infty^{2r}} \prod_{i=1}^N \big(-qg^{\rho_i}f^{\sigma_i}, -qg^{-\rho_i};q\big)_\infty\,. \ee
Each module $\CF_\mu^N$ is generated by the state $|\mu\rangle$, which has charges as in \eqref{mon-charges2}, agreeing with the expected charges of monopole operators  \eqref{mon-charges}. The vector $|\mu\rangle$ thus contributes $q^{\frac12|\rho\mu|^2} g^{\rho^T\rho\mu}t^\mu f^{\sigma^T\rho\mu}$ to the character. It can be seen from the OPE's \eqref{mon-OPE} that the remainder of the module has a PBW basis given by words in the non-positive modes of $N(z)$ and $E(Z)$ (since the zero-modes act as constants, the positive modes act as zero, and the negative modes act freely); as well as modes $\psi_k^{i,\pm}$ satisfying
\be \psi_k^{i,+}:\quad k < -(\rho\mu)^i\,,\qquad \psi_k^{i,-}:\quad k < (\rho\mu)^i\,. \ee
These modes of $\psi^\pm$ act freely on $|\mu\rangle$, while the remaining ones act as constants. Thus, $\CF_\mu^N$ is essentially a Verma module, modulo the shift in moding of the $\psi^\pm$'s. Accounting for the shift, its character is
\be \mathds{I}_{\CF_\mu^N} = q^{\frac12|\rho\mu|^2} g^{\rho^T\rho\mu}t^\mu f^{\sigma^T\rho\mu} \times  \frac{1}{(q;q)_\infty^{2r}} \prod_{i=1}^N \big(-q^{1+(\rho\mu)_i}g^{\rho_i}f^{\sigma_i}, -q^{1-(\rho\mu)_i}g^{-\rho_i};q\big)_\infty\,. \ee

Putting everything together, we find
\be \mathds{I}_{\mc{V}_\rho^B}(q^{\frac12},f,g,t) = \frac{1}{(q;q)_\infty^{2r}} \sum_{\mu\in \Z^r} q^{\frac12|\rho\mu|^2} g^{\rho^T\rho\mu}t^\mu f^{\sigma^T\rho\mu}  \prod_{i=1}^N \big(-q^{1+(\rho\mu)_i}g^{\rho_i}f^{\sigma_i}, -q^{1-(\rho\mu)_i}g^{-\rho_i};q\big)_\infty\,. \label{VB-index} \ee
A corollary of Theorem \ref{thm:VOA} is
\begin{Cor} For any $\rho$ (and corresponding $\rho^\vee$, as well as choices of splitting $\sigma$ and co-splitting $\sigma^\vee$) we have
\be \mathds{I}_{\CV_\rho^A}(q^{\frac12},f,b,t) \,=\, \mathds{I}_{\CV_{\rho^\vee}^B}(q^{\frac12},t,b,f) \ee
\end{Cor}

\begin{Exp}
Consider SQED with two hypermultiplets of charge 1, with $r=1$, $n=2$, and charges
\be \rho = \begin{pmatrix} 1\\1 \end{pmatrix}\,,\quad \rho^\vee = \begin{pmatrix} 1 \\ -1\end{pmatrix}\,,\qquad \sigma = \begin{pmatrix} 1 \\ 0 \end{pmatrix}\,,\quad \sigma^\vee = \begin{pmatrix} 0 \\ 1 \end{pmatrix}\,. \ee
We find characters
\be\mathds{I}_{\CV_\rho^A}(q^{\frac12},f,b,t) = \frac{1}{(q;q)_\infty^2} \oint \frac{ds}{2\pi i s} \frac{\big(-q^{\frac12} s f b, -q^{\frac12}(sfb)^{-1}, -q^{\frac12}s b^{-1}t,-q^{\frac12}(s b^{-1}t)^{-1};q\big)_\infty}
{ \big(q^{\frac12} s f, q^{\frac12}(sf)^{-1}, q^{\frac12}s,q^{\frac12}(s)^{-1};q\big)_\infty}
\ee
\be \mathds{I}_{\CV_{\rho^\vee}^B}(q^{\frac12},f,g,t) = \frac{1}{(q;q)_\infty^2} \sum_{\mu\in \Z} q^{\mu^2}g^{2\mu}t^\mu f^{-\mu} \big(-q^{1+\mu}g,-q^{1-\mu}g^{-1},-q^{1-\mu}g^{-1}f,-q^{1+\mu}gf^{-1};q\big)_{\infty} \ee
and can verify order-by-order in $q$ that
\begin{align} \mathds{I}_{\CV_\rho^A}(q^{\frac12},f,b,t)  &= \mathds{I}_{\CV_{\rho^\vee}^B}(q^{\frac12},f,g,t) \\
 &= 1+\Big( 2 + \tfrac1b+b+\tfrac1f+\tfrac{1}{bf}+f+bf+\tfrac bf+\tfrac{bf}{t}+\tfrac{b^2f}{t}+\tfrac tb+\tfrac{t}{b^2f}+\tfrac{t}{bf}\Big)q + O(q^2) \notag
 \end{align}
\end{Exp}

\newpage

\bibliographystyle{ytamsalpha}
\bibliography{VOAlines}

\end{document}